\documentclass[manuscript, screen,nonacm]{jair}

\RequirePackage[
  datamodel=acmdatamodel,
  style=acmauthoryear,
  backend=biber,
  giveninits=true,
  uniquename=init
  ]{biblatex}
  
\usepackage{xspace}
\usepackage{booktabs}
\usepackage{multirow}
\usepackage{xcolor}
\usepackage{bbold}
\usepackage{graphicx}
\usepackage{tikz}
\usepackage{enumerate}
\usepackage{thm-restate} 
\usepackage{algorithmic}
\usepackage[ruled, noend,linesnumbered]{algorithm2e}

\SetAlFnt{\small}
\SetAlCapFnt{\small}
\SetAlCapNameFnt{\small}
\SetAlCapHSkip{0pt}
\IncMargin{-\parindent}
\SetKwInput{KwInput}{Input}
\SetKwInput{KwOutput}{Output}

\begin{document}

\title{Semiring Provenance for Lightweight Description Logics}

\author{Camille Bourgaux}
\orcid{0000-0002-8806-6682}
\email{camille.bourgaux@ens.fr}
\affiliation{%
  \institution{DI ENS, ENS, CNRS, PSL University \& Inria}
  \city{Paris}
  \country{France}}

\author{Ana Ozaki}
\orcid{0000-0002-3889-6207}
\email{ana.ozaki@uib.no}
\affiliation{%
  \institution{University of Oslo}
  \city{Oslo}
  \country{Norway}}

\author{Rafael Pe\~naloza}
\orcid{0000-0002-2693-5790}
\email{rafael.penaloza@unimib.it}
\affiliation{%
  \institution{University of Milano-Bicocca}
  \city{Milan}
  \country{Italy}
}

\renewcommand{\shortauthors}{Bourgaux, Ozaki, \& Pe\~naloza}

\newcommand{\homo}{\ensuremath{g}\xspace}

\newcommand{\mn}[1]{\ensuremath{\mathsf{#1}}}
\newcommand{\mi}[1]{\ensuremath{\mathit{#1}}}

\newcommand{\CB}[1]{\textcolor{orange}{CB: #1}}
\newcommand{\AO}[1]{\textcolor{blue!40!green}{AO: #1}}
\newcommand{\RP}[1]{\textcolor{blue!40!red}{RPN: #1}}

\newcommand{\posbool}{\ensuremath{\sf{PosBool}[\semiringVars]}\xspace}
\newcommand{\lin}{\ensuremath{\sf{Lin}[\semiringVars]}\xspace}
\newcommand{\why}{\ensuremath{\sf{Why}[\semiringVars]}\xspace}
\newcommand{\whybis}[1]{\ensuremath{\sf{Why}[#1]}\xspace}
\newcommand{\sorp}{\ensuremath{\sf{Sorp}[\semiringVars]}\xspace}
\newcommand{\trio}{\ensuremath{\sf{Trio}[\semiringVars]}\xspace}
\newcommand{\boolseries}{\ensuremath{\mathbb{B}\llbracket\semiringVars\rrbracket}\xspace}
\newcommand{\trioseries}{\ensuremath{\sf{Trio}\llbracket\semiringVars\rrbracket}\xspace}
\newcommand{\series}{\ensuremath{\mathbb{N}^\infty\llbracket\semiringVars\rrbracket}\xspace}
\newcommand{\polynomials}{\ensuremath{\mathbb{N}[\semiringVars]}\xspace}
\newcommand{\boolpolynomials}{\ensuremath{\mathbb{B}[\semiringVars]}\xspace}

\newcommand{\semiringVars}{\ensuremath{{\sf X}}\xspace}
\newcommand{\polynomsemiring}{(\mathbb{N}[\semiringVars], + ,\times,0,1)}
\newcommand{\powerseriessemiring}{\powerseriessemiringshort= (\Nbb^\infty\llbracket\semiringVars\rrbracket, + ,\times,0,1)}

\newcommand{\Just}{\ensuremath{{\sf Just_\Omc}}\xspace}

\newcommand{\Rew}{\mn{Rew}}
\newcommand{\ext}[1]{\boldsymbol{#1}}

\newcommand{\plusidem}{$\oplus$\mbox{-}idempotent\xspace}
\newcommand{\timesidem}{$\otimes$-idempotent\xspace}

\makeatletter
\newcommand{\bigplus}{%
  \DOTSB\mathop{\mathpalette\mattos@bigplus\relax}\slimits@
}
\newcommand\mattos@bigplus[2]{%
  \vcenter{\hbox{%
    \sbox\z@{$#1\sum$}%
    \resizebox{!}{0.9\dimexpr\ht\z@+\dp\z@}{\raisebox{\depth}{$\m@th#1+$}}%
  }}%
  \vphantom{\sum}%
}
\makeatother

\newcommand{\p}[2]{\ensuremath{{\sf prov}_{#1}({#2})}\xspace}

\newcommand{\provdb}{\Pmc_{DB}}
\newcommand{\provdat}{\Pmc_{Dat}}
\newcommand{\elem}{\kappa}
\newcommand{\elema}{\kappa}
\newcommand{\elemb}{\kappa'}
\newcommand{\elemc}{\chi}
\newcommand{\eleme}{\epsilon}
\newcommand{\semiringset}{\ensuremath{K}}
\newcommand{\zero}{\ensuremath{\mathbb{0}}\xspace}
\newcommand{\one}{\ensuremath{\mathbb{1}}\xspace}
\newcommand{\semiring}{(\semiringset,\oplus,\otimes,\zero,\one)}
\newcommand{\semiringshort}{\mathbb{K}}

\newcommand{\annot}{\mi{label}}
\newcommand{\monomial}{\ensuremath{m}\xspace}
\newcommand{\monomials}{\ensuremath{\mn{mon}}\xspace}

\newcommand{\query}{q}
\newcommand{\ans}{\vec{a}}

\newcommand{\constSet}{\bf{C}}
\newcommand{\varSet}{\bf{V}}
\newcommand{\predSet}{\bf{R}}
\newcommand{\pred}{p}
\newcommand{\program}{\Pmc}
\newcommand{\database}{\Dmc}
\newcommand{\goal}{\mn{goal}}
\newcommand{\inter}{I}
\newcommand{\tuple}{\ensuremath{\mathbf{t}}}
\newcommand{\derivationtree}{\tau}
\newcommand{\leaves}{\mn{leaves}}

\newcommand{\NC}{\ensuremath{{\sf N_C}}\xspace}
\newcommand{\NI}{\ensuremath{{\sf N_I}}\xspace}
\newcommand{\NR}{\ensuremath{{\sf N_R}}\xspace}
\newcommand{\NM}{\ensuremath{{\sf N_M}}\xspace} 
\newcommand{\NV}{\ensuremath{{\sf N_V}}\xspace} 
   
\newcommand{\pair}[2]{\ensuremath{({#1},{#2})}} 
\newcommand{\polynomial}{\ensuremath{p}\xspace} 
\newcommand{\sfm}{\ensuremath{{\sf m}}\xspace}

\newcommand{\KB}{\Kmc}
\newcommand{\ontology}{\Omc}
\newcommand{\ABox}{\Amc}
\newcommand{\TBox}{\Tmc}
\newcommand{\axiom}{\ensuremath{\alpha}\xspace}
\newcommand{\consequence}{\ensuremath{c}\xspace}
\newcommand{\justification}{\Mmc}
\newcommand{\ontologyp}{\Pmc}
\newcommand{\ABoxp}{\Bmc}
\newcommand{\TBoxp}{\Smc}
\newcommand{\axiomp}{\ensuremath{\beta}\xspace}

%
%
\newcommand{\tup}[1]{\langle #1 \rangle}
\newcommand{\ie}{i.e.\ }
\newcommand{\wrt}{w.r.t.\ }
\newcommand{\cf}{cf.\ }
\newcommand{\eg}{e.g.\ }
\newcommand{\modapprox}{\mathrel{\scalebox{1}[1.5]{$\shortmid$}\mkern-3.1mu\raisebox{0.1ex}{$\approx$}}}
%
\newcommand{\ELHIbot}{\ELHI^n_\bot}
\newcommand{\ELHIbotrestr}{\ELHI^{n,-}_\bot}
\newcommand{\logic}{\Lmc}
\newcommand{\DLLITE}{\text{DL-Lite}\xspace}
\newcommand{\DLLiteR}{\text{DL-Lite$_\Rmc$}\xspace}
\newcommand{\ELHlhs}{\ensuremath{\mathcal{ELH}_{lhs}}\xspace}
\newcommand{\ELHrhs}{\ensuremath{\mathcal{ELH}_{rhs}}\xspace}
\newcommand{\ELHr}{\ensuremath{\mathcal{ELH}^r}\xspace}
\newcommand{\ELH}{\ensuremath{\mathcal{ELH}}\xspace}
\newcommand{\ELHI}{\ensuremath{\mathcal{ELHI}}\xspace}
\newcommand{\FLo}{\ensuremath{{\cal F\!LE}}\xspace}
\newcommand{\FLE}{\ensuremath{{\cal F\!LE}}\xspace}
\newcommand{\EL}{\ensuremath{\mathcal{EL}}\xspace}
\newcommand{\ALC}{\ensuremath{\mathcal{ALC}}\xspace}
\newcommand{\ALCH}{\ensuremath{{\cal ALCH}}\xspace}
\newcommand{\ALCD}{\ensuremath{{\cal ALC(D)}}\xspace}
\newcommand{\ALCO}{\ensuremath{{\cal ALCO}}\xspace}
\newcommand{\ALCK}{\ensuremath{{\cal ALCK}}\xspace}
\newcommand{\ALCI}{\ensuremath{{\cal ALCI}}\xspace}
\newcommand{\ALCIO}{\ensuremath{{\cal ALCIO}}\xspace}
\newcommand{\ALCQO}{\ensuremath{{\cal ALCQO}}\xspace}
\newcommand{\ALCQI}{\ensuremath{{\cal ALCQI}}\xspace}
\newcommand{\ALCQIO}{\ensuremath{{\cal ALCQIO}}\xspace}
\newcommand{\ALCOD}{\ensuremath{{\cal ALCO(D)}}\xspace}
\newcommand{\ALCKD}{\ensuremath{{\cal ALCK(D)}}\xspace}
\newcommand{\ALCOKD}{\ensuremath{{\cal ALCOK(D)}}\xspace}
\newcommand{\ALCID}{\ensuremath{{\cal ALC}^-(\Dmc)}\xspace}
\newcommand{\ALCDmin}{\ensuremath{{\cal ALC}f(\Dmc)}\xspace}
\newcommand{\ALCFD}{\ensuremath{{\cal ALCF(D)}}\xspace}
\newcommand{\ALCPD}{\ensuremath{{\cal ALCP(D)}}\xspace}
\newcommand{\ALCRPD}{\ensuremath{{\cal ALC}^{rp}({\cal D})}\xspace}
\newcommand{\ALCRPID}{\ensuremath{{\cal ALC}^{rp,-}({\cal D})}\xspace}
\newcommand{\BIGL}{\ensuremath{{\cal ALCP}^{rp,-,\sqcap}({\cal D})}\xspace}
\newcommand{\ALCF}{\ensuremath{{\cal ALCF}}\xspace}
\newcommand{\ALCFI}{\ensuremath{{\cal ALCFI}}\xspace}
\newcommand{\ALCFIO}{\ensuremath{{\cal ALCFIO}}\xspace}
\newcommand{\ALCNR}{\ensuremath{{\cal ALCNR}}\xspace}
\newcommand{\ALCQ}{\ensuremath{{\cal ALCQ}}\xspace}
\newcommand{\ALCN}{\ensuremath{{\cal ALCN}}\xspace}
\newcommand{\HS}{\ensuremath{{\sf HS}}\xspace}
\newcommand{\shiq}{\ensuremath{{\cal SHIQ}}\xspace}
\newcommand{\shoq}{\ensuremath{{\cal SHOQ}}\xspace}
\newcommand{\TDL}{\ensuremath{{\cal TDL}}\xspace}
\newcommand{\OIL}{\ensuremath{{\sf OIL}}\xspace}
\newcommand{\DAMLOIL}{\ensuremath{{\sf DAML{+}OIL}}\xspace}
\newcommand{\shiqt}{\ensuremath{\mathbbm{Q}\text{-}{\cal SHIQ}}\xspace}
\newcommand{\shiqtva}{\ensuremath{\shiqt^\ast}\xspace}
\newcommand{\shiqtvb}{\ensuremath{\shiqt^{\ast\ast}}\xspace}
\newcommand{\SHIQ}{\ensuremath{{\cal SHIQ}}\xspace}
\newcommand{\SHIQO}{\ensuremath{{\cal SHIQO}}\xspace}
\newcommand{\SHOQ}{\ensuremath{{\cal SHOQ}}\xspace}
\newcommand{\SHOQD}{\ensuremath{{\cal SHOQ(D)}}\xspace}
\newcommand{\SHOQKD}{\ensuremath{{\cal SHOQK(D)}}\xspace}

%
%
\newcommand{\cclass}{\textsc{C}\xspace}
\newcommand{\PTime}{\textsc{PTime}\xspace}
\newcommand{\NP}{\textsc{NP}\xspace}
\newcommand{\coNP}{\textsc{coNP}\xspace}
\newcommand{\NL}{\textsc{NLogSpace}\xspace}
\newcommand{\PSpace}{\textsc{PSpace}\xspace}
\newcommand{\NPSpace}{\textsc{NPSpace}\xspace}
\newcommand{\ExpTime}{\textsc{ExpTime}\xspace}
\newcommand{\TwoExpTime}{\textsc{2-ExpTime}\xspace}
\newcommand{\NExpTime}{\textsc{NExpTime}\xspace}
\newcommand{\TwoNExpTime}{\textsc{2-NExpTime}\xspace}
\newcommand{\exptime}{\textsc{exptime}\xspace}
\newcommand{\nexptime}{\textsc{nexptime}\xspace}
\newcommand{\ExpSpace}{\textsc{ExpSpace}\xspace}
\newcommand{\TwoExpSpace}{\textsc{2-ExpSpace}\xspace}
\newcommand{\expspace}{\textsc{expspace}\xspace}
\newcommand{\NExpSpace}{\textsc{NExpSpace}\xspace}
\newcommand{\TwoNExpSpace}{\textsc{2-NExpSpace}\xspace}

%
%

%
%
\newcommand{\Amc}{\ensuremath{\mathcal{A}}\xspace}
\newcommand{\Bmc}{\ensuremath{\mathcal{B}}\xspace}
\newcommand{\Cmc}{\ensuremath{\mathcal{C}}\xspace}
\newcommand{\Dmc}{\ensuremath{\mathcal{D}}\xspace}
\newcommand{\Emc}{\ensuremath{\mathcal{E}}\xspace}
\newcommand{\Fmc}{\ensuremath{\mathcal{F}}\xspace}
\newcommand{\Gmc}{\ensuremath{\mathcal{G}}\xspace}
\newcommand{\Hmc}{\ensuremath{\mathcal{H}}\xspace}
\newcommand{\Imc}{\ensuremath{\mathcal{I}}\xspace}
\newcommand{\Jmc}{\ensuremath{\mathcal{J}}\xspace}
\newcommand{\Kmc}{\ensuremath{\mathcal{K}}\xspace}
\newcommand{\Lmc}{\ensuremath{\mathcal{L}}\xspace}
\newcommand{\Mmc}{\ensuremath{\mathcal{M}}\xspace}
\newcommand{\Nmc}{\ensuremath{\mathcal{N}}\xspace}
\newcommand{\Omc}{\ensuremath{\mathcal{O}}\xspace}
\newcommand{\Pmc}{\ensuremath{\mathcal{P}}\xspace}
\newcommand{\Qmc}{\ensuremath{\mathcal{Q}}\xspace}
\newcommand{\Rmc}{\ensuremath{\mathcal{R}}\xspace}
\newcommand{\Smc}{\ensuremath{\mathcal{S}}\xspace}
\newcommand{\Tmc}{\ensuremath{\mathcal{T}}\xspace}
\newcommand{\Umc}{\ensuremath{\mathcal{U}}\xspace}
\newcommand{\Vmc}{\ensuremath{\mathcal{V}}\xspace}
\newcommand{\Wmc}{\ensuremath{\mathcal{W}}\xspace}
\newcommand{\Xmc}{\ensuremath{\mathcal{X}}\xspace}
\newcommand{\Ymc}{\ensuremath{\mathcal{Y}}\xspace}
\newcommand{\Zmc}{\ensuremath{\mathcal{Z}}\xspace}

%
%
\newcommand{\Amf}{\ensuremath{\mathfrak{A}}\xspace}
\newcommand{\Bmf}{\ensuremath{\mathfrak{B}}\xspace}
\newcommand{\Cmf}{\ensuremath{\mathfrak{C}}\xspace}
\newcommand{\Dmf}{\ensuremath{\mathfrak{D}}\xspace}
\newcommand{\Emf}{\ensuremath{\mathfrak{E}}\xspace}
\newcommand{\Fmf}{\ensuremath{\mathfrak{F}}\xspace}
\newcommand{\Gmf}{\ensuremath{\mathfrak{G}}\xspace}
\newcommand{\Hmf}{\ensuremath{\mathfrak{H}}\xspace}
\newcommand{\Imf}{\ensuremath{\mathfrak{I}}\xspace}
\newcommand{\Jmf}{\ensuremath{\mathfrak{J}}\xspace}
\newcommand{\Kmf}{\ensuremath{\mathfrak{K}}\xspace}
\newcommand{\Lmf}{\ensuremath{\mathfrak{L}}\xspace}
\newcommand{\Mmf}{\ensuremath{\mathfrak{M}}\xspace}
\newcommand{\Nmf}{\ensuremath{\mathfrak{N}}\xspace}
\newcommand{\Omf}{\ensuremath{\mathfrak{O}}\xspace}
\newcommand{\Pmf}{\ensuremath{\mathfrak{P}}\xspace}
\newcommand{\Qmf}{\ensuremath{\mathfrak{Q}}\xspace}
\newcommand{\Rmf}{\ensuremath{\mathfrak{R}}\xspace}
\newcommand{\Smf}{\ensuremath{\mathfrak{S}}\xspace}
\newcommand{\Tmf}{\ensuremath{\mathfrak{T}}\xspace}
\newcommand{\Umf}{\ensuremath{\mathfrak{U}}\xspace}
\newcommand{\Vmf}{\ensuremath{\mathfrak{V}}\xspace}
\newcommand{\Wmf}{\ensuremath{\mathfrak{W}}\xspace}
\newcommand{\Xmf}{\ensuremath{\mathfrak{X}}\xspace}
\newcommand{\Ymf}{\ensuremath{\mathfrak{Y}}\xspace}
\newcommand{\Zmf}{\ensuremath{\mathfrak{Z}}\xspace}

%
%
\newcommand{\Abf}{\ensuremath{\mathbf{A}}\xspace}
\newcommand{\Bbf}{\ensuremath{\mathbf{B}}\xspace}
\newcommand{\Cbf}{\ensuremath{\mathbf{C}}\xspace}
\newcommand{\Dbf}{\ensuremath{\mathbf{D}}\xspace}
\newcommand{\Ebf}{\ensuremath{\mathbf{E}}\xspace}
\newcommand{\Fbf}{\ensuremath{\mathbf{F}}\xspace}
\newcommand{\Gbf}{\ensuremath{\mathbf{G}}\xspace}
\newcommand{\Hbf}{\ensuremath{\mathbf{H}}\xspace}
\newcommand{\Ibf}{\ensuremath{\mathbf{I}}\xspace}
\newcommand{\Jbf}{\ensuremath{\mathbf{J}}\xspace}
\newcommand{\Kbf}{\ensuremath{\mathbf{K}}\xspace}
\newcommand{\Lbf}{\ensuremath{\mathbf{L}}\xspace}
\newcommand{\Mbf}{\ensuremath{\mathbf{M}}\xspace}
\newcommand{\Nbf}{\ensuremath{\mathbf{N}}\xspace}
\newcommand{\Obf}{\ensuremath{\mathbf{O}}\xspace}
\newcommand{\Pbf}{\ensuremath{\mathbf{P}}\xspace}
\newcommand{\Qbf}{\ensuremath{\mathbf{Q}}\xspace}
\newcommand{\Rbf}{\ensuremath{\mathbf{R}}\xspace}
\newcommand{\Sbf}{\ensuremath{\mathbf{S}}\xspace}
\newcommand{\Tbf}{\ensuremath{\mathbf{T}}\xspace}
\newcommand{\Ubf}{\ensuremath{\mathbf{U}}\xspace}
\newcommand{\Vbf}{\ensuremath{\mathbf{V}}\xspace}
\newcommand{\Wbf}{\ensuremath{\mathbf{W}}\xspace}
\newcommand{\Xbf}{\ensuremath{\mathbf{X}}\xspace}
\newcommand{\Ybf}{\ensuremath{\mathbf{Y}}\xspace}
\newcommand{\Zbf}{\ensuremath{\mathbf{Z}}\xspace}

%
%
\newcommand{\Asf}{\ensuremath{\mathsf{A}}\xspace}
\newcommand{\Bsf}{\ensuremath{\mathsf{B}}\xspace}
\newcommand{\Csf}{\ensuremath{\mathsf{C}}\xspace}
\newcommand{\Dsf}{\ensuremath{\mathsf{D}}\xspace}
\newcommand{\Esf}{\ensuremath{\mathsf{E}}\xspace}
\newcommand{\Fsf}{\ensuremath{\mathsf{F}}\xspace}
\newcommand{\Gsf}{\ensuremath{\mathsf{G}}\xspace}
\newcommand{\Hsf}{\ensuremath{\mathsf{H}}\xspace}
\newcommand{\Isf}{\textnormal{I}\xspace}
\newcommand{\Jsf}{\ensuremath{\mathsf{J}}\xspace}
\newcommand{\Ksf}{\ensuremath{\mathsf{K}}\xspace}
\newcommand{\Lsf}{\ensuremath{\mathsf{L}}\xspace}
\newcommand{\Msf}{\ensuremath{\mathsf{M}}\xspace}
\newcommand{\Nsf}{\ensuremath{\mathsf{N}}\xspace}
\newcommand{\Osf}{\ensuremath{\mathsf{O}}\xspace}
\newcommand{\Psf}{\ensuremath{\mathsf{P}}\xspace}
\newcommand{\Qsf}{\ensuremath{\mathsf{Q}}\xspace}
\newcommand{\Rsf}{\ensuremath{\mathsf{R}}\xspace}
\newcommand{\Ssf}{\ensuremath{\mathsf{S}}\xspace}
\newcommand{\Tsf}{\ensuremath{\mathsf{T}}\xspace}
\newcommand{\Usf}{\ensuremath{\mathsf{U}}\xspace}
\newcommand{\Vsf}{\ensuremath{\mathsf{V}}\xspace}
\newcommand{\Wsf}{\ensuremath{\mathsf{W}}\xspace}
\newcommand{\Xsf}{\ensuremath{\mathsf{X}}\xspace}
\newcommand{\Ysf}{\ensuremath{\mathsf{Y}}\xspace}
\newcommand{\Zsf}{\ensuremath{\mathsf{Z}}\xspace}
%
%
\newcommand{\asf}{\ensuremath{\mathsf{a}}\xspace}
\newcommand{\bsf}{\ensuremath{\mathsf{b}}\xspace}
\newcommand{\csf}{\ensuremath{\mathsf{c}}\xspace}
\newcommand{\dsf}{\ensuremath{\mathsf{d}}\xspace}
\newcommand{\esf}{\ensuremath{\mathsf{e}}\xspace}
\newcommand{\fsf}{\ensuremath{\mathsf{f}}\xspace}
\newcommand{\gsf}{\ensuremath{\mathsf{g}}\xspace}
\newcommand{\hsf}{\ensuremath{\mathsf{h}}\xspace}
\newcommand{\isf}{\ensuremath{\mathsf{i}}\xspace}
\newcommand{\jsf}{\ensuremath{\mathsf{j}}\xspace}
\newcommand{\ksf}{\ensuremath{\mathsf{k}}\xspace}
\newcommand{\lsf}{\ensuremath{\mathsf{l}}\xspace}
\newcommand{\msf}{\ensuremath{\mathsf{m}}\xspace}
\newcommand{\nsf}{\ensuremath{\mathsf{n}}\xspace}
\newcommand{\osf}{\ensuremath{\mathsf{o}}\xspace}
\newcommand{\psf}{\ensuremath{\mathsf{p}}\xspace}
\newcommand{\qsf}{\ensuremath{\mathsf{q}}\xspace}
\newcommand{\rsf}{\ensuremath{\mathsf{r}}\xspace}
\newcommand{\ssf}{\ensuremath{\mathsf{s}}\xspace}
\newcommand{\tsf}{\ensuremath{\mathsf{t}}\xspace}
\newcommand{\usf}{\ensuremath{\mathsf{u}}\xspace}
\newcommand{\vsf}{\ensuremath{\mathsf{v}}\xspace}
\newcommand{\wsf}{\ensuremath{\mathsf{w}}\xspace}
\newcommand{\xsf}{\ensuremath{\mathsf{x}}\xspace}
\newcommand{\ysf}{\ensuremath{\mathsf{y}}\xspace}
\newcommand{\zsf}{\ensuremath{\mathsf{z}}\xspace}

\newcommand{\Nbb}{\ensuremath{\mathbb{N}}\xspace}

\newcommand{\queryrewriting}{\ensuremath{q^\ast}\xspace}
\newcommand{\equivalent}[1]{\ensuremath{[{#1}]}\xspace}
\newcommand{\equivalentclass}[1]{\ensuremath{[[{#1}]]}\xspace}
\newcommand{\auxs}[1]{\ensuremath{\mathsf{{\sf aux}}(#1)}\xspace} 
\newcommand{\roles}[1]{\mathsf{rol}(#1)}
\newcommand{\representative}[1]{\ensuremath{[{#1}]}\xspace}
\newcommand{\nrepresentative}[1]{\ensuremath{[\![{#1}]\!]}\xspace}
\newcommand{\aux}[3]{\ensuremath{d^{{#1}}_{{#2}{#3}}}\xspace}
\newcommand{\individuals}[1]{\mathsf{ind}(#1)}
\newcommand{\signature}[1]{\mathsf{sig}(#1)}
\newcommand{\NMrep}{\ensuremath{{\sf N_{[M]}}}\xspace}
\newcommand{\NMnrep}{\ensuremath{{\sf N_{[\![M]\!]}}}\xspace}
\newcommand{\cupdot}{\mathbin{\mathaccent\cdot\cup}}
\newcommand{\nonomial}{\ensuremath{n}\xspace}
\newcommand{\onomial}{\ensuremath{o}\xspace}
\newcommand{\variable}{\ensuremath{v}\xspace}

\newcommand{\NF}{\mn{NF}}
\newcommand{\CR}{\mn{CR}}
\newcommand{\R}{\mn{R}}

\newcommand{\blue}[1]{#1}
\newcommand{\new}[1]{#1}
\newtheorem{remark}[theorem]{Remark}
\newtheorem{claim}[theorem]{Claim}

\begin{abstract}
We investigate semiring provenance---a successful framework originally defined in the relational database setting---for description logics. In this context, the ontology axioms are annotated with elements of a commutative semiring and these annotations are propagated to the ontology consequences in a way that reflects how they are derived. We define a provenance semantics for a language that encompasses several lightweight description logics and show its relationships with semantics that have been defined for ontologies annotated with a specific kind of annotation (such as fuzzy degrees). We show that under some restrictions on the semiring, the semantics satisfies desirable properties (such as extending the semiring provenance defined for databases). We then focus on the well-known why-provenance, for which we study the complexity of problems related to the provenance of an assertion or a conjunctive query answer. Finally, we consider two more restricted cases which correspond to the so-called positive Boolean provenance and lineage in the database setting. For these cases, we present relationships with well-known notions related to explanations in description logics and complete our complexity analysis. As a side contribution, we provide conditions on an $\ELHI_\bot$ ontology that guarantee tractable reasoning. 
\end{abstract}

\maketitle


\section{Introduction}
Description logics (DLs) are a well-known family of formalisms, \new{typically equivalent to fragments of first-order logic}, in which conceptual knowledge about a particular domain 
and facts about specific individuals are expressed in an ontology, using unary and binary predicates called concepts and 
roles \shortcite{dlhandbook}. 
Important reasoning tasks performed over DL ontologies are axiom entailment, \ie deciding whether a given DL axiom follows from the ontology; and query answering, focussing in particular on database-style conjunctive queries. 
Since scalability \new{of reasoning methods} is crucial when using large ontologies, DLs with favorable computational properties have been investigated. 
In particular, the DL-Lite and \EL families propose many dialects that allow for axiom entailment in polynomial time, and Boolean conjunctive query (BCQ) entailment in \NP \shortcite{DBLP:journals/jar/CalvaneseGLLR07,BBL-IJCAI05,DBLP:conf/dlog/Rosati07}. Many real-world ontologies use languages from these families, which underlie the OWL 2 QL and OWL 2 EL profiles of the Semantic Web standard ontology language \shortcite{owl2-profiles}. 

Several extensions of DLs go beyond simple axiom or query entailment and enrich the results with additional information. One can mention fuzzy DLs to express vagueness \cite{DBLP:conf/aaai/Straccia98,DBLP:conf/sum/BorgwardtP17}, possibilistic DLs to handle uncertainty \cite{DBLP:journals/ijar/Hollunder95,DBLP:journals/ijis/QiJPD11,DBLP:journals/logcom/BenferhatB17}, or the bag semantics, which associates multiplicities to query results \shortcite{DBLP:conf/ijcai/NikolaouKKKGH17,DBLP:journals/ai/NikolaouKKKGH19}. 
Another kind of information that may be expected to accompany reasoning results is an \emph{explanation}. It may indeed be crucial to know how a consequence---e.g.\ an axiom or a query answer---has been derived from the ontology. 
\new{In DLs, this problem} has been studied mostly focusing on explaining axiom entailment, in particular concept subsumption, through \emph{axiom pinpointing}, which consists in finding one or all minimal subsets of the ontology that entail the consequence, called \emph{justifications} \shortcite{DBLP:conf/ijcai/SchlobachC03,DBLP:conf/semweb/KalyanpurPHS07,Pena-AP20,DBLP:conf/sum/OzakiP18}. 
A similar approach was investigated for explaining ontology-mediated query answers, but focussing only on the minimal subsets of facts which together with the conceptual knowledge entail the query 
\shortcite{DBLP:conf/ijcai/CeylanLMV19,DBLP:conf/ecai/CeylanLMV20}. 
Alternative approaches that consider the whole ontology provide more involved proof-based explanations for BCQ entailment 
 \shortcite{BorgidaCR08,CroceL18,DBLP:conf/dlog/AlrabbaaBKK22}. Investigation of proofs for axiom entailment that go beyond axiom pinpoiting is also an active line of research \shortcite{DBLP:conf/lpar/AlrabbaaBBKK20,DBLP:conf/ruleml/AlrabbaaBHKKRW22}.

\new{Within the context of databases, the framework of \emph{semiring provenance}  \cite{Green07-provenance-seminal,GreenT17} generalizes the semantics of queries over databases annotated with different kinds of annotations, such as multiplicities, trust levels, 
costs, clearance levels, etc. 
Indeed, it defines the semantics of positive relational algebra queries over databases annotated by elements of an arbitrary commutative semiring, which is an algebraic structure with two binary \blue{commutative} operators called the \emph{addition} and the \emph{multiplication}. 
Intuitively, joint use of the data corresponds to the semiring multiplication while alternative use of the data corresponds to the semiring addition\blue{, and the commutativity of the two operations ensures that the queries $q\wedge q'$ and $q'\wedge q$ are equivalent, and similarly for $\vee$}. 
The semantics is inductively defined on the structure of the query, in a similar way as the semantics of such queries over non-annotated databases.
}

\new{The name ``provenance'' refers to the original idea of tracing the origin of a query answer. 
Indeed, when applied to databases whose tuples are annotated with identifiers (called variables, or provenance tokens), the semiring provenance framework allows to capture many notions of provenance or explanation that have been considered in the database community for nearly 30 years, such as lineage or why-provenance (see, \eg surveys and discussion papers by~\citet{DBLP:journals/ftdb/CheneyCT09},~\citet{Bun2013},~\citet{DBLP:journals/sigmod/Senellart17}, and~\citet{DBLP:journals/ftdb/Glavic21}). This is done by introducing so-called \emph{provenance semirings} 
\cite{DBLP:journals/mst/Green11,GreenT17}. Provenance semiring elements are expressions built from variables associated with each tuple of the database\blue{, using the semiring addition and multiplication}. Depending on the provenance semiring used, these elements may be, for instance, polynomials with coefficients from $\mathbb{N}$, polynomials with Boolean coefficients, etc. 
Such a provenance expression provides a representation of how tuples 
\blue{can be used (jointly or alternatively) to obtain} 
a query result.} 
\new{An important feature of the semiring provenance framework is that a provenance expression can be used to compute the annotation that would be associated to the query result if the database was annotated by elements of any commutative semiring in which the considered provenance semiring can be homomorphically embedded (for example, if a query result has provenance $x\times y + z$, one can obtain the multiplicity of this result in the database in which the tuple with identifier $x$ is annotated by $2$, the one with identifier $y$ by $3$, and the one with identifier $z$ by $4$, by replacing the variables by these values and evaluating the resulting expression in the semiring of the natural integers, i.e., $2\times 3+4=10$).} 

Semiring provenance has also been studied for Datalog queries, for which it is defined based on the set of all derivation trees for the query \shortcite{Green07-provenance-seminal,DBLP:conf/icdt/DeutchMRT14,DBLP:journals/vldb/DeutchGM18},   
and has drawn interest beyond relational databases, 
notably in the context of the Semantic Web \shortcite{Dividino2009,DBLP:conf/semweb/BunemanK10,DBLP:journals/internet/TheoharisFKC11,DBLP:journals/ws/ZimmermannLPS12,Geerts16-provenance}, but also 
XML \cite{DBLP:conf/pods/FosterGT08}, 
graph databases \cite{RamusatMS18,DBLP:phd/hal/Ramusat22}, and 
expressive logics \shortcite{DBLP:journals/siglog/Tannen17,DBLP:conf/csl/DannertGNT21,DBLP:journals/corr/abs-2412-07986}.

In this work, we investigate semiring provenance for description logic ontologies. 
An important feature that distinguishes DL ontologies from the relational database setting is that not only  facts but also  axioms expressing conceptual knowledge can be annotated and, thus, \new{taken into account in the computation of the provenance}. 
\blue{Note that even though Datalog rules also express conceptual knowledge, they are generally considered as part of the query in the database setting, hence not annotated. It is however easy to emulate annotated rules by adding one annotated fact per rule and including it in the rule premises, but this is not permitted by the DL syntax.}  
We also consider not only query answering---as 
done for Datalog queries, which can be seen as \blue{pairs of a rule-based ontology and an atomic query}
---but also axiom entailment, such as concept subsumption, which is 
a classical reasoning task for DL. 
\new{There have already been several proposals to use some kinds of semiring provenance with description logic ontologies \shortcite{attributedDL,provenance-DL-dannert-gradel,provenance-DLLite,provenance-ELHr,Penaloza2023}, with different goals and semantics. We refer to Section~\ref{sec:relatedwork} for a detailed comparison between these proposals and the one we introduce in this paper.}

\new{Our main goal is to define a \blue{semiring provenance semantics for DL ontologies, \ie a} semantics for ontologies annotated with elements of an arbitrary commutative semiring, which distinguishes us from~\citet{attributedDL},~\citet{provenance-DLLite}, and~\citet{provenance-ELHr} who consider ontologies annotated with variables only. 
Our semantics should satisfy some desirable properties (such as being consistent with the semiring provenance defined in the relational database context, or allowing to evaluate a provenance expression \blue{from a given provenance semiring} in another semiring \blue{into which the provenance semiring can be homomorphically embedded,} and obtain the correct value),  
at least under some restrictions on the semiring (such as idempotency of the operations). 
This distinguishes us from~\citet{provenance-DL-dannert-gradel} and~\citet{Penaloza2023}, whose notions of provenance do not coincide with the one defined for relational databases, even if the semiring is fully idempotent \blue{(\cf Section~\ref{sec:relatedwork})}. 
Our semantics should also be defined independently from any specific reasoning algorithm, as the usual model-theoretic classical DL semantics (in the same way that the \blue{semiring provenance} 
of \blue{queries evlatuated over} annotated databases is independent from how the (Datalog) query is evaluated).} 
\blue{The DL language we consider in this paper, denoted by $\ELHIbot$, is a syntactic restriction of $\ELHI_\bot$ \cite{DBLP:conf/rweb/BienvenuO15}, the DL  which features conjunctions, existential restrictions, inverse roles and role inclusions, and encompasses several dialects of the \DLLITE and \EL families, hence the main lightweight DLs of the literature. The syntactic restriction requires that the ontology is normalized to avoid conjunctions and qualified existential restrictions on the right-hand side of concept inclusions.} 
\blue{In addition to the definition of a semiring provenance semantics for $\ELHIbot$ ontologies, we conduct a preliminary complexity analysis of the problems of deciding whether an annotated ontology entails an annotated query and computing the provenance of a query, focussing on some specific provenance semirings such that at least some of the desirable properties we consider are satisfied.} 

The rest of this work is organised as follows. 
\begin{itemize}
\item \new{In Section~\ref{sec:prelim}, we provide the \blue{relevant} background on description logics, semirings, and semiring provenance in the context of databases. 
We also introduce the $\ELHIbotrestr$ fragment of $\ELHIbot$ and \blue{show} that in this language, satisfiability and axiom entailment are \PTime-complete and BCQ entailment is \NP-complete (Theorem~\ref{theo:ELHIrestrcomplexity}). \blue{The definition of} this fragment is a side contribution of the paper, since it provides some insights on the complexity of reasoning with $\ELHI_\bot$ ontologies. }

\item \new{In Section~\ref{sec:prov-definition}, we define a semiring provenance semantics for $\ELHIbot$ (Section~\ref{sec:semanticdef}) and explain our design choices and the restrictions of our setting (Sections \ref{sec:designchoices} and~\ref{sec:syntacticRestrictions}). Specifically, we define the semantics of entailment of annotated axioms and BCQs from an annotated ontology, as well as the provenance of an axiom or a BCQ \wrt an annotated ontology. We show that satisfiability of an ontology does not depend on the annotations (Lemma~\ref{lem:relationship-annotated-standard-models}) and that computing the provenance of a rooted tree-shaped BCQ can be reduced to computing the provenance of a concept assertion (Theorem~\ref{th:instancequeries}).  We then investigate whether the semantics captures some well-known semantics for ontologies annotated with specific kinds of information. In particular, we show that under some conditions, it captures the Zadeh semantics of fuzzy DLs~\blue{\cite{Zadeh-IC65}} as well as a possibilistic semantics~\blue{\cite{DBLP:journals/ijis/QiJPD11,DBLP:journals/logcom/BenferhatB17}} (Proposition~\ref{prop:poss:prov}), and a notion of boundary that has been defined in the context of access control~\blue{\cite{BaKP-JWS12}} (Proposition~\ref{prop:access:prov}), but does not capture the product-based possibilistic semantics~\blue{\cite{BBKN17}} nor the bag semantics~\blue{\cite{DBLP:conf/ijcai/NikolaouKKKGH17,DBLP:journals/ai/NikolaouKKKGH19}}.}

\item \new{In Section~\ref{sec:classical-res-counterparts}, we show that, under some conditions, some classical results for $\ELHIbot$ ontologies can be transferred to annotated ontologies. 
\begin{itemize}
\item We show that an annotated $\ELHIbot$ ontology can be \blue{translated to a normal form} while preserving the provenance of its consequences (Theorem~\ref{th:normal-form}). 
\blue{This result will be used by the algorithms presented in Section~\ref{sec:why}.}
\item Given a satisfiable annotated $\ELHIbot$ ontology, we define \blue{its} canonical model\blue{,} which satisfies exactly the annotated assertions and BCQs that are entailed by the ontology (Theorem~\ref{thm:can-model-main}). 
We also define \blue{the} canonical model of the ontology and a concept $C$ (resp.\ role $P$) satisfiable \wrt the ontology. 
\blue{It allows} us to check whether the ontology entails an annotated concept inclusion between basic concepts (resp.\ role inclusion) with $C$ (resp.\ $P$) in the left-hand side, under the condition that the semiring is multiplicatively idempotent and some additional conditions on the ontology in the concept case (Theorems~\ref{thm:can-model-main-gci} and \ref{thm:can-model-main-ri}). 
\item We show that the problems of deciding entailment of an annotated concept assertion and of deciding entailment of an annotated concept inclusion are \blue{polynomially reducible to each other} if the semiring is multiplicatively idempotent (under some \blue{assumptions about} the ontology) (Theorem~\ref{th:red-concept}), and that a similar result holds for role assertions and role inclusions even if the semiring is not idempotent (Theorem~\ref{th:red-role}).
\end{itemize}}

\item \new{In Section~\ref{sec:expected-properties}, we show that our provenance semantics satisfies some desirable properties when restricted to commutative semirings that are additively and  multiplicatively idempotent and discuss the problems that arise beyond this setting. 
\begin{itemize}
\item We show that when the semiring is positive, an axiom or a BCQ has a non-zero provenance iff it is entailed by the non-annotated version of the ontology 
(with the additional condition that the semiring is multiplicatively idempotent in the case of a concept inclusion) (Theorem~\ref{th:sem-entailment}). \blue{This property ensures that the semiring provenance semantics reflects the entailment or non-entailment of axioms
and queries from the non-annotated ontology.}
\item We show that if the semiring is additively idempotent, our notion of provenance is consistent with the one defined for relational algebra queries or Datalog queries over annotated databases. This \blue{means} that computing the provenance of a query in our setting can be reduced to computing the provenance of a (Datalog) query as defined by~\citet{Green07-provenance-seminal} in the following cases: 
for a BCQ if the ontology contains only annotated assertions (Theorem~\ref{th:sem-algebra-cons-query}); 
for a concept assertion query $A(a)$ and an ontology consisting of annotated assertions and non-recursive concept inclusions with concept $A$ as right-hand side (Theorem~\ref{th:sem-algebra-cons}); 
and for a BCQ and an ontology that does not have existential role restrictions on the right nor role disjointness (that is, it is equivalent to a Datalog program) (Theorem~\ref{th:sem-cons-datalog}). 
We also show that our notion of provenance coincides with an alternative provenance semantics that has been defined for Datalog even if the semiring is not idempotent (Lemma~\ref{th:sem-cons-datalog-SAM}). 
\item We show that if a semiring can be homomorphically embedded in another one, one can obtain the provenance value of an axiom or a BCQ \wrt the ontology in the second semiring from the provenance value in the first one \blue{by applying the homomorphism} if the semirings are additively idempotent (and multiplicatively idempotent in the case of a concept inclusion, under some conditions on the ontology) (Theorem~\ref{th:sem-commut-hom}).
\end{itemize}}

\item \new{In Section~\ref{sec:why}, we focus on the problems of deciding annotated entailment and computing provenance using 
the \why provenance semiring that corresponds to the well-known why-provenance in the database setting. 
Since the \why semiring is additively idempotent but not multiplicatively idempotent \blue{(see Section~\ref{sec:database-prov} and Figure~\ref{fig:hierarchy})}, \blue{and we show in Section~\ref{sec:expected-properties} that multiplicative idempotence is necessary to satisfy some desirable properties in the case of concept inclusion entailment,} 
we focus on assertion and BCQ entailment. We provide algorithms and complexity results for computing the provenance of assertions and BCQs in this semiring for $\ELHIbot$ and $\ELHIbotrestr$.
\begin{itemize}
\item  We present a completion algorithm that derives in exponential time all annotated assertions that are entailed by a \why-annotated $\ELHIbot$ ontology (Theorem~\ref{prop:completionalgorithmELHI}), which allows us to compute the provenance of assertions in exponential time (Corollary~\ref{cor:complexity:provenanceassertion-ELHI}). 
\item For the case of $\ELHIbotrestr$\!, we adapt the completion algorithm so that \blue{it derives in polynomial time all entailed annotated assertions whose annotation size is bounded by a constant} 
(Theorem~\ref{prop:completionalgorithmELHIrestr}), which gives an upper bound for deciding annotated assertion entailment \blue{that is} exponential in the size of the annotation but polynomial in the size of the ontology (Corollary~\ref{th:complexity:provmonomial}). We then improve this complexity upper bound to \PSpace  (Theorem~\ref{th:pspace}).
\item For conjunctive queries, we present an algorithm based on an annotation-aware rewriting of the query using the completion of the ontology, and obtain exponential complexity upper bounds for the problem of deciding annotated BCQ entailment and of computing the provenance of a BCQ in $\ELHIbot$ (Theorem~\ref{th:CQ-answering-algorithm} and Corollaries~\ref{th:complexity-cq-elhi-entailment} and \ref{th:complexity-cq-elhi}) and an \NP upper bound for the problem of deciding annotated BCQ entailment in $\ELHIbotrestr$ if the annotation size is fixed (Theorem~\ref{th:complexity:provBCQmonomial}). 
\end{itemize}}

\item \new{In Section~\ref{sec:posbool-lin}, \blue{we further consider two provenance semirings, \posbool and \lin, which correspond to the so-called positive Boolean provenance and lineage in the database setting.} 
Since 
\posbool and \lin are additively and multiplicatively idempotent, provenance in these semirings satisfies all properties studied in Section~\ref{sec:expected-properties}. Hence \blue{we also consider the provenance of concept inclusions} 
in this section. 
We exhibit relationships with well-known notions related to explanations in DLs and also provide algorithms and complexity results for $\ELHIbot$ and $\ELHIbotrestr$. 
\begin{itemize}
\item We show the relationships between annotated entailments in \why, \posbool and \lin, and that the provenance in \posbool and \lin can be obtained from the provenance in \why (Propositions~\ref{prop:why-lin-pos} and \ref{prop:why-lin}).
\item We show that, under some conditions, provenance in the \posbool semiring can be computed via the set of justifications (Proposition~\ref{prop:ap:pos}). One can thus take advantage of the \blue{already mentioned} large body of work on DL axiom pinpointing. 
\item We show that provenance in the \lin semiring can be computed in polynomial time for $\ELHIbotrestr$ ontologies (and in exponential time for $\ELHIbot$) (Theorem~\ref{th:complexityRelevance}). This is done via an adaptation of the completion algorithm of Section~\ref{sec:why} (Theorem~\ref{completionAlgoRelevance}).
\end{itemize}}

\item \new{In Section~\ref{sec:relatedwork}, we review the other frameworks that use some form of semiring provenance for DLs and compare them with ours. }

\item \new{In Section~\ref{sec:conclusion}, we conclude with a discussion of some of our results and evoke possible future work.}
\end{itemize}

\new{This paper is \blue{closely related to} previous work by~\citet{provenance-DLLite} and~\citet{provenance-ELHr}. In particular, it adapts several ideas, examples or algorithms given by~\citet{provenance-ELHr} (\cf Section~\ref{sec:relatedwork} for details).} 
\new{ \citet{provenance-DLLite} consider provenance for ontology-based data access, a setting where a database is enriched with (i)~a DL-Lite$_\Rmc$ ontology and (ii)~mappings between 
the database and the ontology, and~\citet{provenance-ELHr} consider axiom entailment and query answering in \ELHr. 
In these two papers, the ontology axioms are annotated with provenance variables, and consequences of the ontology (axioms or queries) are annotated with provenance polynomials expressing their provenance information.} 
\blue{In contrast, we define the semantics of an ontology annotated by elements of an arbitrary commutative semiring, investigate whether it captures some of the existing semantics for annotated ontologies (such as the fuzzy semantics or the bag semantics), and under which conditions the polynomials we obtain when the axioms are annotated with provenance variables can be faithfully evaluated in some other semiring, among other desirable properties for a semiring provenance semantics for DL. 
}  

For the ease of reading, we chose to delegate most of the proofs (which may be long and technical but mostly rely on classical techniques) to \blue{Appendices \ref{app:proofs-prov-def}-\ref{app:ELHIrestrcomplexity} and provide proof sketches in the main text}.


\section{Preliminaries}\label{sec:prelim}
\blue{In this section, we recall the syntax and semantics of description logics, focussing on $\ELHIbot$, a syntactic restriction of $\ELHI_\bot$. We also introduce the lightweight $\ELHIbotrestr$ fragment of $\ELHIbot$. We then provide the relevant background on semirings and semiring provenance in the context of databases.}

\subsection{Description Logics and the $\ELHIbot$ Language}\label{sec:prelimDL}

Our main focus \new{is} on a syntactic restriction of the DL language $\ELHI_\bot$ \blue{\cite{DBLP:conf/rweb/BienvenuO15}} \new{that we call $\ELHIbot$ (where $n$ stands for ``normalized''),} 
which 
\blue{can express conservative extensions of ontologies in}
  dialects of the well-known 
\EL \cite{BBL-IJCAI05} and \DLLITE \shortcite{ACKZ-JAIR09} families of 
lightweight DLs.  

\paragraph{Syntax} Let \NC, \NR and \NI be three mutually disjoint countable sets of
\emph{concept-}, \emph{role-}, and \emph{individual names}, respectively. 
Given a DL language \logic, an \logic \emph{ontology} \Omc is a finite set of \emph{axioms} whose form depends on \logic. 
We consider the following kinds of axioms. 
\begin{itemize}
\item Concept and role \emph{assertions} are of the form $A(a)$
 and $R(a,b)$, respectively, with $A\in \NC$, $R\in \NR$, $a,b\in \NI$. 
\item \emph{General concept inclusions} (GCIs) are expressions of the form $C\sqsubseteq D$, \new{where $C$ and $D$ are \emph{$\ELHIbot$ concepts}} 
built according to the grammar rules:
\[
  C::= A\mid \exists P.C \mid C\sqcap C\mid \top \quad\quad\quad D::= A\mid\exists P.\top\mid \bot \quad\quad\quad P::= R\mid R^-
\]
where $R\in \NR$,~$A\in \NC$. We often use $\exists P$ as a shorthand for $\exists P.\top$, and we call \emph{basic concepts} the concepts of the form $A$ or $\exists P$. 
\item Positive \emph{role inclusions} (RIs) are of the form $P\sqsubseteq Q$ with $P,Q::=R\mid R^-$ for $R\in\NR$ and 
\emph{negative role inclusions} are of the form $P\sqcap Q\sqsubseteq \bot$.%
\footnote{$\ELHI_\bot$ normally does not allow negative RIs. 
\blue{We include them in the definition of $\ELHIbot$ so that DL-Lite$_\Rmc$ is a fragment of  $\ELHIbot$. Indeed, DL-Lite$_\Rmc$
allows for axioms of the form $P\sqsubseteq \neg Q$ \cite{DBLP:journals/jar/CalvaneseGLLR07}, equivalently 
expressed as $P\sqcap Q\sqsubseteq \bot$.}}
\end{itemize}
An $\ELHIbot$ \emph{ontology} is a finite set of axioms of any of the available forms. 
\blue{We are in particular interested in two fragments of $\ELHIbot$: \DLLiteR and (a syntactic restriction of) $\EL$ extended with role inclusions and range restrictions, \ELHr.} 
\begin{itemize}
\item A \emph{\DLLiteR ontology} may contain assertions, positive and negative role inclusions, and GCIs of one of the 
following restricted forms: $D_1\sqsubseteq D_2$ or $D_1\sqcap D_2\sqsubseteq \bot$ where $D_1$ and $D_2$ 
are 
\emph{basic concepts}. 
\item 
An \ELHr \emph{ontology} may contain all the considered axioms with the restrictions that 
(i)~$\bot$ does not occur in the ontology 
and (ii)~there is no inverse role (\ie $P::=R$), 
except for \emph{range restrictions} of the form $\exists R^-\sqsubseteq A$. 
\end{itemize}

In the standard DL literature, languages from the \EL family usually allow for GCIs of the general form 
\new{$C_1\sqsubseteq C_2$ with $C_1,C_2$ constructed as in the grammar rule for $C$ introduced above} \blue{\cite{BBL-IJCAI05,BBL-EL08,DBLP:conf/rweb/BienvenuO15}}.  
\new{A set of} such GCIs \new{$\Tmc$} can be translated into \new{a conservative extension $\Tmc'$ of $\Tmc$ in} our syntax \new{in polynomial time} 
by exhaustively applying the following
rules: 
\begin{itemize}
	\item replace $C\sqsubseteq C_1\sqcap C_2$ 
	by $C\sqsubseteq C_1$ and $C\sqsubseteq C_2$, 
	\item replace $C_1\sqsubseteq \exists P.C_2$ 
	by $C_1\sqsubseteq \exists S$, \mbox{$S\sqsubseteq P$} and $\exists S^-\sqsubseteq C_2$ 
	where $S$ is a fresh role name.   
\end{itemize}
The reason for this syntactic restriction 
is that conjunctions or qualified restrictions of a role on the right-hand side 
of GCIs lead to counter-intuitive behavior when dealing with provenance. 
We discuss this issue in further detail in Section~\ref{sec:syntacticRestrictions}.

Given an ontology \Omc, we denote by $\signature{\Omc}$ the \emph{signature} of \Omc; that is, the set of concept and role names 
that occur in $\Omc$; and by $\individuals{\Omc}$ the set of individual names that occur in $\Omc$. For $R\in\NR$, we let 
$\mn{inv}(R)=R^-$ and $\mn{inv}(R^-)=R$.

\paragraph{Semantics} The semantics of DL languages is defined through \emph{interpretations}, in the spirit of first-order
logic. 
An interpretation is a pair $\Imc=(\Delta^\Imc ,\cdot^\Imc)$
where $\Delta^\Imc$ is a non-empty set (called the \emph{domain} of \Imc), 
and $\cdot^\Imc$ is the \emph{interpretation function}, which maps every $a\in\NI$ to $a^\Imc\in\Delta^\Imc$; 
every $A\in\NC$ to $A^\Imc\subseteq \Delta^\Imc$; 
and every   $R\in\NR$ to $R^\Imc\subseteq \Delta^\Imc\times\Delta^\Imc$. 
The interpretation function 
$\cdot^\Imc$ is extended to complex expressions as follows: 
  \begin{align*} 
	(\top)^\Imc = {} & \Delta^\Imc; \\ 
	(\bot)^\Imc = {} & \emptyset; \\ 
	(R^-)^\Imc = {} & \{(e,d)\mid (d,e)\in R^\Imc\};\\	
	    (\exists P.C)^\Imc = {} & \{d\mid \new{\exists e\in C^\Imc}
    \text{ s.t. }(d,e)\in P^\Imc\};\\
    (C\sqcap D)^\Imc = {} &C^\Imc\cap D^\Imc.
  \end{align*} 
The interpretation \Imc \emph{satisfies} the concept assertion $A(a)$ if $a^\Imc\in A^\Imc$; 
the role assertion $R(a,b)$ if $(a^\Imc,b^\Imc)\in R^\Imc$; 
the GCI $C\sqsubseteq D$ if $C^\Imc\subseteq D^\Imc$; 
the positive RI $P\sqsubseteq Q$ if $P^\Imc\subseteq Q^\Imc$; and 
the negative RI $P\sqcap Q\sqsubseteq \bot$ if $P^\Imc\cap Q^\Imc=\emptyset$. 
The satisfaction of an axiom $\alpha$ by $\Imc$ is denoted $\Imc\models \alpha$. 
The interpretation $\Imc$ is a \emph{model} of the ontology $\Omc$, denoted by $\Imc\models\Omc$, iff 
$\Imc\models \alpha$ for every $\alpha\in\Omc$. 
Finally, $\Omc$ \emph{entails} an axiom $\alpha$ if $\Imc\models \alpha$ for every model $\Imc$ of $\Omc$.

\paragraph{Queries}
A \emph{conjunctive query} (CQ) is an existentially quantified \new{first-order} formula of the form 
$\exists \vec{y}\, \phi(\vec{x}, \vec{y})$ where $\phi(\vec{x}, \vec{y})$ is a conjunction of atoms of the form $A(t)$ or $R(t,t')$ with $A\in\NC$, $R\in\NR$, and terms $t,t'\in \NI\cup\vec{x}\cup\vec{y}$.  We denote by $\mn{atoms}(q)$ and $\mn{terms}(q)$ the sets of atoms and terms of a query $q$. A CQ consisting of a single atom is called an \emph{atomic query}.  
A \emph{union of conjunctive queries} (UCQ) is a finite disjunction of CQs (over the same free variables). A query is \emph{Boolean} if it has no free variables. 
An interpretation $\Imc$ satisfies a Boolean CQ (BCQ) $q:=\exists \vec{y}\, \phi(\vec{y})$, written $\Imc\models q$, iff there is a \emph{match for $q$ in $\Imc$}, where a match for $q$ in $\Imc$ is a function $\pi:\mn{terms}(q) \rightarrow \Delta^\Imc$ such that $\pi(t)=t^\Imc$ for every \new{$t \in \NI\cap\mn{terms}(q)$}, \new{and for every $t,t'\in \mn{terms}(q)$:}
\begin{itemize}
\item $\pi(t)\in A^\Imc$ for every $A(t) \in \mn{atoms}(q)$, and 
\item  $(\pi(t), \pi(t')) \in R^\Imc$ for every $R(t,t') \in \mn{atoms}(q)$. 
\end{itemize}
A BCQ $q$ is \emph{entailed} by an ontology \Omc, written $\Omc\models q$, if and only if $\Imc\models q$ for every model $\Imc$ of $\Omc$. 
A tuple of constants $\ans$ is an \emph{answer} to a CQ $q(\vec{x}):=\exists \vec{y}\, \phi(\vec{x}, \vec{y})$ over $\Omc$ if $\ans$ and $\vec{x}$ have the same \new{length} and $\Omc\models q(\ans)$ where $q(\ans)$ is the BCQ obtained by replacing the variables from $\vec{x}$ with the corresponding constants from $\ans$. 

\new{We will sometimes use \emph{rooted tree-shaped queries}. Such a query is a CQ $q$ such that:
\begin{itemize}
\item $\mn{terms}(q)$ contains exactly one constant or free variable, all other terms being existentially quantified variables,
\item the undirected graph $\{(x,y)\mid R(x,y)\in\mn{atoms}(q)\}$ is a tree whose root is a constant or a free variable,
\item for every $x,y\in\mn{terms}(q)$, there is at most one role name $R$ such that $R(x,y)\in\mn{atoms}(q)$ or $R(y,x)\in\mn{atoms}(q)$,
and only one of these is the case.
\end{itemize}
It is well-known that there is a 
 correspondence between $\ELHIbot$ concepts and rooted tree-shaped queries (see, e.g, \cite[page 2]{DBLP:journals/jair/GlimmLHS08}): 
given an $\ELHIbot$ concept $C$, there is a rooted tree-shaped query $q_C(x)$ that retrieves all instances of $C$ and 
can be built by structural induction on $C$; and given a rooted tree-shaped query $q(x)$, there is an $\ELHIbot$ concept $C_q$ such that $q(x)$ retrieves all instances of $C_q$ and $C_q$ can be build by structural induction on the tree-structure of $q(x)$. For example, the $\ELHIbot$ concept $A\sqcap\exists R.(B\sqcap\exists S^-.D)$ corresponds to the CQ $q(x)=\exists yz\, A(x)\wedge R(x,y)\wedge B(y)\wedge S(z,y)\wedge D(z)$.}

\begin{table*}
\caption{Combined complexity of axiom and BCQ entailment in different DL languages. }
\label{tab:comb:complexity}
\centering
\begin{tabular}{@{}lccc@{}}
\toprule
& axiom entailment   &  BCQ entailment \\ 
\midrule
\multirow{2}{*}{$\ELHIbot$} & \ExpTime-complete& \ExpTime-complete
\\
& \new{ \shortcite{DBLP:conf/rweb/BienvenuO15}}& \new{ 
\shortcite{DBLP:conf/rweb/BienvenuO15}}
\\
\midrule
\multirow{2}{*}{\EL}  & \PTime-complete & \NP-complete\\
 & \new{ \shortcite{BBL-IJCAI05,DBLP:conf/kr/CalvaneseGLLR06}}& \new{ \shortcite{DBLP:conf/dlog/Rosati07}}
\\
\midrule
\multirow{2}{*}{\blue{\ELHr}}  & \blue{\PTime-complete} & \blue{\NP-complete}\\
 & \blue{ \shortcite{BBL-EL08,DBLP:conf/kr/CalvaneseGLLR06}}& \blue{ \shortcite{LTW:elcqrewriting09}}
\\
\midrule
\multirow{2}{*}{\DLLiteR}  & in \PTime & \NP-complete
\\
 & \new{\shortcite{DBLP:journals/jar/CalvaneseGLLR07}}& \new{ \shortcite{DBLP:journals/jar/CalvaneseGLLR07}}
\\
\midrule
\multirow{2}{*}{$\ELHIbotrestr$} & \PTime-complete & \NP-complete
\\
& \new{(Theorem \ref{theo:ELHIrestrcomplexity})}& \new{(Theorem \ref{theo:ELHIrestrcomplexity})}
\\
\bottomrule
\end{tabular}
\end{table*}

\paragraph{Complexity}
\new{Given an ontology $\Omc$ (a BCQ $q$, or an axiom $\alpha$), we denote by $|\Omc|$ ($|q|$, or $|\alpha|$, 
respectively) and call the \emph{size} of $\Omc$ ($q$, $\alpha$, resp.) the length of the string representing $\Omc$ 
($q$, $\alpha$, resp.), where elements of $\NC$, $\NR$, $\NI$ and variables are considered of length one.} 
The complexity of axiom entailment and BCQ entailment in the different
languages we consider \blue{(including the new fragment of $\ELHIbot$ we introduce in the next paragraph)} is summarised in Table~\ref{tab:comb:complexity}. 
\new{These results refer} only to the \emph{combined complexity}, where 
everything 
is part of the input and measured accordingly \new{(\ie the input size is $|\Omc|+|\alpha|$ in the case of axiom entailment, $|\Omc|+|q|$ in the case of BCQ entailment)}.

\paragraph{\new{The $\ELHIbotrestr$ fragment of $\ELHIbot$}} To discuss lightweight DLs such as \ELHr and \DLLiteR, we 
define a fragment of $\ELHIbot$ which extends them and shares their good computational properties\blue{, which we call $\ELHIbotrestr$}. 
\blue{To define this fragment, we first need to introduce the normal form of an $\ELHIbot$ ontology.} An $\ELHIbot$ ontology $\Omc$ is in \emph{normal form} if all its GCIs are of the form
$$A\sqsubseteq B,\ A\sqcap A'\sqsubseteq B,\  A\sqsubseteq \exists R, \ A\sqsubseteq \exists R^-, \ \exists R.A\sqsubseteq B, \text{ or } \exists R^-.A\sqsubseteq B$$ with $R\in\NR$, $A,A'\in\NC\cup\{\top\}$, $B\in\NC\cup\{\bot\}$. Every $\ELHIbot$ ontology $\Omc$ can be translated  in linear time into an ontology $\Omc'$ in normal form which is a conservative extension of $\Omc$ (\cf Section~\ref{subsection:normalization} for the normalization process in the more general case where the ontology can be annotated). We \blue{also need to} define a binary relation $\sqsubseteq_\Omc$ over $(\signature{\Omc}\cap\NR)\cup\{R^-\mid R\in \signature{\Omc}\cap\NR\}$ as the transitive closure of the relation defined by $$\{(S,P), (\mn{inv}(S),\mn{inv}(P))\mid S\sqsubseteq P\in\Omc\}\cup\{(R,R), (R^-,R^-)\mid R\in\signature{\Omc}\cap\NR\}.$$

\begin{definition}[$\ELHIbotrestr$]\label{def:ELHIrestr}
An $\ELHIbot$ ontology $\Omc$ \emph{belongs to} $\ELHIbotrestr$ if 
\begin{enumerate}
\item $\Omc$ is in normal form, and   
\item \new{if $C\sqsubseteq \exists P_1\in\Omc$, $P_1\sqsubseteq_\Omc P_2$, and $\exists\mn{inv}(P_2).A\sqsubseteq B \in \Omc$, then $A=\top$. }
\end{enumerate}
\end{definition}
\new{In particular, every \DLLiteR ontology is an $\ELHIbotrestr$ ontology (its GCIs are already in normal form and do not contain qualified role restrictions) and every \ELHr ontology can be normalized into a conservative extension in $\ELHIbotrestr$ (since it does not contain inverse roles except in GCIs of the form $\exists R^-\sqsubseteq A$). 
	Since the binary relation $\sqsubseteq_\Omc $ can be constructed in polynomial time and all other checks are based on the syntax of $\ELHIbotrestr$\!, one can verify whether an ontology belongs to $\ELHIbotrestr$ in polynomial time \wrt its size.}
	
\blue{By restricting the language in this way, we can decrease the complexity of axiom and BCQ entailment from \ExpTime-complete in $\ELHIbot$ to \PTime-complete (axiom entailment) and \NP-complete (BCQ entailment) in $\ELHIbotrestr$, as in \DLLiteR and \ELHr (\cf Table~\ref{tab:comb:complexity}). Indeed,} the algorithms we develop in Section \ref{sec:why} will allow us to prove the following result \blue{(\cf Appendix~\ref{app:ELHIrestrcomplexity})}.

\begin{restatable}{theorem}{theoELHIrestrcomplexity}\label{theo:ELHIrestrcomplexity}
For ontologies that belong to $\ELHIbotrestr$\!, satisfiability and axiom entailment are  \PTime-complete and BCQ entailment is \NP-complete.
\end{restatable}
	
The intuition behind the definition of $\ELHIbotrestr$ is that we want to avoid situations in which the 
ontology entails pairs of consequences of the form $C\sqsubseteq \exists P$ and $\exists \mn{inv}(P).A\sqsubseteq B$,
which together imply $C\sqcap A\sqsubseteq \exists P.B$. Observe that the subsumer in this consequence is a qualified
existential restriction. The idea of the restriction is then that each of the ``anonymous parts'' of a model should be 
constrained by a single atomic concept; that is, adding $A(a)$ should not allow to qualify the $P$-successor required by 
$C(a)$ and $C\sqsubseteq \exists P$.

\subsection{\new{Semirings and Infinite Sums}}

In semiring provenance, database facts are annotated with elements of algebraic structures known as \emph{commutative semirings}. 
\new{We recall here definitions related to semirings that will be useful in this work. We refer the interested reader to Section 2.2 of the PhD thesis of \citet{DBLP:phd/hal/Ramusat22} for references and discussion of alternative definitions, and to the recent work by \citet{DBLP:conf/birthday/BrinkeGMN24} for a discussion of the properties infinitary operations must satisfy to provide informative provenance analysis over infinite domains.}

\paragraph{Semirings} A semiring $\semiringshort=\semiring$ consists of a set $\semiringset$ 
equipped with two binary \blue{operations: the \emph{addition},~ $\oplus$, which is associative and commutative and has 
an identity element called~$\zero$, and the \emph{multiplication},~$\otimes$, which is associative, has 
an identity element called~$\one$, and is such that 
$\otimes$ distributes over $\oplus$ and $\zero$ is annihilating for~$\otimes$.}  
\new{Explicitly, a semiring satisfies the following properties for all $a,b,c\in\semiringset$:
\begin{itemize}
\item $(a\oplus b)\oplus c=a\oplus(b\oplus c)$ (associativity of $\oplus$);
\item $a\oplus b=b\oplus a$ (commutativity of $\oplus$);
\item $a\oplus\zero=\zero\oplus a=a$ (\zero identity of $\oplus$);
\item $(a\otimes b)\otimes c=a\otimes(b\otimes c)$ (associativity of $\otimes$);
\item $a\otimes\one=\one\otimes a=a$ (\one identity of $\otimes$); 
\item $a\otimes(b\oplus c) = (a\otimes b)\oplus (a\otimes c)$ (distributivity of $\otimes$ over $\oplus$); and
\item $a\otimes\zero=\zero\otimes a=\zero$ (\zero annihilating for $\otimes$).
\end{itemize}}
A semiring is called \emph{commutative} when $\otimes$ is commutative. 
We use the convention according to which multiplication has precedence over addition to omit parentheses 
(\ie $a\oplus b\otimes c$ is $a\oplus (b\otimes c)$). \new{If additive inverses exist (\ie for every $a\in\semiringset$, there exists $b\in\semiringset$ such that $a\oplus b = \zero$), then $\semiringshort$ is in fact a \emph{ring}.}

A semiring is \emph{$\oplus$-idempotent} (resp.\  \emph{$\otimes$-idempotent}) if for every $a\in K$, 
$a\oplus a=a$ (resp.\ $a\otimes a=a$). 
It is \emph{absorptive} if for every $a,b\in K$, $a \otimes b \oplus a= a$. Note that if a semiring is absorptive, then it is also \plusidem (since $a\oplus a=a\otimes\one\oplus a = a$). 
Finally, a semiring is \emph{positive} if for every $a,b\in K$, (i)~$a \otimes b =\zero$ iff 
$a=\zero$ or $b=\zero$, and 
(ii)~$a \oplus b =\zero$ iff 
$a=b=\zero$. 

\begin{example}\label{ex:semirings}
We will consider the following commutative semirings in our running example.
\begin{itemize}
\item The counting semiring $\mathbb{N}=(\mathbb{N}, +, \times, 0, 1)$ is the set of the natural integers equipped with the usual addition and multiplication between integers. 

\item The tropical semiring $\mathbb{T}=(\mathbb{R}^\infty_+, \min , +, \infty, 0)$, used to represent costs \footnote{This version of tropical semiring has been considered in the semiring provenance literature (to represent data access costs), but the min-plus tropical semiring is typically defined with $\mathbb{R}^\infty$ instead of $\mathbb{R}^\infty_+$ as domain and there exist other tropical semirings, such as the max-plus one $(\mathbb{R}^{-\infty}, \max,+, -\infty,0)$.}, is the set of \new{non-negative} real numbers extended with the symbol $\infty$, with the minimum operation as addition ($a\oplus b=\min(a,b)$) and standard addition as multiplication ($a\otimes b=a+b$). 
\item The Viterbi semiring $\mathbb{V}=([0,1], \max, \times, 0,1)$, used to represent confidence scores, is the set of real numbers between $0$ and $1$ with the maximum operation as addition ($a\oplus b=\max(a,b)$) and the usual multiplication between reals ($a\otimes b=a\times b$). 

\item The fuzzy semiring $\mathbb{F}=([0,1],\max,\min ,0,1)$, used to represent truth degrees, is the set of real numbers between $0$ and $1$ with the maximum operation as addition ($a\oplus b=\max(a,b)$) and the minimum operation as multiplication ($a\otimes b=\min(a,b)$). 

\item  The \L ukasiewicz semiring $\mathbb{L}=([0,1],\max,\star_L,0,1)$, used to represent truth values, is the set of real numbers between $0$ and $1$ with the maximum operation as addition ($a\oplus b=\max(a,b)$) and multiplication defined by $\blue{a\otimes b=} a\star_L b=\max(0,a+b-1)$. 

\item The access control semiring $\mathbb{A}=(\{P<C<S<T<0\},\min,\max,0,P)$, used to represent clearance levels required to access data \cite{GreenT17}, is the set $\{P,C,S,T,0\}$ (for ``public'', ``confidential'', ``secret'', ``top-secret'', and ``nobody knows'') with the minimum and maximum operations \wrt the obvious total ordering $P<C<S<T<0$ as addition and multiplication respectively. 
\end{itemize}
From these examples, the counting semiring is the only one which is not absorptive nor $\oplus$-idempotent;  
the fuzzy and the access control semirings are the only ones which are $\otimes$\mbox{-}idempotent; 
and the \L ukasiewicz semiring is the only one that is not positive since,
for example, $0.1\star_L 0.1=\max(0,0.1+0.1-1)=0$. 
\end{example}

\new{An example of a \emph{non-commutative} semiring is the semiring of \blue{formal} languages over a finite alphabet $A$: $(2^{A^*},\cup,\cdot,\emptyset,\{\epsilon\})$ where $2^{A^*}$ is the set of \blue{all} 
sets of words over $A$, $\epsilon$ is the empty word and $\cdot$ is element-wise concatenation: $S_1\cdot S_2=\{w_1w_2\mid w_1\in S_1, w_2\in S_2\}$.}

\paragraph{Infinite sums} 
\new{In this work we need to compute the addition over possibly uncountably many semiring elements and thus consider the notion of \emph{complete} 
semirings~\cite{Krob1987,Karner1992}.} 
\new{A semiring $\semiringshort=\semiring$ is \emph{complete} (resp.\ \emph{$\omega$-complete}) if for every (resp.\ at most countable) family $(a_i)_{i\in I}$ of elements in $\semiringshort$ indexed by $I$, we can define an element $\bigoplus_{i\in I}a_i$ in $\semiringset$ satisfying:}
\new{
\begin{itemize}
\item if $I$ is empty: $\bigoplus_{i\in I}a_i=\zero$;
\item if $I=\{1,\dots,n\}$ is non-empty finite: $\bigoplus_{i\in \{1,\dots,n\}}a_i=a_1\oplus\dots\oplus a_n$;
\item $\bigoplus_{i\in I} a_i=\bigoplus_{j\in J}(\bigoplus_{i\in I_j} a_i)$ if $\bigcup_{j\in J}I_j=I$ and $I_j\cap I_{j'}=\emptyset$ for every $j\neq j'$;
\item for every $b\in\semiringset$: $b\otimes(\bigoplus_{i\in I}a_i)=\bigoplus_{i\in I}(b\otimes a_i)$ and $(\bigoplus_{i\in I}a_i)\otimes b=\bigoplus_{i\in I}(a_i\otimes b)$. 
\end{itemize}
}

Given a semiring $\semiringshort=\semiring$, let $\leq$ 
be the binary relation over $\semiringset$ such that $a\leq b$ if and only if there exists $c\in \semiringset$ with $a\oplus c =b$. 
\new{The semiring $\semiringshort$ is \emph{naturally ordered} if $\leq$ is a partial order, $\oplus$ and $\otimes$ are monotone in each argument, and $\zero$ is the least element \blue{\wrt $\leq$}. Datalog provenance was originally defined for \emph{$\omega$-continuous} semirings~\cite{Green07-provenance-seminal}, which are naturally ordered $\omega$-complete semirings that additionally satisfy some conditions related to the existence of least upper bounds according to $\leq$ \cite{Karner1992,DBLP:reference/hfl/Kuich97}.}

\begin{example}
\blue{All semirings introduced in Example \ref{ex:semirings} are naturally ordered and the counting semiring $\mathbb{N}$ is the only one which} 
is not \new{$\omega$-complete, hence not complete}. \new{A complete version $\mathbb{N}^{\infty}$ of $\mathbb{N}$ is obtained by extending it with $\infty$, with both operations extended as expected, and $\Sigma_{i\in I} n_i$ defined as follows for every family $(n_i)_{i\in I}$ of elements from $\mathbb{N}\cup\{\infty\}$: 
$\Sigma_{i\in I} n_i=\infty$ if there exists $i\in I$ such that $n_i=\infty$;
$\Sigma_{i\in I} n_i=\infty$ if there are infinitely many $i\in I$ such that $n_i\neq 0$;
and $\Sigma_{i\in I} n_i= \Sigma_{i\in I_f} n_i$ otherwise, with $I_f=\{i\in I\mid n_i\neq 0\}$.} 
\end{example}

\paragraph{Homomorphisms} A \emph{semiring homomorphism} from a semiring $\semiringshort = \semiring$ to another $\mathbb{K'} = (K', + ,\cdot, 0, 1)$ is a mapping $h : K \rightarrow K'$ such that $h(\zero) = 0$, $h(\one) = 1$, and for all $a,b\in \semiringset$, $h(a \oplus b) = h(a) + h(b)$ and $h(a \otimes b) = h(a) \cdot h(b)$. 
\new{A semiring homomorphism between ($\omega$-)complete semirings is \emph{($\omega$-)complete} if $h(\bigoplus_{i\in I}a_i)=\bigplus_{i\in I}h(a_i)$ for every (countable) family $(a_i)_{i\in I}$.}  
A semiring homomorphism between $\omega$-continuous semirings is \emph{$\omega$-continuous} if it preserves least upper bounds. 

\subsection{Semiring Provenance in the Database Setting}\label{sec:database-prov}
In the database setting, the framework of \emph{semiring provenance} generalizes the semantics of queries over 
annotated databases.    
We explain the main notions of this field next, assuming a basic understanding of relational databases \cite{DBLP:books/aw/AbiteboulHV95}. 

\paragraph{Queries over annotated databases} 
In semiring provenance, database facts are annotated with elements of \emph{commutative semirings}. 
The semantics of positive relational algebra queries \new{(using operations select, project, natural join, rename and union, which are known to be equivalent to union of conjunctive queries and SQL {\sc select-from-where-union} queries)} over databases annotated with elements of any commutative 
semiring is defined inductively on the structure of the query \cite{Green07-provenance-seminal,GreenT17}. Intuitively, joint 
use of the data corresponds to the semiring multiplication while alternative use of the data corresponds to the semiring 
addition. In other words, we \emph{multiply} the annotations of
 all the facts that together produce an answer, and 
\emph{add} such products over all combinations yielding the same answer.

\begin{example}\label{ex:annotated-database}
\new{\sloppy{Consider the three facts $\alpha_1=\mn{Deity}(\mn{Dionysus})$, $\alpha_2=\mn{mother}(\mn{Dionysus},\mn{Semele})$ and $\alpha_3=\mn{mother}(\mn{Dionysus},\mn{Demeter})$}, stating that Dionysus is a deity who has mothers Semele and Demeter, and the query $q:=\exists xy\, \mn{Deity}(x) \wedge \mn{mother}(x,y)$ that asks if there is a deity who has a mother. 
The answer is \emph{yes} and it can be derived either using $\alpha_1$ and $\alpha_2$ or using $\alpha_1$ and $\alpha_3$, interpreting $y$ by either $\mn{Semele}$ or $\mn{Demeter}$. We now annotate the facts with elements of the semirings introduced in Example~\ref{ex:semirings}. 
Table~\ref{tab:exa25} shows possible annotations and the resulting annotation of the query answer in each semiring.\qedhere}

\begin{table*}[t]
\caption{\new{The results of provenance computation over different semirings and annotations.}}
\label{tab:exa25}
\centering
\new{\begin{tabular}{@{}c c c c c@{}}
\toprule
& $\alpha_1$ & $\alpha_2$ & $\alpha_3$ & $q$\\
\midrule
Multiplicities & \multirow{2}{*}{$3$} & \multirow{2}{*}{$2$} & \multirow{2}{*}{$1$} & \multirow{2}{*}{$3\times 2 +3\times 1= 9$}\\
$\mathbb{N}=(\mathbb{N}, +, \times, 0, 1)$\\
\midrule
Costs & \multirow{2}{*}{$1$} & \multirow{2}{*}{$5$} & \multirow{2}{*}{$8$} & \multirow{2}{*}{$\min(1+5, 1+8)=6$}\\
$\mathbb{T}=(\mathbb{R}^\infty_+, \min , +, \infty, 0)$\\
\midrule
Confidence & \multirow{2}{*}{$0.9$} & \multirow{2}{*}{$0.6$} & \multirow{2}{*}{$0.4$} & \multirow{2}{*}{$\max(0.9\times 0.6, 0.9\times 0.4)=0.54$}\\
$\mathbb{V}=([0,1], \max, \times, 0,1)$\\
\midrule
Truth degrees & \multirow{2}{*}{$0.9$} & \multirow{2}{*}{$0.8$} & \multirow{2}{*}{$0.2$} & \multirow{2}{*}{$\max(\min(0.9,0.8),\min(0.9,0.2))=0.8$}\\
$\mathbb{F}=([0,1],\max,\min ,0,1)$\\
\midrule
Truth values & \multirow{2}{*}{$0.9$} & \multirow{2}{*}{$0.8$} & \multirow{2}{*}{$0.2$} &\multicolumn{1}{l}{$\max(\max(0,0.9+0.8-1),$}\\
$\mathbb{L}=([0,1],\max,\star_L,0,1)$ &&&& \multicolumn{1}{r}{$\max(0, 0.9+0.2-1))=0.7$}\\
\midrule
Clearance levels & \multirow{2}{*}{$P$} & \multirow{2}{*}{$C$} & \multirow{2}{*}{$S$} & \multirow{2}{*}{$\min(\max(P,C),\max(P,S))=C$}\\
$\mathbb{A}$ ($P<C<S<T<0$)\\
\bottomrule
\end{tabular}}
\end{table*}
\end{example}

\begin{remark}[Probabilistic databases]
\blue{There is no semiring $\semiringshort=([0,1],\oplus,\otimes,0,1)$ that directly captures the semantics of databases annotated with probabilities. Indeed, if we consider two facts $A(a)$ and $B(b)$, both annotated with probability $0.5$, the provenance of both queries $\exists xy\ A(x)\wedge A(y)$ and $\exists xy\ A(x)\wedge B(y)$ will be $0.5\otimes 0.5$ while the probabilities of these queries (according to the classical semantics of probabilistic databases) differ ($0.5$ for the first one versus $0.25$ for the second one). The reason is that
probabilities are not truth functional, while provenance is. However, query probabilities can be computed using positive Boolean provenance, which is captured by a semiring (\cf Section~\ref{sec:applications-of-posbool-provenance}) \cite{DBLP:journals/sigmod/Senellart17}.}
\end{remark}

\citet{Green07-provenance-seminal} also define semiring provenance for 
Datalog queries \new{(\ie queries formed by a \blue{finite} set of Datalog rules with a distinguished output predicate)}, using derivation trees. The approach associates to the query the sum over all its derivation trees of 
the \new{products of the annotations of the tree leaves (which correspond to database facts)}. 
\new{To be able to handle infinitely many derivation trees and relate the provenance of a Datalog query with  the least fixpoint of a system of fixpoint equations, they consider $\omega$-continuous semirings. However, their definition of provenance for Datalog queries is already well-defined for semirings that are only $\omega$-complete. \blue{Indeed, the sum over all derivation trees is well-defined as soon as infinite sums of semiring elements are well-defined}.} 

\paragraph{Provenance semirings} 

Provenance semirings were introduced to abstract from a specific semiring and compute a representation of the provenance. 
\new{Given a finite set $\semiringVars$ of 
\emph{variables}\footnote{\new{In the literature $\semiringVars$ is often only assumed to be countable. Following \citet{DBLP:journals/sigmod/Senellart17}, we assume $\semiringVars$ to be finite for simplicity, since variables are only used to annotate the elements of finite sets.} \blue{More precisely, for every database $\Dmc$, we can let $\semiringVars_\Dmc=\{x_1,\dots,x_n\}$ where $n$ is the number of facts in $\Dmc$, and compute provenance \wrt $\Dmc$ using $\mi{Prov}[\semiringVars_\Dmc]$.}} 
which are used to annotate the database facts and can be thought of as identifiers, a \emph{provenance semiring} $\mi{Prov}[\semiringVars]=(\mi{Prov}[\semiringVars],\oplus,\otimes,\zero,\one)$ parametrized by $\semiringVars$ 
is a \blue{commutative} semiring over a space of 
\emph{provenance expressions} built from variables from~$\semiringVars$, $\zero$, $\one$, $\oplus$ and $\otimes$.}
\new{ 
We recall below the definitions of provenance semirings that have been considered in the literature \cite{DBLP:journals/mst/Green11,DBLP:conf/icdt/DeutchMRT14,GreenT17}. 
\blue{We will see that all the considered provenance expressions can be written as \emph{polynomials}. Assuming that $\semiringVars=\{x_1,\dots,x_n\}$, recall that a \emph{monomial} is a formal product of these variables, possibly raised to a nonnegative power: $x_1^{e_1}\dots x_n^{e_n}$ with $e_i\in\mathbb{N}$ for $1\leq i\leq n$. Exponents equal to $1$ can be omitted, as well as variables whose exponent is equal to $0$ (\eg $x_1^2x_2^1x_3^0$ can be written $x_1^2x_2$). A \emph{polynomial with variables from $\semiringVars$ and coefficients from a commutative ring $C=(C,\oplus^C,\otimes^C,\zero^C,\one^C)$} is a finite linear combination of monomials: $\bigoplus^C_{(e_1,\dots,e_n)\in I}\mi{coef}{(e_1,\dots,e_n)}x_1^{e_1}\dots x_n^{e_n}$ where $I$ is a finite subset of $\mathbb{N}^n$ and $\mi{coef}{(e_1,\dots,e_n)}\in C$ for every $(e_1,\dots,e_n)\in I$. Monomials whose coefficient is equal to $\zero^C$ can be omitted.}
\begin{itemize}
\item The \emph{provenance polynomials semiring for $\semiringVars$} is $\polynomials=(\polynomials,+,\times, 0,1)$ where $\polynomials$ is the set of polynomials with variables from $\semiringVars$ and coefficients from $\mathbb{N}$, with the operations defined as usual.
\item The \emph{Boolean provenance polynomials semiring for $\semiringVars$} is $\boolpolynomials=(\boolpolynomials,+,\times, 0,1)$ where $\boolpolynomials$ is the set of polynomials over variables $\semiringVars$ with Boolean coefficients.
\item The \emph{trio semiring for $\semiringVars$}, denoted $\trio$, is the quotient semiring of $\polynomials$ by $\approx_f$, where $\approx_f$ is the congruence relation defined by $p_1\approx_f p_2$ iff $f(p_1)=f(p_2)$ with $f$ the function that ``drops exponents''. \blue{Recall that the elements of the quotient semiring are the equivalence classes for the congruence relation. In practice, we represent them by polynomials without exponents.}
\item  The \emph{sorp semiring for $\semiringVars$}, denoted $\sorp$, is the quotient semiring of $\polynomials$ by $\approx$, where $\approx$ is the smallest congruence relation on $\polynomials$ that identifies polynomials according to absorption.
\item The \emph{why-provenance semiring for $\semiringVars$} is $\why=(\why, \cup, \Cup, \emptyset,\{\emptyset\})$ where $\why$ consists of the set of all possible sets of subsets of $\semiringVars$ and $\Cup$ denotes pairwise union: $S_1\Cup S_2=\{s_1\cup s_2\mid s_1\in S_1, s_2\in S_2\}$. 
\item The \emph{semiring of positive Boolean functions over $\semiringVars$} is $\posbool=(\posbool, \vee, \wedge, 0, 1)$ where $\posbool$ is the set of classes of equivalent positive Boolean expressions over variables $\semiringVars$ 
(which involve only disjunction, conjunction, and constants $1$ and $0$ for true and false but without any negations). 
We identify elements of \posbool with their irredundant disjunctive normal form. \posbool is isomorphic to the bounded distributive lattice freely generated by~$\semiringVars$.
\item The \emph{lineage semiring for $\semiringVars$} is $\lin=(\lin,\cup,\cup^*, \emptyset^*,\emptyset)$ where $\lin=2^\semiringVars\cup\{\emptyset^*\}$ is the set of all subsets of $\semiringVars$ extended with an element $\emptyset^*$, $\cup$ is the usual union, $S_1\cup^*S_2=S_1\cup S_2$ if $S_1,S_2\neq\emptyset^*$ and $\emptyset^*\cup S=S\cup\emptyset^*=S$, $\emptyset^*\cup^* S=S\cup^*\emptyset^*=\emptyset^*$. Intuitively, $\emptyset^*$ and $\cup^*$ are introduced because taking $\oplus=\otimes=\cup$ and $\zero=\one=\emptyset$ would not satisfy the requirement that in a semiring, $\zero$ is annihilating for $\otimes$. $\lin$ is a semiring such that for every $a,b\neq \zero$, $a\oplus b=a\otimes b$ (denoted by $\oplus\approx\otimes$).
\end{itemize}
We will often write a set of variables $\{x_1,\dots,x_n\}$ as their product, \ie as the \emph{monomial} $x_1\times\dots \times x_n$ (or $x_1\cdots x_n$) and see a set of such sets (\ie an element of \why) as a \emph{sum of monomials}. Similarly we can write, e.g., the formula $x_1\wedge x_2\vee x_3$ as $x_1\times x_2 + x_3$, and represent all provenance expressions from \why, \posbool or \lin as \emph{polynomials}. 
Figure \ref{fig:hierarchy} shows the hierarchy of expressiveness for these provenance semirings. Intuitively, annotating a query result by $x^2+xy\in\polynomials$ expresses that it can be obtained either by using twice a fact annotated by $x$ or once a fact annotated by $x$ and once a fact annotated with $y$, while annotation $x+xy\in\why$ only indicates that the query result can be obtained by using a fact annotated with $x$ or two facts annotated with $x$ and $y$, thus is less informative.} \blue{More formally, $\semiringshort_1$ is considered more expressive than $\semiringshort_2$ if there exists a surjective semiring homomorphism from $\semiringshort_1$ to $\semiringshort_2$.}  
In this work, we mostly focus on the provenance semirings \why, \posbool and \lin and discuss the difficulties that arise for 
more expressive provenance semirings \blue{when they are used for annotating ontologies}. These provenance semirings correspond to notions of query explanation that have long been considered in the 
context of databases: \why corresponds to the \emph{why-provenance}, or \emph{witness basis}, \posbool corresponds to 
the \emph{minimal witness basis}, and \lin corresponds to 
\emph{lineage}~\cite{BKT01:dbprovenance,DBLP:journals/ftdb/CheneyCT09}. 

\begin{example}[Example~\ref{ex:annotated-database} continued]\label{ex:annotated-database-bis}
\blue{Assume that the three facts $\alpha_1$, $\alpha_2$ and $\alpha_3$ are annotated by the variables $x_1$, $x_2$ and $x_3$ respectively. The annotation of the query $q$ is $x_1x_2+x_1x_3$ for all provenance semirings we consider except for \lin, for which it is $x_1x_2x_3$. If we now consider the query $q':=\exists xyz\, \mn{Deity}(x) \wedge \mn{mother}(x,y)\wedge \mn{mother}(x,z)$, the provenance of $q'$ is $x_1x_2^2+x_1x_3^2+2x_1x_2x_3$ in \polynomials, $x_1x_2^2+x_1x_3^2+x_1x_2x_3$ in $\boolpolynomials$ and \sorp, $x_1x_2+x_1x_3+2x_1x_2x_3$ in \trio, $x_1x_2+x_1x_3+x_1x_2x_3$ in \why, $x_1x_2+x_1x_3$ in \posbool, and $x_1x_2x_3$ in \lin.}
\end{example}

\begin{figure}
\begin{center}
\begin{tikzpicture}
[level distance=2.2cm,
level 1/.style={sibling distance=4.5cm},
level 2/.style={sibling distance=4.5cm}]
\node (poly) {{\polynomials}}
	child {node (boolpoly) {{\boolpolynomials}}
		child {node (sorp) {{\sorp}}}
		child {node (why) {{\why}}
			child {node (posbool) {{\posbool}}}
			child {node (lin) {{\lin}}}
			}
		}
	child {node (trio) {{\trio}}};
\draw [thick, ->] (poly) -- (boolpoly) node[near start,left] {\small\plusidem};	
\draw [thick, ->] (poly) -- (trio) node[near start,right] {\small\new{drop exponents}};	

\draw [thick, ->] (boolpoly) -- (sorp) node[near start,left] {\small absorptive};
\draw [thick, ->] (boolpoly) -- (why) node[near start,right] {\small\new{drop exponents}};

\draw [thick, ->] (trio) -- (why) node[near start,right] {\small\plusidem};

\draw [thick, ->] (why) -- (posbool) node[near start,left] {\small absorptive};
\draw [thick, ->] (why) -- (lin) node[near start,right] {\small$\oplus \approx \otimes$};

\draw [thick, ->] (sorp) -- (posbool) node[near start,left] {\small\timesidem};
\end{tikzpicture}
\end{center}
\caption[]{A hierarchy of provenance semirings
	 \cite[Fig.~2]{GreenT17}.\footnotemark 
	\blue{An arrow from $\semiringshort_1$ to $\semiringshort_2$}  
	means that there exists a surjective semiring homomorphism from $\semiringshort_1$ to $\semiringshort_2$\blue{, i.e., that $\semiringshort_1$ is more expressive than $\semiringshort_2$}.}
\label{fig:hierarchy}
\end{figure}
\footnotetext{\new{The original picture by \citet{GreenT17} has a typo: \trio and \why are not $\otimes$-idempotent (since \blue{$(x+y)\times(x+y)=x+y+ 2xy$ in \trio and $(x+y)\times(x+y)=x+y+xy$ in \why}) but rather can be obtained by ``dropping exponents'' (\cf \cite{DBLP:journals/mst/Green11} for the precise definition of \trio).}\label{footonotewhytrionotidem}}

A provenance semiring $\mi{Prov}[\semiringVars]$ \blue{(such that $\semiringVars\subseteq\mi{Prov}[\semiringVars]$)}  \emph{specializes} correctly to a semiring $\semiringshort$, if any \new{function} $\nu : \semiringVars \rightarrow K$ extends uniquely to a ($\omega$-continuous if $\mi{Prov}[\semiringVars]$ and $\semiringshort$ are $\omega$-continuous) semiring homomorphism $h : \mi{Prov}[\semiringVars] \rightarrow K$ \new{such that $h(x)=\nu(x)$ for every $x\in\semiringVars$}, allowing the computations for $\semiringshort$ to factor through the computations for $\mi{Prov}[\semiringVars]$ \shortcite{DBLP:conf/icdt/DeutchMRT14,DBLP:conf/kr/BourgauxBPT22}. 

\begin{example}
If we consider the semirings from Example \ref{ex:semirings} and the provenance semirings in Figure \ref{fig:hierarchy}, only 
\polynomials specializes correctly to the counting semiring $(\mathbb{N}, +, \times, 0, 1)$. 
\blue{Indeed, if $\nu : \semiringVars \rightarrow \mathbb{N}$ is such that $\nu(x)=2$, there is no homomorphism $h : \boolpolynomials \rightarrow \mathbb{N}$ such that $h(x)=2$ (since it would require that $2=h(x)=h(x+x)=h(x)+h(x)=2+2=4$ as $x+x=x$ in $\boolpolynomials$), and similarly, there is no homomorphism $h : \trio \rightarrow \mathbb{N}$ such that $h(x)=2$ (since it would require that $2=h(x)=h(x\times  x)=2\times 2=4$ as $x\times x=x$ in $\trio$). 
}

\sorp and all provenance 
semirings above it specialize correctly to the tropical semiring $(\mathbb{R}^\infty_+, \min , +, \infty, 0)$, the Viterbi semiring 
$([0,1], \max, \times, 0,1)$, and the \L ukasiewicz semiring $([0,1],\max,\star_L,0,1)$; and \posbool along with all provenance 
semirings above it specialize correctly to the fuzzy semiring $([0,1],\max,\min ,0,1)$ and the access control semiring 
$(\{P<C<S<T<0\},\min,\max,0,P)$. 
\end{example}

\begin{remark}
\new{Interestingly, \trio and \why do not specialize correctly to themselves \blue{(if $\ \semiringVars$ contains at least two variables)}: if $\nu(x)=x+y$ and $\nu(z)=z$ for every $z\in\semiringVars\setminus\{x\}$, there is no semiring homomorphism $h$ \blue{such that $h(v)=\nu(v)$ for every $v\in\semiringVars$}  
because it would imply that \blue{$x+y=h(x)=h(x\times x)=h(x)\times h(x)=(x+y)\times(x+y)$ and $(x+y)\times(x+y)\neq x+y$ in \trio and \why (\cf footnote \ref{footonotewhytrionotidem})}.}
\end{remark}

\new{As mentioned before,} for Datalog queries, infinite provenance expressions may be needed. 
\new{For example, the Datalog query that consists of the recursive rule $A(x)\leftarrow A(x)$ and the rule 
$\mn{goal}\leftarrow A(x)$ where $\mn{goal}$ is the output predicate, evaluated 
over the database that contains a single fact $A(a)$ annotated with $x_0\in\semiringVars$, has infinitely many derivation trees and its provenance as defined by \citet{Green07-provenance-seminal} is an infinite sum of $x_0$, which is not an element of $\polynomials$.} 
\blue{Such infinite provenance expressions are expressed using \emph{formal power series}.} 
A formal power series with variables from $\semiringVars$ and coefficients from $C$ is a mapping \blue{$S$ from $\monomials(\semiringVars)$ to $C$, where 
$\monomials(\semiringVars)$ is the set of all monomials over $\semiringVars$. It} 
can be written as a possibly 
infinite sum $S=\Sigma_{\monomial\in\monomials(\semiringVars)}S(\monomial)\monomial$. 
$C\llbracket\semiringVars\rrbracket$ denotes the set of formal power series with variables from $\semiringVars$ 
and coefficients from~$C$.  
\citet{Green07-provenance-seminal} define the 
\emph{Datalog provenance semiring} as $\series$, the commutative $\omega$-continuous semiring of formal power 
series with coefficients from $\Nbb^\infty=\Nbb\cup\{\infty\}$. 
We can obtain a hierarchy similar to that of Figure \ref{fig:hierarchy} \new{for $\omega$-continuous commutative provenance semirings} where \polynomials is replaced by \series, \boolpolynomials is replaced by 
$\boolseries$, 
\new{and $\trio$ is replaced by $\trioseries$ (obtained from $\series$ as \trio is obtained from $\polynomials$).  Note that since we assume that $\semiringVars$ is finite, there is no difference between polynomials and formal power series when considering \why, \sorp, \posbool and \lin. 
There is no hierarchical relationship between \series and \polynomials because there is no \emph{surjective} semiring 
homomorphism from $\polynomials$ to $\series$ nor from $\series$ to $\polynomials$.}

A provenance semiring $\mi{Prov}[\semiringVars]$ is \emph{universal for a \new{class} of semirings} if it specializes correctly to each semiring of this \new{class. This is equivalent to say that $\mi{Prov}[\semiringVars]$ has the universal mapping property for this class over $\semiringVars$, or that $\mi{Prov}[\semiringVars]$ is the free algebra generated by $\semiringVars$ for this class}. 
\new{\citet{Green07-provenance-seminal} showed that $\polynomials$ is universal for commutative semirings and $\series$ is universal for commutative $\omega$-continuous semirings. 
It is also well-known that 
$\posbool$ specializes correctly to every commutative semiring that is \timesidem and absorptive. 
Note that a provenance semiring specializing correctly to another does not mean that the former is more expressive than the latter. For instance, $\polynomials$ specializes correctly to $\series$, but is not more expressive than $\series$.}


\section{\new{A Semiring Provenance Semantics for Description Logics}}\label{sec:prov-definition}

\new{In this section, we define a provenance semantics for $\ELHIbot$ ontologies. We  explain our design choices and the restrictions of our setting. Finally, we investigate whether it captures some well-known semantics for ontologies annotated with specific kinds of information. A more detailed analysis of the properties of our semantics is conducted in Sections~\ref{sec:classical-res-counterparts} and~\ref{sec:expected-properties}.}

\subsection{Provenance Semantics for Annotated Ontologies }\label{sec:prov-ontology}

Let $\semiringshort=\semiring$ be a \new{commutative semiring} and \logic be a DL language 
encompassed by $\ELHIbot$. 
A \emph{\mbox{($\semiringshort$-)}annotated \logic ontology}  is a pair $\Omc^\semiringshort=\tup{\Omc,\lambda}$ where 
$\Omc$ is an \logic ontology and $\lambda$ is an annotation function 
$\lambda:\Omc\mapsto \semiringset\setminus\{\zero\}$. We often treat $\Omc^\semiringshort$ as the set of pairs
$\{(\alpha,\lambda(\alpha))\mid \alpha\in \Omc\}$ and call such pairs $(\alpha,\lambda(\alpha))$ \emph{annotated axioms}. 
\new{We next define the semantics of $\semiringshort$-annotated \logic ontologies before showing that annotations do not impact satisfiability and discussing our design choices.}

\subsubsection{\blue{Basic Requirements}}\label{sec:semanticreq}
\blue{Let us start with a few basic requirements and intuitions that guide our design choices, focusing on entailment and provenance of assertions, as they are closer than GCIs or RIs to the database (Datalog) queries for which the provenance notion that inspires us has been defined.
\begin{itemize}
\item We want to define a model-theoretic semantics for annotated ontologies, as classically done in DL. Specifically, in this semantics 
\begin{itemize}
\item the satisfaction of an annotated axiom by an interpretation should be defined independently from the rest of the ontology;
\item the satisfiability of an annotated ontology should not depend on its annotations; and
\item assertion entailments should not depend on the annotations: Given a $\semiringshort$-annotated ontology $\Omc^\semiringshort$ and an assertion $\alpha$ such that $\Omc\models\alpha$, there should exist some $\elem\in\semiringset$ such that $\Omc^\semiringshort\models (\alpha,\elem)$.
\end{itemize}
\item We want to define the provenance of a DL axiom or query so that it consists in a single semiring element obtained from  its annotations in the annotated ontology models (as provenance in database is obtained from the annotations associated to the query matches or derivation trees).
\item The provenance of a DL axiom or query should be such that joint use of the axioms corresponds to the semiring multiplication while alternative use of the axioms corresponds to the semiring addition. In particular, the following two conditions should hold.
\begin{itemize}
\item ``Irrelevant'' axioms should not influence the provenance of a consequence: Given a $\semiringshort$-annotated ontology $\Omc^\semiringshort$ and an assertion $A(a)$ such that $A$ and $a$ do not occur in $\Omc$, for every assertion $\alpha$, the provenance of $\alpha$ \wrt $\Omc^\semiringshort$ should be the same as its provenance \wrt $\Omc^\semiringshort\cup\{(A(a),\elem)\}$, for any $\elem\in\semiringset$.
\item ``Irrelevant'' semiring elements should not influence the provenance of a consequence, except maybe in a few cases (i.e., trivial consequences): Given a satisfiable ontology $\Omc$ and an assertion $\alpha$, the provenance of $\alpha$ \wrt $\Omc^\semiringshort=\tup{\Omc,\lambda}$ should be the same as its provenance \wrt $\Omc^{\semiringshort'}=\tup{\Omc,\lambda}$ where $\semiringshort'$ is the commutative semiring obtained by extending $\semiringshort$ with an element $\infty_\semiringshort$ such that $\infty_\semiringshort\otimes\elem=\infty_\semiringshort\oplus\elem=\infty_\semiringshort$ for every $\elem\neq\zero$.
\end{itemize}
\end{itemize}
}
\subsubsection{\new{Semantics Definition}}\label{sec:semanticdef}

The semantics of annotated 
ontologies extends the classical notion of
interpretations with annotations. 
A \emph{($\semiringshort$-)annotated interpretation} is a triple
$\Imc=(\Delta^\Imc, K ,\cdot^\Imc)$
where $\Delta^\Imc$ is a non-empty set (the
\emph{domain} of \Imc), 
and $\cdot^\Imc$ maps
\begin{itemize}
\item every $a\in\NI$ to $a^\Imc\in\Delta^\Imc$;
\item every $A\in\NC$ to 
$A^\Imc\subseteq \Delta^\Imc\times K$;
\item every   $R\in\NR$ to 
$R^\Imc\subseteq \Delta^\Imc\times\Delta^\Imc\times K$.
\end{itemize}
We extend
$\cdot^\Imc$ to complex expressions as follows: 
  \begin{align*} 
	(\top)^\Imc = {} & \Delta^\Imc\times \{\one\}; \\ 
	(\bot)^\Imc = {} & \emptyset; \\ 
	(R^-)^\Imc = {} & \{(e,d,\elem)\mid (d,e,\elem)\in R^\Imc\};\\	
	(\exists P.C)^\Imc = {} & \{(d,\elema\otimes \elemb)\mid \exists e\in\Delta^\Imc  
    \text{ s.t. }(d,e,\elema)\in \new{P}^\Imc, (e, \elemb)\in C^\Imc\};
    \\
    (C\sqcap D)^\Imc = {} & \{(d,\elema\otimes \elemb)\mid (d,\elema)\in C^\Imc, (d,\elemb)\in D^\Imc\};
    \\
    \new{(P\sqcap Q)^\Imc={}}  & \new{\{(e,d,\elema\otimes\elemb)\mid (e,d,\elema)\in P^\Imc, (e,d,\elemb)\in Q^\Imc\}}. 
  \end{align*}
The annotated interpretation \Imc \emph{satisfies}:
  \begin{align*}  
    & \pair{C\sqsubseteq D}{\elem}, && \new{\text{if, for all } d\in\Delta^\Imc,\elemb \in K, }\  
    (d,\elemb)\in C^\Imc\text{ implies  } (d,\elem\otimes \elemb)\in D^\Imc\!;
   	\\
	& \pair{P\sqsubseteq Q}{\elem}, &&  \new{\text{if, for all  } d, e\in\Delta^\Imc,\elemb \in K, }\ 
    (d,e,\elemb)\in P^\Imc \text{ implies  } (d,e,\elem\otimes \elemb)\in Q^\Imc\!; \\
	& \pair{P\sqcap Q\sqsubseteq \bot}{\elem}, && \text{if } 
   \new{ (P\sqcap Q)^\Imc}=\emptyset;  \\
& \pair{A(a)}{\elem}, && \text{if } (a^\Imc,\elem)\in A^\Imc\!;  \\ 
&   \pair{R(a,b)}{\elem}, && \text{if } (a^\Imc,b^\Imc,\elem)\in R^\Imc\!.  	
  \end{align*}
Note that since $\bot^\Imc=\emptyset$, the satisfaction of a GCI of the form $(C\sqsubseteq \bot,\elem)$ (\ie $ (d,\elemb)\in C^\Imc$ implies  
$(d,\elem\otimes \elemb)\in \bot^\Imc$) is equivalent to $C^\Imc=\emptyset$. 
\new{An annotated interpretation \Imc is a \emph{model} of the annotated ontology $\Omc^\semiringshort$, denoted $\Imc\models\Omc^\semiringshort$, if
it satisfies all annotated axioms in $\Omc^\semiringshort$;   
\blue{$\Omc^\semiringshort$ is \emph{satisfiable} if it has a model;}
$\Omc^\semiringshort$ \emph{entails} an annotated axiom 
$(\alpha,\elem)$, denoted $\Omc^\semiringshort\models (\alpha,\elem)$,
if $\Imc\models (\alpha,\elem)$ for every model $\Imc$ of~$\Omc^\semiringshort$.}

\new{Some remarks are in \blue{order}. Intuitively, each annotation of a domain element $d\in\Delta^\Imc$ in the interpretation $C^\Imc$ of an $\ELHIbot$ concept $C$ corresponds to some match of the corresponding rooted tree-shaped query $q_C(d)$ in $\Imc$ seen as an annotated database.  We explain in Section~\ref{sec:designchoices} why we choose to allow for multiple annotations rather than using a single annotation that sums over all such matches, \blue{and why we use sets of annotations (so that the same domain element $d\in\Delta^\Imc$ cannot be annotated several times by the same semiring element $\elem$ in the interpretation of a concept) rather than multisets}. 
GCIs propagate these annotations, multiplying them by their own annotation \blue{(following the intuition that the multiplication corresponds to the ``joint use'' of $C\sqsubseteq D$ and the fact that $d$ is in the interpretation of $C$ to obtain that $d$ is in the interpretation of $D$)}. That is why the interpretation of a concept 
name may contain several pairs $(d,\elem)$ with the same domain element $d\in\Delta^\Imc$, allowing for several provenance annotations to be dealt with simultaneously. 
Interpretations of role names behave similarly.
}

\begin{example}\label{ex:running-onto}
Consider the following ontology:
\begin{align*}
\Omc=\{&\mn{Deity}(\mn{Dionysus}), \mn{mother}(\mn{Dionysus},\mn{Semele}),\mn{mother}(\mn{Dionysus},\mn{Demeter}),\\&
 \mn{Deity}(\mn{Demeter}), \mn{father}(\mn{Dionysus},\mn{Zeus}), \mn{Deity}(\mn{Zeus}), \\&
 \exists\mn{parent}.\mn{Deity}\sqsubseteq \mn{Deity}, \mn{mother}\sqsubseteq\mn{parent}, \mn{father}\sqsubseteq\mn{parent}
\}.
\end{align*}
We define the annotated ontologies $\Omc^{\why}=\tup{\Omc,\lambda_\semiringVars}$, $\Omc^{\mathbb{T}}=\tup{\Omc,\lambda_{\mathbb{T}}}$, and $\Omc^{\mathbb{F}}=\tup{\Omc,\lambda_{\mathbb{F}}}$, annotated with variables, costs and truth degrees respectively, as follows.
\begin{center}
\begin{tabular}{l c c c}
& $\lambda_\semiringVars$ & $\lambda_{\mathbb{T}}$ & $\lambda_{\mathbb{F}}$ \\
\mn{Deity}(\mn{Dionysus})& $x_1$ & 1 & 0.9\\
\mn{mother}(\mn{Dionysus},\mn{Semele})& $x_2$ & 5 & 0.8 \\
\mn{mother}(\mn{Dionysus},\mn{Demeter})& $x_3$ & 8 & 0.2\\
\mn{Deity}(\mn{Demeter})& $x_4$ & 1 & 0.9\\
\mn{father}(\mn{Dionysus},\mn{Zeus})& $x_5$ & 2 & 0.9\\
\mn{Deity}(\mn{Zeus})& $x_6$ & 1 & 1\\
$\exists\mn{parent}.\mn{Deity}\sqsubseteq \mn{Deity}$& $y_1$ & 2 & 0.5\\
$\mn{mother}\sqsubseteq\mn{parent}$& $y_2$ & 1 & 1\\
$\mn{father}\sqsubseteq\mn{parent}$& $y_3$ & 1 & 1
\end{tabular}
\end{center}
Let $\Imc $ be a ${\why}$-annotated interpretation such that $\Imc \models\Omc^{\why}$. \new{Recall that ${\why}$, \blue{as all semirings we consider}, is commutative, so that the order of variables in monomials does not matter (we write them below in the lexicographic order).}
\begin{itemize}
\item Since $\Imc \models (\mn{Deity}(\mn{Dionysus}),x_1)$, then $(\mn{Dionysus}^{\Imc },x_1)\in \mn{Deity}^{\Imc }$. 
\item Since $\Imc \models (\mn{mother}(\mn{Dionysus},\mn{Demeter}),x_3)$, then 
	$(\mn{Dionysus}^{\Imc },\mn{Demeter}^{\Imc },x_3)\in \mn{mother}^{\Imc }$. 
	
	It follows that $\Imc \models (\mn{mother}\sqsubseteq\mn{parent}, y_2)$ yields 
	$(\mn{Dionysus}^{\Imc },\mn{Demeter}^{\Imc },x_3y_2)\in \mn{parent}^{\Imc }$. 
	
	Moreover, $\Imc \models (\mn{Deity}(\mn{Demeter}),x_4)$ implies 
	$(\mn{Demeter}^{\Imc },x_4)\in \mn{Deity}^{\Imc }$. 
	
	Hence from 
	$\Imc \models (\exists\mn{parent}.\mn{Deity}\sqsubseteq \mn{Deity}, y_1)$ we obtain 
	$(\mn{Dionysus}^{\Imc },x_3x_4y_1y_2)\in \mn{Deity}^{\Imc }$. 
\item Similarly, since $\Imc \models (\mn{father}(\mn{Dionysus},\mn{Zeus}),x_5)$, 
	$\Imc \models (\mn{Deity}(\mn{Zeus}),x_6)$, 
	$\Imc \models (\mn{father}\sqsubseteq\mn{parent}, y_3)$ and $\Imc \models (\exists\mn{parent}.\mn{Deity}\sqsubseteq \mn{Deity}, y_1)$, it holds that $(\mn{Dionysus}^{\Imc },x_5x_6y_1y_3)\in \mn{Deity}^{\Imc}$. 
\end{itemize}
\blue{Hence, $\Omc^{\why}$ entails $(\mn{Deity}(\mn{Dionysus}), x_1)$, $(\mn{Deity}(\mn{Dionysus}), x_3x_4y_1y_2)$ and $(\mn{Deity}(\mn{Dionysus}), x_5x_6y_1y_3)$.}

We obtain the following results similarly, by considering $\mathbb{T}$-annotated and $\mathbb{F}$-annotated interpretations that are models of $\Omc^{\mathbb{T}}$ and $\Omc^{\mathbb{F}}$ respectively. Recall that the multiplication of the tropical semiring $\mathbb{T}$ is the usual addition and its addition is $\min$, while the multiplication of the fuzzy semiring $\mathbb{F}$ is $\min$ and its addition is $\max$.
\begin{itemize}
\item \blue{$\Omc^{\mathbb{T}}$ entails $(\mn{Deity}(\mn{Dionysus}), 1)$, $(\mn{Deity}(\mn{Dionysus}), 12)$ and $(\mn{Deity}(\mn{Dionysus}), 6).$}

\item \blue{$\Omc^{\mathbb{F}}$ entails $(\mn{Deity}(\mn{Dionysus}), 0.9)$, $(\mn{Deity}(\mn{Dionysus}), 0.2)$ and $(\mn{Deity}(\mn{Dionysus}), 0.5)$}. \qedhere
\end{itemize}
\end{example}

\new{Contrary to the ontology annotations, we allow for $\zero$ to occur in annotated interpretations. \blue{Indeed, if the semiring is not positive and there are some $\elema,\elemb\neq\zero$ such that $\elema\otimes\elemb=\zero$, 
every model $\Imc$ of $\Omc^\semiringshort=\{(A(a),\elema),(A\sqsubseteq B,\elemb)\}$ is such that $(a^\Imc,\elema\otimes\elemb)=(a^\Imc,\zero)\in B^\Imc$, so disallowing $\zero$ would make $\Omc^\semiringshort$ unsatisfiable, which goes against our basic requirement that annotations do not change the satisfiability of the ontology (\cf~Lemma~\ref{lem:relationship-annotated-standard-models})}. We shall however see in Theorem~\ref{th:sem-entailment} that when the semiring is positive, 
\blue{for every $\Omc^\semiringshort$, $\Omc\models A(a)$ iff there exists $\elem\neq\zero$ such that $\Omc^\semiringshort\models(A(a),\elem)$,} 
and similarly for role assertions.} 

The interpretation of $\top$ by $\Delta^\Imc\times \{\one\}$ yields that $A^\Imc\not\subseteq \top^\Imc$ in general, in contrast with the classical semantics of DL, and that $\emptyset\not\models (A\sqsubseteq \top,\one)$. However, interpreting $\top$ by $\Delta^\Imc\times \semiringset$ would ``flood'' the annotations of models of annotated ontologies such that $\top$ occurs in the left-hand side of some GCIs, \blue{as exemplified below.}
\begin{example}
\blue{Consider $\Omc^\semiringshort=\{(A(a),\one), (A\sqsubseteq \exists R,\one),(\exists R.\top\sqsubseteq B,\one)\}$. For every model $\Imc$ of $\Omc^\semiringshort$, $(a^\Imc,\one)\in A^\Imc$, so since $\Imc\models (A\sqsubseteq \exists R,\one)$, there exists $e\in\Delta^\Imc$ such that $(a^\Imc,e,\one\otimes\one)\in R^\Imc$, \ie $(a^\Imc,e,\one)\in R^\Imc$. Since $(a^\Imc,e,\one)\in R^\Imc$ and $(e,\one)\in\top^\Imc$, we obtain that $(a^\Imc,\one)\in (\exists R.\top)^\Imc$. Hence, since $\Imc\models (\exists R.\top\sqsubseteq B,\one)$, it follows that $(a^\Imc,\one)\in B^\Imc$. Thus, $\Omc^\semiringshort\models (B(a),\one)$. Moreover, the annotated interpretation defined by $\Delta^\Imc=\{a\}$, $a^\Imc=a$, $A^\Imc=\{(a,\one)\}$, $B^\Imc=\{(a,\one)\}$, and $R^\Imc=\{(a,a,\one)\}$ is a model of $\Omc^\semiringshort$ so for every $\elem\neq\one$, $\Omc^\semiringshort\not\models (B(a),\elem)$.} 

\blue{In contrast, if we interpret $\top$ by $\Delta^\Imc\times \semiringset$, then for every model $\Imc$ of $\Omc^\semiringshort$ we would obtain that for every $\elem\in \semiringset$, $(a^\Imc,\one\otimes\elem)\in (\exists R.\top)^\Imc$, \ie $(a^\Imc,\elem)\in (\exists R.\top)^\Imc$, so that $(a^\Imc,\one\otimes\elem)\in B^\Imc$, \ie $(a^\Imc,\elem)\in B^\Imc$, and $\Omc^\semiringshort\models (B(a),\elem)$. In particular, this would go against our requirement that ``irrelevant'' semiring elements should not influence the provenance of assertions entailed by a satisfiable ontology (\cf provenance definition below).}
\end{example}

\paragraph{\new{Provenance}} \new{The \emph{provenance annotations} of an axiom $\alpha$ in an annotated interpretation \Imc is the set $\p{\Imc}{\alpha}=\{\elem \mid \Imc\models (\alpha,\elem)\}$. 
If the semiring $\semiringshort$ is complete (or if $\semiringshort$ is $\omega$-complete and $\semiringset$ is countable),} we define the \emph{provenance} of $\alpha$ \wrt $\Omc^\semiringshort$ 
as 
\new{
\begin{align}\label{eq:provenancesum}
\Pmc(\alpha, \Omc^\semiringshort) := 
\bigoplus_{\Omc^\semiringshort\models (\alpha,\elem)} \elem. 
\end{align}}%
\new{Thanks to the conditions \blue{imposed above} on~$\semiringshort$, the sum in Equation~\eqref{eq:provenancesum} is well-defined even if there are infinitely many $\elem\in\semiringset$ such that 
$\Omc^\semiringshort\models (\alpha,\elem)$ (note that there can be uncountably many such $\elem$ only if $\semiringset$ is uncountable). Moreover, if there is no 
\blue{$\elem\in\semiringset$ such that 
$\Omc^\semiringshort\models (\alpha,\elem)$}
(\ie the sum is empty), then $\Pmc(\alpha, \Omc^\semiringshort)=\zero$.} 

\begin{remark}\label{rem:unsat-onto}
	\blue{It holds that $\Pmc(\alpha, \Omc^\semiringshort) =\bigoplus_{\elem\in\bigcap_{\Imc\models \Omc^{\semiringshort}}\p{\Imc}{\alpha}} \elem$. 
		Indeed, for every $\elem\in\semiringset$, $\elem\in \bigcap_{\Imc\models \Omc^{\semiringshort}}\p{\Imc}{\alpha}$ iff $\Imc\models (\alpha,\elem)$ for every model $\Imc$ of $\Omc^{\semiringshort}$, which is equivalent to $\Omc^{\semiringshort}\models (\alpha,\elem)$. 
		Note that when $\Omc^\semiringshort$ is unsatisfiable, \ie has no model, $\Omc^\semiringshort\models (\alpha,\elem)$ holds for every $\elem\in\semiringset$ and $\Pmc(\alpha, \Omc^\semiringshort)= \bigoplus_{\elem\in\semiringset}\elem$.}
\end{remark}

\begin{example}[Example~\ref{ex:running-onto} continued]
\blue{We have seen in Example~\ref{ex:running-onto} that $\Omc^{\why}$ entails $(\mn{Deity}(\mn{Dionysus}), x_1)$, $(\mn{Deity}(\mn{Dionysus}), x_3x_4y_1y_2)$ and $(\mn{Deity}(\mn{Dionysus}), x_5x_6y_1y_3)$. Moreover, we can verify that there is no other element $\elem\in\why$ such that $\Omc^{\why}$ entails $(\mn{Deity}(\mn{Dionysus}), \elem)$ by considering the model $\Imc$ of $\Omc^{\why}$ defined by $\Delta^\Imc=\{dio,dem,sem,zeu\}$, $\mn{Dionysus}^\Imc=dio$, $\mn{Demeter}^\Imc=dem$, $\mn{Semele}^\Imc=sem$ and $\mn{Zeus}^\Imc=zeu$ and} 
\begin{align*}
\blue{\mn{Deity}^\Imc=}&\blue{\{(dio,x_1), (dio,x_3x_4y_1y_2), (dio,x_5x_6y_1y_3), (dem,x_4), (zeu,x_6)\},}\\
\blue{\mn{mother}^\Imc=}&\blue{\{(dio,sem,x_2), (dio, dem,x_3)\},}\\
\blue{\mn{father}^\Imc=}&\blue{\{(dio,zeu,x_5)\},}\\
\blue{\mn{parent}^\Imc=}&\blue{\{(dio,sem,x_2y_2), (dio, dem,x_3y_2),(dio,zeu,x_5y_3)\}.}
\end{align*}
\blue{Hence, the provenance of $\mn{Deity}(\mn{Dionysus})$ \wrt $\Omc^{\why}$ is }
$$\blue{\Pmc(\mn{Deity}(\mn{Dionysus}), \Omc^{\why})=x_1+x_3x_4y_1y_2+x_5x_6y_1y_3.}$$
\blue{Similarly, we can show that }
\blue{$\Pmc(\mn{Deity}(\mn{Dionysus}), \Omc^{\mathbb{T}})=1$, 
and 
$\Pmc(\mn{Deity}(\mn{Dionysus}), \Omc^{\mathbb{F}})=0.9.$} 
\end{example}

\begin{remark}[Unsatisfiable concept/role]\label{rem:unsat-left}
If an $\ELHIbot$ concept $C$ is unsatisfiable \wrt $\Omc$, \ie such that for every model $\Imc$ of $\Omc$, 
$C^\Imc=\emptyset$, then for every $\semiringshort$-annotated version 
$\Omc^\semiringshort=\tup{\Omc,\lambda}$, $\ELHIbot$ concept $D$ and $\elem\in\semiringset$, it holds that 
$\Omc^\semiringshort\models (C\sqsubseteq D,\elem)$, so that 
$\Pmc(C\sqsubseteq D, \Omc^\semiringshort)=\bigoplus_{\elem\in\semiringset}\elem$. The same holds for unsatisfiable roles. 
\end{remark}
\blue{Remark~\ref{rem:unsat-left} shows that our notion of provenance may not be informative (in the sense that the provenance is not connected to the annotations of the ontology axioms that yield the consequence) in some cases, in particular when some axioms trivially hold (which are, arguably, not the kind of consequences one is  generally interested in). We   discuss in details the properties of this notion of provenance in subsequent sections but can already illustrate below that some axioms that are usually regarded as tautologies may not be entailed in annotated ontologies, hence have provenance $\zero$. 
	 }
\begin{example}\label{ex:tautologies-with-prov-zero}
\blue{Consider the empty ontology $\Omc=\emptyset$ and the Viterbi semiring $\mathbb{V}=([0,1], \max, \times, 0,1)$. We show that $\Pmc(A\sqsubseteq \top,\Omc^{\mathbb{V}})=0$ and $\Pmc(A\sqcap B\sqsubseteq A,\Omc^{\mathbb{V}})=0$ for every $B\in\NC$ (in particular, $\Pmc(A\sqcap A\sqsubseteq A,\Omc^{\mathbb{V}})=0$), \ie these tautologies have the same zero-provenance as axioms that are not entailed by $\Omc$. Indeed, the annotated interpretation $\Imc$ defined by $\Delta^\Imc=\{a\}$ and $B^\Imc=\{(a,0.5)\}$ for every $B\in\NC$ is a model of $\Omc^{\mathbb{V}}$ such that  for any $\elem\in[0,1]$, $\Imc\not\models (A\sqsubseteq\top,\elem)$ (since $\top^\Imc=\{(a,1)\}$ so that $\{(a,\elem\times 0.5)\}\not\subseteq \top^\Imc$) and  $\Imc\not\models (A\sqcap B\sqsubseteq A,\elem)$ (since $(A\sqcap B)^\Imc=\{(a,0.25)\}$ and $\{(a,\elem\times 0.25)\}\not\subseteq A^\Imc$).}
\end{example}
\blue{Note, however, that for every $\semiringshort$-annotated interpretation $\Imc$ and $\ELHIbot$ concept $C$, $\Imc\models (C\sqsubseteq C, \one)$, so that these ``most basic'' tautologies always have a non-zero provenance.}

\subsubsection{\new{Satisfiability}}
The next lemma shows how annotated and non-annotated models are related, and that satisfiability does not depend on the annotations.

\begin{definition}
A (classical) interpretation $\Imc=(\Delta^\Imc,\cdot^\Imc)$ and a $\semiringshort$-annotated interpretation 
$\Imc^\semiringshort=(\Delta^{\Imc^\semiringshort},K,\cdot^{\Imc^\semiringshort})$ 
\emph{coincide on their non-annotated part} iff (i)~$\Delta^\Imc=\Delta^{\Imc^\semiringshort}$;
(ii)~for all $a\in\NI$, $a^\Imc=a^{\Imc^\semiringshort}$; 
(iii)~for all $A\in\NC$ and $d\in\Delta^\Imc$, it holds that $d\in A^\Imc$ iff there exists $\elem\in\semiringset$ such that 
	$(d,\elem)\in A^{\Imc^\semiringshort}$; and 
(iv)~for all $R\in\NR$ and $d,e\in\Delta^\Imc$, it holds that $(d,e)\in R^\Imc$ iff there exists $\elem\in\semiringset$ such 
	that $(d,e,\elem)\in R^{\Imc^\semiringshort}$. 
\end{definition}

\begin{restatable}{lemma}{lemrelationshipannotatedstandardmodels}
\label{lem:relationship-annotated-standard-models}
Let $\Omc^\semiringshort=\tup{\Omc,\lambda}$ be a $\semiringshort$-annotated ontology. The following claims hold.
\begin{enumerate}[(i)]
\item For each model $\Imc$ of $\Omc$, there exists a model $\Imc^\semiringshort$ of $\Omc^\semiringshort$ such that $\Imc$ and $\Imc^\semiringshort$ coincide on their non-annotated~part. 
\label{claim-add-annot}
\item For each model $\Imc^\semiringshort$ of $\Omc^\semiringshort$ there exists a model $\Imc$ of $\Omc$ such that $\Imc$ and $\Imc^\semiringshort$ coincide on their non-annotated~part. 
\label{claim-remove-annot}
\end{enumerate}
\end{restatable}
\begin{proof}
\new{(i) Given a model $\Imc$ of $\Omc$, we obtain a model $\Imc^\semiringshort$ of $\Omc^\semiringshort$ as follows: $a^{\Imc^\semiringshort}=a^\Imc$ for every $a\in\NI$, $A^{\Imc^\semiringshort}=A^\Imc\times\semiringset$ for every $A\in\NC$ and $R^{\Imc^\semiringshort}=R^\Imc\times\semiringset$  for every $R\in\NR$. 
\noindent (ii) Given a model $\Imc^\semiringshort$ of $\Omc^\semiringshort$, we obtain a model $\Imc$ of $\Omc$ by dropping the annotations from $\Imc^\semiringshort$.}
\end{proof}

\subsubsection{\new{Provenance Design Choices}}\label{sec:designchoices}
\new{One could think of many other ways of defining provenance for description logics. We explain the choices we made here.}
\new{\paragraph{Use of annotated models} The most established notion of provenance for Datalog queries is based on proof trees \cite{Green07-provenance-seminal,DBLP:conf/icdt/DeutchMRT14}, and one could try to extend such definition to DLs instead of proposing a definition based on annotated models. 
Defining proof trees in DL in a way that leads to a reasonable notion of provenance is however far from being straightforward. In particular, the set of \emph{tree proofs for $\Omc\models\alpha$} defined by \citet{DBLP:conf/lpar/AlrabbaaBBKK20,DBLP:conf/cade/AlrabbaaBBKK21} contains all trees consisting of a single hyperedge $(\Omc',\alpha)$ with $\Omc'\subseteq\Omc$ such that $\Omc'\models\alpha$, hence using this set to define provenance would \blue{go against our basic requirement that ``irrelevant'' axioms should not influence the provenance (except in the case of }
\timesidem and absorptive semirings, for which it will amount to consider only minimal such $\Omc'$). To obtain meaningful tree proofs, \citet{DBLP:conf/lpar/AlrabbaaBBKK20,DBLP:conf/cade/AlrabbaaBBKK21} thus consider only those that can be found in a \emph{derivation structure produced by a deriver}, which defines the class of allowed inference steps. 
However, a semantics for annotated DLs based on such tree proofs would depend on the choice of the deriver, which does not seem a proper definition for a \emph{semantics} and is not in line with the way semantics of existing annotated DLs (such as fuzzy DLs) are defined.
} 
\new{\paragraph{Definition of annotated models} One could think of many options for the definition of annotated interpretations.} 
First, one could interpret complex concepts by, e.g., tuples of the form $(d,\bigoplus_{(d,\elem)\in C^\Imc, (d,\elem')\in D^\Imc}\elem\otimes\elem')$ for $(C\sqcap D)^\Imc$ (\ie using the ``relational database provenance'' of the query $q_{C\sqcap D}(x)=C(x)\wedge D(x)$ in $\Imc$ to obtain the annotations).  However, \blue{``non-minimal models''} 
would then make most entailments of annotated assertions fail. 
\blue{For example, consider $\Omc^\semiringshort=\{(A(a),\elema), (B(a),\elemb), (A\sqcap B\sqsubseteq C,\one)\}$. Both $\Imc$ defined by $A^\Imc=\{(a^\Imc,\elema)\}$, $B^\Imc=\{(a^\Imc,\elemb)\}$ and $C^\Imc=\{(a^\Imc,\elema\otimes\elemb)\}$ and $\Jmc$ defined by $A^\Jmc=\{(a^\Jmc,\elema),(a^\Jmc,\epsilon)\}$, $B^\Jmc=\{(a^\Jmc,\elemb)\}$ and $C^\Jmc=\{(a^\Jmc,\elema\otimes\elemb\oplus\epsilon\otimes\elemb)\}$ would be models of $\Omc^\semiringshort$ under this definition. Hence, if $\elema\otimes\elemb\oplus\epsilon\otimes\elemb\neq \elema\otimes\elemb$, $\Omc^\semiringshort$ would not entail $C(a)$ annotated with any semiring element, which goes against our basic requirement that entailment of assertions does not depend on the annotations.}

\new{A workaround to avoid this problem while using a single semiring annotation per tuple of domain elements 
has been considered in the context of Datalog queries, in which ~\citet{DBLP:conf/kr/BourgauxBPT22} proposed two alternative ways of 
defining provenance semantics based on annotated models. One (called \emph{set-annotated model-based semantic}) is 
similar to the one we choose here (\cf Lemma~\ref{th:sem-cons-datalog-SAM}) while the other (called \emph{annotated 
model-based semantics}) is defined on a restricted class of semirings where every set of semiring elements has a \emph{greatest lower bound}. This second semantics requires that each tuple in an interpretation is annotated by a unique semiring element, with the condition that, e.g., $\Imc\models (A(a),\elem)$ if $(a^\Imc,\chi)\in A^\Imc$ for some $\chi\geq \elem$ and $\Imc\models (C\sqsubseteq D,\elem)$ if $(d,\elem')\in C^\Imc$ implies $(d,\chi)\in D^\Imc$ for some $\chi\geq \elem\otimes\elem'$. The provenance of $\alpha$ is then defined as the infimum of $\elem$ such that $\Imc\models (\alpha,\elem)$ over all models $\Imc$ of $\Omc^\semiringshort$. However, this definition is more complex and it has been shown for the case of Datalog that both semantics coincide when $\oplus$ is idempotent and suffer from undesirable behaviors when $\oplus$ is not idempotent.}

\new{We shall indeed see that our semantics may have some counter-intuitive behaviors when $\oplus$ is not idempotent. 
This is due to the use of sets of annotated tuples to interpret concepts and roles, which does not allow us to account for 
the fact that the same semiring element can be used multiple times: e.g., if 
$\Omc^\semiringshort=\{(R(a,b),\elem),(R(a,c),\elem),(\exists R.\top\sqsubseteq A,\one)\}$, then 
$\Pmc(A(a),\Omc^\semiringshort)=\elem$ while we could arguably expect it to be $\elem\oplus\elem$ to take into account 
both assertions, \blue{since the entailment of $A(a)$ from $\Omc$ is equivalent to $a$ being an answer to the query $A(x) \vee \exists y\, R(x,y)$ evaluated over the assertions of $\Omc$ and the database provenance of this query is $\elem\oplus\elem$} (we will come back to this question in more details in Section~\ref{sec:expected-properties}). 
A natural way to obtain the desired provenance in this specific example would be to use \emph{multisets} instead of sets to interpret concepts and roles (so that $(\exists R.\top)^\Imc$ would be defined by the multiset $\{(d,\chi)\mid (d,e,\chi)\in R^\Imc\}$ and  contain $(a^\Imc,\elem)$ at least twice in a model of $\Omc^\semiringshort$), and set conditions such as $\Imc\models (C\sqsubseteq D,\elem)$ means that if $(d,\elem')$ occurs $k$ times in $C^\Imc$, then $(d,\elem\otimes\elem')$ must occur at least $k$ times in $D^\Imc$. 
\blue{However, using multisets would not solve this issue in all cases: if $\Omc^\semiringshort=\{(B_1(a),\elem),(B_2(a),\elem), (B_1\sqsubseteq C,\one), (B_2\sqsubseteq C,\one)\}$, one would obtain a provenance of $\elem$ for $C(a)$ while the database provenance of the query $C(a)\vee B_1(a)\vee B_2(a)$ over the assertions of $\Omc^\semiringshort$ is $\elem\oplus\elem$. To take into account the two assertions, we would need to require that $(B_1\sqsubseteq C,\one)$ and $(B_2\sqsubseteq C,\one)$ are satisfied by $\Imc$ iff $C^\Imc$ contains the (multiset) union of $B_1^\Imc$ and $B_2^\Imc$. However, this goes against our basic requirement that the satisfaction of an axiom by an interpretation is independent from potential other axioms.} We thus choose to stick with sets which allow for simpler definitions and a homogeneous treatment of the above examples.
}

Finally, one could require that annotations in interpretations are closed under $\oplus$, \eg that $(c,\elem)\in A^\Imc$ and $(c,\elem')\in A^\Imc$ implies that $(c,\elem\oplus\elem')\in A^\Imc$. However, \blue{this would cause problems when $\semiringshort$ is not \plusidem (\eg $\Omc^{\mathbb{N}^\infty}=\{(A_1(a),1), (A_2(a),2), (A_1\sqsubseteq A,1), (A_2\sqsubseteq A,1)\}$ would entail $(A(a),n)$ for every $n\in\mathbb{N}$, so that the provenance of $A(a)$ \wrt $\Omc^{\mathbb{N}^\infty}$ would be $\infty$), and} this would lead to arguably undesirable entailments \blue{(or lack of desirable entailments, depending on how the semantics of complex concepts is defined)} 
even if $\oplus$ is required to be idempotent\blue{, as illustrated in Example~\ref{ex:interpretationclosedunderplus}}. 
\begin{example}\label{ex:interpretationclosedunderplus}
\blue{Consider the following \why-annotated ontology
\begin{align*}
\Omc^{\why}=\{ & (A_1(a),x_1), (A_2(a),x_2), (A_1\sqsubseteq A,y_1), (A_2\sqsubseteq A,y_2), (A\sqsubseteq \exists R,y),\\
	&(\exists R^-.\top\sqsubseteq B_1,z),  (\exists R^-.\top\sqsubseteq B_2,u), (B_1\sqcap B_2\sqsubseteq C,v), (\exists R.C\sqsubseteq D,w)\}.
\end{align*}
Let us first show that under our semantics, 
$$\Pmc(D(a),\Omc^{\why})=x_1y_1yzuvw+x_2y_2yzuvw=(x_1y_1+x_2y_2)yzuvw,$$ 
witnessing as expected (in line with the intuition that in the why-provenance, monomials correspond to ``witnesses'' for a consequence) the two possible ways to obtain $D(a)$, using either $A_1(a)$ and $A_1\sqsubseteq A$ or $A_2(a)$ and $A_2\sqsubseteq A$, together with all the other GCIs. Recall that in \why, monomials represent subsets of $\semiringVars$ so we can drop the exponents in monomials.} 
\blue{Let $\Imc$ be a model of $\Omc^{\why}$. For $i\in\{1,2\}$, we have that $(a^\Imc,x_i)\in A_i^\Imc$ and $(a^\Imc,x_iy_i)\in A^\Imc$. Since $\Imc\models  (A\sqsubseteq \exists R,y)$, there must exist $e_i\in\Delta^\Imc$ such that $(a^\Imc, e_i, x_iy_iy)\in R^\Imc$. Hence $(e_i, x_iy_iy)\in (\exists R^-.\top)^\Imc$ so $(e_i, x_iy_iyz)\in B_1^\Imc$ and $(e_i, x_iy_iyu)\in B_2^\Imc$, so that $(e_i, x_iy_iyzu)\in (B_1\sqcap B_2)^\Imc$. It follows that $(e_i, x_iy_iyzuv)\in C^\Imc$. Hence 
$(a^\Imc, x_iy_iyzuv)\in (\exists R.C)^\Imc$ so $(a^\Imc, x_iy_iyzuvw)\in D^\Imc$. We thus have $\Omc^{\why}\models (D(a), x_iy_iyzuvw)$ for $i\in\{1,2\}$. Moreover, the following annotated interpretation $\Jmc$ shows that there is no other element $\elem\in\why$ such that $\Omc^{\why}\models (D(a), \elem)$.
\begin{align*}
A_1^\Jmc=&\{(a^\Jmc,x_1)\} && B_1^\Jmc=\{(e_1,x_1y_1yz),(e_2,x_2y_2yz)\}
\\
 A_2^\Jmc=&\{(a^\Jmc,x_2)\} && B_2^\Jmc=\{(e_1,x_1y_1yu),(e_2,x_2y_2yu)\}
\\
A^\Jmc=&\{(a^\Jmc,x_1y_1),(a^\Jmc,x_2y_2)\} && C^\Jmc=\{(e_1,x_1y_1yzuv),(e_2,x_2y_2yzuv)\}
\\
R^\Jmc=&\{(a^\Jmc,e_1,x_1y_1y),(a^\Jmc,e_2,x_2y_2y) \}&&D^\Jmc=\{(a^\Jmc,x_1y_1yzuvw),(a^\Jmc,x_2y_2yzuvw)\}
\end{align*}
Now, if we require that interpretations are closed under $+$, we would obtain that for every annotated model $\Imc$ of $\Omc^{\why}$, $(a^\Imc,x_1y_1)\in A^\Imc$, $(a^\Imc,x_2y_2)\in A^\Imc$, and $(a^\Imc,x_1y_1+x_2y_2)\in A^\Imc$. 
If we keep the definitions of complex concepts interpretations $(\exists P.E)^\Imc$ and $(E\sqcap F)^\Imc$ as they are, by applying the same steps as above to $(a^\Imc,x_1y_1+x_2y_2)\in A^\Imc$, we additionally obtain that there exists $(a^\Imc,e, (x_1y_1+x_2y_2)y)\in R^\Imc$, $(e,(x_1y_1+x_2y_2)yz)\in B_1^\Imc$ and $(e, (x_1y_1+x_2y_2)yu)\in B_2^\Imc$.  Thus, $(e, (x_1y_1+x_2y_2)yz\times (x_1y_1+x_2y_2)yu)\in (B_1\sqcap B_2)^\Imc$, \ie $(e, (x_1y_1+x_2y_2+x_1y_1x_2y_2)yzu)\in (B_1\sqcap B_2)^\Imc$, and $(e, (x_1y_1+x_2y_2+x_1y_1x_2y_2)yzuv)\in C^\Imc$. Thus, $(a^\Imc, (x_1y_1+x_2y_2)y\times(x_1y_1+x_2y_2+x_1y_1x_2y_2)yzuv)\in(\exists R.C)^\Imc$, \ie $(a^\Imc,(x_1y_1+x_2y_2+x_1y_1x_2y_2)yzuv)\in (\exists R.C)^\Imc$, which yields $(a^\Imc,(x_1y_1+x_2y_2+x_1y_1x_2y_2)yzuvw)\in D^\Imc$. Hence the provenance of $D(a)$ \wrt $\Omc^{\why}$ would be $$x_1y_1yzuvw+x_2y_2yzuvw+(x_1y_1+x_2y_2+x_1y_1x_2y_2)yzuvw=x_1y_1yzuvw+x_2y_2yzuvw+x_1y_1x_2y_2yzuvw.$$ 
The monomial $x_1y_1x_2y_2yzuvw$ is not in line with the notion of witness used in the why-provenance: intuitively, $D(a)$ only follows from the fact that $a$ has an $R$-successor, which has two independent causes, $A_1(a)$ and $A_2(a)$, so there is no way to use both assertions in a witness for the entailment of $D(a)$.
}

\blue{To avoid the entailment of $(D(a),x_1y_1x_2y_2yzuvw)$, we can alternatively define the interpretation of the complex concept $\exists P.E$ as the \emph{closure by $\oplus$} of $\{(d,\elem\otimes\elem')\mid \exists e\in\Delta^\Imc\text{ s.t. }(d,e,\elem)\in P^\Imc, (e,\elem')\in E^\Imc\}$, and similarly for conjunction, so that the interpretation $\Jmc'$ defined from the above $\Jmc$ by adding $(a^\Jmc, x_1y_1+x_2y_2)$ to $A^\Jmc$ and $(a^\Jmc, x_1y_1yzuvw+x_2y_2yzuvw)$ to $D^\Jmc$ would be a model of $\Omc^{\why}$ (since, in particular, we would have $(\exists R)^{\Jmc'}=\{(a^\Jmc,x_1y_1y),(a^\Jmc,x_2y_2y), (a^\Jmc,x_1y_1y+x_2y_2y)\}$ thanks to the closure of $(\exists R)^{\Jmc'}$ by $+$). However, we next show that this definition would go against our basic requirement that assertion entailment does not depend on the annotations.} 
\blue{Indeed, consider the following \why-annotated ontology
\begin{align*}
\Omc^{\why}=\{ & (A(a),x+y), (B(a),z), (A\sqsubseteq \exists R,1),
	(\exists R^-.B\sqsubseteq C,1), (\exists R.C\sqsubseteq D,1)\}.
\end{align*}
Clearly, $\Omc\models D(a)$. However, if we consider annotated interpretations closed under $\oplus$ such that the interpretations of complex concepts are also defined with the closure operation, $\Jmc$ and $\Hmc$ below would both be models of $\Omc^{\why}$, so there would be no $\elem\in\why$ such that $(a^\Imc,\elem)\in D^\Imc$ for every model $\Imc$ of $\Omc^{\why}$. 
\begin{align*}
A^\Jmc=&\{(a^\Jmc, x+y)\}&& A^\Hmc=\{(a^\Hmc, x+y)\}
\\
B^\Jmc=&\{(a^\Jmc, z)\} &&B^\Hmc=\{(a^\Hmc, z)\} 
\\
R^\Jmc=&\{(a^\Jmc, e, x+y)\}&&R^\Hmc=\{(a^\Hmc, e_1, x),(a^\Hmc, e_2, y)\}
\\ 
C^\Jmc=&\{(e, (x+y)z)\}&&C^\Jmc=\{(e_1, xz),(e_2, yz)\}
\\
D^\Jmc=&\{(a^\Jmc, (x+y+xy)z)\}&&D^\Hmc=\{(a^\Hmc, xz),(a^\Hmc, yz),(a^\Hmc, (x+y)z)\}
\end{align*}}
\end{example}

\subsubsection{\new{Choice of DL Language}}\label{sec:syntacticRestrictions}

\new{We explain here the syntactic restrictions we impose on $\ELHIbot$ and the difficulties of extending the semantics to \blue{more expressive languages}.}

\paragraph{Syntactic restrictions}
We \new{start with} the restrictions on the form of the right-hand side of the $\ELHIbot$ GCIs \new{(namely, GCIs have right-hand side $D::= A\mid\exists P.\top\mid \bot$, forbidding to use a conjunction or qualified role existential restriction contrary to what is usually allowed in the \EL family)}. 
Example~\ref{ex:conj-right} illustrates \new{a counter-intuitive behavior} when a conjunction occurs on the right-hand side of a GCI, \new{if we extend the semantics to GCI with conjunctions in the right-hand side as expected: $\Imc\models (C\sqsubseteq C_1\sqcap C_2,\elem)$ if $(d,\elemb)\in C^\Imc$ implies  $(d,\elem\otimes \elemb)\in (C_1\sqcap C_2)^\Imc$.} 
 Qualified role restrictions lead to the same kind of behavior as they can be seen as a kind of implicit conjunction. 
 \begin{example}\label{ex:conj-right}
 \new{Consider $\Omc^{\semiringshort}=\{(A\sqsubseteq B\sqcap C,\elema), (A(a),\elemb)\}$ where $\semiringshort$ is a commutative \mbox{\timesidem} semiring}. The following interpretations 
which interpret  $a$ by itself are models of~$\Omc^{\semiringshort}$:
\begin{equation*}
\begin{array}{l@{\ }l@{\ }l}
A^{\Imc_1}=\{(a,\elemb)\}, &\quad B^{\Imc_1}=\{(a,\elema\otimes \elemb)\}, \quad & C^{\Imc_1}=\{(a,\elema\otimes \elemb)\}\\
A^{\Imc_2}=\{(a,\elemb)\}, &\quad  B^{\Imc_2}=\{(a,\elema)\}, \quad & C^{\Imc_2}=\{(a,\elemb)\}
\\
A^{\Imc_3}=\{(a,\elemb)\}, &\quad  B^{\Imc_3}=\{(a,\one)\}, \quad & C^{\Imc_3}=\{(a,\elema\otimes \elemb)\}
\end{array}
\end{equation*}
Since the semantics does not provide a unique way to ``split'' the semiring element $\elema\otimes \elemb$ between the two elements of the conjunction, if $\elema\neq\one$ or $\elemb\neq\one$, then $\Omc^{\semiringshort}\not\models (B(a), \elemc)$ for any $\elemc\in\semiringset$, 
\blue{which goes against our basic requirement that assertion entailment should not depend on the annotations. }
\end{example}
\new{A similar assumption has been made in previous work on provenance for knowledge bases. In particular, the most similar works considered either DL-Lite, which does not allow for conjunction or qualified role restriction \cite{provenance-DLLite}, or $\ELHr$ with the same syntactic restriction \cite{provenance-ELHr}.  In the same way,~\citet{DBLP:conf/kr/BourgauxBPT22} defined several notions of provenance for Datalog assuming that rules are normalized to have only one atom in the head. Dealing with conjunction in the head has also proven difficult in other contexts related to annotated databases or knowledge bases. For example,~\citet{DBLP:conf/lics/HernichK17} defined two bag semantics in the context of data exchange and proved that in the case where the mappings do not have existentially quantified variables in the head, query answering under both semantics is in \PTime \wrt data complexity when mappings have only one atom in the head and becomes \coNP-complete when two atoms are allowed. \citet{DBLP:journals/ai/NikolaouKKKGH19} defined bag semantics for ontology-based data access, where the assertions of a DL-Lite ontology are created via mappings from a database, only for the case where the mappings have one atom in the head.
Moreover, the decision algorithm for fuzzy \EL by~\citet{BoPe-DL13} is only correct if conjunctions on the right are disallowed and the undecidability results for the cases where $\otimes$ is not
idempotent \cite{BoDP-AIJ15} rely heavily on the right-hand side conjunctions and existential restrictions. 
Finally, in the context of explanations for DLs,~\citet{DBLP:journals/ai/PenalozaS17} showed that even if existential
restrictions are disallowed, conjunctions on the right-hand side of GCIs increase the complexity of computing justifications.
}

One could argue that it would be better to define the semantics so that only $\Imc_1$ was a model of $\Omc^{\semiringshort}$ in Example~\ref{ex:conj-right}, instead of restricting the language as we do. We explain next why this is not so simple. 

One possibility is to change the definition of satisfaction of a GCI by an interpretation such that $\Imc\models (A\sqsubseteq B\sqcap C,\elem)$ iff for every $(d,\elemb)\in A^\Imc$, then $(d,\elem\otimes\elemb)\in B^\Imc$ and $(d,\elem\otimes\elemb)\in C^\Imc$, and similarly for qualified role restrictions. 
This approach leads to a counter-intuitive behavior. For instance if $\Omc^{\semiringshort}=\emptyset$, then $\Omc^{\semiringshort}\not\models (A\sqcap B\sqsubseteq A\sqcap B, \one)$, since there is a model $\Imc$ of $\Omc^{\semiringshort}$ such that $A^\Imc=\{(a,\elema)\}$ and $B^\Imc=\{(a,\elemb)\}$, so that $(A\sqcap B)^\Imc=\{(a, \elema\otimes\elemb)\}$, and $(a, \elema\otimes\elemb)\notin A^\Imc$, $(a, \elema\otimes\elemb)\notin B^\Imc$. 
In contrast, our definition of satisfaction ensures that for every interpretation $\Imc$ and concept $C$, $\Imc\models (C\sqsubseteq C, \one)$. 

Another possibility is to modify the interpretation of conjunctions and qualified role restrictions such that
\begin{align*}
(C\sqcap D)^\Imc & {} =\{(d,\elem)\mid (d,\elem)\in C^\Imc, (d,\elem)\in D^\Imc\} \\ 
\text{and }(\exists R.C)^\Imc & {} = \{(d,\elem)\mid \exists e\in\Delta^\Imc    \text{ s.t. }(d,e,\elem)\in R^\Imc, (e, \elem)\in C^\Imc\}. 
\end{align*}
However, this would lead to the loss of many natural and wanted entailments, for example 
\blue{\[\{(A(a),\elema), (B(a),\elemb), (A\sqcap B\sqsubseteq C,\one)\}\models (C(a), \chi)\] would not hold for any $\chi\in\semiringset$, going against our requirement that assertion entailment does not depend on the annotations.}

Hence, restricting the syntax to prevent conjunctions on the right and defining the semantics as usual in DLs seems to be the most natural way of handling provenance in DL languages with conjunction. 
\blue{Compared to \EL ontologies in normal form,} the main restriction in our language is the avoidance of qualified existential restrictions on the right-hand side.

\paragraph{\blue{Semantics of other DLs}} 
\new{It is not clear how to extend the semantics to a DL language featuring disjunction ($C \sqcup D$) and negation ($\neg C$) such as $\mathcal{ALC}$. Indeed, while it seems natural to define $(C \sqcup D)^\Imc=C^\Imc\cup D^\Imc$, so that, e.g., $\{(C\sqcup D\sqsubseteq A, \elem)\}$ is equivalent to $\{(C\sqsubseteq A, \elem),(D\sqsubseteq A, \elem)\}$, and $\Imc\models (A\sqsubseteq C \sqcup D, \elem)$ iff $(e,\elem')\in A^\Imc$ implies that either $(e,\elem\otimes\elem')\in C^\Imc$ or $(e,\elem\otimes\elem')\in D^\Imc$, it is unclear how to define $(\neg C)^\Imc$. In particular, we would like that \blue{$(\neg(A\sqcap B))^\Imc=(\neg A\sqcup \neg B)^\Imc$, or that} $(A\sqsubseteq C \sqcup D, \elem)$ is equivalent to $(A\sqcap \neg C\sqsubseteq D,\elem)$, which would not hold if we define $(\neg C)^\Imc$ by $(\Delta^\Imc\times\semiringset)\setminus C^\Imc$. \blue{Note that already in the case of databases, defining provenance for queries with negation has proven challenging (\cf discussion in Section~\ref{sec:future-work}). Regarding universal role restriction ($\forall P.C$), $\ELHIbot$ already allows us to express $A\sqsubseteq \forall R.B$ as $\exists R^-.A\sqsubseteq B$ but it is unclear how to define $(\forall P.C)^\Imc$.}}

\subsection{Querying Annotated Ontologies}

	To define the provenance of (Boolean) CQs, we consider \emph{extended CQs}, which replace concept and role predicates with binary and ternary predicates respectively,  where the last term of the tuple is used for provenance information. 
Given a CQ $q(\vec{x})=\exists \vec{y}\, \phi(\vec{x},\vec{y})$, its \emph{extended version} is $\ext{q}(\vec{x})=\exists \vec{y}\exists\new{\vec{z}}\, \phi'(\vec{x},\vec{y},\new{\vec{z}})$ where $\new{\vec{z}}$ is a tuple of variables disjoint from $\vec{x}\cup\vec{y}$ and $\phi'(\vec{x},\vec{y},\new{\vec{z}})$ is obtained from $\phi(\vec{x},\vec{y})$ by replacing each $A(\new{t})$ by $A(\new{t},\new{z})$ and each $R(\new{t,t'})$ by $R(\new{t,t',z})$ with $\new{z\in\vec{z}}$ such that $\new{z}$ does not occur anywhere else in $\phi'(\vec{x},\vec{y},\new{\vec{z}})$. 
We use $P(\new{\vec{t},t})$ to refer to an atom which is either $A(\new{t_1,t})$
or $R(\new{t_1,t_2,t})$, and $P(\new{\vec{t},t})\in \ext{q}$ if $P(\new{\vec{t},t})$ occurs in~$\ext{q}$. 
A $\semiringshort$-annotated interpretation $\Imc$ satisfies an extended BCQ $\ext{q}:=\exists \vec{y}\exists \vec{z}\, \phi(\vec{y},\vec{z})$, written $\Imc\models \ext{q}$, iff there is a \emph{match for $\ext{q}$ in $\Imc$}, where a match for $\ext{q}$ in $\Imc$ is a function $\pi:\mn{terms}(\ext{q}) \rightarrow \Delta^\Imc\cup \semiringset$ such that $\pi(t)=t^\Imc$ for every $t \in \NI\cap\mn{terms}(\ext{q})$, and $\pi(\vec{t},t)\in P^\Imc$ for every $P(\vec{t},t)\in \mn{atoms}(\ext{q})$, where $\pi(\vec{t},t)$ is a shorthand for $(\pi(t_1),\pi(t))$ or $(\pi(t_1),\pi(t_2),\pi(t))$ depending on the arity of $P$. 
An extended BCQ $\ext{q}$ is \emph{entailed} by an annotated ontology $\Omc^\semiringshort$, written $\Omc^\semiringshort\models \ext{q}$, if and only if $\Imc\models \ext{q}$ for every model $\Imc$ of $\Omc^\semiringshort$. 
The following proposition is a direct consequence of the definition of (extended) BCQ entailment and Lemma \ref{lem:relationship-annotated-standard-models}.
\begin{proposition}\label{prop:cq-entailment-annotated-non-annotated}
	For every annotated ontology $\Omc^\semiringshort$ and BCQ $q$, $\Omc\models q$ iff $\Omc^\semiringshort\models \ext{q}$ where $\ext{q}$ is the extended version of $q$.
\end{proposition}
Given an extended BCQ $\ext{q}$ and a $\semiringshort$-annotated interpretation \Imc, let $\nu_\Imc(\ext{q})$ denote
the set of all matches of $\ext{q}$ in \Imc. 
The \emph{provenance annotations} of $\ext{q}$ in \Imc is defined as the set 
$\p{\Imc}{\ext{q}}=\left\{\bigotimes_{P(\vec{t},t)\in \ext{q}} \pi(t) \mid \pi\in\nu_\Imc(\ext{q})\right\}$ of elements of $\semiringshort$,  where $\pi(t)$ is the last element of the tuple $\pi(\vec{t},t)\in P^\Imc$ \new{(which is a semiring element)}. 
\new{We write $(q,\elem)$ to denote an \emph{annotated BCQ}.
The semantics of annotated BCQs is as expected:
$\Imc\models (q,\elem)$ if $\elem\in \p{\Imc}{\ext{q}}$  where $\ext{q}$ is the extended version of $q$, and $\Omc^\semiringshort\models (q,\elem)$
if  $\Imc \models\Omc^\semiringshort $ implies $\Imc\models (q,\elem)$, for all $\semiringshort$\mbox{-}annotated interpretations \Imc. If $\semiringshort$ is complete (or $\semiringshort$ is $\omega$-complete and $\semiringset$ is countable), we define the \emph{provenance} of $q$ \wrt $\Omc^\semiringshort$ as 
\begin{align*}
	\Pmc(q, \Omc^\semiringshort) := \bigoplus_{ \Omc^{\semiringshort} \models(q,\elem)} \elem.
\end{align*}

\blue{Analogous to Remark~\ref{rem:unsat-onto},  $\Pmc(q, \Omc^\semiringshort) = \bigoplus_{\elem\in\bigcap_{\Imc\models \Omc^{\semiringshort}}\p{\Imc}{\ext{q}}} \elem$ and   when $\Omc^\semiringshort$ is unsatisfiable, $\Pmc(q, \Omc^\semiringshort) = \bigoplus_{ \elem\in\semiringset} \elem$.}

\begin{example}[Example \ref{ex:running-onto} continued]
Consider the extended version of the query that asks for the deities having some parent, $\ext{q}(x)=\exists yt_1t_2\, \mn{Deity}(x,t_1)\wedge \mn{parent}(x,y,t_2)$. 

We saw in Example \ref{ex:running-onto} that all models $\Imc$ of $\Omc^{\why}$ satisfy
$(\mn{Dionysus}^{\Imc},x_1)\in \mn{Deity}^{\Imc}$\negmedspace, 
$(\mn{Dionysus}^{\Imc},x_3x_4y_1y_2)\in \mn{Deity}^{\Imc}$\negmedspace, and 
$(\mn{Dionysus}^{\Imc},x_5x_6y_1y_3)\in \mn{Deity}^{\Imc}$\negmedspace. In addition, it holds
that 
$(\mn{Dionysus}^{\Imc},\mn{Semele}^{\Imc},x_2y_2)\in \mn{parent}^{\Imc}$, $(\mn{Dionysus}^{\Imc},\mn{Demeter}^{\Imc},x_3y_2)\in \mn{parent}^{\Imc}$, and
$(\mn{Dionysus}^{\Imc},\mn{Zeus}^{\Imc},x_5y_3)\in \mn{parent}^{\Imc}$. 
Hence, for every model $\Imc$ of $\Omc^{\why}$, there are the following matches for $\ext{q}(\mn{Dionysus})$ in $\Imc$:
\begin{itemize}
\item $\pi_1(\mn{Dionysus})=\mn{Dionysus}^{\Imc}$, $\pi_1(y)=\mn{Semele}^{\Imc}$, $\pi_1(t_1)=x_1$, $\pi_1(t_2)=x_2y_2$;
\item $\pi_2(\mn{Dionysus})=\mn{Dionysus}^{\Imc}$, $\pi_2(y)=\mn{Demeter}^{\Imc}$, $\pi_2(t_1)=x_1$, $\pi_2(t_2)=x_3y_2$;
\item $\pi_3(\mn{Dionysus})=\mn{Dionysus}^{\Imc}$, $\pi_3(y)=\mn{Zeus}^{\Imc}$, $\pi_3(t_1)=x_1$, $\pi_3(t_2)=x_5y_3$;
\item $\pi_4(\mn{Dionysus})=\mn{Dionysus}^{\Imc}$, $\pi_4(y)=\mn{Semele}^{\Imc}$, $\pi_4(t_1)=x_3x_4y_1y_2$, $\pi_4(t_2)=x_2y_2$;
\item $\pi_5(\mn{Dionysus})=\mn{Dionysus}^{\Imc}$, $\pi_5(y)=\mn{Demeter}^{\Imc}$, $\pi_5(t_1)=x_3x_4y_1y_2$, $\pi_5(t_2)=x_3y_2$;
\item $\pi_6(\mn{Dionysus})=\mn{Dionysus}^{\Imc}$, $\pi_6(y)=\mn{Zeus}^{\Imc}$, $\pi_6(t_1)=x_3x_4y_1y_2$, $\pi_6(t_2)=x_5y_3$;
\item $\pi_7(\mn{Dionysus})=\mn{Dionysus}^{\Imc}$, $\pi_7(y)=\mn{Semele}^{\Imc}$, $\pi_7(t_1)=x_5x_6y_1y_3$, $\pi_7(t_2)=x_2y_2$;
\item $\pi_8(\mn{Dionysus})=\mn{Dionysus}^{\Imc}$, $\pi_8(y)=\mn{Demeter}^{\Imc}$, $\pi_8(t_1)=x_5x_6y_1y_3$, $\pi_8(t_2)=x_3y_2$;
\item $\pi_9(\mn{Dionysus})=\mn{Dionysus}^{\Imc}$, $\pi_9(y)=\mn{Zeus}^{\Imc}$, $\pi_9(t_1)=x_5x_6y_1y_3$, $\pi_9(t_2)=x_5y_3$.
\end{itemize}
It follows that  
\begin{align*}
\Pmc(q(\mn{Dionysus}), \Omc^{\why}) = {} & x_1x_2y_2+x_1x_3y_2+x_1x_5y_3 + x_2x_3x_4y_1y_2 + x_3x_4y_1y_2 + {}\\
 & x_3x_4x_5y_1y_2y_3
+x_2x_5x_6y_1y_3y_2+x_3x_5x_6y_1y_2y_3+x_5x_6y_1y_3.
\end{align*}

We obtain the following results similarly, by considering the matches for $\ext{q}(\mn{Dionysus})$ in the $\mathbb{T}$-annotated and $\mathbb{F}$-annotated interpretations that are models of $\Omc^{\mathbb{T}}$ and $\Omc^{\mathbb{F}}$ respectively.
\begin{itemize}
\item $\Pmc(q(\mn{Dionysus}), \Omc^{\mathbb{T}})=4$;

\item $\Pmc(q(\mn{Dionysus}), \Omc^{\mathbb{F}})=0.9$. \qedhere
\end{itemize}
\end{example}

\new{Given a possibly complex $\ELHIbot$ concept $C$, let $q_C(x)$ be the rooted tree-shaped query that retrieves all instances of $C$ (\cf Section~\ref{sec:prelimDL}).} 
The following lemma 
establishes the relationship between the answers of $\ext{q_C}(x)$ and their provenance annotations in an interpretation and the interpretation of $C$.

\begin{restatable}{lemma}{LemComplexConceptQueryInter}\label{lem:complex-concept-query-inter}
For every $\ELHIbot$ concept $C$, if $\ext{q_C}(x)$ is the extended version of the rooted tree-shaped query ${q_C}(x)$, then 
for every annotated interpretation $\Imc$ and for every $d\in\Delta^\Imc$, we have that $\{\bigotimes_{P(\vec{t},t)\in \ext{q_C}(x)} \pi(t) \mid \pi\in\nu_\Imc(\ext{q_C}(x)), \pi(x)=d\}=\{\elem\mid (d, \elem)\in C^\Imc\}$.
\end{restatable}

Theorem \ref{th:instancequeries} establishes that computing the provenance of query answers for CQs that are tree-shaped and rooted can be reduced in polynomial time to computing the provenance of assertions, \new{using the fact 
that a rooted tree-shaped CQ can be seen as a syntactic variant of a possibly complex $\ELHIbot$ concept $C$ (\cf Section~\ref{sec:prelimDL}). The idea of the reduction is to introduce a fresh concept name $A_C$ that subsumes $C$. 
This theorem allows us to extend all results given for assertion entailment in this paper to rooted tree-shaped BCQ entailment.}

\begin{restatable}{theorem}{Theorembasic}
\label{th:instancequeries}
For every annotated $\ELHIbot$ ontology $\Omc^\semiringshort$, $\ELHIbot$ concept $C$ and $a\in\NI$, 
we have $$\Pmc(q_C(a), \Omc^\semiringshort)=\Pmc(A_C(a), \Omc^\semiringshort\cup\{(C\sqsubseteq A_C,\one)\}),$$ where $A_C$ is a fresh concept name.
\end{restatable}

\subsection{Relationship with Semantics for Specific Annotations}
\label{sec:sem:ann}

We now discuss some weighted reasoning tasks which have been studied in the literature, and show how some of them fall 
into our general approach. 
Regarding the use of provenance as explanation, we will discuss in details  in Section~\ref{sec:pin:prov} the relationship between our provenance semantics and one of the most studied related problems in the DL community, axiom 
pinpointing \cite{Pena-AP20}. 

\paragraph{Fuzzy DL}
The original work on fuzzy DLs introduced by~\citet{DBLP:conf/aaai/Straccia98} considers fuzzy assertions of the
form $(\alpha,n)$ where $n\in[0,1]$ represents the membership degree for the assertion $\alpha$, and the
terminological axioms are all considered to be classical; that is, they are precise.
Fuzzy interpretations
generalise classical interpretations by considering concepts and roles as \emph{fuzzy} unary and binary predicates,
respectively. More precisely, every concept name $A\in\NC$ is interpreted as a function $A^\Imc:\Delta^\Imc\to[0,1]$ and
likewise role names $R\in\NR$ are interpreted as (binary) functions $R^\Imc:\Delta^\Imc\times\Delta^\Imc\to[0,1]$. 
Following the
then-popular Zadeh semantics~\cite{Zadeh-IC65}, which for the scope of $\ELHIbot$ in the setting of
\citet{DBLP:conf/aaai/Straccia98} coincides with the standard G\"odel semantics of mathematical fuzzy logic \cite{Hajek-98}, conjunctions
are interpreted through the minimum ($\min$) operator over the interval $[0,1]$, while disjunctions and existential 
restrictions are based on the maximum ($\max$) \new{operator}.
A GCI $C\sqsubseteq D$ is satisfied by a fuzzy interpretation \Imc if the membership degree for $C$ is less than or
equal to the degree for $D$ at every domain element; i.e., if for every 
$d\in\Delta^\Imc$ it holds that $C^\Imc(d)\le D^\Imc(d)$, 
\new{and a fuzzy assertion $(A(a),n)$ is satisfied by $\Imc$ if $A^\Imc(a^\Imc)\geq n$.} 
The entailment relation in this logic is defined in the standard manner: given a fuzzy ontology
$\Omc^{\new{f}}$, we get that $\Omc^{\new{f}}\modapprox(C(a),n)$ iff $C^\Imc(a^\Imc)\geq n$ in every fuzzy
model \Imc of $\Omc^{\new{f}}$. 
The general notions have later been extended to consider also membership degrees in the 
TBox~\cite{BoDP-AIJ15} and reasoning tasks such as query answering~\cite{MaTu-JIST14,MaTZ-DL15}. 
In such a setting, \citet{PaPe-TPLP22}  showed that for DL-Lite and BCQs, $\Omc^{\new{f}}\modapprox(\alpha,n)$ iff 
$\Omc_{\ge n}\models\alpha$ where $\Omc_{\ge n}$ is the classical ontology---called the \emph{$n$-cut}---containing only those axioms of $\Omc^{\new{f}}$ annotated with a degree greater or equal to $n$. 
Hence, the next proposition and Theorem~19 by \citet{PaPe-TPLP22} allow us to reduce BCQ (hence also assertion) entailments over fuzzy DL-Lite ontologies to computation of provenance in the fuzzy~semiring. 

\begin{restatable}{proposition}{proppossprov}
\label{prop:poss:prov}
\new{Let $\mathbb{F}=([0,1],\max,\min ,0,1)$ be the fuzzy semiring. For every satisfiable $\mathbb{F}$-annotated $\ELHIbot$ ontology 
$\Omc^{\mathbb{F}}$ and $n\in[0,1]$, $$\Pmc(\alpha,\Omc^{\mathbb{F}})\geq n\text{ iff }\Omc_{\ge n}\models \alpha\quad\text{holds if:}$$ 
\begin{enumerate}
\item $\alpha$ is a BCQ, an assertion, or an RI whose left-hand side is satisfiable \wrt $\Omc$; or
\item $\alpha$ is a GCI between basic concepts whose left-hand side is satisfiable \wrt $\Omc$ and $\Omc$ does not contain any GCI with $\top$ as left-hand side.
\end{enumerate}}
\end{restatable}

\new{The next example shows that we cannot extend the proposition to GCIs or RIs with unsatisfiable left-hand~side.}

\begin{example}\label{ex:fuzzy-unsat-left}
 If $\Omc^\mathbb{F}=\{(A\sqsubseteq B, 0.5),(A\sqsubseteq C, 0.5), (B\sqcap C\sqsubseteq \bot, 0.9)\}$, then 
 we have $\Omc_{\geq 0.9}\not\models A\sqsubseteq D$ for every $D\in\NC\new{\setminus\{A\}}$; but 
 \new{since $A$ is unsatisfiable \wrt $\Omc$}, it follows that $\Pmc(A\sqsubseteq D,\Omc^{\mathbb{F}})=1$ 
 (\cf Remark \ref{rem:unsat-left}). \new{Note that one can replace concept names by role names in this example.}
 \end{example}
\new{The reason for the requirement on the GCIs of $\Omc$ in the case where $\alpha$ is a GCI is illustrated below.}

\begin{example}\label{ex:fuzzy-GCI-top}
\new{If $\Omc^\mathbb{F}=\{(\top\sqsubseteq D,1)\}$, then $\Omc_{\geq 1}\models C\sqsubseteq D$ for every $C\in\NC$, but for every $C\in\NC\setminus\{D\}$, $\Pmc(C\sqsubseteq D,\Omc^{\mathbb{F}})=0$. Indeed, the annotated model $\Imc$ of $\Omc^\mathbb{F}$ defined by $D^\Imc=\Delta^\Imc\times\{1\}$ and $C^\Imc=\Delta^\Imc\times\{0.5\}$ does not satisfy $(C\sqsubseteq D,n)$ for any $n\in[0,1]$.}
\end{example}

The $n$-cuts are widely used in the fuzzy logic community. In the fuzzy DL literature, it is more common to use a
method called \emph{crispification} in which different concept names are used in a classical interpretation to represent the 
truth degrees assigned by a fuzzy interpretation \shortcite{BDG-URSW08,BDGS-IJUF12,BMPT-JoDS16}.
While crispification allows for a more fine-grained semantic analysis in general, by ``cutting'' also the interpretations at
the desired degree $n$, one obtains an interpretation which is a model of the $n$-cut whenever the fuzzy interpretation
is a model of the original ontology, thus preserving the result.

\paragraph{Possibilistic DL} 
\new{\citet{DBLP:journals/ijar/Hollunder95} originally introduced a possibilistic extension of the DL $\mathcal{ALCN}$ by annotating each axiom with either a \emph{necessity} or a \emph{possibility} value but most recent works (covering various DLs such as $\mathcal{ALC}$ or DL-Lite) only consider one kind of annotation, which corresponds to necessity value and is also called certainty or confidence degree \cite{DBLP:journals/ijis/QiJPD11,DBLP:journals/logcom/BenferhatB17}. 
In this context, a possibilistic ontology $\Omc^p$ is a set of annotated axioms $(\alpha,n)$ with $n\in]0,1]$. 
The semantics is based on a possibility distribution $\pi_{\Omc^p}$ over all the possible  interpretations \Imc, which depends on the degrees of the axioms not satisfied by \Imc: $\pi_{\Omc^p}(\Imc)=1$ if $\Imc\models \alpha$ for every $(\alpha, n)\in {\Omc^p}$, and  $\pi_{\Omc^p}(\Imc)=\min\{1-n\mid(\alpha,n)\in\Omc^p,\Imc\not\models\alpha\}$ otherwise. For a \emph{satisfiable} ontology, $\Omc^p\models_\pi (\alpha,n)$ iff $n\geq 1-\sup\{\pi_{\Omc^p}(\Imc)\mid\Imc\not\models\alpha\}$. 
Although at first sight this semantics does not exactly fit our definition, it has been shown that reasoning in possibilistic DLs can be reduced to classical reasoning over $n$-cuts,  where the $n$-cut $\Omc_{\geq n}$ of $\Omc^p$ is defined as in the fuzzy case: $\Omc^p\models_\pi (\alpha,n)$ iff $\Omc_{\geq n}\models \alpha$ \cite[Section 2.2]{DBLP:journals/ijis/QiJPD11}. This comes from a general result on possibilistic first-order logic knowledge bases \cite[Proposition 11]{PossibilisticLogicHandbook}. 
A tight necessity degree can be computed by taking the maximum over $n$ such that $\Omc_{\geq n}\models \alpha$.} 
Hence, by Proposition~\ref{prop:poss:prov}, possibilistic reasoning can be reduced
to provenance computation over the 
fuzzy semiring. 

\new{The same is not true for the \emph{product-based} variant of possibilistic semantics, where $\pi_{\Omc^p}(\Imc)=1$ if $\Imc\models \alpha$ for every $(\alpha, n)\in {\Omc^p}$, and  $\pi_{\Omc^p}(\Imc)=\Pi_{(\alpha,n)\in\Omc^p,\Imc\not\models\alpha}(1-n)$ otherwise \shortcite{BBKN17}. Indeed, if $\Omc^p=\{(A(a), 0.5), (B(a),0.5), (A\sqsubseteq C,1),(B\sqsubseteq C,1)\}$, then $\Omc^p\models_\pi (C(a), 0.75)$ under this semantics while $\Omc_{\geq 0.75}\not\models C(a)$. Actually, for every commutative semiring $\semiringshort=([0,1], \oplus, \otimes, 0,1)$, if $\Omc^\semiringshort$ is $\Omc^p$ interpreted as a $\semiringshort$-annotated ontology, $\Pmc(C(a), \Omc^\semiringshort)=0.5$ (since $0.5\otimes 1=0.5$ and there is an annotated model $\Imc$ of $\Omc^\semiringshort$ such that $C^\Imc=\{(a^\Imc, 0.5\otimes 1)\}$). 
} 

\new{Possibilistic semantics often target reasoning with unsatisfiable ontologies, by using the inconsistency degree $1-\sup\{\pi_{\Omc^p}(\Imc)\}$ to define $\Omc^p\models_\pi(\alpha,n)$ when the ontology may be unsatisfiable. 
Since the provenance of any axiom over an unsatisfiable ontology is the sum of all the semiring elements, we restrict our comparison to the case of satisfiable ontologies.}

\paragraph{Access control} 
\citet{BaKP-JWS12}  proposed to annotate axioms in an ontology with a label belonging
to a so-called \emph{access lattice} \new{$(L,\leq)$, which can be seen as a generalization of the dual of the access control semiring of Example~\ref{ex:semirings}. In this context, a clearance level $\ell\in L$ gives access to all axioms of greater level in the annotated ontology $\Omc^a$, \ie to $\Omc_{\geq\ell}=\{\alpha\mid (\alpha,v)\in \Omc^a, v\geq\ell\}$. Contrary to the access control semiring, clearance levels may be incomparable. Recall that a lattice is a partially ordered set in which every set of elements has a unique supremum (join) and a unique infimum (meet).} 
\new{Given an ontology with access levels $\Omc^a$ and an axiom $\alpha$, $\ell\in L$ is called an $(\Omc^a,\alpha)$-boundary if for every $\ell'\in L$ that is join prime\footnote{\blue{Given a lattice $(L,\leq)$ and a set $L'\subseteq L$, $\ell'\in L$ is join prime relative to $L'$ if for every $M'\subseteq \{\bigotimes_{m\in M}m\mid M\subseteq L'\}$, $\ell'\leq \bigoplus_{m\in M'}m$ implies that there is $m_0\in M'$ such that $\ell'\leq m_0$.}} relative to the set of labels that annotate $\Omc^a$, it holds that $\ell'\leq\ell$ iff $\Omc_{\geq\ell'}\models \alpha$. Each ontology defines a lattice element 
computed as the meet of the labels of its axioms. One of the main results of \citet{BaKP-JWS12} (Theorem~4.3) shows that a boundary of  $\alpha$ is the join of the labels of all subontologies that entail $\alpha$: formally, $\sup\{ \inf_{\beta\in\Mmc}(\lambda(\beta)) \mid \Mmc\subseteq\Omc, \Mmc\models \alpha\}$ is an $(\Omc^a,\alpha)$-boundary. A bounded distributive lattice $(L,\leq)$ is a lattice such that  the operations of join and meet distribute over each other and there exists a greatest element $\top$ and a least element $\bot$ such that $\bot\leq \ell\leq \top$ for every $\ell\in L$. It naturally corresponds to a semiring $\mathbb{L_A}=(L,\new{\sup,\inf},\bot,\top)$ where both operations (the join \new{($\sup$)} and meet \new{($\inf$)} of the lattice) are idempotent. Proposition~\ref{prop:access:prov} relates the $(\Omc^a,\alpha)$-boundary given by \citet{BaKP-JWS12} and the provenance of $\alpha$ in $\Omc^{\mathbb{L_A}}$, where $\Omc^{\mathbb{L_A}}$ is $\Omc^a$ interpreted as an $\mathbb{L_A}$-annotated ontology, with the same restrictions as  Proposition~\ref{prop:poss:prov}.} 

\begin{restatable}{proposition}{propaccessprov}
\label{prop:access:prov}
\new{Let $(L,\leq)$ be a bounded distributive lattice and $\mathbb{L_A}=(L,\new{\sup,\inf},\bot,\top)$. For every satisfiable $\mathbb{L_A}$-annotated $\ELHIbot$ ontology 
$\Omc^{\mathbb{L_A}}=\tup{\Omc,\lambda}$, $$\Pmc(\alpha,\Omc^{\mathbb{L_A}})=\sup\{ \inf_{\beta\in\Mmc}(\lambda(\beta)) \mid \Mmc\subseteq\Omc, \Mmc\models \alpha\}\quad\text{holds if:}$$ 
\begin{enumerate}
\item $\alpha$ is a BCQ, an assertion, or an RI whose left-hand side is satisfiable \wrt $\Omc$; or
\item $\alpha$ is a GCI between basic concepts whose left-hand side is satisfiable \wrt $\Omc$ and $\Omc$ does not contain any GCI with $\top$ as left-hand side.
\end{enumerate}
}
\end{restatable}

\paragraph{Bag semantics} 
\new{
\citet{DBLP:conf/ijcai/NikolaouKKKGH17,DBLP:journals/ai/NikolaouKKKGH19} defined a bag semantics for \DLLiteR ontologies whose assertions are annotated by multiplicities from $\mathbb{N}\cup\{\infty\}$. 
The semantics is based on bag interpretations, which are very similar to fuzzy interpretations, except that concept and role names are interpreted by bags $A^\Imc:\Delta^\Imc\rightarrow \mathbb{N}\cup\{\infty\}$ and $R^\Imc:\Delta^\Imc\times\Delta^\Imc\rightarrow \mathbb{N}\cup\{\infty\}$. 
The interpretation of inverse roles is as expected and $(\exists R)^\Imc$ maps each $d\in\Delta^\Imc$ to $\Sigma_{e\in\Delta^\Imc} R^\Imc(d,e)$. 
A bag interpretation $\Imc$ satisfies an assertion $A(a)$ with multiplicity $n$ if $A^\Imc(a^\Imc)\geq n$ and a GCI $C\sqsubseteq D$ if $C^\Imc(d)\leq D^\Imc(d)$ for every $d\in\Delta^\Imc$. 
Given a bag ontology $\Omc^{\new{b}}$, entailment of an assertion $A(a)$ \blue{with multiplicity $n$} under bag certain semantics is defined by $\Omc^b\models^b (A(a),n)$ if $n$ is the minimum $A^\Imc(a^\Imc)$ over all bag models of $\Omc^b$. This semantics is not captured by provenance in $\Nbb^{\infty}=(\Nbb\cup\{\infty\}, +, \times, 0, 1)$. Consider $\Omc_1^b=\{(A(a),2), (B(a),3), (A\sqsubseteq C,1), (B\sqsubseteq C,1)\}$. It holds that 
$\Omc_1^b\models^b (C(a),3)$ while if $\Omc_1^{{\Nbb^{\infty}}}$ is $\Omc_1^b$ interpreted as a $\Nbb^{\infty}$-annotated ontology, $\Pmc(C(a),\Omc_1^{{\Nbb^{\infty}}})=5$ (since every annotated model $\Imc$ of $\Omc_1^{\Nbb^{\infty}}$ is such that $\{(a^\Imc,2), (a^\Imc,3)\}\subseteq C^\Imc$). 
Note that replacing $\Nbb^\infty$ by $\semiringshort=(\Nbb\new{\cup\{\infty\}}, \max, \times, 0, 1)$ does not capture the bag semantics either. While we would have $\Pmc(C(a),\Omc_1^{\new{\semiringshort}})=3$, if we consider $\Omc_2^b=\{(R(a,b),2), (R(a,c),3), (\exists R\sqsubseteq C,1)\}$, then $\Omc_2^b\models^b (C(a),5)$ but $\Pmc(C(a),\Omc_2^{\semiringshort})=3$. 
Actually, for every $\semiringshort=(\mathbb{N}\cup\{\infty\},\oplus,\otimes, 0,1)$, we can see that \blue{(i) $\{(a^\Imc,2), (a^\Imc,3)\}\subseteq C^\Imc$ for every annotated model $\Imc$ of $\Omc_1^\semiringshort$ or of $\Omc_2^\semiringshort$} (since $n\otimes 1=n$), \blue{and (ii) that there is a model $\Imc$ of $\Omc_1^\semiringshort$ and $\Omc_2^\semiringshort$ such that $C^\Imc=\{(a^\Imc,2), (a^\Imc,3)\}$}, so that $\Pmc(C(a),\Omc_1^\semiringshort)=\Pmc(C(a),\Omc_2^\semiringshort)=2\oplus 3$, while the multiplicity of $C(a)$ under the bag semantics is not the same in $\Omc_1^b$ and in $\Omc_2^b$. One could modify the provenance semantics to capture bag semantics by adapting the so-called annotated model-based semantics for Datalog defined by \citet{DBLP:conf/kr/BourgauxBPT22} for a restricted class of semirings, but as already discussed in Section~\ref{sec:designchoices}, both semantics suffer from undesirable behaviors when $\oplus$ is not idempotent and coincide otherwise.
} 

\section{\new{Extension of Results from Classical DL} }\label{sec:classical-res-counterparts}
\new{In this section, we investigate under which conditions some classical results for 
DLs of the \EL family 
can be transferred to $\semiringshort$-annotated ontologies.}

\subsection{\new{Normal Form}}\label{subsection:normalization}
\new{In DL, it is customary to \blue{convert} ontologies \blue{into some} \emph{normal form} when working with languages from the $\EL$ family, to simplify the reasoning methods \cite{BBL-IJCAI05,BBL-EL08}. \blue{We present a similar conversion for annotated ontologies.} 
An annotated $\ELHIbot$ ontology $\Omc^\semiringshort$ is in \emph{normal form} \blue{if $\Omc$ is in normal form (as defined in Section~\ref{sec:prelimDL}), i.e.,} if for every GCI $(\alpha,\elem)\in \Omc^\semiringshort$, $\alpha$ is of one of the forms 
$$A\sqsubseteq B,\ A\sqcap A'\sqsubseteq B,\  A\sqsubseteq \exists R, \ A\sqsubseteq \exists R^-, \ \exists R.A\sqsubseteq B, \text{\ or \ } \exists R^-.A\sqsubseteq B$$ with $R\in\NR$, $A,A'\in\NC\cup\{\top\}$, $B\in\NC\cup\{\bot\}$. Every annotated $\ELHIbot$ ontology $\Omc^\semiringshort$ 
can be transformed, in polynomial time, into an ontology in normal form, which entails the same annotated axioms as $\Omc^\semiringshort$ over the signature of $\Omc^\semiringshort$.  \blue{Such an ontology can be} built by applying exhaustively the following rules, where $\widehat{C},\widehat{D}\notin\NC\cup\{\top\}$ and \blue{a fresh concept $A$ (not appearing so far in the ontology) is used at each transformation step.}
\begin{equation*}
\begin{array}{l@{\ }r@{\ }l@{\ }l}
\NF_1: &  (C\sqcap\widehat{D}\sqsubseteq E,\, \elem) &\longrightarrow&(\widehat{D}\sqsubseteq A,\, \one), (C\sqcap A\sqsubseteq E, \, \elem)\\
 \NF_2: &  (\widehat{C}\sqcap D\sqsubseteq E,\, \elem) &\longrightarrow&(\widehat{C}\sqsubseteq A,\, \one), (A\sqcap D\sqsubseteq E, \, \elem)\\
 \NF_3: &   (\exists P.\widehat{C}\sqsubseteq D,\, \elem) &\longrightarrow&(\widehat{C}\sqsubseteq A,\, \one), (\exists P. A\sqsubseteq D, \, \elem)\\
 \NF_4: &   (\widehat{C}\sqsubseteq\exists P,\, \elem) &\longrightarrow&(\widehat{C}\sqsubseteq A,\, \one), (A\sqsubseteq \exists P, \, \elem)
\end{array}
\end{equation*}}%
Note that \blue{the resulting set of annotated axioms} does not contain any pair of annotated axioms $(\alpha,\elem)$, $(\alpha,\elem')$ with $\elem\neq\elem'$, so 
\blue{it} is a proper annotated ontology. 
\blue{Indeed, observe that for every pair of axioms among those introduced by the rules, either they have been introduced by different rule applications, so the last introduced axiom contains some fresh concept name $A$ that does not appear in the other, or they have been introduced by the same rule application and the fresh concept name $A$ introduced by this rule occurs in different sides of $\sqsubseteq$ in each of them. Hence each introduced axiom (without annotation) is new.}
\begin{restatable}{theorem}{NormalisationELHI}\label{th:normal-form}
\new{Let $\semiringshort$ be a commutative semiring, $\Omc^\semiringshort$ a $\semiringshort$-annotated $\ELHIbot$ ontology, $\alpha$ an axiom, and $\elem\in\semiringset$. 
Let $\NF(\Omc^\semiringshort)$ be obtained by applying exhaustively Rules $\NF_1$-$\NF_4$ to~$\Omc^\semiringshort$.
\begin{itemize}
\item If $\Omc^\semiringshort\models (\alpha,\elem)$, then $\NF(\Omc^\semiringshort)\models (\alpha,\elem)$.
\item If $\NF(\Omc^\semiringshort)\models (\alpha,\elem)$ and every concept name occurring in $\alpha$ also occurs in $\Omc^\semiringshort$, then $\Omc^\semiringshort\models (\alpha,\elem)$.
\end{itemize}}
\end{restatable}
\begin{corollary}\label{cor:normal-form}
\new{If $\Omc^\semiringshort$ and $\NF(\Omc^\semiringshort)$ are as in Theorem \ref{th:normal-form}, then for every $\ELHIbot$ axiom $\alpha$ over the signature of $\Omc^\semiringshort$, $\Pmc(\alpha, \Omc^\semiringshort)=\Pmc(\alpha, \NF(\Omc^\semiringshort))$.}
\end{corollary}
\begin{remark}
\blue{There may be several ontologies in normal form obtained by applying exhaustively Rules $\NF_1$-$\NF_4$ to~$\Omc^\semiringshort$, depending on the order in which we apply the rules. For example $\{(B\sqcap C\sqcap D\sqsubseteq E,\elem)\}$ can be normalized in $\{(C\sqcap D\sqsubseteq A,\one),(B\sqcap A\sqsubseteq E,\elem)\}$ by applying $\NF_1$ or in $\{(B\sqcap C\sqsubseteq A,\one),(A\sqcap D\sqsubseteq E,\elem)\}$ by applying $\NF_2$.}
\end{remark}

\subsection{Canonical Model}\label{sec:canonical-model}

In DLs of the \EL family, every ontology has a \emph{canonical (or universal) model} which can be homomorphically mapped into any other model of the ontology \blue{(see, e.g., \cite[Section 3.2]{DBLP:conf/rweb/KontchakovZ14}, for a definition of the canonical model of an \EL ontology). Such a model allows us to investigate what is  entailed by the ontology without the need to consider all the possible models. In particular, a BCQ is entailed by an ontology if and only if it holds in its canonical model. 
It can be constructed via a forward chaining procedure similar to the \emph{oblivious chase} for databases \cite{DBLP:journals/ws/CaliGL12} (see also the definition of the oblivious chase given by~\citet{DBLP:conf/ijcai/BednarczykFO20}, which we took inspiration from). 
We adapt this construction 
to annotated $\ELHIbot$ ontologies.}

Given an annotated $\ELHIbot$ ontology $\Omc^\semiringshort$ that is satisfiable, we define a canonical model $\Imc_{\Omc^\semiringshort}$ of  $\Omc^\semiringshort$ as follows. 
Let $\Imc_0$ be \new{the $\semiringshort$-annotated interpretation} such that:
\begin{itemize}
	\item $\Delta^{\Imc_0}:=\NI$;
	\item $a^{\Imc_0}:=a$,  
	for all $a\in\NI$;
	\item $(a,\elem)\in A^{\Imc_0}$ iff 
$(A(a),\elem)\in\Omc^{\semiringshort}$; 
	\item $(a,b,\elem)\in R^{\Imc_0}$ iff 
$(R(a,b),\elem)\in\Omc^{\semiringshort}$. 
\end{itemize}
\blue{
Given two annotated interpretations $\Imc_n$ and $\Imc_{n+1}$ such that $a^{\Imc_n}=a^{\Imc_{n+1}}=a$ for all $a\in\NI$, we say that $\Imc_{n+1}$ is obtained from $\Imc_n$ by applying the ``chase rule'' to $(\alpha,\elem)\in\Omc^\semiringshort$ and $(\vec{d},\elem')\in E^{\Imc_n}$  if one of the following conditions holds (note that $P$ and $Q$ can be role names or inverse roles):
\begin{itemize}
\item $\alpha=P\sqsubseteq Q$, $(\vec{d})=(d,d')$, $E=P$, $Q^{\Imc_{n+1}}=Q^{\Imc_n}\cup\{(d,d',\elem\otimes \elemb)\}$,\\ $\Delta^{\Imc_{n+1}}=\Delta^{\Imc_{n}}$, and $F^{\Imc_{n+1}}=F^{\Imc_{n}}$ for all concept and role names $F\neq Q$;
\item  $\alpha=C\sqsubseteq A$, $(\vec{d})=(d)$, $E=C$, $A^{\Imc_{n+1}}=A^{\Imc_{n}}\cup\{(d,\elem\otimes \elemb)\}$,\\ $\Delta^{\Imc_{n+1}}=\Delta^{\Imc_{n}}$, and $F^{\Imc_{n+1}}=F^{\Imc_{n}}$ for all concept and role names $F\neq A$;
\item $\alpha=C\sqsubseteq \exists P$,  $(\vec{d})=(d)$, $E=C$, $\Delta^{\Imc_{n+1}}=\Delta^{\Imc_{n}}\cup\{d_f\}$ with $d_f\notin\Delta^{\Imc_{n}}$,  
 $P^{\Imc_{n+1}}=P^{\Imc_{n}}\cup\{(d,d_f,\elem\otimes \elemb)\}$,\\ and $F^{\Imc_{n+1}}=F^{\Imc_{n}}$ for all concept and role names $F\neq P$. 
\end{itemize}
There exists a (potentially infinite) sequence $\Imc_0,\Imc_1,\dots$ such that 
\begin{enumerate}[(i)]
\item for every $i\geq 0$, $\Imc_{i+1}$ is obtained from $\Imc_i$ by applying the chase rule to $(\alpha_i,\elem_i)\in\Omc^\semiringshort$ and $(\vec{d}_i,\elem'_i)\in E_i^{\Imc_i}$, 
\item for every $i,j\geq 0$ such that $i\neq j$, $(\alpha_i,\elem_i)\neq(\alpha_j,\elem_j)$ or $(\vec{d}_i,\elem'_i)\neq (\vec{d}_j,\elem'_j)$ (the chase rule is not applied twice to the same $(\alpha,\elem)$ and $(\vec{d},\elem')$), and 
\item for every $i\geq 0$ if there is $(\alpha,\elem)\in\Omc^\semiringshort$ and $(\vec{d},\elem')\in E^{\Imc_i}$ to which the chase rule can be applied (\ie such that $E$ is the left-hand side of $\alpha$), then there is $j\geq 0$ such that $\Imc_{j+1}$ is obtained from $\Imc_j$ by applying  the chase rule with $(\alpha,\elem)\in\Omc^\semiringshort$ and $(\vec{d},\elem')\in E^{\Imc_j}$ (fairness condition, the rule is applied at some point in the sequence, note that it may be the case that $0\leq j\leq i$).
\end{enumerate} 
} 
Existence of such a sequence can be shown by applying the rule in \blue{a level-saturating fashion 
(\blue{as done by \citet{DBLP:journals/ws/CaliGL12} in the database context}): define the level of a tuple $(\vec{d},\elem')\in E^{\Imc_n}$ as $i$ such that $(\vec{d},\elem')\in E^{\Imc_i}$ and $(\vec{d},\elem')\notin E^{\Imc_{i-1}}$ and apply the chase rule to some tuple of degree $i$ only if it has already been applied to all $(\vec{e},\chi)\in F^{\Imc_n}$ of degree $j<i$ to which it can be applied.} 
We define $\Imc_{\Omc^\semiringshort} := \bigcup_{n \geq 0} \Imc_n$, \ie $\Delta^{\Imc_{\Omc^\semiringshort}}= \bigcup_{n \geq 0}\Delta^{\Imc_n}$; for every $a\in\NI$, $a^{\Imc_{\Omc^\semiringshort}}=a$; for every $A\in\NC$, $A^{\Imc_{\Omc^\semiringshort}}=\bigcup_{n \geq 0}A^{\Imc_n}$; and for every $R\in\NR$,  $R^{\Imc_{\Omc^\semiringshort}}=\bigcup_{n \geq 0}R^{\Imc_n}$. Even if there may be different sequences $\Imc_0,\Imc_1,\dots$, \blue{the resulting} $\Imc_{\Omc^\semiringshort}$ is unique up to renaming of the fresh domain elements introduced (by the same arguments that show uniqueness of the oblivious chase), 
which allows us to talk about \emph{the} canonical model of $\Omc^\semiringshort$.} 
Proposition~\ref{lem:can} is an easy consequence of the definition of $\Imc_{\Omc^\semiringshort}$.
\begin{proposition}\label{lem:can}
	If $\Omc^\semiringshort$ is a satisfiable $\semiringshort$-annotated $\ELHIbot$ ontology then $\Imc_{\Omc^\semiringshort}\models\Omc^\semiringshort$.
\end{proposition}

Let $\Imc=(\Delta^\Imc,K,\cdot^\Imc)$ and $\Jmc=(\Delta^\Jmc,K,\cdot^\Jmc)$ be $\semiringshort$-annotated interpretations and $d,e \in \Delta^\Imc$. 
A \emph{homomorphism} $\homo: \Imc \rightarrow \Jmc$ is a function from $\Delta^\Imc$ to $\Delta^\Jmc$ such that:

\begin{itemize}
	\item for all individual names $a \in \NI$,   $\homo(a^{\Imc})=a^{\Jmc}$;
	\item for all concept names $A \in \NC$, if $(d,\elem) \in A^{\Imc}$ then $(\homo(d),\elem) \in A^{\Jmc}$;
	\item for all role names $R \in \NR$, if $(d,e,\elem) \in R^{\Imc}$ then $(\homo(d),\homo(e),\elem) \in R^{\Jmc}$.
\end{itemize} 

Lemma \ref{lem:homomo} states that $\Imc_{\Omc^\semiringshort}$ has the usual property of canonical models.
\begin{restatable}{lemma}{lemhomomo}\label{lem:homomo} For every model $\Imc$ of  a satisfiable $\Omc^\semiringshort$, 
	there is a homomorphism 
	$\homo: \Imc_{\Omc^\semiringshort} \rightarrow \Imc$.
\end{restatable}

Theorem~\ref{thm:can-model-main}, \blue{which relies on Lemma~\ref{lem:homomo},}  states that the canonical model behaves as expected for annotated assertions and BCQs entailment (note that annotated BCQs actually subsume annotated assertions, so that the two first points are consequences of the last one). 

\begin{restatable}{theorem}{thmcanmodelmain}\label{thm:can-model-main}
\new{Let $\semiringshort$ be a commutative semiring.} For every  satisfiable $\semiringshort$-annotated $\ELHIbot$ ontology $\Omc^\semiringshort$  the following hold:
 \begin{itemize}
 	\item for every $\elem\in K$, 
 	$a\in\NI$, and $A\in\NC$, $\Omc^\semiringshort\models (A(a),\elem)$ iff $\Imc_{\Omc^\semiringshort}\models (A(a),\elem)$; 
 	\item for every $\elem\in K$,  $a,b\in\NI$, and $R\in\NR$, $\Omc^\semiringshort\models (R(a,b),\elem)$ iff $\Imc_{\Omc^\semiringshort}\models (R(a,b),\elem)$;
 		\item for every $\elem\in K$ and  
 	  BCQ $q$, 
 	$\Omc^\semiringshort\models (q,\elem)$ iff $\Imc_{\Omc^\semiringshort}\models (q,\elem)$. 
 \end{itemize}
\end{restatable}

\new{In the classical $\EL$ family, 
	 one can define a canonical model $\Imc_{C,\Omc}$ of a concept $C$ and an ontology $\Omc$ such that if $C$ is satisfiable \wrt $\Omc$ (\ie there exists a model $\Imc$ of $\Omc$ such that $C^\Imc\neq\emptyset$), then for every   concept $D$, $\Omc\models C\sqsubseteq D$ iff $d_C\in D^{\Imc_{C,\Omc}}$ where $d_C$ is a distinguished domain element that ``represents'' $C$~(see e.g.~\cite{DBLP:journals/jsc/LutzW10}). Intuitively, $\Imc_{C,\Omc}$ is built as the canonical model of $\Omc$ but starts from a modified $\Imc_0$ which is such that $d_C\in C^{\Imc_0}$. For example, if $C=A\sqcap \exists R.B$, $d_C$ is added to $A^{\Imc_0}$, $(d_C, x_{RB})$ to $R^{\Imc_0}$ and $x_{RB}$ to $B^{\Imc_0}$.} 
\new{Example~\ref{ex:problem-gcis-sem-entail} illustrates the difficulty to extend this result to the annotated case, \ie obtain a canonical model $\Jmc$ of $C$ and $\Omc^\semiringshort$ such that $(d_C,\elem)\in D^\Jmc$ iff $\Omc^\semiringshort\models (C\sqsubseteq D,\elem)$.}  
\begin{example}\label{ex:problem-gcis-sem-entail}
	\blue{It would be difficult to obtain a canonical model $\Jmc$ of $A\sqcap B$ and $\Omc^{\posbool}=\{(A\sqsubseteq D, x)\}$ such that $(d_{A\sqcap B},\elem)\in D^\Jmc$ iff $\Omc^{\posbool}\models (A\sqcap B\sqsubseteq D,\elem)$. 
		We first argue that there is no $\elem\in\posbool$ such that $\Omc^{\posbool}\models (A\sqcap B\sqsubseteq D, \elem)$. Indeed, the \posbool-annotated interpretation $\Imc$ below is a model of $\Omc^{\posbool}$ and there is no $\elem$ such that $x=\elem\times y$.
		\begin{align*}
			A^{\Imc}=\{(e,1)\} \quad B^{\Imc}=\{(e,y)\} \quad (A\sqcap B)^{\Imc}=\{(e,y)\} \quad D^{\Imc}=\{(e,x)\}
		\end{align*}	 
		For every annotated model $\Jmc$ of $\Omc^{\posbool}$, if there is $(d,\chi)\in (A\sqcap B)^\Jmc$ as we would expect for a representative $d$ of $A\sqcap B$, there must be some $(d,\chi')\in A^\Jmc$ so that $(d,\chi'\times x)\in D^\Jmc$, while, by the above argument, we have that $\Omc^{\posbool}\not\models (A\sqcap B\sqsubseteq D, \chi'\times x)$.}
\end{example}
\new{We thus restrict our attention to \emph{basic concepts} (\ie concept names and concepts of the form $\exists P$). For such a basic concept $C$, we define the canonical model} 
$\Imc_C$ as the interpretation 
with domain $\Delta^{\Imc_C}:=\{d_C\}\cup \NI$ if $C$ is a concept name, and $\Delta^{\Imc_C}:=\{d_C, d_f\}\cup \NI$ otherwise, such that $a^{\Imc_C}:=a$ for all
$a\in\NI$, 
$A^{\Imc_C}:=\{(d_A,\one)\}$ if $C=A\in\NC$, 
$R^{\Imc_{C}}:=\{(d_{\exists R},d_f,\one)\}$ if $C=\exists R$, and $R^{\Imc_{C}}:=\{(d_f,d_{\exists R^-},\one)\}$ if $C=\exists R^-$, and all other concept and role names are mapped to the empty set. 
The canonical model $\Imc_{C,\Omc^\semiringshort}$ for a basic concept $C$ and a
$\semiringshort$-annotated $\ELHIbot$ ontology $\Omc^\semiringshort$
\new{such that $C$ is satisfiable \wrt $\Omc^\semiringshort$} is defined in the same way as the canonical model for $\Omc^\semiringshort$
except that $\Imc_0$ \blue{is replaced by} $\Imc_C$ (and  assertions if they are present in the ontology).

\begin{restatable}{theorem}{thmcanmodelmaingci}\label{thm:can-model-main-gci}
Assume $\semiringshort$ is a commutative \timesidem semiring.  
	Let $\Omc^\semiringshort$
	be a  satisfiable $\semiringshort$-annotated $\ELHIbot$ ontology 
	\new{such that $\Omc$ does not contain any GCI with $\top$ as left-hand side,}
	and let $C,D$ be \emph{basic} concepts \new{such that $C$ is} satisfiable 
	w.r.t. $\Omc^\semiringshort$. 
	Then, 
for every $\elem\in K$,   $(d_C,\elem)\in D^{\Imc_{C,\Omc^\semiringshort}}$ iff ${\Omc^\semiringshort}\models (C\sqsubseteq D,\elem)$.
\end{restatable}

\new{Note that ${\Omc^\semiringshort}\models (C\sqsubseteq D,\elem)$ implies $(d_C,\elem)\in D^{\Imc_{C,\Omc^\semiringshort}}$ even without the restrictions on $\Omc$ and $\semiringshort$, since $\Imc_{C,\Omc^\semiringshort}$ is a model of $\Omc^\semiringshort$. However, the other direction of Theorem~\ref{thm:can-model-main-gci} does not hold 
if $\Omc$ contains some GCI with $\top$ as its left-hand side; for example, if $\Omc^\semiringshort=\{(\top\sqsubseteq D,\elem)\}$ (since in this case one can show that $\Omc\not\models (C\sqsubseteq D,\elem)$ for $C\in\NC\setminus\{D\}$ as we did in  Example~\ref{ex:fuzzy-GCI-top}). Example~\ref{ex:timesidemforcanmodelgci} shows that the $\otimes$-idempotency condition is also necessary.}

\begin{example}\label{ex:timesidemforcanmodelgci}
\new{Let $\Omc^{\why}=\{(A\sqsubseteq B_1,1), (A\sqsubseteq B_2, 1), (B_1\sqcap B_2\sqsubseteq C, 1)\}$. 
It is easy to see that $(d_A,1)\in C^{\Imc_{A,\Omc^{\why}}}$. 
However, there is no $\elem\in\why$ such that $\Omc^{\why}\models (A\sqsubseteq C, \elem)$. 
Indeed, the \why-annotated interpretation $\Imc$ below is a model of $\Omc^{\why}$ and there is no $\elem$ such that $x+y+xy=\elem\times (x+y)$. 
\begin{align*}
		A^{\Imc}=&\{(e,x+y)\} \quad B_1^{\Imc}=\{(e,x+y)\} \quad B_2^{\Imc}=\{(e,x+y)\} \quad C^{\Imc}=\{(e,x+y+xy)\}
\end{align*}
	}
	Intuitively, this is because when the semiring is not multiplicatively idempotent, $A\sqcap A$ and $A$ are not equivalent: ${\Omc^{\why}}\models (A\sqcap A\sqsubseteq C,1)$ but ${\Omc^{\why}}\not\models (A\sqsubseteq C,1)$ \blue{(recall also Example~\ref{ex:tautologies-with-prov-zero})}. 	
\end{example}

\new{
Finally, we briefly show how role inclusions can be captured using a canonical model. For a role $P$ (with $P\in\NR$ or $P=R^-$ for some $R\in\NR$), we define $\Imc_P$ as the interpretation with domain $\Delta^{\Imc_P}:=\{d_1,d_2\}\cup\NI$ such that $a^{\Imc_P}:=a$ for all $a\in\NI$, $P^{\Imc_P}:=\{(d_1,d_2,\one)\}$, and all other concept and role names mapped to the empty set. The canonical model $\Imc_{P,\Omc^\semiringshort}$ for a role $P$ and a
$\semiringshort$-annotated $\ELHIbot$ ontology $\Omc^\semiringshort$
such that $P$ is satisfiable \wrt $\Omc^\semiringshort$ (\ie there exists a model $\Imc$ of $\Omc^\semiringshort$ such that $P^\Imc\neq\emptyset$) is defined in the same way as the canonical model for $\Omc^\semiringshort$ 
except that $\Imc_0$ starts with $\Imc_P$ (and  assertions if they are present in the ontology). 
Contrary to the GCI case, we do not need the semiring to be \timesidem to obtain the following result. }
\begin{restatable}{theorem}{thmcanmodelmainri}\label{thm:can-model-main-ri}
\new{Let $\semiringshort$ be a commutative semiring.  
	Let $\Omc^\semiringshort$
	be a  satisfiable $\semiringshort$-annotated $\ELHIbot$ ontology
	and let $P,Q$ be two roles such that $P$ is satisfiable 
	w.r.t. $\Omc^\semiringshort$. 
	Then, 
for every $\elem\in K$,   $(d_1,d_2,\elem)\in Q^{\Imc_{P,\Omc^\semiringshort}}$ iff ${\Omc^\semiringshort}\models (P\sqsubseteq Q,\elem)$.}
\end{restatable}

\subsection{\new{Reduction Between Assertion and GCI or RI Entailment}}\label{sec:reduction-tasks}
\new{It is well-known that in $\ELHI_\bot$, GCI or RI entailment can be reduced to assertion entailment in polynomial time, and reciprocally. 
In the case of $\semiringshort$-annotated $\ELHIbot$ ontologies, we obtain similar reductions, with a restriction on the semiring in the case of GCI entailment.}

\begin{restatable}{theorem}{thReductionConcept}\label{th:red-concept}
\new{If $\semiringshort$ is a commutative \timesidem semiring then, for every $\semiringshort$-annotated $\ELHIbot$ ontology $\Omc^\semiringshort=\tup{\Omc,\lambda}$ such that $\Omc$ does not contain any GCI with $\top$ as left-hand side, the following hold.
\begin{itemize}
\item For every GCI between basic concepts $C\sqsubseteq D$ and $\elem_0\in\semiringset$, $${\Omc^\semiringshort}\models (C\sqsubseteq D,\elem_0)\text{ iff }{\Omc^\semiringshort}\cup\Tmc^\semiringshort_D\cup\Amc^\semiringshort_C\models (E(a_0), \elem_0)$$ where $\Tmc^\semiringshort_D=\{(D\sqsubseteq E, \one)\}$,
 $\Amc^\semiringshort_C=\{(C(a_0),\one)\}$ if $C\in\NC$, 
and $\Amc^\semiringshort_C=\{(P(a_0,b_0),\one)\}$ if $C=\exists P$ (where $P(a_0,b_0)$ denotes $R(b_0,a_0)$ if $P=R^-$), 
with $a_0, b_0\in\NI\setminus\individuals{\Omc}$ and $E\in\NC\setminus\signature{\Omc}$.
\item For every concept assertion $B(a_0)$ and $\elem_0\in\semiringset$, $${\Omc^\semiringshort}\models (B(a_0),\elem_0)\text{ iff }\Tmc^\semiringshort\models (C_{a_0}\sqsubseteq B,\elem_0)$$ 
where $\Tmc^\semiringshort$ is defined as follows (assuming that for all $a,b\in\individuals{\Omc}$ and $R\in\NR$, $C_a$ and $R_{ab}$ are fresh concept and role names respectively):
\begin{align*}
\Tmc^\semiringshort=&\Omc^\semiringshort\cup\bigcup_{a\in\individuals{\Omc}}\Tmc^\semiringshort_{C_a}\\
\Tmc^\semiringshort_{C_a}=&\{(C_a\sqsubseteq A, \elem)\mid (A(a),\elem)\in{\Omc^\semiringshort}\}\ \cup\\
&\{(R_{ab}\sqsubseteq R, \elem), (C_a\equiv \exists R_ {ab},\one),    (C_b\equiv \exists R_{ab}^-, \one)\mid (R(a,b),\elem)\in{\Omc^\semiringshort}\}.
\end{align*}
\end{itemize}}
\end{restatable}
\new{The reason for the restriction of the use of $\top$ in the left-hand side of GCIs is the same as in Theorem~\ref{thm:can-model-main-gci}: if $\Omc^\semiringshort=\{(\top\sqsubseteq D,\elem)\}$, it holds that $\Omc\not\models (C\sqsubseteq D,\elem)$ for every $C\in\NC\setminus\{D\}$ but $\Omc^\semiringshort\cup\Tmc^\semiringshort_D\cup\Amc^\semiringshort_C\models (E(a_0), \elem)$ because in every model $\Imc$ of $\Omc^\semiringshort\cup\Tmc^\semiringshort_D\cup\Amc^\semiringshort_C$, $(a_0^\Imc,\one)\in\top^\Imc$ so $(a_0^\Imc,\elem)\in D^\Imc\subseteq E^\Imc$. 
The annotated ontology $\Omc^{\why}$ defined in Example~\ref{ex:timesidemforcanmodelgci} shows that the $\otimes$-idempotency condition is necessary for the first item of the theorem: $\Omc^{\why}\not\models (A\sqsubseteq C,1)$ while $\Omc^{\why}\cup\{(C\sqsubseteq E,1)\}\cup\{(A(a_0),1)\}\models (E(a_0), 1)$. For the second item, we can consider $\Omc^{\why}\cup\{(A(a_0),1)\}$ and check that it entails $(C(a_0),1)$ while $\Tmc^{\why}=\Omc^{\why}\cup\{(C_{a_0}\sqsubseteq A,1)\}$ does not entail $(C_{a_0}\sqsubseteq C,1)$. 
However, the idempotency requirement can be lifted for RIs.}

\begin{restatable}{theorem}{thReductionRole}\label{th:red-role}
\new{If $\semiringshort$ is a commutative semiring then, for every $\semiringshort$-annotated $\ELHIbot$ ontology $\Omc^\semiringshort=\tup{\Omc,\lambda}$, the following hold.
\begin{itemize}
\item For every positive role inclusion $P_1\sqsubseteq P_2$ and $\elem_0\in\semiringset$, $$\Omc^\semiringshort{\models} (P_1\sqsubseteq P_2,\elem_0)\text{ iff }\Omc^\semiringshort\cup\{(P_1(a_0,b_0), \one)\}{\models}(P_2(a_0,b_0),\elem_0)$$ where $a_0,b_0\in\NI\setminus\individuals{\Omc}$ and $P_i(a_0,b_0)$ denotes $R(b_0,a_0)$ if $P_i=R^-$.
\item For every role assertion $R(a_0,b_0)$ and $\elem_0\in\semiringset$, $${\Omc^\semiringshort}\models(R(a_0,b_0),\elem_0)\text{ iff }\Tmc^\semiringshort_{S_{a_0,b_0}}\models (S\sqsubseteq R,\elem_0)$$ where $S\in\NR\setminus\signature{\Omc}$ and 
\begin{align*}
\Tmc^\semiringshort_{S_{a_0,b_0}}={} & \Omc^\semiringshort\cup\{(S\sqsubseteq R',\elem)\mid (R'(a_0,b_0),\elem)\in\Omc^\semiringshort\}\cup{} \\ 
& \{(S\sqsubseteq R'^-,\elem)\mid (R'(b_0,a_0),\elem)\in\Omc^\semiringshort\}.
\end{align*}
\end{itemize}}
\end{restatable}

\section{\new{Properties of the Provenance Semantics}}\label{sec:expected-properties}

\new{Inspired by the generic definition of a provenance semantics for Datalog queries and properties proposed by~\citet[Definition~3 and Properties~1 and~3]{DBLP:conf/kr/BourgauxBPT22}, we show in this section} 
that under some restrictions on the semiring, our semiring 
provenance definition satisfies properties that are expected for a provenance semantics.

\subsection{Preservation of Entailment} Theorem \ref{th:sem-entailment} ensures that the semantics reflects the entailment or non-entailment of axioms and queries \new{from the non-annotated ontology}.

\begin{restatable}{theorem}{thsementailment}\label{th:sem-entailment}
Let $\semiringshort$ be a commutative \new{complete semiring (or $\omega$-complete with $\semiringset$ countable)}. For every $\semiringshort$-annotated $\ELHIbot$ ontology $\Omc^\semiringshort=\tup{\Omc,\lambda}$ and every 
$\alpha$ that is an assertion, a BCQ, \new{an RI}, or a GCI between basic concepts: 
\begin{itemize}
\item $\Omc\not\models\alpha$ implies $\Pmc(\alpha,\Omc^\semiringshort) = \zero$; 
\item if $\semiringshort$ is positive\new{, $\alpha$ is an assertion, a BCQ, or an RI,} and $\Pmc(\alpha,\Omc^\semiringshort) = \zero$, then $\Omc\not\models\alpha$; 
\item if $\semiringshort$ is positive \new{and \timesidem,  $\Omc$ does not contain any GCI with $\top$ as left-hand side,} and $\Pmc(\alpha,\Omc^\semiringshort) = \zero$, then $\Omc\not\models\alpha$. 
\end{itemize}
\end{restatable}
Unfortunately, the third point of Theorem~\ref{th:sem-entailment} does not hold if $\alpha$ is a GCI with a complex concept on its left-hand side, as we can see on Example~\ref{ex:problem-gcis-sem-entailtwo}.
\begin{example}\label{ex:problem-gcis-sem-entailtwo}
\new{Consider $\Omc^{\posbool}=\{(A\sqsubseteq D,x)\}$. We have already shown 
(Example~\ref{ex:problem-gcis-sem-entail}) that there is no $\elem\in\posbool$ such that $\Omc^{\posbool}\models (A\sqcap B\sqsubseteq D, \elem)$. Hence it follows that $\Pmc( A\sqcap B\sqsubseteq D,\Omc^{\posbool}) = 0$, while $\Omc\models A\sqcap B\sqsubseteq D$ and $\posbool$ is positive and idempotent.}
\end{example}	
\new{The reason for the requirement on GCIs without $\top$ as left-hand side in the third point is the same explained
in the previous section: if $\Omc^\semiringshort=\{(\top\sqsubseteq D,\elem)\}$, then $\Omc\models C\sqsubseteq D$ for every $C\in\NC$, but for every $C\in\NC\setminus\{D\}$, it may be the case that $\Pmc(C\sqsubseteq D,\Omc^\semiringshort)=\zero$ (\cf Example~\ref{ex:fuzzy-GCI-top}).} 
The next example illustrates the impact of $\otimes$-idempotency 
\new{in the case where $\alpha$ is a GCI}.

\begin{example}\label{ex:idempotent}
	Let $\semiringshort$ be a \timesidem \new{commutative} semiring and consider the ontology $$\Omc^{\semiringshort}=\{(A\sqsubseteq B_1,\elem_1), (A\sqsubseteq B_2, \elem_2), (B_1\sqcap B_2\sqsubseteq C, \elem_3)\}.$$
	If $\Imc$ is a model of $\Omc^{\semiringshort}$ and 
	$(e,\eleme)\in A^\Imc$, then $(e,\eleme\otimes \elem_1)\in B_1^\Imc$ and $(e,\eleme\otimes \elem_2)\in B_2^\Imc$ so $(e,\eleme\otimes \elem_1\otimes \eleme\otimes \elem_2)\in (B_1\sqcap B_2)^\Imc$, i.e. $(e, \eleme\otimes \elem_1\otimes\elem_2)\in (B_1\sqcap B_2)^\Imc$ by $\otimes$-idempotency, which implies $(e, \eleme\otimes \elem_1\otimes\elem_2\otimes\elem_3)\in C^\Imc$. 
	Thus we get that $\Omc^{\semiringshort}\models (A\sqsubseteq C, \elem_1\otimes\elem_2\otimes\elem_3)$ and $\Pmc(A\sqsubseteq C,\Omc^{\semiringshort})=\elem_1\otimes\elem_2\otimes\elem_3$. 
	
\new{We have seen in Example~\ref{ex:timesidemforcanmodelgci} that this intuitive behavior is lost if we consider the
semiring $\semiringshort=\why$ and $\elem_1=\elem_2=\elem_3=1$, since in this case there is no $\elem\in\why$ such that $\Omc^{\why}\models (A\sqsubseteq C, \elem)$, so that $\Pmc(A\sqsubseteq C,\Omc^{\why})=0$. }
\end{example}
\new{Requiring $\otimes$-idempotency disregards the number of times an axiom is used in a derivation. 
Consider the ontology $\Omc^{\mi{Prov}[\semiringVars]}=\{(A\sqsubseteq B, x_1), (B\sqsubseteq A, x_2)\}$. 
If $\mi{Prov}[\semiringVars]=\series$, which is  not \timesidem, 
$\Pmc(A\sqsubseteq B,\Omc^{\series}) =\sum_{n\geq 0}x_1(x_1x_2)^n$, while if we consider the $\otimes$-idempotent semirings 
\posbool or \lin, we obtain $\Pmc(A\sqsubseteq B,\Omc^{\posbool}) =x_1$ and $\Pmc(A\sqsubseteq B,\Omc^{\lin}) = x_1 x_2$ respectively. }

\subsection{Consistency with Semiring Provenance for UCQs}
Theorems \ref{th:sem-algebra-cons-query} and \ref{th:sem-algebra-cons} ensure that our notion of provenance is 
consistent with the one defined for relational algebra queries over annotated databases. 
Given an annotated set of facts $\Dmc^\semiringshort=\tup{\Dmc,\lambda}$ and a Boolean UCQ 
$q=\exists\vec{x} \bigvee_{i=1}^n \varphi_i(\vec{x})$, we denote the relational provenance of $q$ over 
$\Dmc^\semiringshort$ seen as an annotated database by 
\[
\provdb(q,\Dmc^\semiringshort) = 
	\bigoplus_{i=1}^n\bigoplus_{\pi\in\Pi(\varphi_i,\Dmc) }\bigotimes_{P(\vec{t})\in \varphi_i}\lambda(\pi(P(\vec{t}))), 
\]
where $\Pi(\varphi_i,\Dmc)$ is the set of matches of $\varphi_i$ in $\Dmc$ \cite{Green07-provenance-seminal}. \new{Note that in this section we do not need to assume that the semiring is ($\omega$-)complete since the number of matches in $\Dmc$ is finite.}

\begin{restatable}{theorem}{thsemalgebraconsquery}\label{th:sem-algebra-cons-query}
Let $\semiringshort$ be a commutative \plusidem semiring. For every $\semiringshort$-annotated 
ontology $\Omc^\semiringshort=\tup{\Omc,\lambda}$ containing only assertions \new{(\ie being an annotated set of facts)}, and every 
BCQ $q$, \new{$\Pmc(q,\Omc^{\semiringshort})=\provdb(q,\Omc^\semiringshort)$.}
\end{restatable}
\new{The proof of Theorem~\ref{th:sem-algebra-cons-query} is based on the facts that (i) when $\Omc^{\semiringshort}$ is a set of annotated assertions, the canonical model $\Imc_{\Omc^{\semiringshort}}$  of $\Omc^\semiringshort$ is simply the model with domain $\NI$ that satisfies exactly these assertions, so the matches for $q$ in $\Omc$ correspond with the matches of the extended version $\ext{q}$ of $q$ in $\Imc_{\Omc^{\semiringshort}}$, and (ii) $\Omc^{\semiringshort}\models (q,\elem)$ iff $\Imc_{\Omc^{\semiringshort}}\models (q,\elem)$ (by Theorem~\ref{thm:can-model-main}).} 
Example~\ref{ex:pb-non-plus-idem-query} shows that $\oplus$-idempotency is a necessary condition of Theorem~\ref{th:sem-algebra-cons-query}.  
 
\begin{example}\label{ex:pb-non-plus-idem-query}
\new{	Let $\semiringshort$ be a commutative semiring which is not \plusidem and $\elem\in\semiringset$ such that $\elem\oplus\elem\neq\elem$. 
	Let $\Omc^\semiringshort=\{(A(a),\elem), (A(b),\elem)\}$ and $q=\exists y\, A(y)$. 
	There are exactly two matches for the extended version $\ext{q}=\exists yt\, A(y,t)$ of $q$ in $\Imc_{\Omc^\semiringshort}$: $\pi_1(y)=a$ and  $\pi_1(t)=\elem$, and $\pi_2(y)=b$ and  $\pi_2(t)=\elem$. 
	Hence $\Pmc(q, \Omc^\semiringshort) =\bigoplus_{\Omc^\semiringshort\models(q,\chi)} \chi= \bigoplus_{\Imc_{\Omc^\semiringshort}\models(q,\chi)} \chi=\elem$, which is different from  the relational provenance of $q$ over $\Omc^\semiringshort$, $\provdb(q,\Omc^\semiringshort)=\elem\oplus\elem$.}
	
\new{Theorem~\ref{th:sem-algebra-cons-query} does not hold when $\oplus$ is not idempotent even if all 
assertions have distinct annotations. Consider 
$\Omc^{\new{\series}}=\{(A(a),x_1), (A(b),x_2)\}$ and $q=\exists yz\, A(y)\land A(z)$. 
There are exactly four matches for the extended version $\ext{q}=\exists yzt_1t_2\, A(y,t_1)\wedge A(z,t_2)$ of $q$ in $\Imc_{\Omc^{\series}}$, obtained by mapping $(y, t_1)$ either to $(a,x_1)$ or to $(b,x_2)$, and similarly for $(z, t_2)$. 
We thus get 
$\Pmc(q, \Omc^{\new{\series}} )= x_1^2+x_2^2+x_1x_2$ while $\provdb(q,\Omc^{\new{\series}})= x_1^2+x_2^2+\mathbf{2}x_1x_2$. Our notion of provenance does not distinguish between the match that maps $y$ to $a$ and $z$ to $b$ and the one that maps $y$ to $b$ and $z$ to $a$.}

\new{This is actually not surprising since our provenance is defined by summing over a set of semiring elements (\cf Section~\ref{sec:designchoices} for the discussion of this design choice).}
\end{example}

The following theorem extends the \new{comparison between our notion of provenance and provenance for relational databases} 
to the case of concept assertion queries where the ontology is allowed to have a restricted form of GCIs. The intuition behind the theorem is that one can
rewrite atomic concept queries into UCQs by unfolding the wanted concept name into the prerequisites for deriving it.

\begin{restatable}{theorem}{thsemalgebracons}\label{th:sem-algebra-cons}
Let $\semiringshort$ be a commutative \plusidem semiring, and $A$ a fixed (but arbitrary)
concept name. For every 
$\semiringshort$-annotated ontology $\Omc^\semiringshort=\tup{\Omc,\lambda}$ that contains only (i)~assertions and 
(ii)~GCIs of the form $C\sqsubseteq A$, labelled with $\one$, where 
$C$ is an \ELHI concept not containing $A$, for every 
$a\in\NI$, \new{$\Pmc(A(a),\Omc^{\semiringshort})=\provdb(q,\Dmc^\semiringshort)$ where $\Dmc^\semiringshort$ is the set of the annotated assertions of $\Omc^\semiringshort$ and}
\[
q:= A(a)\vee\bigvee_{C\sqsubseteq A\in\Omc}q_{C}(a),
\]
with $q_{C}(x)$ the rooted tree-shaped query that retrieves all instances of $C$ \new{(\cf Section~\ref{sec:prelimDL})}.
\end{restatable}

\begin{remark}[On the possibility of extending Theorem \ref{th:sem-algebra-cons} to UCQ-rewritable queries]
\new{A query $q$ is \emph{UCQ-rewritable} \wrt a set of GCIs and RIs (TBox) $\Tmc$ if there exists a UCQ $q'$ such that for every ontology $\Omc=\Tmc\cup\Amc$ with $\Amc$ a set of assertions (ABox), it holds that $\Omc\models q$ iff $q'$ is satisfied by the interpretation $\Imc_\Amc$ with domain $\NI$ that satisfies exactly the assertions in $\Amc$. For example, it is well-known that every BCQ is UCQ-rewritable \wrt a \DLLITE TBox \cite{DBLP:journals/jar/CalvaneseGLLR07}.} 
\new{One could try to extend Theorem~\ref{th:sem-algebra-cons} to $\semiringshort$-annotated $\ELHIbot$ ontologies whose GCIs and RIs are annotated by $\one$ and BCQs that are UCQ-rewritable \wrt the set of GCIs and RIs in~\Omc, in order to obtain $\Pmc(q,\Omc^{\semiringshort})=\provdb(q',\Dmc^\semiringshort)$ with $q'$ a UCQ rewriting of $q$. 
However, Theorem~\ref{th:sem-algebra-cons} cannot be extended in this way without further restrictions. Consider $\Omc^{\boolseries}=\{(A(a),x_1),(B(a),x_2),(A\sqsubseteq C,1),(B\sqcap C\sqsubseteq A,1)\}$ and $q=A(a)$. Clearly, $q$ is UCQ-rewritable \wrt the GCIs in $\Omc$ (a UCQ-rewriting is $A(a)\vee (B(a)\wedge C(a))$) but $\Pmc(A(a),\Omc^{\boolseries})=\Sigma_{i\geq 0}x_1x_2^i$ while the provenance of a UCQ over a finite set of annotated facts in \boolseries is always a \emph{finite} sum of monomials.}

\new{Even if we add some kind of non-recursivity condition on $\Omc$, extending Theorem~\ref{th:sem-algebra-cons} would not be straightforward. Indeed, existing rewriting algorithms crucially rely on unification and minimization to allow, \eg to rewrite $\exists yz R(x,y)\wedge R(z,y)$ into $A(x)$ \wrt $A\sqsubseteq \exists R$. 
Hence, if $q'$ is obtained with such an algorithm, the result would not hold. 
Consider the following case: $\Omc^{\semiringshort}=\{(A(a),\elem), (A\sqsubseteq \exists R,\one), (\exists R^-\sqsubseteq B,\one), (\exists R.B\sqsubseteq C,\one)\}$. 
A UCQ rewriting of $C(a)$ \wrt $\Omc$ is a disjunction of the following CQs (or of a subset of irredundant queries): $C(a)$,  $\exists y R(a,y)\wedge B(y)$, $\exists yz R(a,y)\wedge R(z,y)$, $\exists y R(a,y)\wedge R(a,y)$, $\exists y R(a,y)$, and $A(a)$. Hence its relational provenance over $\{(A(a),\elem)\}$ is $\elem$. 
However, $\Pmc(C(a),\Omc^{\semiringshort})=\elem\otimes\elem$. Extending Theorem~\ref{th:sem-algebra-cons} to UCQ-rewritable queries would thus require to design rewriting algorithms that rewrite, \eg $\exists yz R(x,y)\wedge R(z,y)$ into $A(x)\wedge A(x)$ in this example.}
\end{remark}

\subsection{Consistency with Semiring Provenance for Datalog Queries}
Theorem \ref{th:sem-cons-datalog} ensures that our notion of provenance is also 
consistent with the one defined for Datalog queries over annotated databases \cite{Green07-provenance-seminal}. 
Following the notation of \citet{DBLP:conf/kr/BourgauxBPT22}, given an annotated database $\Dmc^\semiringshort=\tup{\Dmc,\lambda}$ and Datalog progam 
$\Sigma$, we denote the Datalog provenance of a fact $\alpha$ \wrt $\Sigma$ and 
$\Dmc^\semiringshort$ by 
\[
\provdat(\Sigma,\Dmc^\semiringshort,\alpha) = 
\bigoplus_{ t\in  T^\Sigma_\Dmc(\alpha)}  \  \bigotimes_{v\text{ is a leaf of $t$}} \lambda(v)
\]
where $T^\Sigma_\Dmc(\alpha)$ is the set of all derivation trees for $\alpha$ \wrt $\Sigma$ and $\Dmc$. 
\new{In the Datalog provenance literature, semirings are usually assumed to be $\omega$-continuous, but 
$\omega$-completeness is sufficient for $\provdat(\Sigma,\Dmc^\semiringshort,\alpha)$ to be well-defined. Hence, in this section, we consider $\omega$-complete semirings. We do not need to require that the semirings are complete because the ontologies we consider (corresponding to Datalog programs) are trivially satisfiable, 
which means that the set $\{\elem\mid\Omc^\semiringshort\models(\alpha,\elem)\}$ is guaranteed to be countable even if the semiring domain is not.}

\begin{theorem}\label{th:sem-cons-datalog}
For each  \ELHI concept $C$, let $q_C(x)$ be the rooted tree-shaped CQ that corresponds to $C$, and for every GCI or RI $\alpha$, let $F_\alpha$ be a \new{fresh} nullary predicate. 

Let $\semiringshort$ be a commutative \plusidem \new{$\omega$-complete} semiring. For every $\semiringshort$-annotated 
ontology $\Omc^\semiringshort=\tup{\Omc,\lambda}$ whose GCIs have only concept names as right-hand sides and RIs are positive, for every 
BCQ $q$, \new{$\Pmc(q,\Omc^{\semiringshort})=\provdat(\Sigma,\Dmc^\semiringshort,\mn{goal})$ with}  
\begin{align*}
&\Sigma=\{q_C(x)\wedge F_{C\sqsubseteq A}\rightarrow A(x) \mid C\sqsubseteq A\in\Omc\}
\\&\quad\quad\cup \{P(x,y)\wedge F_{P\sqsubseteq Q}\rightarrow Q(x,y) \mid P\sqsubseteq Q\in\Omc\}
\\&\quad\quad\cup\{q\rightarrow\mn{goal}\}\\
&\Dmc^\semiringshort=\tup{\Dmc,\lambda'}\\
\text{ with }&\Dmc=\{A(a)\mid A(a)\in\Omc\}\cup\{R(a,b)\mid R(a,b)\in\Omc\}\cup\{F_\alpha\mid \alpha\text{ is a GCI or RI of }\Omc\}\\
\text{ and }&\lambda'(\alpha)=\lambda(\alpha)\text{ if }\alpha\text{ is an assertion of }\Omc\\
&\lambda'(F_{\alpha})=\lambda(\alpha)\text{ if }\alpha\text{ is a GCI or RI of }\Omc.
\end{align*}
\end{theorem}
Since the Datalog provenance of a CQ seen as a Datalog query coincides with its relational provenance \new{(formally, for every BCQ $q=\exists \vec{y}\phi(\vec{y})$, and annotated set of facts $\Dmc^\semiringshort$, $\provdb(q,\Dmc^\semiringshort)=\provdat(\{\exists \vec{y}\phi(\vec{y})\rightarrow\mn{goal}\},\Dmc^\semiringshort,\mn{goal})$)}  \cite{Green07-provenance-seminal,DBLP:conf/kr/BourgauxBPT22}, Theorem~\ref{th:sem-cons-datalog} does not hold when the addition is not idempotent. 

\begin{example}\label{ex:datalog-provenance}
\new{Consider $\Omc^{\new{\series}}=\{(A(a),x_1), (A(b),x_2)\}$ defined in Example~\ref{ex:pb-non-plus-idem-query} and $\Sigma=\{\exists yz\, A(y)\land A(z)\rightarrow\mn{goal}\}$. We have $\provdat(\Sigma,\Omc^{\new{\series}},\mn{goal}) = x_1^2+x_2^2+2x_1x_2$ because $\mn{goal}$ has the following derivation trees \cite[Definition 1]{DBLP:conf/kr/BourgauxBPT22}.}

\new{\begin{center}
\begin{tabular}{llll}
	\begin{tikzpicture}
	[level distance=0.75cm]
	\node {$\mn{goal}$}
	child {node {$A(a)$}} 
	child {node {$A(a)$}};
	\end{tikzpicture}
	&
		\begin{tikzpicture}
	[level distance=0.75cm]
	\node {$\mn{goal}$}
	child {node {$A(b)$}} 
	child {node {$A(b)$}};
	\end{tikzpicture}
	&
		\begin{tikzpicture}
	[level distance=0.75cm]
	\node {$\mn{goal}$}
	child {node {$A(a)$}} 
	child {node {$A(b)$}};
	\end{tikzpicture}
	&
		\begin{tikzpicture}
	[level distance=0.75cm]
	\node {$\mn{goal}$}
	child {node {$A(b)$}} 
	child {node {$A(a)$}};
	\end{tikzpicture}
\end{tabular} 
\end{center}}
\end{example}

We \new{prove} Theorem \ref{th:sem-cons-datalog} \new{with the help of} Proposition 3 of \citet{DBLP:conf/kr/BourgauxBPT22} which states that when $\semiringshort$ is \plusidem, then
$\provdat(\Sigma,\Dmc^\semiringshort,\mn{goal})= \provdat^{\texttt{SAM}}(\Sigma,\Dmc^\semiringshort,\mn{goal})
$ where $\provdat^{\texttt{SAM}}$ is an alternative notion of provenance semantics for Datalog (set-annotated model-based provenance semantics) defined as follows:  
$\provdat^{\texttt{SAM}}(\Sigma,\Dmc^\semiringshort,\mn{goal})= \bigoplus_{\elem\in \bigcap_{(I,\mu^I)\models (\Sigma,\Dmc^\semiringshort)} \mu^I(\mn{goal})} \elem,$
where $(I,\mu^I)$--with $I$ a set of facts and $\mu^I$ a function from $I$ 
to the power-set of $K$--is such that $(I,\mu^I)\models (\Sigma,\Dmc^\semiringshort)$ if 
\begin{enumerate}
	\item$\Dmc\subseteq I$, and for every $\alpha\in\Dmc$, $\lambda'(\alpha)\in\mu^I(\alpha)$; 
	\item   for every $\phi(\vec{x},\vec{y} ) \rightarrow H(\vec{x})$ in $\Sigma$ \new{with $\phi(\vec{x},\vec{y})=\alpha_1\wedge\dots\wedge\alpha_n$}, if there is a homomorphism $h:\{\alpha_i\mid 1\leq i\leq n\}\mapsto I$, 
	then \mbox{$h(H(\vec{x}))\in I$} and if \new{
	$h(\alpha_1)\wedge\dots\wedge h(\alpha_n)=\beta_1\wedge\dots\wedge\beta_n$, then} $\{\bigotimes_{i=1}^n \elem_i \mid (\elem_1,\dots,\elem_n)\in\mu^I(\beta_1)\times\dots\times\mu^I(\beta_n)\}\subseteq \mu^I(h(H(\vec{x})))$. 
\end{enumerate}

\begin{example}
\new{Consider the same $\Sigma$ and $\Omc^{\new{\series}}$ as in Example~\ref{ex:datalog-provenance}. It then holds 
that $\provdat^{\texttt{SAM}}(\Sigma,\Omc^{\series},\mn{goal}) = x_1^2+x_2^2+x_1x_2$. Indeed, 
$(I,\mu^I)\models (\Sigma,\Omc^{\series})$ means that:
\begin{enumerate}
\item $\{A(a), A(b)\}\subseteq I$, $x_1\in\mu^I(A(a))$ and $x_2\in\mu^I(A(a))$;
\item whenever there is a homomorphism $h:\{A(y), A(z)\}\mapsto I$, then $\mn{goal}\in I$ and if $h(A(y))\wedge h(A(z))=\beta_1\wedge\beta_2$, then $\{p_1\times p_2\mid (p_1,p_2)\in \mu^I(\beta_1)\times\mu^I(\beta_2)\}\subseteq \mu^I(\mn{goal})$.
\end{enumerate}
It follows that for every $(I,\mu^I)\models (\Sigma,\Omc^{\series})$, $\mn{goal}\in I$ and $\{x_1^2, x_2^2,x_1x_2\}\subseteq \mu^I(\mn{goal})$. Moreover, one can check that $(I,\mu^I)$ with $I=\{A(a),A(b),\mn{goal}\}$ and $\mu^I(A(a))=\{x_1\}$, $\mu^I(A(b))=\{x_2\}$, $\mu^I(\mn{goal})=\{x_1^2,x_2^2,x_1x_2\}$ is such that $(I,\mu^I)\models (\Sigma,\Omc^{\series})$. Hence, $\bigcap_{(I,\mu^I)\models (\Sigma,\Omc^{\series})}\mu^I(\mn{goal})=\{x_1^2,x_2^2,x_1x_2\}$.}

\new{Notice that in this example we use a semiring which is \emph{not} \plusidem and that $\provdat^{\texttt{SAM}}(\Sigma,\Omc^{\series},\mn{goal})\neq\provdat(\Sigma,\Omc^{\series},\mn{goal})$.
}
\end{example}

Theorem \ref{th:sem-cons-datalog} is thus a direct consequence of the following lemma, which shows the correspondance between our provenance semantics and the set-annotated model-based Datalog provenance semantics (note that the semiring is not required to be \plusidem here).

\begin{restatable}{lemma}{thsemconsdatalogSAM}\label{th:sem-cons-datalog-SAM}
If $\semiringshort$ is a commutative \new{$\omega$-complete} semiring, then for every $\semiringshort$-annotated 
ontology $\Omc^\semiringshort=\tup{\Omc,\lambda}$ whose GCIs have only concept names as right-hand sides and RIs are positive, for every 
BCQ $q$, $\Pmc(q,\Omc^{\semiringshort})=\provdat^{\texttt{SAM}}(\Sigma,\Dmc^\semiringshort,\mn{goal})$ where $\Sigma$ and 
$\Dmc^\semiringshort$  are defined as in Theorem \ref{th:sem-cons-datalog}.
\end{restatable}

\subsection{Commutation with Homomorphisms}
Theorem \ref{th:sem-commut-hom} \new{states the relationship between the provenance values computed in different semirings: under some conditions, provenance computation ``commutes'' with homomorphisms in the sense that if a $\semiringshort_2$-annotated ontology is obtained from a $\semiringshort_1$-annotated ontology via a semiring homomorphism, then the consequence provenance in $\semiringshort_2$ can be obtained by applying this homomorphism to the consequence provenance in $\semiringshort_1$. This is particularly useful to ensure} that the provenance expression computed in a provenance semiring \new{(such as $\boolseries$ or $\posbool$)} can be used to compute the provenance value in any semiring to which it specializes correctly. \new{Since the theorem applies to satisfiable ontologies, we only need to require the semiring to be $\omega$-complete since $\{\elem\mid\Omc^\semiringshort\models(\alpha,\elem)\}$ is guaranteed to be countable.}

\begin{restatable}{theorem}{thsemcommuthom}\label{th:sem-commut-hom}
Let $\semiringshort_1$ and $\semiringshort_2$ be commutative \new{$\omega$-complete} semirings such that there is a \new{$\omega$-complete} semiring homomorphism $h$ from $\semiringshort_1$ to $\semiringshort_2$. For every satisfiable $\ELHIbot$ ontology $\Omc$ and annotated versions $\Omc^{\semiringshort_1}=\tup{\Omc,\lambda}$ and $\Omc^{\semiringshort_2}=\tup{\Omc,h\circ\lambda}$, \new{$$h(\Pmc(\alpha,\Omc^{\semiringshort_1})) = \Pmc(\alpha, \Omc^{\semiringshort_2})\quad\text{holds if:}$$
\begin{enumerate}
\item $\semiringshort_1$ and $\semiringshort_2$ are \plusidem and $\alpha$ is a BCQ, an assertion, or an RI whose left-hand side is satisfiable \wrt $\Omc$; or
\item $\semiringshort_1$ and $\semiringshort_2$ are \plusidem and \timesidem, $\Omc$ does not contain any GCI with $\top$ as left-hand side, and $\alpha$ is a GCI whose left-hand side is satisfiable \wrt $\Omc$.	
\end{enumerate}}
\end{restatable}

\new{The proof of Theorem~\ref{th:sem-commut-hom} relies on the canonical models of $\Omc^{\semiringshort_1}$ and $\Omc^{\semiringshort_2}$ and Theorems~\ref{thm:can-model-main}, \ref{thm:can-model-main-gci} and \ref{thm:can-model-main-ri}. The following example illustrates how this theorem can be used and shows why the assumption that $\semiringshort_1$ and $\semiringshort_2$ are \timesidem is necessary for the GCI case.}

\begin{example}
\new{Consider $\Omc=\{A(a),\ A\sqsubseteq B_1,\ A\sqsubseteq B_2,\ B_1\sqcap B_2\sqsubseteq C\}.$}
\begin{itemize}
\item \new{It is well-known that there exists a semiring homomorphism from the tropical semiring 
	$\mathbb{T}=(\mathbb{R}^\infty_+, \min , +, \infty, 0)$ to the Viterbi semiring $\mathbb{V}=([0,1], \max, \times, 0,1)$ 
	defined by $h(x)=e^{-x}$ for every $x\in\mathbb{R}_+$ and $h(\infty)=0$. Let 
	$\Omc^{\mathbb{T}}=\tup{\Omc,\lambda_\mathbb{T}}$ with $\lambda_\mathbb{T}(A(a))=0$, 
	$\lambda_\mathbb{T}(A\sqsubseteq B_1)=2$, $\lambda_\mathbb{T}(A\sqsubseteq B_2)=3$, and 
	$\lambda_\mathbb{T}(B_1\sqcap B_2\sqsubseteq C)=0$, and consider the $\mathbb{V}$-annotated ontology
	$\Omc^{\mathbb{V}}=\tup{\Omc,\lambda_\mathbb{V}}$ with $\lambda_\mathbb{V}=h\circ\lambda_\mathbb{T}$. Since $\mathbb{T}$ and $\mathbb{V}$ are additively idempotent, $\Pmc(C(a),\Omc^{\mathbb{V}})=h(\Pmc(C(a),\Omc^{\mathbb{T}}))=h(5)=e^{-5}$ by point 1 of Theorem~\ref{th:sem-commut-hom}.}
\item\new{Consider now fully idempotent semirings $\posbool$ and $\mathbb{F}=([0,1],\max,\min ,0,1)$, and let $\Omc^{\posbool}=\tup{\Omc,\lambda_\semiringVars}$ with $\lambda_\semiringVars(A(a))=y$, $\lambda_\semiringVars(A\sqsubseteq B_1)=x_1$, $\lambda_\semiringVars(A\sqsubseteq B_2)=x_2$, and $\lambda_\semiringVars(B_1\sqcap B_2\sqsubseteq C)=x_3$, and $\Omc^{\mathbb{F}}=\tup{\Omc,\lambda_\mathbb{F}}$ with $\lambda_\mathbb{F}=h\circ\lambda_\semiringVars$ for the unique semiring homomorphism $h:\posbool\rightarrow\mathbb{F}$ such that $h(y)=0.2$, $h(x_1)=0.5$, $h(x_2)=0.6$, $h(x_3)=1$, and $h(z)=1$ for every other $z\in\semiringVars$. In this case we obtain $\Pmc(C(a),\Omc^{\mathbb{F}})=h(\Pmc(C(a),\Omc^{\mathbb{\posbool}}))=h(yx_1x_2x_3)=0.2$ by point 1 and $\Pmc (A\sqsubseteq C,\Omc^{\mathbb{F}})=h(\Pmc(A\sqsubseteq C,\Omc^{\mathbb{\posbool}}))=h(x_1x_2x_3)=0.5$ by point 2 of Theorem~\ref{th:sem-commut-hom}.}
\item \new{To see why $\semiringshort_1$ has to be \timesidem when $\alpha$ is a GCI, take 
	$\Omc^{\why}=\tup{\Omc,\lambda_\semiringVars}$ and $\Omc^{\posbool}=\tup{\Omc,\lambda_\semiringVars}$ with 
	$\lambda_\semiringVars$ as in the previous point. There is a unique semiring homomorphism  $h:\why\rightarrow \new{\posbool}$ such that $h$ is the identity over $\semiringVars$. Since the semirings are \plusidem, it holds that $\Pmc(C(a),\Omc^{\posbool})=h(\Pmc(C(a),\Omc^{\why}))=h(yx_1x_2x_3)=yx_1x_2x_3$. However, since \why is not \timesidem, one can show (in the same way as in Example~\ref{ex:timesidemforcanmodelgci}) that $\Pmc(A\sqsubseteq C,\Omc^{\why})=0$. Hence $h(\Pmc(A\sqsubseteq C,\Omc^{\why}))=h(0)=0$ is different from $\Pmc(A\sqsubseteq C,\Omc^{\posbool})=x_1x_2x_3$.}
\item \new{Finally, to see why $\semiringshort_2$ has to be \timesidem for the GCI case, we define the 
	\plusidem commutative semiring $\semiringshort=(\{0,1,c\},\oplus,\otimes,0,1)$ where $1\oplus c=1$ and 
	$c\otimes c=0$ (all other values of sums or products are implied by the properties of a \plusidem commutative 
	semiring). One can check that this  indeed defines a semiring which is not \timesidem since $c\otimes c\not=c$.
	We then consider the ontology $\Omc^{\posbool}=\tup{\Omc,\lambda_\semiringVars}$ with $\lambda_\semiringVars$ 
	as before and $\Omc^{\semiringshort}=\tup{\Omc,\lambda_{\semiringshort}}$ with 
	$\lambda_{\semiringshort}=h\circ\lambda_\semiringVars$ for the unique semiring homomorphism 
	$h:\posbool\rightarrow \semiringshort$ such that $h(z)=1$ for every $z\in\semiringVars$ (\ie $h(p)=1$ for every 
	$p\in\posbool\setminus\{0\}$). 
	As we have seen already, $\Pmc(A\sqsubseteq C,\Omc^{\posbool})=x_1x_2x_3$ so $h(\Pmc(A\sqsubseteq C,\Omc^{\posbool}))=1$. However, one can show that $\Pmc(A\sqsubseteq C,\Omc^{\semiringshort})=0$ by considering the following $\semiringshort$-annotated model of $\Omc^\semiringshort$.
\begin{align*}
		A^{\Imc}=&\{(a^\Imc,1),(e,c)\} \ B_1^{\Imc}=\{(a^\Imc,1),(e,c)\} \ B_2^{\Imc}=\{(a^\Imc,1),(e,c)\} \ C^{\Imc}=\{(a^\Imc,1),(e,0)\}
\end{align*}
Indeed, since there is no $\elem\in\{0,1,c\}$ such that $1\otimes\elem=1$ and $c\otimes\elem=0$, there is no $\elem$ such that $\Imc\models (A\sqsubseteq C,\elem)$.\qedhere
 }
\end{itemize}
\end{example}

Example~\ref{ex:pb-non-plus-idem} shows that Theorem
~\ref{th:sem-commut-hom} does not hold without the assumption that the addition is idempotent, even for simple ontologies with only concept name inclusions.

\begin{example}\label{ex:pb-non-plus-idem}
	Let $\semiringshort$ be a commutative $\omega$-complete semiring which is not \plusidem and $\elem\in\semiringset$ such that $\elem\oplus\elem\neq\elem$. 
	Let 
	\begin{align*}
		\Omc=&\{A(a), B(a), A\sqsubseteq C, B\sqsubseteq C\}\\
		\Omc^{\new{\series}}=&\{(A(a),x), (B(a),y), (A\sqsubseteq C,u), (B\sqsubseteq C,v)\}\\
		\Omc^\semiringshort=&\{(A(a),\elem), (B(a),\elem), (A\sqsubseteq C,\one), (B\sqsubseteq C,\one)\}. 
	\end{align*}
	The provenance of $C(a)$ \wrt $\Omc^{\new{\series}}$ is $\Pmc(C(a),\Omc^{\new{\series}})=x u+y v$, hence, if we consider the semiring homomorphism $h$ such that $h(x)=h(y)=\elem$ and $h(u)=h(v)=\one$, $h(\Pmc(C(a),\Omc^{\new{\series}}))=\elem\oplus\elem$.  
	However, $\Pmc(C(a),\Omc^\semiringshort)=\elem$ since there exists a model $\Imc$ of  $\Omc^\semiringshort$ such that $C^\Imc=\{(a^\Imc,\elem)\}$. 
	Hence $h(\Pmc(C(a),\Omc^{\new{\series}}))\neq \Pmc(C(a), \Omc^\semiringshort)$. 
\end{example}
  
\new{Theorem~\ref{th:sem-commut-hom} does not hold if the ontology $\Omc$ is unsatisfiable or if $\alpha$ is a GCI or an RI whose left-hand side is unsatisfiable \wrt $\Omc$ since in these cases the provenance of $\alpha$ is the sum of all the elements of the semiring (see Remarks~\ref{rem:unsat-onto} and~\ref{rem:unsat-left}). Indeed, if $\semiringshort_1=\posbool$, $\semiringshort_2=(\{0,1,a\}, \oplus,\otimes,0,1)$ is the commutative semiring such that both operations are idempotent and $1\oplus a=a$ (all other values of sums or products are implied by the properties of a fully idempotent commutative semiring) and 
$h$ is the semiring homomorphism defined by $h(0)=0$, $h(1)=1$ and $h(x)=1$ for every $x\in\semiringVars$, then 
we have that 
$h(\Pmc(\alpha,\Omc^{\semiringshort_1}))=h(0\vee 1\vee\bigvee_{S\subseteq \semiringVars}\bigwedge_{x\in S}x)=1$ while $\Pmc(\alpha, \Omc^{\semiringshort_2})=0\oplus 1\oplus a=a$. However, if we require $\semiringshort_1$ and $\semiringshort_2$ to be absorptive, the theorem still holds for unsatisfiable ontologies: if $\semiringshort_1 = (K_1, \oplus ,\otimes, \zero, \one)$ and $\semiringshort_2 = (K_2, + ,\cdot, 0, 1)$ are commutative complete (or $\omega$-complete and $K_1$, $K_2$ are countable) and absorptive semirings, $h$ is a ($\omega$-)complete homomorphism from $\semiringshort_1$ to $\semiringshort_2$, and $\Omc$ is unsatisfiable, then 
$$h(\Pmc(\alpha,\Omc^{\semiringshort_1})) =h(\bigoplus_{\elem\in K_1}\elem)=h(\one)=1=\bigplus_{\elem'\in K_2}\elem'.$$ 
Indeed, if $\semiringshort$ is absorptive, $\one\oplus \elem=\one\oplus(\one\otimes \elem)=\one$ for every $\elem\in\semiringset$, which means that 
$\bigoplus_{\elem\in K}\elem=\one\oplus\bigoplus_{\elem\in K\setminus\{\one\}}\elem=\one$.}


\section{Computing \new{Query} Provenance in the \why Semiring}\label{sec:why}
\new{We now focus on computing the provenance of assertions and BCQs in \mbox{\why-}annotated ontologies. 
Recall from the previous section that since \why is \plusidem, all desirable properties we considered hold for the provenance of assertions and BCQs. 
While the annotations of such an ontology $\Omc^{\why}$ can normally take any values from $\why\setminus\{0\}$, \emph{we focus here on the special case where $\Omc$ is annotated 
by a function $\lambda_\semiringVars:\Omc\mapsto\semiringVars\cup\{1\}$}. Indeed, we are interested in capturing the well-known \emph{why-provenance} defined in the database context, \blue{which is a standard way of providing \emph{explanations} for a query result by associating it to the sets of database tuples from which it can be derived, representing tuples with identifiers. This notion has also} 
recently received interest in the context of Datalog \shortcite{DBLP:conf/ruleml/ElhalawatiKM22,DBLP:journals/pacmmod/CalauttiLPS24,DBLP:conf/aaai/CalauttiLPS24,DBLP:journals/pacmmod/CalauttiLPS24b}.  
\blue{Why-provenance can be captured by }
annotating each axiom of $\Omc$ by a distinct variable from $\semiringVars$ (assuming that $\semiringVars$ has a greater cardinality than $\Omc$). However, since we will require the ontologies to be in \emph{normal form}, and the normalization process described in Section~\ref{subsection:normalization} introduces axioms annotated with $1$, we allow $\lambda_\semiringVars$ to also take this value.} 
\new{Recall that 
by Theorem~\ref{th:normal-form} and Corollary~\ref{cor:normal-form}, one can normalize an annotated ontology in polynomial time while preserving annotated entailments and provenance of axioms and queries.} 

\new{Thus, for the rest of this section, a $\why$-annotated $\ELHIbot$ ontology is a pair $\tup{\Omc,\lambda_\semiringVars}$ where 
$\lambda_\semiringVars:\Omc\mapsto\semiringVars\cup\{1\}$. For brevity, we omit the superscript and identify $\Omc$ and 
$\tup{\Omc,\lambda_\semiringVars}$.  We denote by $\mn{Card}(\Omc)$ the cardinality of $\Omc$ (that is, the number of axioms in $\Omc$), and given $p\in\why$, we denote by $|p|$ the length of the string that represents it as a sum of monomials, where variables from $\semiringVars$ are considered of length one.}

We start with a general proposition that shows that we can focus on \emph{monomials}. \new{Recall that we denote by $\monomials(\semiringVars)$ the set of all monomials over $\semiringVars$.} \blue{The proof of Proposition~\ref{prop:FocusMonomialForWhy} uses the canonical model of $\Omc$.}

\begin{restatable}{proposition}{propFocusMonomialForWhy}\label{prop:FocusMonomialForWhy}
Let $\Omc$ be a satisfiable $\why$-annotated $\ELHIbot$ ontology \blue{(annotated 
by $\lambda_\semiringVars:\Omc\mapsto\semiringVars\cup\{1\}$)} and $\monomial$ be an element of $\why$. 
If $\Omc\models(\alpha,\monomial)$ where $\alpha$ is \new{an assertion or a BCQ, 
then $\monomial\in\monomials(\semiringVars)$.}
\end{restatable}
We can thus compute $\Pmc(\alpha,\Omc)$ by finding all monomials $\monomial$ over $\semiringVars$ such that $\Omc\models(\alpha,\monomial)$.


\subsection{Annotated \new{Assertion} Entailment from $\ELHIbot$ Ontologies}\label{sec:completion}

In this section, we present a completion algorithm for deriving \new{annotated assertions entailed} from a \why-annotated $\ELHIbot$ ontology. \new{For DLs of the \EL family, completion algorithms are a classical way to derive axioms by saturating an ontology in normal form by applying so-called \emph{completion, or saturation, rules} \cite{BBL-IJCAI05,DBLP:conf/rweb/BienvenuO15}.} 

\subsubsection{Completion Algorithm for $\ELHIbot$}\label{sec:subcompletionalgo}

\begin{table*}[tb]
\caption{Completion rules. $A_{(i)},B_{(i)},C,D\in\NC\cup\{\top,\bot\}$, $R,S \in\NR$, $P_{(i)},Q:=R\mid R^-$, $M, N$ are (possibly empty) conjunctions over $\NC\cup\{\top\}$, and $\monomial_{(i)},\nonomial_{(i)}, \onomial_{(i)}  \in\monomials(\semiringVars)$ (recall that \new{by definition of \why,} 
monomials do not contain repeated variables). 
Empty conjunctions are written $\top$ and non-empty conjunctions are written without $\top$ (in particular in $\CR^T_{3}$, if some but not all $A_i$ are equal to $\top$, they do not explicitly occur in the conjunction $A\sqcap A_1\sqcap\dots\sqcap A_k$).
Conjunctions are treated as \new{multisets with maximal multiplicity $\mn{Card}(\Omc)$: the order does not matter but there may contain repetitions such as $A\sqcap A$ and each concept name occurs at most $\mn{Card}(\Omc)$ times (\ie $\underbrace{A\sqcap\dots\sqcap A}_{\mn{Card}(\Omc)+k \text{ times}}$ is written as 
$\underbrace{A\sqcap\dots\sqcap A}_{\mn{Card}(\Omc) \text{ times}}$)}. 
}
\label{tab:completionRules}
\resizebox{\textwidth}{!}{
\begin{tabular}{@{}lll@{}}
\toprule
& if   &  then $\Smc \leftarrow \Smc\cup\{\Phi\}$ \\ 
\hline
\midrule
$\CR^T_{0}$ &  $(A\sqsubseteq \exists P,\monomial_0),(P\sqsubseteq P_1,\monomial_1),(P\sqsubseteq P_2,\monomial_2),(P_1\sqcap P_2\sqsubseteq \bot,\monomial_3)\in\Smc$ 
 &  $\Phi=(A\sqsubseteq \bot,{\monomial_0\times\monomial_1\times \monomial_2\times\monomial_3})$  \\
 
$\CR^T_{1}$ &  $(P_1\sqsubseteq P_2,\monomial_1),(P_2\sqsubseteq P_3,\monomial_2)\in\Smc$ 
 &  $\Phi=(P_1\sqsubseteq P_3,{\monomial_1\times \monomial_2})$  \\
 
$\CR^T_{2}$ &  $(M\sqsubseteq A,\monomial_1),(A\sqcap N\sqsubseteq C,\monomial_2)\in\Smc$ 
& $\Phi=(M\sqcap N \sqsubseteq C,{\monomial_1\times \monomial_2})$  \\	 
 		
$\CR^T_{3}$ &  
$(A\sqsubseteq \exists Q,\monomial_0)$, $(Q\sqsubseteq P, \monomial)$,  
&\\									
&  $(Q\sqsubseteq P_i,\monomial_i ), (\exists \mn{inv}(P_i).A_i\sqsubseteq B_i,\nonomial_i)$, $1\leq i\leq k$, $k\geq 0$,   
&\\	
&
$(\top\sqsubseteq B'_i,\onomial_i)$, $1\leq i\leq k'$, $k'\geq 0$,
&\\
&  
$(B_1\sqcap\dots\sqcap B_k\sqcap B'_1\sqcap\dots\sqcap B'_{k'}\sqsubseteq C,\nonomial),$ 
& 
$\Phi=(A\sqcap A_1\sqcap\dots\sqcap A_k \sqsubseteq D,$
\\
&
$(\exists P.C\sqsubseteq D, \onomial)\in\Smc$
&
$\monomial\times\nonomial\times\onomial\times\monomial_0\times\Pi_{i=1}^k(\monomial_i\times\nonomial_i)\times\Pi_{i=1}^{k'}\onomial_i)$ \\

\midrule

$\CR^A_{1}$ &  $(A_i(a),\monomial_i),1\leq i\leq k ,(A_1\sqcap \dots\sqcap A_k\sqsubseteq B,\monomial)\in\Smc$  
& $\Phi=(B(a),{\monomial_1\times \dots\times\monomial_k\times \monomial})$ \\	

$\CR^A_{2}$ &  $(R(a,b),\monomial_1),(A(b),\monomial_2),(\exists R.A\sqsubseteq B,\monomial_3)\in\Smc$  
& $\Phi=(B(a),{\monomial_1\times \monomial_2\times \monomial_3})$   \\

$\CR^A_{3}$ &  $(R(b,a),\monomial_1),(A(b),\monomial_2),(\exists R^-.A\sqsubseteq B,\monomial_3)\in\Smc$  
& $\Phi=(B(a),{\monomial_1\times \monomial_2\times \monomial_3})$  \\

$\CR^A_{4}$ &  $(R(a,b),\monomial_1),(R\sqsubseteq S,\monomial_2)\in\Smc$  
& $\Phi=(S(a,b),{\monomial_1\times \monomial_2})$  \\

$\CR^A_{5}$ &  $(R(a,b),\monomial_1),(R\sqsubseteq S^-,\monomial_2)\in\Smc$  
& $\Phi=(S(b,a),{\monomial_1\times \monomial_2})$ \\

$\CR^A_{6}$ &  $(R(a,b),\monomial_1),(S(a,b),\monomial_2),(R\sqcap S\sqsubseteq\bot,\monomial_3)\in\Smc$  
& $\Phi=(\bot(a),{\monomial_1\times \monomial_2\times\monomial_3})$ \\

$\CR^A_{7}$ &  $(R(a,b),\monomial_1),(S(b,a),\monomial_2),(R\sqcap S^-\sqsubseteq\bot,\monomial_3)\in\Smc$  
& $\Phi=(\bot(a),{\monomial_1\times \monomial_2\times\monomial_3})$ \\

\bottomrule
\end{tabular}
}
\end{table*}

We design completion rules for $\why$-annotated $\ELHIbot$ ontologies, inspired by the completion rules given by~\citet[Section 4.3]{DBLP:conf/rweb/BienvenuO15}  for $\ELHI_\bot$ \blue{(without annotations)}.  
\blue{As usual with completion algorithms, the algorithm receives as input an annotated $\ELHIbot$ ontology $\Omc$ in normal form and exhaustively applies  completion rules to construct a set $\mn{saturate}(\Omc)$, called the \emph{saturation} of $\Omc$, which contains all entailed annotated assertions over the signature of $\Omc$ and only those (see Theorem~\ref{prop:completionalgorithmELHI} for the complete statement regarding the completion algorithm).}

Given a \why-annotated $\ELHIbot$ ontology in normal form $\Omc$ \blue{(annotated 
by $\lambda_\semiringVars:\Omc\mapsto\semiringVars\cup\{1\}$)}, the completion algorithm starts with $\Omc$ completed with some trivially entailed axioms expressed in an extension of $\ELHIbot$ (featuring assertions of the form $\top(a)$ and GCIs of the form $\exists P.\bot\sqsubseteq \bot$, \blue{whose semantics is as expected: $(\top(a),\elem)$ is satisfied by an interpretation $\Imc$ if $(a^\Imc,\elem)\in\top^\Imc$, \ie if $\elem=\one$, and $(\exists P.\bot\sqsubseteq \bot,\elem)$ is satisfied if $\{(d,\elem\otimes\elem')\mid \exists e\in\Delta^\Imc \text{ s.t. } (d,e,\elem)\in P^\Imc, (e,\elem')\in\bot^\Imc\}$ is empty, which is always true since $\bot^\Imc=\emptyset$}):
\begin{align*}
\Smc:= \Omc&\cup\{(\top(a),1)\mid a\in\individuals{\Omc}\new{\cup\{a_\top\}} \}
\\&
\cup\{(A\sqsubseteq A,1) \mid A\in(\NC\cap\signature{\Omc})\cup\{\top,\bot\}\}
\\&
\cup\{(R\sqsubseteq R,1),(R^-\sqsubseteq R^-,1),(\exists R.\bot \sqsubseteq \bot,1), (\exists R^-.\bot \sqsubseteq \bot,1) \mid R\in\NR\cap\signature{\Omc}\}
\\&
\cup\{(\mn{inv}(P_1)\sqsubseteq \mn{inv}(P_2),v)\mid (P_1\sqsubseteq P_2,v)\in\Omc\}
\\&
\cup\{(\mn{inv}(P_1)\sqcap\mn{inv}(P_2)\sqsubseteq \bot,v)\mid (P_1\sqcap P_2\sqsubseteq \bot,v)\in\Omc\},
\end{align*}
where $a_\top$ is a fresh individual name from $\NI\setminus\individuals{\Omc}$. Contrary to the saturation procedure described by \citet{DBLP:conf/rweb/BienvenuO15}, we do not add \blue{axioms of the form} $A\sqsubseteq \top$: since $\top^\Imc=\Delta^\Imc\times\{1\}$, $(A\sqsubseteq \top, 1)$ does not hold in every interpretation. 

After this initialization step, the completion algorithm extends $\Smc$ with axioms through an iterative application of the rules from Table~\ref{tab:completionRules} until $\Smc$ becomes \emph{saturated}, i.e., no more rules are applicable. The resulting set is $\mn{saturate}(\Omc)$. To ensure termination, a rule is applicable only if its conclusion is not in~\Smc. 

\blue{Some remarks are in order. }
\begin{itemize}
\item The algorithm adds to $\Smc$ axioms annotated with \emph{monomials} and the definition of \why ensures that all monomials have at most $\mn{Card}(\Omc)$ variables (since they represent sets of variables that annotate $\Omc$). 
It may add annotated axioms such that the axiom is already present with some other annotation, so $\mn{saturate}(\Omc)$ is not an annotated ontology per say, but rather a \emph{set of annotated axioms}. 

\item \blue{The rules from Table~\ref{tab:completionRules} add four kinds of axioms: concept and role assertions ($\CR^A_{1}$-$\CR^A_{5}$), assertions of the form $(\bot(a),\monomial)$ ($\CR^A_{6}$ and $\CR^A_{7}$), which are used to indicate that the ontology is unsatisfiable, positive RIs ($\CR^T_{1}$), and GCIs of the form $C\sqsubseteq D$ where $C$ is either a conjunction of concept names or $\top$ and $D\in\NC\cup\{\bot\}$  ($\CR^T_0$, $\CR^T_2$ and $\CR^T_3$). 
These GCIs respect the syntactic restriction of $\ELHIbot$ but may not be in normal form. In particular, conjunctions of arbitrary size (bounded by \new{$\mn{Card}(\Omc)*|\signature{\Omc}|$}) can be introduced in the left-hand side of GCIs by $\CR^T_{2}$ and $\CR^T_{3}$.}  

\item \blue{More precisely, when applying rule $\CR^T_2$, the left-hand side conjunction is obtained as follows (and similarly for $\CR^T_3$). If the premises of an instantiation of $\CR^T_2$ are such that $M=\bigsqcap_{i=1}^{k^M} B^M_i$ and $N=\bigsqcap_{i=1}^{k^N} B^N_i$ with $B^M_i, B^N_i\in\NC$ and, $k^M, k^N\geq 0$ (with $M=\top$ if $k^M=0$ and similarly for $N$), then $M\sqcap N$ designates the conjunction obtained from the multiset $\{B^M_1,\dots, B^M_{k_M}, B^N_1,\dots, B^N_{k_N}\}$ by (i) writing the concept names in a fixed order (for example, in the lexicographic order) and (ii) limiting the number of repetitions of each concept name to $\mn{Card}(\Omc)$. This way of rewriting conjunctions of concept names is harmless because (i) for every $\ELHIbot$ concepts $C_1,C_2, C_3$, commutative semiring $\semiringshort$, and $\semiringshort$-annotated interpretation $\Imc$, $(C_1\sqcap (C_2\sqcap C_3))^\Imc=((C_1\sqcap C_2)\sqcap C_3)^\Imc$ and $(C_1\sqcap C_2)^\Imc=(C_2\sqcap C_1)^\Imc$, by associativity and commutativity of $\otimes$; and (ii) monomials produced by the algorithm have at most $\mn{Card}(\Omc)$ variables (and such a monomial with $\mn{Card}(\Omc)$ variables indicates that the axiom it annotates has been derived using all axioms of $\Omc$), so the provenance information brought by additional repetitions would be redundant.}  
\end{itemize}

\blue{The following example illustrates why we cannot directly adapt the rules presented by \citet{DBLP:conf/rweb/BienvenuO15}: they use qualified role restrictions on the right side of GCIs that may introduce wrong provenance information.}

\begin{example}\label{ex:pb-with-standard-rules}
Consider the following annotated ontology:
\begin{align*}
\new{\Omc=\{(\top\sqsubseteq A,x),\ (B\sqsubseteq \exists R.\top,y),\ (\exists R.\top\sqsubseteq C,z),\ (R\sqsubseteq R,1),\ (B(a),v)\}.}
\end{align*}
If we apply the completion rules \new{T6, T7 and A1 below, which were originally} given by \citet{DBLP:conf/rweb/BienvenuO15}, 
\new{while propagating the provenance annotations as expected, we obtain (among others) the following axioms:
$(B\sqsubseteq \exists R.A, xy)$ (by rule T6), $(B\sqsubseteq C, xyz)$ (by rule T7) and $(C(a), vxyz)$ (by rule A1). }
$$
\new{\text{T6}\ \ \frac{M\sqsubseteq \exists R.(N\sqcap N')\quad N\sqsubseteq A}{M\sqsubseteq \exists R.(N\sqcap N'\sqcap A)} \quad\quad \text{T7}\ \ \frac{M\sqsubseteq \exists R.(N\sqcap A)\quad \exists S.A\sqsubseteq B\quad R\sqsubseteq S}{M\sqsubseteq B}}$$
$$\new{\text{A1}\ \ \frac{A_1\sqcap \dots\sqcap A_n\sqsubseteq B\quad A_i(a)\ (1\leq i\leq n)}{B(a)}\quad\text{ with }M,N^{(')}\text{ conjunctions over }\NC\cup\{\top,\bot\}}
$$
\new{However, $\Omc\not\models (C(a), vxyz)$ because the following model of $\Omc$ does not satisfy $(C(a), vxyz)$: $\Delta^\Imc=\NI\cup\{\star\}$, $b^\Imc=b$ for all $b\in\NI$, and
$$A^\Imc=\Delta^\Imc\times\{x\},\quad B^\Imc=\{(a,v)\},\quad C^\Imc=\{(a,vyz)\},\quad R^\Imc=\{(a,\star,vy)\}.$$ 
Intuitively, $(\top\sqsubseteq A,x)$ is not needed when considering the
	 provenance of elements of $C$, 
	but T7 can be applied on $B\sqsubseteq \exists R.A$ (which \blue{is derived by T6 using} $\top\sqsubseteq A$) and $\exists R.\top\sqsubseteq C$ to get $B\sqsubseteq C$.} 

\new{In contrast, \blue{rules from Table~\ref{tab:completionRules}} cannot produce a conjunction or qualified role restriction on the right-hand side. When applied to $\Omc$, \blue{they yield}: 
\begin{itemize}
\item $\CR^T_{3}$ (with $k=0$ and $k'=1$) applied to $(B\sqsubseteq \exists R,y)$, $(R\sqsubseteq R, 1)$, $(\top\sqsubseteq \top,1)$ and $(\exists R.\top\sqsubseteq C,z)$ produces $(B\sqsubseteq C, yz)$ \blue{(recall that $(R\sqsubseteq R, 1)$ and $(\top\sqsubseteq \top,1)$ are added before applying the completion rules, see definition of $\Smc$ above)};
\item $\CR^A_{1}$ (with $k=1$) applied to $(B\sqsubseteq C, yz)$ and $(B(a),v)$ produces $(C(a), yzv)$;
\item $\CR^A_{1}$ (with $k=1$) applied to $(\top\sqsubseteq A, x)$ and $(\top(a),1)$ produces $(A(a), x)$, and similarly $(A(a_\top), x)$ is produced from $(\top(a_\top),1)$ (\blue{again, recall that $(\top(a),1)\in\Smc$ and $(\top(a_\top),1)\in\Smc$});
\item no other rule is applicable. \qedhere
\end{itemize}
}
\end{example}

The following example shows why it is important to allow for repetitions in the left-hand side of the GCIs introduced by $\CR^T_{2}$ and $\CR^T_{3}$.

\begin{example}\label{ex:repetitionsMatter}
\blue{Let 
\begin{align*}
\Omc=\{&(A_1(a), x_1), (A_2(a),x_2), (A_1\sqsubseteq A, y_1), (A_2\sqsubseteq A, y_2), (A\sqsubseteq\exists P, u), (\exists P^-.A\sqsubseteq C, v), (\exists P.C\sqsubseteq D, w)\}.
\end{align*}
For every model $\Imc$ of $\Omc$, one can check that for $i\in\{1,2\}$, $(a^\Imc,x_iy_i)\in A^\Imc$, so there exists some $(a^\Imc,e_i,x_iy_iu)\in P^\Imc$. Hence, for $i\in\{1,2\}$, it holds that $(e_i,x_iy_iuv)\in C^\Imc$ and $(e_i,x_1x_2y_1y_2uv)\in C^\Imc$. It follows that $(a^\Imc,x_1y_1uvw)\in D^\Imc$, $(a^\Imc,x_2y_2uvw)\in D^\Imc$ and $(a^\Imc,x_1x_2y_1y_2uvw)\in D^\Imc$. 
Hence 
$\Omc\models (D(a), x_1y_1uvw)$, $\Omc\models (D(a), x_2y_2uvw)$ and $\Omc\models (D(a), x_1x_2y_1y_2uvw)$. 
Our algorithm derives these annotated assertions as follows:
\begin{itemize}
\item by $\CR^T_{3}$ (with $k=1$ and $k'=0$) applied to $(A\sqsubseteq\exists P, u)$, $(P\sqsubseteq P, 1)$, $(\exists P^-.A\sqsubseteq C, v)$, $(C\sqsubseteq C,1)$, and $(\exists P.C\sqsubseteq D, w)$, we obtain $(A\sqcap A\sqsubseteq D, uvw)$ ;
\item  for $i\in\{1,2\}$, by $\CR^A_{1}$ (with $k=1$) applied to $(A_i(a), x_i)$ and $(A_i\sqsubseteq A, y_i)$, we obtain $(A(a), x_iy_i)$;
\item for $i\in\{1,2\}$, by $\CR^A_{1}$ (with $k=2$) applied to $(A(a), x_iy_i)$ and $(A\sqcap A\sqsubseteq D, uvw)$,  we obtain $(D(a), x_iy_iuvw)$;
\item by $\CR^A_{1}$ ($k=2$) applied to $(A(a), x_1y_1)$, $(A(a), x_2y_2)$, and $(A\sqcap A\sqsubseteq D, uvw)$, we obtain $(D(a), x_1x_2y_1y_2uvw)$.
\end{itemize}
We would not obtain the last annotated assertion (whose annotation witnesses the fact that we can use the two ontology assertions \emph{together} to obtain $D(a)$) if we wrote $A\sqsubseteq D$ instead of $A\sqcap A\sqsubseteq D$. Intuitively, the repetition of $A$ means that it plays two roles (requiring a $P$-successor, and enforcing that this successor belongs to $C$), and each of these roles can be fulfilled by a different ``cause'' for $A$.}
\end{example}

\blue{Before showing that the completion algorithm is sound and complete for deciding annotated assertion entailment, we stress out that it is not sound for \emph{annotated GCI entailment}, because \why is not multiplicatively idempotent, as shown in the following example.}

\begin{example}\label{rem:completion-not-sound-GCI}
Consider $$\Omc=\{(A\sqsubseteq \exists P, x), (\exists P^-.B\sqsubseteq C,y), (\exists P.C\sqsubseteq D,z)\}.$$ 
\blue{The completion algorithm derives $(A\sqcap B\sqsubseteq D, xyz)$ by $\CR^T_{3}$ (with $k=1$, $k'=0$) applied to $(A\sqsubseteq \exists P, x)$, $(P\sqsubseteq P,1)$, $(\exists P^-.B\sqsubseteq C,y)$, $(C\sqsubseteq C,1)$ and $(\exists P.C\sqsubseteq D,z)$. This is actually the only annotated axiom derived by the completion algorithm besides the tautologies added in the initialization step.}  However, \blue{$\Omc\not\models(A\sqcap B\sqsubseteq D, xyz)$. Indeed,} the \why-annotated interpretation $\Imc$ such that $A^\Imc=\{(e,u+v)\}$, $B^\Imc=\{(e,w)\}$, $P^\Imc=\{(e,d,ux+vx)\}$, $C^\Imc=\{(d, uwxy+vwxy)\}$, and $D^\Imc=\{(e, uwxy\blue{z}+vwxy\blue{z}+uvwxy\blue{z})\}$ is a model of $\Omc$ such that $(e,uw+vw)\in(A\sqcap B)^\Imc$ while $(e,uwxyz+vwxyz)\notin D^\Imc$ because $(u+v)^2\neq(u+v)$ in the \why semiring.

\blue{We can also use this example to show that the completion algorithm would not work if the ontology was annotated with arbitrary elements of \why. Consider $\Omc'=\Omc\cup\{(A(a),u+v), (B(a),w)\}$. As above, the completion algorithm would derive $(A\sqcap B\sqsubseteq D, xyz)$, then $(D(a), uwxyz+vwxyz)$ by $\CR^A_{1}$. However, $\Omc'\not\models (D(a), uwxyz+vwxyz)$, as witnessed by the above \why-annotated interpretation extended by $a^\Imc=e$.}
\end{example}

\blue{Hence, the annotated GCIs derived by the completion algorithm should only be considered as syntactic steps needed to compute the set of entailed \emph{annotated assertions}. We manage to nevertheless prove soundness of the algorithm by using the weaker property for GCIs that if $(C\sqsubseteq D,\monomial)\in\mn{saturate}(\Omc)$, then for every model $\Imc$ of $\Omc$, for every \emph{monomial} $\nonomial$ and domain element $e\in\Delta^\Imc$, $(e,\nonomial)\in C^\Imc$ implies $(e, \nonomial\times\monomial)\in D^\Imc$.}

\begin{restatable}{theorem}{CorrectnessELHI}
\label{prop:completionalgorithmELHI}
Let \new{$\Omc$ be a \why-annotated $\ELHIbot$ ontology \blue{(annotated 
by $\lambda_\semiringVars:\Omc\mapsto\semiringVars\cup\{1\}$)} in normal form,} $\mn{saturate}(\Omc)$ the result of saturating \Omc, $\monomial\in\monomials(\semiringVars)$, and $\alpha$ an \emph{assertion}.
\begin{enumerate}
\item If $(\alpha,\monomial)\in\mn{saturate}(\Omc)$ then $\Omc\models (\alpha,\monomial)$. \new{Moreover, if $\alpha=A(a_\top)$, then for every $c\in\NI$, $\Omc\models (A(c),\monomial)$.}

\item If $\Omc$ is satisfiable and $\Omc\models (\alpha,\monomial)$, 
\begin{enumerate}
\item \new{if $\alpha$ is of the form $A(a)$ or $R(a,b)$ with $a,b\in\individuals{\Omc}$,  then $(\alpha,\monomial)\in\mn{saturate}(\Omc)$;}
\item \new{if $\alpha$ is of the form $A(c)$ with $c\in\NI\setminus\individuals{\Omc}$,  then $(A(a_\top),\monomial)\in\mn{saturate}(\Omc)$.}
\end{enumerate}
\item If $\Omc$ is unsatisfiable, then $(\bot(a),\nonomial)\in \mn{saturate}(\Omc)$ for some $a\in\individuals{\Omc}\new{\cup\{a_\top\}}$ and $\nonomial\in\monomials(\semiringVars)$.

\item $\mn{saturate}(\Omc)$ is computable in exponential time w.r.t.\ the size of $\Omc$ \blue{(in $O(e^{p(|\Omc|)})$ with $p$ a polynomial function)}.
\end{enumerate}
\end{restatable}
\begin{proof}[\new{Proof Sketch}]
\new{For (1), we show that for every model $\Imc$ of $\Omc$, 
for every $(\alpha, {\monomial})\in\mn{saturate}(\Omc)$: (i) if $\alpha$ is an assertion or an RI (positive or negative), then $\Imc\models(\alpha, {\monomial})$, and 
(ii) if $\alpha$ is a GCI of the form $C\sqsubseteq D$, then for every \emph{monomial} $\nonomial$ and domain element $e\in\Delta^\Imc$, $(e,\nonomial)\in C^\Imc$ implies $(e, \nonomial\times\monomial)\in D^\Imc$. 
The proof is by induction on the number of completion rule applications before $(\alpha,\monomial)$ is added to $\mn{saturate}(\Omc)$.}

\new{We show (2.a) by contrapositive: 
assuming that $(\alpha, {\monomial})\notin \mn{saturate}(\Omc)$, we show that the canonical model $\Imc_\Omc=\bigcup_{i\geq 0}\Imc_i$ of $\Omc$ is such that $\Imc_\Omc\not\models(\alpha,\monomial)$. We show by induction on $i$ that for every assertion $\beta$ on the vocabulary of $\Omc$ and every monomial $\nonomial$, $(\beta, {\nonomial})\notin \mn{saturate}(\Omc)$ implies that $\Imc_i\not\models (\beta,\nonomial)$. In a nutshell, the induction step shows that if $\Imc_{i+1}\models (\beta,\nonomial)$, then $(\beta, {\nonomial})\in \mn{saturate}(\Omc)$ via a case analysis on the chase rule applied to construct $\Imc_{i+1}$ from $\Imc_i$. The most technical part is to handle the case where this chase rule uses an axiom of the form $(\exists P.A\sqsubseteq B,\onomial)$ and an anonymous domain element $x\in\Delta^{\Imc_i}\setminus\NI$. For (2.b), we remark that if $\Omc\models (A(c),\monomial)$ for some $c\in\NI\setminus\individuals{\Omc}$, then $\Omc\models (A(d),\monomial)$ for every $d\in\NI$, and, in particular, $\Omc\models (A(a_\top),\monomial)$. We then consider an ontology $\Omc'$ that extends $\Omc$ with an assertion $(\mi{Top}(a_\top),1)$ for some fresh concept $\mi{Top}$ and conclude using~(2.a).}

\new{For (3), we consider the set of annotated axioms $\Omc'$ obtained from $\mn{saturate}(\Omc)$ by replacing $\bot$ by a fresh concept name $\mi{Bot}$ in assertions and GCIs and removing negative RIs and assertions on $a_\top$. We show that if $\Omc$ is unsatisfiable, there exists $(\mi{Bot}(a), \nonomial)\in\mn{saturate}(\Omc')$ for some $a\in\individuals{\Omc}\new{\cup\{a_\top\}}$, which implies that $(\bot(a), \nonomial)\in\mn{saturate}(\Omc)$. Otherwise, we could get a model of $\Omc$ via the canonical model construction with the annotated axioms in $\Omc'$.}

Finally, for (4), we observe that the completion rules add at most an exponential number of axioms ($O(|\Omc|^{2|\Omc|+1})$) annotated with an exponential number of monomials ($O(2^{|\Omc|})$), hence in total 
\blue{$O(e^{p_1(|\Omc|)})$} 
annotated axioms\blue{, with $p_1$ a polynomial function}. 
\blue{Moreover, the number of rule premises is bounded by $p_2(|\Omc|)$ for some polynomial function $p_2$ (since the size of the conjunctions in the left-hand side of GCIs is bounded by $|\Omc|^2$), so for each rule, the number of rule instantiations \wrt the axioms already derived is bounded by $O(e^{p_1(|\Omc|)*p_2(|\Omc|)})$. Hence, for each annotated axiom in $\mn{saturate}(\Omc)$, there have been at most $O(e^{p_1(|\Omc|)*p_2(|\Omc|)})$ rule instantiation evaluations during the step that added this axiom, so we can bound the total run time of the algorithm by $O(e^{p_1(|\Omc|)+p_1(|\Omc|)*p_2(|\Omc|)})$, and the algorithm runs in exponential time.} 
\end{proof}

Due to Theorems~\ref{prop:completionalgorithmELHI} and~\ref{th:normal-form}, 
we obtain the following corollaries. \new{Note that if an assertion $\alpha$ is of the form $R(a,b)$ with $a$ or $b$ in $\NI\setminus\individuals{\Omc}$, one can directly conclude that $\Omc\not\models (\alpha,\monomial)$ if $\Omc$ is satisfiable, so Theorem~\ref{prop:completionalgorithmELHI} shows that the completion algorithm can be used to decide entailment of any annotated assertion. Moreover, for every assertion $\alpha$, 
if $\monomial$ is a monomial in $\why$ (hence without variable repetition) such that $|\monomial|>|\Omc|$, $\monomial$ contains variables that do not occur in $\Omc$ and one can directly conclude that $\Omc\not\models (\alpha,\monomial)$ if $\Omc$ is satisfiable, so the size of $\monomial$ does not matter.}

\begin{corollary}
\label{cor:complexity:provmonomial-ELHI}
If $\Omc$ is a \why-annotated $\ELHIbot$ ontology \blue{(annotated 
by $\lambda_\semiringVars:\Omc\mapsto\semiringVars\cup\{1\}$)}, for every \new{assertion
$\alpha$} and monomial $\monomial\in\monomials(\semiringVars)$, $\Omc\models (\alpha,\monomial)$ is decidable in exponential time \new{\wrt $|\Omc|$}.
\end{corollary}

\begin{corollary}\label{cor:complexity:provenanceassertion-ELHI}
If $\Omc$ is a \why-annotated $\ELHIbot$ ontology \blue{(annotated 
by $\lambda_\semiringVars:\Omc\mapsto\semiringVars\cup\{1\}$)}, for every \new{assertion
$\alpha$},  
the provenance $\Pmc(\alpha,\Omc)=\sum_{\Omc\models(\alpha,\monomial)}\monomial$ can be computed in exponential time \new{\wrt $|\Omc|$}.
\end{corollary}

Since subsumption checking in $\mathcal{ELI}$ is \ExpTime-hard \cite{BBL-EL08} \new{and assertion entailment in $\mathcal{ELI}$ can be reduced to subsumption checking}, we obtain a matching lower bound for the \new{annotated assertion} entailment problem and cannot hope for computing the why-provenance of axioms in $\ELHIbot$ more efficiently. 
However, there is an additional source of complexity when considering provenance, as illustrated by the following example that shows \new{that even if the ontology is formulated in a simple DL in which entailment can be checked in polynomial time, there may be exponentially many monomials $\monomial$ such that $\Omc\models(\alpha,\monomial)$.}

\begin{example}\label{ex:exponential}
\new{Consider the ontology 
$$\Omc=\{(A\sqsubseteq A_i,v_i), (A_i\sqsubseteq B,u_i)\mid 0\leq i\leq n\}\cup\{(B(a),x),(B\sqsubseteq A,u)\}.$$ For every $S\subseteq \{1,\dots, n\}$, $\Omc\models (A(a), x\times u \times \Pi_{i\in S} u_i\times v_i)$.}
\end{example}

\begin{remark}\label{rem:restriction-completion-for-one-monom}
\new{If one is only interested in checking the entailment of a given annotated assertion $(\alpha,\monomial)$, it is possible
to apply the completion algorithm using only the axioms annotated with relevant variables. More precisely, one can use 
the ontology 
$$\Omc_\monomial=\{(\beta,v)\mid (\beta,v)\in\Omc, v\text{ occurs in }\monomial\}\cup\{(\beta,1)\mid (\beta,1)\in\Omc\}$$ 
instead of $\Omc$. In the case where each axiom of $\Omc$ is annotated with a distinct variable, $\Omc_\monomial$ has exactly $|\monomial|$ axioms and $\Omc\models (\alpha,\monomial)$ is decidable in exponential time  \wrt $|\monomial|$.}
\end{remark}

\subsubsection{The Case of Lightweight DLs}\label{sec:completionELHIrestr}

We show that in the fragment $\ELHIbotrestr$ (\cf Definition \ref{def:ELHIrestr}), the exponential complexity of deciding annotated \new{assertion} entailment comes only from the size of the monomial. 
\new{Intuitively, this comes from the fact that in $\ELHIbotrestr$, the axioms $A\sqsubseteq \exists Q$, $Q\sqsubseteq P_i$ and $\exists\mn{inv}(P_i).A_i\sqsubseteq B_i$ in the premise of $\CR^{T}_3$ have to be such that $A_i=\top$, which means that the conclusion is simply $A\sqsubseteq D$. This allows us to modify the completion rules in a way that avoids introducing conjunctions of more than two concept names, and hence, the number of (non-annotated) GCIs that can be constructed remains polynomial \wrt $|\Omc|$. Recall that by definition of $\ELHIbotrestr$, an ontology that belongs to $\ELHIbotrestr$ is already in normal form.} 

\begin{figure}[tbh]
\framebox{\begin{minipage}{0.99\linewidth}
\begin{itemize}
\item $\CR^T_{2}$ is restricted by fixing $|M|\leq 2$ and $N$ is empty, \ie \blue{(for $M$ of the form $B_1\sqcap B_2$, $B$ or $\top$)}: 

``If $\Smc$ contains $(\blue{M}\sqsubseteq A,\monomial_1)$ and $(A\sqsubseteq C,\monomial_2)$, then $\Smc\leftarrow\Smc\cup\{\Phi\}$ where 
 $\Phi=(\blue{M}\sqsubseteq C,{\monomial_1\times \monomial_2})$.''
 
\item $\CR^T_{3}$ is restricted by fixing $k+k'\leq 2$ and $A_i=\top$ for every $1\leq i\leq k$, \ie: 

``Given $k, k'\geq 0$ such that $ k+k'\leq 2$, if $\Smc$ contains $(A\sqsubseteq \exists Q,\monomial_0)$, $(Q\sqsubseteq P, \monomial)$,  $(Q\sqsubseteq P_i,\monomial_i )$, $(\exists \mn{inv}(P_i).\top\sqsubseteq B_i,\nonomial_i)$ for $1\leq i\leq k$, 
$(\top\sqsubseteq B'_i,\onomial_i)$ for $1\leq i\leq k'$, 
$(B_1\sqcap\dots\sqcap B_k\sqcap B'_1\sqcap\dots\sqcap B'_{k'}\sqsubseteq C,\nonomial)$, and $(\exists P.C\sqsubseteq D, \onomial)$,
 then $\Smc\leftarrow\Smc\cup\{\Phi\}$ where 
$\Phi=(A\sqsubseteq D,\monomial\times\nonomial\times\onomial\times\monomial_0\times\Pi_{i=1}^k(\monomial_i\times\nonomial_i)\times\Pi_{i=1}^{k'}\onomial_i)$.''

\item $\CR^T_{4}$ is defined by: 

``Given $k, k'\geq 0$ such that $ k+k'\leq 2$, if $\Smc$ contains $(P\sqsubseteq P_i,\monomial_i ), (\exists P_i.\top\sqsubseteq B_i,\nonomial_i)$ for $1\leq i\leq k$,   
$(\top\sqsubseteq B'_i,\onomial_i)$ for $1\leq i\leq k'$, and 
$(B_1\sqcap\dots\sqcap B_k\sqcap B'_1\sqcap\dots\sqcap B'_{k'}\sqsubseteq C,\nonomial)$, 
 then $\Smc\leftarrow\Smc\cup\{\Phi\}$ where 
$\Phi=(\exists P.\top\sqsubseteq C,\nonomial\times\Pi_{i=1}^k(\monomial_i\times\nonomial_i)\times\Pi_{i=1}^{k'}\onomial_i)$.''

\item $\CR^T_{5}$ is defined by: 

``If $\Smc$ contains $(\top\sqsubseteq A_1,\monomial_1)$, $(\top\sqsubseteq A_2,\monomial_2)$ and $(A_1\sqcap A_2\sqsubseteq B,\nonomial)$, then $\Smc\leftarrow\Smc\cup\{\Phi\}$ where 
$\Phi=(\top\sqsubseteq B,\monomial_1\times\monomial_2\times\nonomial)$.''
\end{itemize}
\end{minipage}}
\caption{\new{Modified completion rules for the case where $\Omc$ belongs to $\ELHIbotrestr$. Note that in $\CR^T_{3}$ and $\CR^T_{4}$, the conjunction $B_1\sqcap\dots\sqcap B_k\sqcap B'_1\sqcap\dots\sqcap B'_{k'}$ contains at most two concept names since $k+k'\leq 2$.}}\label{fig:completion-rules-restr}
\end{figure}

\new{We modify the completion algorithm presented in Section~\ref{sec:subcompletionalgo} as follows. 
First, we modify the rules from Table~\ref{tab:completionRules} as explained in Figure~\ref{fig:completion-rules-restr}. Intuitively, we restrict the rules that may introduce exponentially many axioms ($\CR^T_{2}$ and $\CR^T_{3}$) so that they cannot introduce conjunctions in the left anymore. We also need to include two additional rules $\CR^T_{4}$ and $\CR^T_{5}$. 
Example \ref{ex:completion-restr} demonstrates how the modified completion rules work, and in particular why $\CR^T_{4}$ and $\CR^T_{5}$ are needed.}
\begin{example}\label{ex:completion-restr}
\new{Let 
\begin{align*}
\Omc=\{&(A\sqsubseteq \exists R,x), (\exists R^-.\top\sqsubseteq B_1,y_1), (\exists R^-.\top\sqsubseteq B_2,y_2), (\top\sqsubseteq B_3,y_3),(\top\sqsubseteq B_4,y_4), \\&(B_1\sqcap B_2\sqsubseteq C_1, z_1), (B_3\sqcap B_4\sqsubseteq C_2, z_2), (C_1\sqcap C_2\sqsubseteq C,z_3), (\exists R.C\sqsubseteq D,u), (A(a),v) \}.
\end{align*}
The completion algorithm for $\ELHIbot$ that uses the rules from Table~\ref{tab:completionRules} derives, among many
others, the following axioms:
\begin{itemize}
\item $(B_1\sqcap B_2\sqcap C_2\sqsubseteq C,z_1z_3)$ with $\CR^T_{2}$;
\item $(B_1\sqcap B_2\sqcap B_3\sqcap B_4\sqsubseteq C,z_1z_2z_3)$ with $\CR^T_{2}$;
\item $(A\sqsubseteq D, xy_1y_2y_3y_4z_1z_2z_3u)$ with $\CR^T_{3}$;
\item $(D(a), xy_1y_2y_3y_4z_1z_2z_3uv)$ with $\CR^A_{1}$.
\end{itemize}
In contrast, the completion algorithm for $\ELHIbotrestr$ that uses the modified rules derives the following axioms (among others):
\begin{itemize}
\item $(\exists R^-.\top\sqsubseteq C_1,y_1y_2z_1)$ with $\CR^T_{4}$;
\item $(\top\sqsubseteq C_2,y_3y_4z_2)$ with $\CR^T_{5}$;
\item $(A\sqsubseteq D, xy_1y_2y_3y_4z_1z_2z_3u)$ with the modified version of $\CR^T_{3}$;
\item $(D(a), xy_1y_2y_3y_4z_1z_2z_3uv)$ with $\CR^A_{1}$.
\end{itemize}
Intuitively, $\CR^T_{4}$ and $\CR^T_{5}$ allow us to derive the ``range information'' for $R$ that was directly handled in $\CR^T_{3}$ when the size of the conjunction in the left was not bounded \blue{by 2}.
}
\end{example}

\new{Second, for $k\in\mathbb{N}$, we let $\mn{saturate}^k(\Omc)$ be the set of annotated axioms obtained from $\Smc$ (defined as in Section~\ref{sec:subcompletionalgo}) through the modified completion rules restricted to \emph{monomials of size at most $k$} (\ie that have at most $k$ variables). We call $\mn{saturate}^k(\Omc)$ the \emph{$k$-saturation of $\Omc$}.} 
\new{The next theorem states that $\mn{saturate}^k(\Omc)$ can be computed in polynomial time \wrt the size of $\Omc$ and exponential time \wrt $k$, 
and that when $\Omc$ is satisfiable, $\mn{saturate}^k(\Omc)$ suffices for deciding entailment of all annotated assertions $(\alpha,\monomial)$ with $|\monomial|\leq k$.}

\begin{restatable}{theorem}{CorrectnessELHIrestr}
\label{prop:completionalgorithmELHIrestr}
\new{Let $\Omc$ be a \why-annotated ontology \blue{(annotated 
by $\lambda_\semiringVars:\Omc\mapsto\semiringVars\cup\{1\}$)} belonging to $\ELHIbotrestr$, $\monomial\in\monomials(\semiringVars)$, and $\alpha$ an assertion.}
\begin{enumerate}
\item \new{For every $k\in\mathbb{N}$, if $(\alpha,\monomial)\in\mn{saturate}^k(\Omc)$, then $\Omc\models (\alpha,\monomial)$. Moreover, if $\alpha=A(a_\top)$, then $\Omc\models (A(c),\monomial)$ for every $c\in\NI$.}

\item \new{If $\Omc$ is satisfiable and $\Omc\models (\alpha,\monomial)$ with  $|\monomial|\leq k$, 
\begin{enumerate}
\item if $\alpha$ is of the form $A(a)$ or $R(a,b)$ with $a,b\in\individuals{\Omc}$, then $(\alpha,\monomial)\in\mn{saturate}^k(\Omc)$;
\item if $\alpha$ is of the form $A(c)$ with $c\in\NI\setminus\individuals{\Omc}$, then $(A(a_\top),\monomial)\in\mn{saturate}^k(\Omc)$.
\end{enumerate}}

\item \new{$\mn{saturate}^k(\Omc)$ is computable in $p(|\Omc|^{k})$ where $p$ is a polynomial function.}
\end{enumerate}
\end{restatable}

\new{Moreover, the algorithm can be modified by ignoring the monomials in the annotated axioms to check in polynomial time whether $\Omc$ is satisfiable (which shows the complexity of satisfiability in $\ELHIbotrestr$ stated in Theorem~\ref{theo:ELHIrestrcomplexity}). It follows that to decide whether $\Omc\models(\alpha,\monomial)$, we can first check that $\Omc$ is satisfiable in polynomial time, then compute $\mn{saturate}^{|\monomial|}(\Omc)$ in polynomial time \wrt the size of $\Omc$ and exponential time \wrt $|\monomial|$.}

\begin{restatable}{corollary}{complexityprovmonomial}
\label{th:complexity:provmonomial}
If $\Omc$ \new{is a \why-annotated ontology \blue{(annotated 
by $\lambda_\semiringVars:\Omc\mapsto\semiringVars\cup\{1\}$) belonging to}}  $\ELHIbotrestr$ then, for every \new{assertion
$\alpha$} and monomial $\monomial$, $\Omc\models (\alpha,\monomial)$ is decidable \new{in $p(|\Omc|^{|\monomial|})$ where $p$ is a polynomial function.}
\end{restatable}

\begin{remark}
\new{As explained in Remark~\ref{rem:restriction-completion-for-one-monom}, one can apply the completion algorithm using only the axioms annotated with variables from $\monomial$ or $1$, and obtain a bound in $p(|\monomial|^{|\monomial|})$ in the case where each axiom of $\Omc$ is annotated with a distinct variable.}
\end{remark}

\blue{We obtain an improved \PSpace upper bound for the problem of annotated assertion entailment in $\ELHIbotrestr$ by adapting the proof of a similar result for $\ELHr$ annotated with monomials (\cf Section~\ref{sec:relatedwork}) \cite{provenance-ELHr}. The construction is} inspired by \citet{DBLP:conf/lata/HutschenreiterP17},  who showed how to transform a completion algorithm into a tree automaton, and \citet{DBLP:journals/iandc/BaaderHP08}, who gave conditions on tree automata to obtain \PSpace algorithms. 

\begin{theorem}\label{th:pspace}
\new{If $\Omc$ is a \why-annotated ontology \blue{(annotated 
by $\lambda_\semiringVars:\Omc\mapsto\semiringVars\cup\{1\}$)} belonging to $\ELHIbotrestr$, then for every assertion
$\alpha$ 
and monomial $\monomial$, $\Omc\models (\alpha,\monomial)$ is decidable in \PSpace.}
\end{theorem}
\begin{proof}
By Theorem~\ref{theo:ELHIrestrcomplexity}, one can check in polynomial time whether $\Omc$ is satisfiable. If it is not the case, $\Omc\models (\alpha,\monomial)$ trivially. Otherwise, by Theorem~\ref{prop:completionalgorithmELHIrestr}, $\Omc\models (\alpha,\monomial)$ iff $(\alpha,\monomial)\in \mn{saturate}^{|\monomial|}(\Omc)$ (or $(A(a_\top),\monomial)\in\mn{saturate}^{|\monomial|}(\Omc)$ in the case where $\alpha=A(c)$ for some $c\in\NI\setminus\individuals{\Omc}$). This is equivalent to the existence of a derivation tree \wrt the completion algorithm,  \ie a labelled tree $t$ such that (i) the root of $t$ is labelled by $(\alpha,\monomial)$, (ii) every leaf of $t$ is labelled by some annotated axiom in $\Smc$ (which can be constructed in polynomial time), (iii) every non-leaf node is labelled by an annotated axiom $(\beta,\nonomial)$ and its children's labels $(\beta_1,\nonomial_1),\dots,(\beta_\ell,\nonomial_\ell)$ are such that $(\beta_1,\nonomial_1),\dots,(\beta_\ell,\nonomial_\ell)$ and $(\beta,\nonomial)$ correspond respectively to the premises and conclusion of the instantiation of a completion rule (using the modified completion rules for $\ELHIbotrestr$)\blue{, and (iv) for every path from the root to a leaf, no label $(\beta,\nonomial)$ is repeated}. Such a tree has arity bounded by 8 (maximum number of premises in the completion rules for $\ELHIbotrestr$), and its depth is polynomially bounded (indeed, for every leaf, the monomial in the node labels on the path from the leaf to the root can only increases \blue{in length (or stay the same), so there are at most as many different monomials on the path as the number of variables}, and the number of (non-annotated) axioms that can be built by the completion rules for $\ELHIbotrestr$ is polynomial). Moreover, the size of each node label (annotated axiom) is polynomial, and for each node, there are only exponentially many options of \blue{sets of} 
children \blue{(since they are subsets of size at most 8 of $\mn{saturate}^{|\monomial|}(\Omc)$, which has exponential size)} so one can guess a node's children in polynomial time. It is thus possible to define a non-deterministic algorithm that checks whether a derivation tree for $(\alpha,\monomial)$ exists using depth-first search. 
\blue{In more details, such an algorithm starts from the root labelled by $(\alpha,\monomial)$, and if $(\alpha,\monomial)\notin\Smc$, guesses a set of at most 8 children labelled by some annotated axioms  $(\beta_1,\nonomial_1),\dots,(\beta_\ell,\nonomial_\ell)$ such that each $(\beta_i,\nonomial_i)$ is different from $(\alpha,\monomial)$ and each $\beta_i$ is either a concept or role assertion, a positive or negative RI, or a GCI of the form $A\sqsubseteq B$, $A\sqcap A'\sqsubseteq B$, $A\sqsubseteq \exists P$, or $\exists P.A\sqsubseteq B$ with $A,A'\in(\NC\cap\signature{\Omc})\cup\{\top\}$ and $B\in(\NC\cap\signature{\Omc})\cup\{\bot\}$ and each $\nonomial_i$ is a monomial over the variables that occur in $\monomial$, and $(\beta_1,\nonomial_1),\dots,(\beta_\ell,\nonomial_\ell)$ and $(\alpha,\monomial)$ correspond respectively to the premises and conclusion of the instantiation of a completion rule. Then the algorithm iterates, and guesses children of $(\beta_1,\nonomial_1)$ if $(\beta_1,\nonomial_1)\notin\Smc$, etc., using depth-first search, so keeping in memory only the path from $(\alpha,\monomial)$ to the current node and the children that remain to check for each node on the path.}
\end{proof}


\subsection{Conjunctive Queries over $\ELHIbot$ Ontologies}\label{sec:why-cq-elhi}
\new{We now provide algorithms to decide entailment of annotated BCQs and compute the provenance of BCQs \wrt \why-annotated ontologies. For the rest of this section, $\Omc$ is a $\why$-annotated $\ELHIbot$ ontology in normal form and $\mn{saturate}(\Omc)$ is the result of saturating \Omc.} 

We have seen that $\ELHIbot$ is expressive enough to reduce entailment of rooted tree-shaped BCQs to assertion entailment (\cf Theorem \ref{th:instancequeries}). 
However, this does not apply to general BCQs that may, for example, contain cycles. To deal with such queries, we adapt the method for conjunctive query answering in $\ELHIbot$ described by  \citet[Section 5.2]{DBLP:conf/rweb/BienvenuO15}. In a nutshell, the idea is to rewrite some parts of the query \wrt the saturated ontology into a query that can be matched in the assertions of the saturated ontology. 
\new{Before introducing the query rewriting process formally, we illustrate the approach on an example.}

\begin{example}
\new{Consider $q(x)=\exists y D(x)\wedge E(y)\wedge P(x,y)\wedge R(y,x)$ and 
\begin{align*}
\Omc=\{&(A\sqsubseteq \exists P, v_1), (P\sqsubseteq R^-,v_2), (\exists R.B\sqsubseteq C,v_3), (\exists P.C\sqsubseteq D,v_4), (C\sqsubseteq E, v_5),\\& (A(a),u_1), (B(a),u_2)\}.
\end{align*}
In every model $\Imc$ of $\Omc$ there exists $e\in\Delta^\Imc$ such that $(a^\Imc,e,u_1v_1)\in P^\Imc$, $(e,a^\Imc,u_1v_1v_2)\in R^\Imc$, $(e,u_1u_2v_1v_2v_3)\in C^\Imc$, $(a^\Imc,u_1u_2v_1v_2v_3v_4)\in D^\Imc$ and $(e,u_1u_2v_1v_2v_3v_5)\in E^\Imc$, so there is a match for $\ext{q(a)}$ in $\Imc$. It follows that $\Omc\models (q(a),u_1u_2v_1v_2v_3v_4v_5)$.}

\new{Our algorithm starts by computing $\mn{saturate}(\Omc)$. One can verify that $\mn{saturate}(\Omc)$ extends $\Omc$ with the tautologies added upon initialization (in particular, $(P\sqsubseteq P,1)$) and the two following axioms: $(A\sqcap B\sqsubseteq D, v_1v_2v_3v_4)$ and $(D(a), u_1u_2v_1v_2v_3v_4)$.}

\new{Then the algorithm rewrites $q$ into a set of CQs using the GCIs and RIs of $\mn{saturate}(\Omc)$, while keeping track of the annotations of the axioms used in the process. The goal of the rewriting is to \emph{eliminate some existentially quantified variables}. Since our example query has only one existential variable $y$, our goal is to rewrite $q$ into a conjunction of atoms that do not contain $y$. Intuitively, this is done with the following steps.
\begin{itemize}
\item Write $q$ as $\exists y D(x)\wedge E(y)\wedge P(x,y)\wedge R^-(x,y)$, so that $y$ occurs in second position of role atoms.
\item Consider the following axioms of $\mn{saturate}(\Omc)$: $(A\sqsubseteq \exists P, v_1)$, $(P\sqsubseteq P, 1)$, $(P\sqsubseteq R^-, v_2)$, $(C\sqsubseteq E, v_5)$, $(\exists R.B\sqsubseteq C,v_3)$. They allow us to rewrite $\exists y E(y)\wedge P(x,y)\wedge R^-(x,y)$ into $A(x)\wedge B(x)$. 
\item We obtain the CQ $q'(x)=D(x)\wedge A(x)\wedge B(x)$ associated with the monomial $v_1v_2v_3v_5$ that keeps track of the axioms used for this rewriting.
\end{itemize}
Since there is no other option to rewrite $\exists y E(y)\wedge P(x,y)\wedge R^-(x,y)$ and there is no other existentially quantified variable to consider, the full rewriting is the set of pairs of a CQ and a monomial $\{(q(x),1), (q'(x),v_1v_2v_3v_5)\}$. 
Finally, we evaluate all CQs of the rewriting in \blue{the annotated assertions of} $\mn{saturate}(\Omc)$. There is no match for $\ext{q(a)}$ but there is one for $\ext{q'(a)}$ that uses the three annotated assertions of $\mn{saturate}(\Omc)$, so we multiply the monomial of $q'(x)$ with the monomials of the annotated assertions to obtain that $\Omc\models (q(a), u_1u_2v_1v_2v_3v_4v_5)$.}
\end{example}

\blue{Importantly, our rewriting algorithm handles queries that may contain repeated atoms such as $C(x)\wedge C(x)$. The following example shows why allowing for repetitions in the rewriting matters.}

\begin{example}
\blue{Consider $\Omc$ defined in Example~\ref{ex:repetitionsMatter} and $q=\exists y P(a,y)\wedge C(y)$. One can easily verify from Example~\ref{ex:repetitionsMatter} that $\Omc\models (q, x_1y_1uv)$, $\Omc\models (q, x_2y_2uv)$ and $\Omc\models (q, x_1x_2y_1y_2uv)$, since for every model $\Imc$ of $\Omc$, for $i\in\{1,2\}$, there exists some $(a^\Imc,e_i,x_iy_iu)\in P^\Imc$ such that $(e_i,x_iy_iuv)\in C^\Imc$ and $(e_i,x_1x_2y_1y_2uv)\in C^\Imc$. Our algorithm produces the rewriting $\{ (q,1), (A(a)\wedge A(a), uv)\}$ and evaluating $A(a)\wedge A(a)$ on the annotated assertions of $\mn{saturate}(\Omc)$ indeed yield the three desired monomials when multiplied by $uv$, since $\mn{saturate}(\Omc)$ contains $(A(a), x_1y_1)$ and $(A(a), x_2y_2)$, so that the extended version of $q$, $\ext{q}=\exists t_1t_2 \ A(a,t_1)\wedge A(a,t_2)$, has four matches that map $t_1$ and $t_2$ to $x_1y_1$ or $x_2y_2$. In contrast, if we rewrite $q$ into $A(a)$, we would miss $\Omc\models (q, x_1x_2y_1y_2uv)$.}
\end{example}

\blue{However, in the same way as we did in the completion algorithm, we limit the number of repetitions of each atom in a rewriting by $\mn{Card}(\Omc)$. This is possible because every monomial $\monomial$ such that $\Omc\models (q,\monomial)$ contains at most $\mn{Card}(\Omc)$ variables so we never need to consider more than $\mn{Card}(\Omc)$ matches for an atom $C(x)$ in the annotated assertions of $\mn{saturate}(\Omc)$ to obtain all variables of $\monomial$ in the product.}

\new{In the following definition, the intuition behind $(q,\monomial)\rightarrow_\Omc (q',\monomial')$ is that the CQ $q'$ can be obtained from $q$ by a rewriting step that removes an existential variable from $q$ and $\monomial'$ is the product of $\monomial$ and the annotations of the axioms used to rewrite $q$ into $q'$.} 
\begin{definition}[Adapted from Definition 8 of \citet{DBLP:conf/rweb/BienvenuO15}]\label{def:rewriting}
For a CQ $q$, a monomial $\monomial$ and a $\why$-annotated $\ELHIbot$ ontology $\Omc$ \blue{(annotated 
by $\lambda_\semiringVars:\Omc\mapsto\semiringVars\cup\{1\}$)}  in normal form, we write $(q,\monomial)\rightarrow_\Omc (q',\monomial')$ if $(q',\monomial')$ can be obtained from $(q,\monomial)$ by applying the following steps.
\begin{itemize}
\item[(S1)]  Select in $q$ an arbitrary existentially quantified variable $x_0$ such that there are no atoms of the form $R(x_0,x_0)$ in $q$.
\item[(S2)] Replace each role atom of the form $R(x_0,y)$ in $q$, where $y$ and $R$ are arbitrary, by the atom $\mn{inv}(R)(y,x_0)$.
\item[(S3)] Let $V_p=\{y\mid Q(y,x_0)\in q\text{ for some }Q\}$ and select some $(A\sqsubseteq \exists P,v)\in\Omc$ (recall that $A\in\NC$ since $\Omc$ is in normal form) such that:
\begin{itemize}
\item[(a)] for every $Q(y,x_0)\in q$, there is some $(P\sqsubseteq Q,\monomial_{P\sqsubseteq Q})\in \mn{saturate}(\Omc)$; 

\item[(b)] for every $C(x_0)\in q$, there exist $p,p'\geq 0$ such that $(B_1{\sqcap}\dots{\sqcap} B_p\sqcap B'_1{\sqcap}\dots{\sqcap} B'_{p'}\sqsubseteq C, \nonomial_C)\in\mn{saturate}(\Omc)$ and 
\begin{itemize}
\item for every $1\leq i\leq p$, there exist $(\exists \mn{inv}(P_i).A_i\sqsubseteq B_i,v_i)\in\Omc$ and $(P\sqsubseteq P_i,\monomial_i)\in\mn{saturate}(\Omc)$ and 
\item for every $1\leq i\leq p'$, there exists $(\top\sqsubseteq B'_i,\onomial_i)\in\mn{saturate}(\Omc)$.
\end{itemize}
\end{itemize}
\new{For the selected $(A\sqsubseteq \exists P,v)\in\Omc$,} \blue{initialize a list of concept names $At(q,x_0,A\sqsubseteq \exists P)=[A]$ and a monomial $mon(q,x_0,A\sqsubseteq \exists P)=v$. Then:}

\begin{itemize}
\item[(a)] For every \blue{occurrence of $Q(y,x_0)$ in $q$}, arbitrarily select one $(P\sqsubseteq Q,\monomial_{P\sqsubseteq Q})$ in $\mn{saturate}(\Omc)$ and update $mon(q,x_0,A\sqsubseteq \exists P)\leftarrow mon(q,x_0,A\sqsubseteq \exists P)\times\monomial_{P\sqsubseteq Q}$.

\item[(b)] For every \blue{occurrence of $C(x_0)$ in $q$}, arbitrarily select some $(B_1\sqcap\dots\sqcap B_p\sqcap B'_1\sqcap\dots\sqcap B'_{p'}\sqsubseteq C, \nonomial_C)\in\mn{saturate}(\Omc)$, for every $1\leq i\leq p$ select arbitrarily a pair $(\exists \mn{inv}(P_i).A_i\sqsubseteq B_i,v_i)\in\Omc$ and $(P\sqsubseteq P_i,\monomial_i)\in\mn{saturate}(\Omc)$, and for every $1\leq i\leq p'$ select arbitrarily some $(\top\sqsubseteq B'_i,\onomial_i)\in\mn{saturate}(\Omc)$.\\ 
Update $At(q,x_0,A\sqsubseteq \exists P)\leftarrow At(q,x_0,A\sqsubseteq \exists P)\blue{\cdot [}A_i\mid 1\leq i\leq p\blue{]}$ and $mon(q,x_0,A\sqsubseteq \exists P)\leftarrow mon(q,x_0,A\sqsubseteq \exists P)\times\nonomial_C\times \monomial_C$ where $\monomial_C=\prod_{i=1}^p (v_i\times\monomial_i)\times\prod_{i=1}^{p'}\onomial_i$.
\end{itemize}

\item[(S4)]  Drop from $q$ every atom that contains $x_0$.
\item[(S5)] Select a variable $y_0\in V_p$ and replace every occurrence of \new{every} $y'\in V_p$ in $q$ by $y_0$.
\item[(S6)] \blue{Add atom $D(y_0)$ to $q$ for each occurrence of concept name $D$ in the list $At(q,x_0,A\sqsubseteq \exists P)$ (i.e., $D(y_0)$ is repeated as many times as $D$ occurs in the list), then limit the total number of occurrences of $D(y_0)$ in the query to $\mn{Card}(\Omc)$}  
and multiply $\monomial$ by $mon(q,x_0,A\sqsubseteq \exists P)$.
\end{itemize}
We write $(q,1)\rightarrow_\Omc^* (q^*,\monomial^*)$ if $(q,1)=(q_0,\monomial_0)$ and $(q^*,\monomial^*)=(q_k,\monomial_k)$ for some finite rewriting sequence $(q_0,\monomial_0)\rightarrow_\Omc (q_1,\monomial_1)\rightarrow_\Omc\dots\rightarrow_\Omc(q_k,\monomial_k)$ for $k\geq 0$, and we say that  
the set $\Rew(q,\Omc)=\{(q^*,\monomial^*)\mid (q,1)\rightarrow_\Omc^* (q^*,\monomial^*)\}$ is the 
\emph{annotated rewriting of $q$ \wrt $\mn{saturate}(\Omc)$}.
\end{definition}

Note that there may be exponentially many pairs of the form $(q^*,\monomial_1), (q^*,\monomial_2),\dots$ in $\Rew(q,\Omc)$ since the same query $q^*$ can be obtained by choosing different annotated inclusions in the rewriting steps and an axiom can be annotated with exponentially many monomials in $\mn{saturate}(\Omc)$ \new{(\eg in Example \ref{ex:exponential}, $(B\sqsubseteq A, u \times \Pi_{i\in S} u_i\times v_i)\in\mn{saturate}(\Omc)$ for every $S\subseteq\{1,\dots,n\}$)}. 

\begin{restatable}{theorem}{ThCQansweringAlgo}\label{th:CQ-answering-algorithm}
Let $\Omc$ be a satisfiable $\why$-annotated $\ELHIbot$ ontology \blue{(annotated 
by $\lambda_\semiringVars:\Omc\mapsto\semiringVars\cup\{1\}$)} in normal form, $q(\vec{x})$ be a CQ \new{that does not contain any individual name from $\NI\setminus\individuals{\Omc}$} and $\vec{a}$ be a tuple of individuals from $\Omc$ of the same \new{length} as $\vec{x}$. We denote by $\Omc'$ the non-annotated version of $\Omc$, \new{by $\Dmc$ the set of annotated assertions in $\mn{saturate}(\Omc)$, by $\Imc_\Dmc$ the annotated interpretation with domain $\NI$ that satisfies exactly the annotated assertions in $\Dmc$,} and by $\Dmc'$ the corresponding set of (non-annotated) assertions. Then 
\begin{enumerate}
\item $\Omc'\models q(\vec{a})$ iff there is $(q^*,\monomial^*)\in\Rew(q,\Omc)$ such that there is a match for $q^*(\vec{a})$ in~$\Dmc'$.
\item For every $\monomial\in\monomials(\semiringVars)$, $\Omc\models (q(\vec{a}),\monomial)$ iff \new{there exist $(q^*,\monomial^*)\in\Rew(q,\Omc)$ and $\onomial\in\p{\Imc_\Dmc}{\ext{q^*(\vec{a})}}$ such that $\monomial=\monomial^*\times\onomial$.} 
\item  \new{$\Pmc(q(\vec{a}),\Omc)=\sum_{(q^*,\monomial^*)\in \Rew(q,\Omc)}(\monomial^*\times \Sigma_{\onomial\in \p{\Imc_\Dmc}{\ext{q^*(\vec{a})}}}\onomial$).} 
\item $\Rew(q,\Omc)$ is computable in exponential time in $|\Omc|+|q|$. 
\end{enumerate}
\end{restatable}

\new{We show that for every $(q^*,\monomial^*)\in\Rew(q,\Omc)$, $(q^*,\monomial^*)$ is of polynomial size. Since by Theorem~\ref{prop:completionalgorithmELHI}, $\Dmc$ can be computed in exponential time, hence is of exponential size, deciding whether $\onomial\in\p{\Imc_\Dmc}{\ext{q^*(\vec{a})}}$ or computing $\Sigma_{\onomial\in \p{\Imc_\Dmc}{\ext{q^*(\vec{a})}}}\onomial$ can be done in exponential time in $|\Omc|+|q|$ (by finding all matches for $\ext{q^*(\vec{a})}$ in $\Imc_\Dmc$). 
This yields the following corollaries of Theorem~\ref{th:CQ-answering-algorithm}.  Again, note that for every BCQ $q$, if $\monomial$ is a monomial in $\why$ such that $|\monomial|>|\Omc|$, $\monomial$ contains variables that do not occur in $\Omc$ and one can directly conclude that $\Omc\not\models (q,\monomial)$ if $\Omc$ is satisfiable, so the size of $\monomial$ does not matter.} 

\begin{restatable}{corollary}{ThCQAnsweringComplexityEntailemnt}\label{th:complexity-cq-elhi-entailment}
If $\Omc$ is a \why-annotated $\ELHIbot$ ontology \blue{(annotated 
by $\lambda_\semiringVars:\Omc\mapsto\semiringVars\cup\{1\}$)}, then for every BCQ $q$ \new{that does not contain any individual name from $\NI\setminus\individuals{\Omc}$} and $\monomial\in\monomials(\semiringVars)$,
$\Omc\models (q,\monomial)$ is decidable in exponential time \new{\wrt $|\Omc|+|q|$}.
\end{restatable}

\begin{restatable}{corollary}{ThCQAnsweringComplexity}\label{th:complexity-cq-elhi}
If $\Omc$ is a \why-annotated $\ELHIbot$ ontology \blue{(annotated 
by $\lambda_\semiringVars:\Omc\mapsto\semiringVars\cup\{1\}$)}, then for every BCQ $q$ \new{that does not contain any individual name from $\NI\setminus\individuals{\Omc}$}, 
the provenance $\Pmc(q,\Omc)$ can be computed in exponential time \new{\wrt $|\Omc|+|q|$}.
\end{restatable}

Since BCQ entailment is already \ExpTime-complete for $\ELHIbot$ without annotations \cite{DBLP:conf/rweb/BienvenuO15}, we cannot hope for a better complexity upper bound. 
However, if we restrict ourselves to ontologies expressed in $\ELHIbotrestr$, we can obtain a result similar to Corollary~\ref{th:complexity:provmonomial} for annotated BCQ entailment. 

\begin{restatable}{theorem}{complexityBCQprovmonomial}
\label{th:complexity:provBCQmonomial}
If $\Omc$ is a \new{\why-annotated} $\ELHIbotrestr$ ontology \blue{(annotated 
by $\lambda_\semiringVars:\Omc\mapsto\semiringVars\cup\{1\}$)}, then for every BCQ $q$ \new{that does not contain any individual name from $\NI\setminus\individuals{\Omc}$} and monomial $\monomial$, 
$\Omc\models (q,\monomial)$ is decidable in \NP \wrt $|\Omc|+|q|$ (with $|\monomial|$ fixed), and in exponential time in~$|\monomial|$.
\end{restatable}

A matching \NP lower bound comes from the combined complexity of standard BCQ entailment over databases. 
Note that for the specific cases of \DLLiteR and \ELHr, \citet{provenance-DLLite}  and \citet{provenance-ELHr}  proposed algorithms for annotated BCQ entailment with some 
annotations \new{based on monomials}, respectively by adapting classical query rewriting algorithm for \DLLiteR \cite{DBLP:journals/jar/CalvaneseGLLR07} and the combined approach based on the computation of a compact canonical model for \ELHr \mbox{\cite{LTW:elcqrewriting09}}.

\section{\new{Provenance in the \posbool and \lin Semirings}}\label{sec:posbool-lin}

\new{In this section, we investigate provenance in the \posbool and \lin semirings, which correspond to well-known notions of provenance in the database setting. 
Similarly to the preceding section, we focus on ontologies annotated with variables (or $1$) rather than any element from $\posbool$ or $\lin$. A $\posbool$ (resp.\ $\lin$)-annotated ontology thus means $\tup{\Omc,\lambda_\semiringVars}$ with $\lambda_\semiringVars:\Omc\mapsto\semiringVars\cup\{1\}$ and we often omit the superscripts and identify $\Omc$ and $\tup{\Omc,\lambda_\semiringVars}$.} 
\new{Recall from Section~\ref{sec:database-prov} that $\posbool=(\posbool, \vee, \wedge, 0, 1)$ is the semiring of positive Boolean functions over $\semiringVars$ and that $\lin=(\lin,\cup,\cup^*, \emptyset^*,\emptyset)$ where $\lin$ is the set of all subsets of $\semiringVars$ extended with $\emptyset^*$ and for every $S_1,S_2\subseteq \semiringVars$, $S_1\cup^*S_2=S_1\cup S_2$. 
Also recall that we often represent elements of \posbool or \lin as polynomials: elements of \posbool are identified with their irredundant disjunctive normal form and represented by a sum of monomials, and elements of \lin are represented by monomials.}

\new{Given a satisfiable $\ELHIbot$ ontology $\Omc$ annotated by elements from $\semiringVars\cup\{1\}$, if we denote by $\Omc^{\why}$, $\Omc^{\posbool}$ and $\Omc^{\lin}$ this ontology interpreted as a \why-, \posbool- and \lin-annotated ontology respectively, it is easy to verify that the three annotated ontologies share the same canonical model (if the annotations are represented by monomials). The following proposition thus follows from Theorems~\ref{thm:can-model-main} and \ref{thm:can-model-main-ri} (since the statement holds trivially in the cases where $\Omc$ is unsatisfiable or $\alpha$ is an RI with an unsatisfiable left-hand side by Remarks~\ref{rem:unsat-onto} and~\ref{rem:unsat-left}).} 
\begin{proposition} \label{prop:why-lin-pos}
\new{Let $\Omc$ be an $\ELHIbot$ ontology, $\lambda_\semiringVars:\Omc\mapsto\semiringVars\cup\{1\}$ and $\Omc^{\why}$, $\Omc^{\posbool}$ and $\Omc^{\lin}$ defined by $\tup{\Omc,\lambda_\semiringVars}$. For every BCQ, assertion, or RI $\alpha$ and monomial $\monomial$, $\Omc^{\why}\models (\alpha,\monomial)$ iff $\Omc^{\posbool}\models (\alpha,\monomial)$ iff $\Omc^{\lin}\models (\alpha,\monomial)$.}
\end{proposition}

\begin{remark}\label{rem:why-lin-pos-gci}
\new{Example~\ref{ex:idempotent} shows that the proposition does not hold if $\alpha$ is a GCI because \why is not \timesidem: if $\Omc=\{(A\sqsubseteq B_1,x_1), (A\sqsubseteq B_2, x_2), (B_1\sqcap B_2\sqsubseteq C, x_3)\}$, since \posbool and \lin are \timesidem, ${\Omc^{\posbool}}\models (A\sqsubseteq C,x_1x_2x_3)$ and ${\Omc^{\lin}}\models (A\sqsubseteq C,x_1x_2x_3)$ but ${\Omc^{\why}}\not\models (A\sqsubseteq C,x_1x_2x_3)$.}
\end{remark}

\new{However, the \emph{provenance} differs depending on the semiring considered, since the additions have different properties. Since \posbool is absorptive (intuitively, $x+xy=x$ since the Boolean function $x\vee(x\wedge y)$ is equivalent to $x$), the \posbool-provenance of a BCQ, assertion or RI $\alpha$ corresponds to the sum of \emph{minimal} monomials that occur in its \why-provenance. Regarding \lin, for all monomials $\monomial$ and $\nonomial$, $\monomial+\nonomial=\monomial\times\nonomial$ so the \lin-provenance of $\alpha$  corresponds to the product of all variables that occur in its \why-provenance. These relationships are well-known in the database setting.} 

\begin{proposition}\label{prop:why-lin}
\new{Let $\Omc^{\why}$, $\Omc^{\posbool}$, $\Omc^{\lin}$ be as in Proposition~\ref{prop:why-lin-pos} and $\alpha$ be a BCQ, assertion, or RI.
\begin{itemize}
\item $\Pmc(\alpha,\Omc^{\posbool})=\Sigma_{\monomial\in M}\monomial$ where $M$ is the set of monomials $\monomial$ such that $\monomial$ occurs in $\Pmc(\alpha,\Omc^{\why})$ and for every $\nonomial\neq\monomial$ that occurs in $\Pmc(\alpha,\Omc^{\why})$, there exists a variable that occurs in $\nonomial$ and not in $\monomial$. 
\item $\Pmc(\alpha,\Omc^{\lin})=\Pi_{v\in V}v$ where $V=\{v\mid v\in\semiringVars\text{ occurs in }\Pmc(\alpha,\Omc^{\why})\}$.
\end{itemize}}
\end{proposition}

\subsection{\posbool-Provenance}
\label{sec:pin:prov}

\subsubsection{Relationship \new{to} Axiom Pinpointing}
Provenance over the \posbool semiring \new{is closely related to what} has been thoroughly studied in the DL literature under the name of 
\emph{axiom pinpointing}  
(originally coined by \citet{DBLP:conf/ijcai/SchlobachC03}). In this context, the goal is
to find one or all the \emph{minimal} (w.r.t.\ set inclusion) subsets of axioms that entail a given consequence. These
sets are called the \emph{justifications} for the consequence \cite{DBLP:conf/semweb/KalyanpurPHS07}.  
This definition of justifications can be  straightforwardly extended to justifications for BCQs. 
\new{Proposition~\ref{prop:ap:pos} establishes the connection between \posbool-provenance and justifications. The Boolean formula $\bigvee_{\Jmc\in\Just(\alpha)}\ \bigwedge_{\beta\in\Jmc}\lambda_\semiringVars(\beta)$ in the proposition statement is a \emph{pinpointing formula}, defined in the DL literature as a monotone Boolean formula whose satisfying valuations correspond exactly to subsets of the ontology that entail the consequence \cite{DBLP:journals/jar/BaaderP10}. Indeed, since description logics are monotone, every superset of a justification also entails the consequence so the formula is equivalent to $\bigvee_{\Mmc\subseteq\Omc',\Mmc\models\alpha}\ \bigwedge_{\beta\in\Mmc}\lambda_\semiringVars(\beta)$.} 

\begin{restatable}{proposition}{propposbooljustif}
\label{prop:ap:pos}
\new{Let $\Omc=\tup{\Omc',\lambda_\semiringVars}$ be a satisfiable \posbool-annotated $\ELHIbot$ ontology \blue{(annotated 
by $\lambda_\semiringVars:\Omc'\mapsto\semiringVars\cup\{1\}$)}. If (i) $\alpha$ is a BCQ, an assertion, or an RI whose left-hand side is satisfiable  \wrt $\Omc$, or (ii) $\alpha$ is a GCI between basic concepts whose left-hand side is satisfiable  \wrt $\Omc$ and $\Omc$ does not contain any GCI with $\top$ as left-hand side, then 
$$\Pmc(\alpha,\Omc)=
	\bigvee_{\Jmc\in\Just(\alpha)}\ \bigwedge_{\beta\in\Jmc}\lambda_\semiringVars(\beta)$$
where $\Just(\alpha)$ denotes the set of all justifications for $\alpha$ \wrt $\Omc'$.}
\end{restatable}
\begin{remark}
\new{One can check that the conditions imposed on $\alpha$ and $\Omc$ are necessary with the usual examples. Example~\ref{ex:problem-gcis-sem-entailtwo} shows that Proposition~\ref{prop:ap:pos} does not apply to GCIs with conjunctions in the left.  
Regarding the satisfiability of the left-hand side of $\alpha$, if $\Omc=\{(A\sqsubseteq B,v_1), (A\sqsubseteq C, v_2), (B\sqcap C\sqsubseteq \bot,v_3)\}$, then $\Omc\models A\sqsubseteq D$ for every $D\in\NC\setminus\{A\}$ and the only justification for $\Omc\models A\sqsubseteq D$ is $\Omc'$ itself. However, $\Pmc(A\sqsubseteq D,\Omc)$ is the sum over all elements of \posbool so is equal to 1. Finally, if $\Omc=\{(\top\sqsubseteq D,x)\}$, for every $C\in\NC\setminus\{D\}$, $\Omc'$ is a justification for $C\sqsubseteq D$ but $\Pmc(C\sqsubseteq D,\Omc)=0$ (\cf Example~\ref{ex:fuzzy-GCI-top}).}
\end{remark}

In the context of axiom pinpointing, it is common to  allow for static axioms to capture cases where, for instance, the assertions or the GCIs and RIs are 
considered to be immutable, among many other situations that appear in different formalisms 
\shortcite{DBLP:conf/ki/BaaderPS07,PrMa-13,LiMa13}. 
Formally, we consider that a DL ontology $\Omc$ is partitioned into two subsets $\Omc:=\Omc_s\cup \Omc_r$ where
$\Omc_s$ is the class of \emph{static} axioms, which are assumed to always hold, and $\Omc_r$ is the class of 
\emph{refutable} axioms which take part in the justifications. 
In particular, to the best of our knowledge, the only  systematic study on
the computation of BCQ justifications considers justifications to be sets of assertions, \ie GCIs and RIs to be static \cite{DBLP:conf/ecai/CeylanLMV20}. 
To handle this case, we may simply \new{annotate the axioms of $\Omc_s$ with $1$ and those of $\Omc_r$ with distinct variables}. 

It is worth noting that the problem of axiom pinpointing has been studied, under different names, in many other 
communities like---among many others---propositional satisfiability \shortcite{LPMM16}, process modelling \shortcite{RCFG22}, and 
answer set programming \shortcite{ADFPR22}.

\subsubsection{Computing \posbool-Provenance}
There exist two main approaches for finding and enumerating the justifications for a given consequence of an
ontology. The \emph{black-box} approach simply calls an existing (classical) reasoner repeatedly to prune out 
superfluous axioms \cite{DBLP:conf/semweb/KalyanpurPHS07}; the \emph{glass-box} approach, on the other hand,
modifies the reasoner to compute the justifications directly. The specific glass-box technique needed obviously depends
on the characteristics of the underlying reasoner, but general frameworks have been developed for tableaux 
\cite{DBLP:journals/logcom/BaaderP10}, automata \cite{DBLP:journals/jar/BaaderP10}, and consequence-based
\cite{DBLP:conf/sum/OzakiP18} methods. Within the context of $\ELHI_\bot$, it has also been proposed to reduce the
problem to an enumeration in SAT \shortcite{AMIMPM16,KaSK17,SeVe09} or ASP \shortcite{PeRi22,HMPR-23}.

Since \posbool-provenance can be obtained from \why-provenance by removing the non-minimal monomials, the algorithms we gave in Section \ref{sec:why} provide a glass-box approach to compute the set of all justifications for an assertion or a BCQ in $\ELHIbot$. This is, in particular, interesting for BCQs for which the problem of finding all justifications that take into account axioms beyond assertions has not been considered yet, \new{while GCIs and RIs are usually crucial to explain why a BCQ is entailed from a DL ontology to a user that may not have in mind all the semantics relationships between concepts and roles expressed in the ontology}. 
\new{Glass-box approaches have the advantage of being streamlined: contrary to black-box methods, a glass-box algorithm needs
to be executed only once to find one or all justifications, making them much more efficient. For instance, the only method
capable of enumerating all justifications for all consequences of the very large \ELH ontology \textsc{Snomed} is the
glass-box based PuLi~\cite{KaSK17}. The cost of this efficiency is that implementations and optimizations need to be
developed anew.}

A direct consequence of the black-box methods  is that one justification can be computed (or verified) 
with a polynomial number of calls to a classical reasoner. This means that a justification in $\ELHIbot$ can be found
in exponential time, and this bound reduces to polynomial time in any sublogic of $\ELHIbot$ which allows for
polynomial-time reasoning. On the other hand, it is known that a single consequence may have exponentially many
justifications \cite{DBLP:conf/ki/BaaderPS07}. This means that enumerating all justifications (that is, finding the full
\posbool-provenance polynomial in extended form) necessarily requires exponential time. Through a more fine-grained complexity
analysis, \citet{DBLP:journals/ai/PenalozaS17} showed that, \new{unless $\PTime=\NP$, 
there exists no algorithm that can compute all justifications \wrt an \EL ontology 
in polynomial time \emph{even if only polynomially many justifications exist}}. 

\subsubsection{\new{Applications of \posbool-Provenance}}\label{sec:applications-of-posbool-provenance}

\new{Recall that $\posbool$ specializes correctly to every commutative, \timesidem and absorptive semiring. Since it is idempotent for both operations, by Theorem~\ref{th:sem-commut-hom}, it follows that for every commutative semiring $\semiringshort$ that is \timesidem and absorptive, every satisfiable $\ELHIbot$ 
$\semiringshort$-annotated ontology $\Omc^\semiringshort=\tup{\Omc,\lambda}$, and every $\alpha$, if (i) $\alpha$ is a BCQ, assertion or RI whose left-hand side is satisfiable \wrt $\Omc$ or (ii) $\alpha$ is a GCI between basic concepts whose left-hand side is satisfiable \wrt $\Omc$ and $\Omc$ does not contain GCI with $\top$ as left-hand side, then the provenance $\Pmc(\alpha,\Omc^{\semiringshort})$ can be computed as follows. First, compute $\Pmc(\alpha,\Omc^{\posbool})$ where $\Omc^{\posbool}=\tup{\Omc,\lambda_\semiringVars}$ with $\lambda_\semiringVars$ an injective function from the axioms in $\Omc$ to the set of variables $\semiringVars$. Then evaluate the obtained polynomial through the unique semiring homomorphism $h$ from $\posbool$ to $\semiringshort$ such that $h(\lambda_\semiringVars(\beta))=\lambda(\beta)$ for every $\beta\in\Omc$ and $h(x)=\zero$ for every $x\in\semiringVars\setminus\{\lambda_X(\beta)\mid\beta\in\Omc\}$. Many useful semirings are \timesidem and absorptive, so \posbool-provenance is particularly interesting. 
For example, we have shown in Section~\ref{sec:sem:ann} that provenance in the fuzzy semiring captures the semantics of fuzzy or possibilistic DLs, and that provenance in the semiring that corresponds to a bounded distributive lattice allows to compute access rights in the setting defined by \citet{BaKP-JWS12}. }

\new{The set of justifications of a query has also proven to be useful to characterize other semantics for annotated databases or DL ontologies. This is in particular the case for \emph{probabilistic} databases \cite{DBLP:journals/sigmod/Senellart17} and probabilistic DLs \cite{DBLP:journals/semweb/RiguzziBLZ15,CePe-17,Ceyl-18}. In this context, axioms are annotated either directly with probability values or with some Boolean events associated with a probability of being true, and the probability of a query is the sum of the probabilities of the worlds in which the query is true. It has been shown that the query probability is equal to the probability of the Boolean formula that corresponds to the \posbool-provenance.}


\subsection{\lin-Provenance}\label{sec:linprov}

\subsubsection{Computing \lin-Provenance of Axioms}
\new{Since the \lin-provenance of an assertion can be obtained from its \why-provenance  by taking the product of all variables that occur in the \why-provenance (Proposition~\ref{prop:why-lin}), it can be computed in exponential time using the techniques from Section~\ref{sec:completion}. However, we show that if we modify the completion algorithm from Section~\ref{sec:completion} to combine all monomials that annotate an axiom instead of storing them separately, one can compute the \mbox{\lin-}provenance of all assertions in \emph{polynomial time} if the ontology belongs to $\ELHIbotrestr$.} 

As in Section \ref{sec:completion}, 
the algorithm assumes normal form and keeps as data structure 
a set \Smc of annotated axioms $(\alpha,\monomial)$, 
where $\alpha$ uses the vocabulary of \Omc \new{(extended with $a_\top$ and possibly $\top$, $\bot$)}, and $\monomial\in \monomials(\semiringVars)$.
\Smc is initialised as in Section \ref{sec:completion} 
and extended by exhaustively applying the rules in Table~\ref{tab:completionRules}, where rule applications are modified \new{by replacing $\Smc\leftarrow\Smc\cup\{\Phi\}$ by $\Smc\leftarrow\Smc\Cup\{\Phi\}$ with}
\[
\Smc\Cup\{(\alpha,\monomial)\} :=
	\begin{cases}
		\Smc\cup\{(\alpha,\monomial)\}  \text{ if there is no $(\alpha,\nonomial)\in\Smc$} \\
		\Smc\setminus\{(\alpha,\nonomial)\}\cup\{(\alpha,\monomial\times\nonomial)\}  \text{ if $(\alpha,\nonomial)\in\Smc$;}
	\end{cases}
\]
i.e., add the
axiom $\alpha$ with an associated monomial if it does not yet appear in \Smc, and modify the monomial
associated to $\alpha$ to include new variables otherwise. \new{Note that since no axiom occurs several times with different annotations when $\Smc$ is initialized, there is always at most one $(\alpha,\nonomial)\in\Smc$ for a given $\alpha$. To ensure termination, a rule is applicable only if its conclusion is not in~\Smc (note that each rule application either adds an annotated axiom or adds some variables to some monomial, and the monomial size is bounded by the number of axioms in $\Omc$).} 
The rules are applied until no new rule is applicable; i.e., \Smc is~\emph{saturated}.

\begin{example}
\new{For} the ontology of Example \ref{ex:exponential}, we obtain the saturated set
\begin{align*}
\Smc=&\{\new{(A(a),x\times\monomial), (B(a),x\times\monomial),} (A\sqsubseteq A,\monomial), (B\sqsubseteq B,\monomial),(A\sqsubseteq B,\monomial), (B\sqsubseteq A,\monomial)\}\ \cup\\&\{(A_i\sqsubseteq B,\monomial), (B\sqsubseteq A_i,\monomial),(A_i\sqsubseteq A,\monomial), (A\sqsubseteq A_i,\monomial)\mid 1\leq i\leq n\}\ \cup\\&\{(A_i\sqsubseteq A_j,\monomial)\mid 1\leq i,j\leq n\}\cup\{\new{(\top(a),1), (\top(a_\top),1)}\}
\end{align*} with $\monomial=u\times \Pi_{i=1}^n u_i\times \Pi_{i=1}^n v_i$.
\end{example}

\begin{restatable}{theorem}{completionAlgoRelevance}\label{completionAlgoRelevance}
If \Omc is a satisfiable \lin-annotated $\ELHIbot$ ontology \blue{(annotated 
by $\lambda_\semiringVars:\Omc\mapsto\semiringVars\cup\{1\}$)}, $\alpha$ is an assertion s.t.\ $\Omc\models \alpha$, and
$\mn{linsat}(\Omc)$ is the result of saturating \Omc: 
\begin{enumerate}
\item \new{if $\alpha$ is of the form $A(a)$ or $R(a,b)$ with $a,b\in\individuals{\Omc}$, then  
$(\alpha,\Pmc(\alpha,\Omc))\in\mn{linsat}(\Omc)$  and there is no other $\monomial$ such that $(\alpha,\monomial)\in\mn{linsat}(\Omc)$;}
\item \new{if $\alpha=A(c)$ for $c\in\NI\setminus\individuals{\Omc}$, then $\Pmc(A(c),\Omc)\!=\!\Pmc(A(a_\top),\Omc)$, 
$(A(a_\top),\Pmc(A(a_\top),\Omc))\!\in\mn{linsat}(\Omc)$  and there is no other $\monomial$ such that $(A(a_\top),\monomial)\in\mn{linsat}(\Omc)$.}
\end{enumerate}
\end{restatable}

We show that if $\Omc$ belongs to $\ELHIbotrestr$, \new{$\mn{linsat}(\Omc)$} can be computed in polynomial time using the completion rules modified for this case as in Theorem~\ref{prop:completionalgorithmELHIrestr}. \new{Indeed, we have seen in Section~\ref{sec:completionELHIrestr} that these rules can build only polynomially many different (non-annotated) axioms and since a rule application that modifies a monomial adds at least one variable, each annotated axiom can be modified only a linear number of times.}
\new{Moreover, since \lin is multiplicatively idempotent, by Theorems~\ref{th:red-concept} and~\ref{th:red-role}, the \lin-provenance of GCIs between basic concepts and positive RIs can computed via a reduction to the \lin-provenance of assertions (with the usual condition for the GCI case).} 

\begin{restatable}{theorem}{complexityRelevance}\label{th:complexityRelevance}
Let $\Omc$ be a satisfiable \lin-annotated $\ELHIbot$ ontology \blue{(annotated 
by $\lambda_\semiringVars:\Omc\mapsto\semiringVars\cup\{1\}$)}. \new{If (i) $\alpha$ is an assertion  
or a positive RI or (ii) $\alpha$ is a GCI between basic concepts and $\Omc$ does not contain any GCI with $\top$ as left-hand side, then the following hold (note that $|\alpha|\leq 3$).}
\begin{itemize}
\item $\Pmc(\alpha,\Omc)$ can be computed in exponential time \new{\wrt $|\Omc|$}.
\item If $\Omc$ belongs to $\ELHIbotrestr$, $\Pmc(\alpha,\Omc)$ can be computed in polynomial time \new{\wrt $|\Omc|$}.
\end{itemize}
\end{restatable}

\subsubsection{\new{Relevant Axioms}}

\new{Given a satisfiable \lin-annotated $\ELHIbot$ ontology $\Omc$,  
we say that a \emph{variable $v\in\semiringVars$ is relevant} to entail an axiom $\alpha$ \wrt $\Omc$ if it occurs in $\Pmc(\alpha,\Omc)$. By definition of the addition of the \lin semiring, this is equivalent to 
the existence of a monomial $\monomial$ such that $v$ occurs in $\monomial$ and $\Omc\models (\alpha,\monomial)$.  
If the annotation function $\lambda_\semiringVars$ is injective and does not map any axiom to $1$, we say that an \emph{axiom $\beta$ is relevant} to entail $\alpha$ \wrt $\Omc$ if $\lambda_\semiringVars(\beta)$ is relevant. 
We relate this notion to \emph{usable facts} that have been defined in the context of Datalog provenance as the database facts that occur in some derivation tree for the query \cite{DBLP:conf/kr/BourgauxBPT22}. 
An equivalent definition of usable facts independent of the notion of derivation tree has also been proposed (see ``adornment-usable facts'' in \cite[Section B.3]{DBLP:journals/corr/abs-2202-10766}). Adapting this definition to $\ELHIbot$ characterizes relevant axioms as follows.
}

\begin{definition}
\new{
Given an $\ELHIbot$ axiom $\gamma$, the \emph{adornment by} $\gamma$ of a concept (resp.\ role) name $A$ (resp.\ $R$) 
is (the fresh name) $A^\gamma$ (resp.\ $R^\gamma$).}

\new{For an assertion $\beta$, let $\beta^\gamma$ be the assertion obtained by replacing the predicate in $\beta$ by its adornment by $\gamma$, and for an $\ELHIbot$ GCI or RI $\beta$, let $\beta^\gamma$ be the axiom obtained by replacing the (unique) predicate in the right-hand side of $\beta$ by its adornment by $\gamma$.}

\new{For an assertion $\beta$, let $f^\gamma(\beta)=\{\beta\}$, and for an $\ELHIbot$ GCI or RI $\beta$, let $f^\gamma(\beta)$ be the set of all GCIs or RIs obtained from $\beta$ by applying the two following steps: (i) replace the predicate in the right-hand side of $\beta$ by its adornment by $\gamma$, and (ii) replace a single predicate in the left-hand side of $\beta$ by its adornment by $\gamma$.} 

\new{Let $\Omc$ be a satisfiable $\ELHIbot$ ontology. Given an $\ELHIbot$ axiom $\alpha$ and an axiom $\gamma\in\Omc$, we say that $\gamma$ is \emph{usable to derive} $\alpha$ \wrt $\Omc$ if $\Omc\cup\Omc^\gamma\cup\{\gamma^\gamma\}\models \alpha^\gamma$ where $\Omc^\gamma=\bigcup_{\beta\in\Omc}f^\gamma(\beta)$.}
\end{definition}

\begin{restatable}{proposition}{Propusableaxioms}\label{prop:usableaxioms}
\new{Let $\Omc=\tup{\Omc',\lambda_\semiringVars}$ be a satisfiable \lin-annotated $\ELHIbot$ ontology such that $\lambda_\semiringVars$ maps all axioms to distinct variables. If (i) $\alpha$ is an assertion or a positive RI whose left-hand side is satisfiable \wrt $\Omc$, or (ii) $\alpha$ is a GCI between basic concepts whose left-hand side is satisfiable \wrt $\Omc$ and $\Omc$ does not contain any GCI with $\top$ as left-hand side, then $\gamma\in\Omc'$ is relevant to entail $\alpha$ \wrt $\Omc$ iff it is usable to derive $\alpha$ \wrt $\Omc'$.}
\end{restatable}

\new{It follows from Theorem~\ref{th:complexityRelevance} (or alternatively from Proposition~\ref{prop:usableaxioms}, since $\Omc^\gamma$, $\gamma^\gamma$ and $\alpha^\gamma$ can be built in polynomial time) that if we only need to know which axioms of an ontology $\Omc$ are relevant, or usable, to entail an axiom $\alpha$ (where $\Omc$ is a satisfiable  $\ELHIbot$ or $\ELHIbotrestr$ ontology and $\alpha$ is as in Proposition~\ref{prop:usableaxioms}), the complexity is the same as classical reasoning in $\ELHIbot$ or $\ELHIbotrestr$. 
This contrasts with the axiom pinpointing setting in which 
deciding whether an axiom belongs to a justification is 
\NP-hard for Horn-\EL\cite{PeSe10-KR}. 
This is because axiom pinpointing requires that justifications are \emph{minimal}: if 
$\Omc=\{(A\sqsubseteq B, v_1), (B\sqsubseteq C, v_2), (C\sqsubseteq B, v_3)\}$, the only justification for $A\sqsubseteq B$ is $\{A\sqsubseteq B\}$ but $\Pmc(A\sqsubseteq B,\Omc)=v_1v_2v_3$ (since $\Omc\models (A\sqsubseteq B, v_1)$ and $\Omc\models(A\sqsubseteq B, v_1v_2v_3)$) and, in particular, $v_2$ and $v_3$ 
are relevant to entail $A\sqsubseteq B$ (\ie $B\sqsubseteq C$ and $C\sqsubseteq B$ are relevant).}

\subsubsection{Relationship with Lean Kernels}

Provenance in the \lin semiring is related to \emph{lean kernels} \shortcite{PMIM17}, 
which approximate the union of justifications. 
\new{\citet{PMIM17} define lean kernels \wrt a \emph{consequence-based method}, defined as an algorithm that works on a set of axioms and uses rules to extend this set, until the set becomes saturated and consequences can be read from the saturated set.} 
The lean kernel of a consequence $\alpha$ (being an assertion or subsumption between two concept names) is the set of axioms appearing in at least one proof of $\alpha$
in a given \new{consequence-based method}. 
\new{This generalizes} the notion from propositional logic, 
where a lean kernel is the set of 
clauses appearing in a resolution proof for unsatisfiability. 
\new{The monomials, or sets of variables, computed by the completion algorithm for \lin-annotated ontologies described in this section correspond to 
the sets of axioms used in the derivations by the completion algorithm that ignores the annotations, 
which is a consequence-based method for $\ELHIbot$, thus correspond to lean kernels \wrt this algorithm}.


\section{\new{Related Work on Semiring Provenance for Description Logics}}\label{sec:relatedwork}

\new{In this section, we review and discuss   other frameworks that use some form of semiring provenance for description logics or very close settings within the semantic web. We refer to Section~\ref{sec:sem:ann} for a discussion about the relationship between our framework and DLs annotated with specific kinds of annotations, to Section~\ref{sec:expected-properties} for a comparison with the semiring provenance framework for relational databases and Datalog, and to Section~\ref{sec:conclusion} for 
   a discussion of our results and possible future work in light of the literature.}

\paragraph{\new{Provenance for lightweight DLs}} 
\new{The closest works to ours are those by~\citet{provenance-DLLite} and~\citet{provenance-ELHr}, who considered \DLLiteR and \ELHr respectively, which are fragments of $\ELHIbot$ (in particular,~\citet{provenance-ELHr} imposed the same syntactic restriction on \ELHr as the one we use in this paper). 
In these papers, ontology axioms are annotated by \emph{monomials} (actually \emph{variables} or $1$ in the latter paper). 
The semantics is defined using annotated interpretations of the form $\Imc=(\Delta^\Imc,\Delta_\text{m}^\Imc,\cdot^\Imc)$ which interpret monomials by elements of the \emph{domain of monomials} $\Delta_\text{m}^\Imc$ with the constraint that two monomials that are mathematically equal are mapped to the same element. 
Such an interpretation satisfies, e.g., an assertion $(A(a), x)$ if $(a^\Imc, x^\Imc)\in A^\Imc$ and a GCI $(C\sqsubseteq D, y)$ if $(e,\monomial^\Imc)\in C^\Imc$ implies that $(e,(\monomial\times y)^\Imc)\in D^\Imc$ (thus ignoring any $(e,\elem)\in C^\Imc$ such that $\elem\in\Delta_\text{m}^\Imc$ is not equal to $\monomial^\Imc$ for some monomial $\monomial$). 
To avoid counter-intuitive behaviors with conjunction in \ELHr,~\citet{provenance-ELHr} additionally assumed that $\times$ is idempotent (\ie that $(x\times x)^\Imc=x^\Imc$). 
If we interpret a  \DLLiteR or \ELHr ontology $\Omc$ annotated with 
monomials as a \posbool- or \lin-annotated ontology $\Omc^{\semiringVars}$ in our framework, then for every assertion, GCI, or RI $\alpha$ and monomial $\monomial$, it holds that $\Omc\models (\alpha,\monomial)$ under this semantics iff $\Omc^{\semiringVars}\models (\alpha,\monomial)$ under our semantics. 
In the case where $\Omc^{\semiringVars}$ is interpreted as a \why-annotated ontology, this result holds if $\alpha$ is an assertion (\cf Proposition~\ref{prop:why-lin-pos} and Remark~\ref{rem:why-lin-pos-gci}). 
However,~\citet{provenance-DLLite} and~\citet{provenance-ELHr} did not define the semantics of an ontology annotated with elements of an arbitrary commutative semiring. 
In contrast, we defined a general semantics, more in line with the database semiring provenance framework where the semantics of annotated databases is defined independently from the specific semiring, and showed how it captures (or does not capture) the semantics of several annotated DLs such as fuzzy or possibilistic DLs (Section~\ref{sec:sem:ann}). 
We also studied the properties of the semantics, and, in particular, showed the relationship between our semantics and the classical semiring provenance of relational or Datalog queries (Section~\ref{sec:expected-properties}), a question that was not considered by~\citet{provenance-DLLite} and~\citet{provenance-ELHr}.}

\new{From an algorithmic perspective,~\citet{provenance-ELHr} showed how to normalize the ontology, reduced annotated GCIs and RIs entailment to annotated assertions entailment, and provided a completion algorithm for computing annotated assertions entailed by an \ELHr ontology. We straightforwardly adapted the normalization rules to the case where the annotations can be elements of an arbitrary commutative semiring. To adapt the reductions between different entailment tasks to our semantics, we had to modify them slightly and 
	restrict our attention to GCIs between basic concepts because the reduction provided by~\citet{provenance-ELHr} relies on the fact that axioms are annotated with variables rather than by the elements of an arbitrary commutative semiring. 
	Our completion algorithm follows the same idea as the one for \ELHr but its rules are different to handle $\ELHIbot$. Finally, \citet{provenance-ELHr} proposed the adaptation of the completion algorithm to compute ``relevant variables'', which we reused to compute the \lin-provenance.}

\new{Regarding complexity results, \citet{provenance-ELHr} showed that in \ELHr, deciding entailment of an axiom annotated by a monomial is in polynomial time \wrt the ontology size (if the monomial size is fixed), 
and in polynomial space \wrt the whole input size. 
Since $\Omc\models (\alpha,\monomial)$ iff $\Omc^{\semiringVars}\models (\alpha,\monomial)$ where $\Omc^{\semiringVars}$ is $\Omc$ interpreted as a \posbool- or \lin-ontology and these semirings are fully idempotent, 
one can use the reductions from annotated GCI or RI entailment to annotated assertion entailment from Theorems~\ref{th:red-concept} and \ref{th:red-role} and 
the complexity results of Corollary~\ref{th:complexity:provmonomial} and Theorem~\ref{th:pspace} to generalize these results to $\ELHIbotrestr$ (with the restriction that our reduction is for GCIs between basic concepts and ontologies that do not contain any GCI with $\top$ as left-hand side, but the reduction given by \citet{provenance-ELHr} without these restrictions could be used for \posbool-annotated ontologies). 
They also showed that deciding whether a variable is relevant to an entailment (\ie whether it occurs in a monomial $\monomial$ such that $\Omc\models (\alpha,\monomial)$) can be decided in polynomial time. Theorem~\ref{th:complexityRelevance} generalizes this result to $\ELHIbotrestr$ (with the same restriction as before). }

\new{Annotated BCQ entailment was considered by \citet{provenance-DLLite} and 
\citet{provenance-ELHr}. Intuitively, given a BCQ $q$ and a \emph{sum of monomials} $p$, they define $\Omc\models (q,p)$ iff $p$ is included in the provenance of $q$ in each model of $\Omc$. If we restrict the comparison to $p$ being a monomial, we obtain the same annotated query entailments under our semantics. However, under the semantics of \citet{provenance-DLLite} and 
\citet{provenance-ELHr}, one can also get $\{(R(a,b), v_1), (R(b,a),v_2)\}\models (\exists xy R(x,y)\wedge R(y,x), v_1\times v_2+v_1\times v_2)$ \cite[Section 2.4]{provenance-ELHr}. 
\citet{provenance-DLLite} proposed a rewriting algorithm for \DLLiteR to compute all monomials such that $\Omc\models(\alpha,\monomial)$, \ie the provenance of the query as we defined it, and implemented it. 
For \ELHr, \citet{provenance-ELHr} described an algorithm based on the computation of a model and query rewriting. \citet{provenance-DLLite} showed that in \DLLiteR, deciding $\Omc\models (q,p)$ is \NP-complete (even if $\Omc$ is specified by an ontology-based data access instance consisting of an ontology, a set of mappings and a relational database) and \citet{provenance-ELHr} showed that $\Omc\models (q,p)$ is decidable in exponential time in \ELHr. Since $p$ is a sum of monomials, this problem differs from the problem of entailment of BCQ annotated with a monomial that we considered for $\ELHIbotrestr$ in Theorem~\ref{th:complexity:provBCQmonomial}, for which we obtained an \NP upper bound if the monomial size is fixed, and an exponential time one \wrt the monomial size.}

\paragraph{\new{Provenance semantics for attributed DL-Lite}} \new{A provenance semantics was considered in the context of \emph{attributed DL-Lite} \cite{attributedDL}. 
Attributed DLs allow for annotating assertions with an arbitrary number of \emph{attribute-value pairs of individual names} \cite{DBLP:conf/semweb/KrotzschMOT17}. For example some assertions may be annotated with source, clearance level and multiplicity, such as $A(a)@[\mn{src}:x_1, \mn{classif}:P, \mn{mult}:3]$, while some are not annotated at all. GCIs and RIs are used to express constraints on annotations (for example by requiring that a premise of a GCI has a given source to use this GCI)}. 
\new{If we consider the case where (i) each assertion is annotated by a single attribute-value pair of the form $[\mn{attr}:x]$ with the same attribute $\mn{attr}$, and (ii) RIs and GCIs only propagate annotations (in attributed DL syntax: $C@X\sqsubseteq D@X$), then we can see an attributed DL-Lite ontology $\Omc^a$ in our setting as an \polynomials-annotated DL-Lite ontology $\Omc^{\polynomials}=\{(\beta,x)\mid \beta@[\mn{attr}:x] \in\Omc^a, \beta\text{ assertion}\}\cup\{(C\sqsubseteq D,1)\mid C@X\sqsubseteq D@X\in\Omc^a\}\cup\{(P\sqsubseteq Q,1)\mid P@X\sqsubseteq Q@X\in\Omc^a\}$ 
with $\semiringVars=\{x\mid \beta@[\mn{attr}:x]\in \Omc^a\}$, 
keeping the unique common attribute implicit. 
The \emph{provenance-interpretations} for attributed DL-Lite defined by \citet{attributedDL} are required to satisfy a property of closure under sum, which amounts in our simplified context to requiring that, e.g., if $(e,s_1)$ and $(e,s_2)$ are in $A^\Imc$ and neither $s_1$ nor $s_2$ can be obtained as the sum of $s$ and $s'$ such that $(e,s)$ and $(e,s')$ are in $A^\Imc$, then $(e,s_1+s_2)$ is in $A^\Imc$. Since GCIs in DL-Lite cannot have conjunction or qualified existential restriction in the left-hand side, if $\Omc$ is satisfiable and $\alpha$ is an assertion such that $\Omc\models\alpha$, the attributed ontology $\Omc^a$ 
entails $\alpha@[\mn{attr}:\Pmc(\alpha, \Omc^{\polynomials})]$.
}

\paragraph{\new{Provenance for expressive DLs}} 
\new{Another notion of provenance was defined by \citet{provenance-DL-dannert-gradel} for the expressive DL $\ALC$ and arbitrary commutative semirings. 
They considered GCIs without annotation and (potentially complex) assertions such that each assertion is associated with a single expression of the form $=\elem$, $>\elem$ or $\geq\elem$ for some $\elem\in\semiringset$. 
In this setting, the semantics is defined by interpretations that map every (negated) simple assertion $(\neg)\alpha$ built from a \emph{finite} interpretation domain $\Delta\subseteq\NI$, $\NC$, and $\NR$, to an element of the semiring. Such an interpretation $\pi$ is required to be such that  $\pi(\alpha)\otimes \pi(\neg\alpha)=\zero$ for every $\alpha$ built from $\Delta$, $\NC$ and $\NR$. It is extended to interpret complex assertions with, e.g., $\pi((C\sqcap D)(e))=\pi(C(e))\otimes \pi(D(e))$ and $\pi((\exists R.C)(e))= \bigoplus_{d\in\Delta} \pi(R(e,d))\otimes \pi(C(d))$.}
 
\new{A (strong) model of an ontology is an interpretation that satisfies all its assertions and such that for each of its GCIs  $C\sqsubseteq D$, for every $e\in\Delta$, $\pi(C(e))\leq \pi(D(e))$ (where $a\leq b$ iff there exists $c\in \semiringset$ with $a\oplus c =b$).  
If we translate our annotated assertions as ``$=\elem$'' assertions in this framework, their semantics is not comparable with ours when we consider the intersection of the two settings (i.e., $\EL_\bot$ ontology that complies with our syntactic restriction, with GCIs annotated with $\one$). Indeed, the \why-annotated ontology $\Omc^{\why}=\{(A(a),x_1),(B(a), x_2), (A\sqcap B\sqsubseteq A, 1)\}$ translates into the ontology $\{\pi(A(a))\!=\!x_1,\ \pi(B(a))\!=\!x_2,\  A\sqcap B\sqsubseteq A\}$ which is not satisfiable according to this semantics (since there is no $\pi$ such that $\pi(A(a))=x_1$ and $\pi(A(a))\geq x_1x_2$). 
If we instead translate our annotated assertions as ``$\geq\elem$'' assertions, then given a satisfiable ontology $\Omc$ and an assertion $\alpha$ such that $\Omc\models\alpha$, the possible provenance values of $\alpha$ in the framework of \citet{provenance-DL-dannert-gradel} are all semiring elements $\eta$ such that $\eta\geq \elem$ for every $\elem$ such that $\Omc^\semiringshort\models (\alpha,\elem)$ under our semantics.} 

\new{
A notion of weak model was also introduced for ontologies put in some specific form. Intuitively, in this case, $C\sqsubseteq D$ requires that $\pi(C(e))\otimes\pi(\neg D(e))=\zero$ and $C\equiv D$ that $\pi(C(e))=\pi(D(e))$. This alternative semantics is still not comparable with ours:  the \why-annotated ontology $\Omc^{\why}=\{(A(a),x), (C(a),y), (A\sqsubseteq B, 1)\}$ translates into $\{\pi(A(a))\!=\!x,\ \pi(C(a))\!=\!y,\  A\sqsubseteq B\}$ or $\{\pi(A(a))\!\geq\!x,\ \pi(C(a))\!\geq\!y,\  A\sqsubseteq B\}$ so the possible provenance values for $B(a)$ in weak models are all elements of $\why$ but $0$ (in particular, $y$ is a possible provenance value for $B(a)$). 
}

\new{From an algorithmic point of view, \citet{provenance-DL-dannert-gradel} defined tableaux rules that can be applied if (i) the semiring is \emph{absorptive and such that $\geq$ is a total order}, and (ii) the ontology does not contain any equality statement (`` $=\elem$'').
}
\new{Finally, \citet{provenance-DL-dannert-gradel} proposed \emph{provenance tracking interpretations}. In such interpretation, for each assertion $C(a)$, $\pi(C(a))$ is a polynomial whose variables represent literals of the form $\alpha$ or $\neg\alpha$ where $\alpha$ is an assertion \emph{which does not occur in the ontology}.}

A recent work by \citet{Penaloza2023} considered GCI entailment from expressive ontologies (with a focus on $\ALC$) annotated with elements of a $\oplus$- and $\otimes$-idempotent commutative semiring. To handle $\ALC$ constructors, and, in particular, negation, the author defined a provenance semantics based on interpretations that do not satisfy the consequence. For each such interpretation, it takes the sum of the labels of the ontology axioms violated by the interpretation, then the provenance is defined as the product of these sums. 
When the semiring is not absorptive, the semantics differs from \new{the one defined in this paper}. For example, if we consider the ontology $\Omc^{\semiringshort}=\{(A\sqsubseteq B, \elem), (C\sqsubseteq D,\elemb)\}$, the provenance of $A\sqsubseteq B$ defined by \citet{Penaloza2023} is $\elem\otimes(\elem\oplus\elemb)=\elem\oplus\elem\otimes\elemb$, \new{while $\Pmc(A\sqsubseteq B, \Omc^{\semiringshort})=\elem$}. 
\new{When the semiring is absorptive, however, the provenance of a consequence $\alpha$ defined by \citet{Penaloza2023} is equal to $\bigoplus_{\Jmc\in\Just(\alpha)}\ \bigotimes_{\beta\in\Jmc}\lambda(\beta)$ \cite[Theorem 5]{Penaloza2023}, hence corresponds to our semantics if $\Omc$ and $\alpha$ are as required by Proposition~\ref{prop:ap:pos}. Indeed, since $\posbool$ specializes correctly to every commutative semiring that is \timesidem and absorptive, there exists a unique semiring homomorphism $h$ from $\posbool$ to $\semiringshort$ such that $h(\lambda_\semiringVars(\alpha))=\lambda(\alpha)$ for every $\alpha\in\Omc$ and $h(x)=\zero$ for every $x\in\semiringVars\setminus\{\lambda_\semiringVars(\alpha)\mid\alpha\in \Omc\}$, \ie $\lambda=h\circ\lambda_\semiringVars$. 
By Theorem~\ref{th:sem-commut-hom}, $h(\Pmc(\alpha,\Omc^{\posbool})) = \Pmc(\alpha, \Omc^{\semiringshort})$ and by Proposition~\ref{prop:ap:pos}, 
$\Pmc(\alpha,\Omc^{\posbool})=\bigvee_{\Jmc\in\Just(\alpha)}\bigwedge_{\beta\in\Jmc}\lambda_\semiringVars(\beta)$, so $\Pmc(\alpha, \Omc^{\semiringshort})=\bigoplus_{\Jmc\in\Just(\alpha)}\ \bigotimes_{\beta\in\Jmc}\lambda(\beta)$.}

 \paragraph{Provenance for the semantic web} 
 Several works focussed on querying annotated RDF data with SPARQL \cite{Dividino2009,DBLP:journals/internet/TheoharisFKC11}. 
In particular, \citet{Geerts16-provenance}  considered SPARQL queries on RDF data annotated with values from some arbitrary annotation domain $K$ equipped with three binary operations (with $\ominus$ to cover SPARQL difference operator) and showed that $(K,\oplus,\otimes,\ominus, 0,1)$ must be an extension of semiring they called spm-semiring and further studied. 
As far as reasoning is concerned, 
 \citet{DBLP:conf/semweb/BunemanK10} defined an algebraic deductive system for RDFS annotated with elements of a \plusidem semiring (not necessarily commutative), and  
\citet{DBLP:journals/ws/ZimmermannLPS12} considered RDFS annotated with elements of a \plusidem commutative semiring and defined a semantics based on annotated interpretations and a deductive system. 
These two deductive systems are similar in spirit to our completion algorithm for the \why semiring, the main difference being the existence of a generalization rule that deduces $(\alpha,\elema\oplus\elemb)$ from $(\alpha,\elema)$ and $(\alpha,\elemb)$, while our completion algorithm only computes the monomials that then need to be added to obtain the full \why-provenance, hence avoiding another blowup of the saturated ontology.


\section{Conclusions and Future Work}\label{sec:conclusion}
In this paper, we defined a semiring provenance semantics for $\ELHIbot$ ontologies, more in line with the classical semiring provenance for relational databases than previous proposals \new{(\cf discussion of the work by \citet{provenance-DLLite} and 
\citet{provenance-ELHr} in Section~\ref{sec:relatedwork})}. After studying its properties in details, we provided algorithms and complexity results for computing \new{the provenance of assertions and conjunctive queries in the case of \why-annotated ontologies}. We also investigated in more details \posbool-provenance and \lin-provenance and discussed connections with notions related to explanations in description logics. Besides $\ELHIbot$, we also considered $\ELHIbotrestr$, a fragment of $\ELHIbot$ \new{that we introduced and which has} the same good computational complexity as DL-Lite and~\EL. 

\subsection{Discussion}
\new{If we restrict further the language to ontologies whose GCIs have only a concept name on the right-hand side and whose RIs are positive, then, by Theorem~\ref{th:sem-cons-datalog}, for commutative \plusidem $\omega$\mbox{-}continuous semirings, the provenance of any BCQ in our framework can be computed using the tools developed for computing the provenance of Datalog queries over databases. This is via the translation of the ontology and query into a Datalog program presented in Theorem~\ref{th:sem-cons-datalog}, and this is applicable, in particular, for the \why-provenance and \posbool-provenance which have attracted a lot of interest in the Datalog literature.   
Datalog provenance has been considered since the seminal work on provenance by \citet{Green07-provenance-seminal}  who gave an algorithm for computing the coefficient of a particular monomial in the provenance series of the query. 
\citet{DBLP:journals/pacmmod/CalauttiLPS24} studied the data complexity of deciding whether a monomial is part of the \why-provenance for Datalog queries and showed that while the problem is in general intractable \new{(\NP-complete \wrt data complexity)}, it is tractable for non-recursive Datalog. 
 \citet{DBLP:conf/icdt/DeutchMRT14}  proposed circuit-based provenance representation as an efficient way to compute provenance with absorptive (hence \plusidem) semirings (\sorp-provenance) as well as \why-provenance. 
Several approaches have also been investigated to approximate Datalog provenance. For example, \citet{DBLP:journals/vldb/DeutchGM18} proposed to compute a compact representation of the top-k derivation trees, ranking the trees using tree patterns and facts and rules they use. 
The implemented approaches either use these kinds of approximations, restrict the language and/or constrain the semiring: 
\citet{DBLP:journals/toplas/ZhaoSS20} considered only minimal depth proof trees computed through semi-naïve evaluation; 
 \citet{DBLP:journals/vldb/LeeLG19} used SQL to compute provenance of non-recursive Datalog queries; 
\citet{DBLP:conf/ruleml/ElhalawatiKM22} used a hypergraph that represents all derivation steps of the Datalog program to compute the why-provenance either via a system of equations or via a translation to an extension of Datalog with sets; 
\citet{DBLP:conf/grades/RamusatMS22} translated a Datalog program into a weighted hypergraph and characterized the semirings where the best-weight derivation in the hypergraph corresponds to the provenance for the initial Datalog program, and used this translation to develop a practical approach to compute Datalog provenance in absorptive semirings that are totally ordered; 
finally, \citet{DBLP:conf/aaai/CalauttiLPS24} proposed a practical SAT-based approach for computing \blue{a variant of the why-provenance} based on a restricted class of proof-trees. }

\new{By Theorem~\ref{th:CQ-answering-algorithm} (point 3), the provenance of a BCQ \wrt a \why-annotated $\ELHIbot$ ontology can be obtained by computing the saturation of the ontology $\mn{saturate}(\Omc)$, rewriting the query \wrt the saturation, and evaluating the rewritings over the set of annotated assertions in $\mn{saturate}(\Omc)$, seen as a database.  
This last step could be performed by using a system for relational provenance management,  
such as 
GProM \shortcite{DBLP:journals/debu/ArabFGLNZ17}, which supports \polynomials-provenance as well as some other types of annotations, or 
ProvSQL \shortcite{2018-vldb-provsql} which supports all provenance semirings  
as well as arbitrary user-defined semirings, and has been used by \citet{provenance-DLLite} to implement a provenance-aware ontology-based data access system \new{(\cf Section~\ref{sec:relatedwork})}.} 

\new{\citet{DBLP:journals/jiis/MailisSSSK12} proposed a completion algorithm for a fuzzy version of $\EL^{++}$. If we adapted the completion algorithms for \why-annotated $\ELHIbot$ and $\ELHIbotrestr$ ontologies presented in Section~\ref{sec:completion} to $\mathbb{F}$-annotated ontologies (by using degrees from $[0,1]$ instead of monomials and the $\min$ operator instead of $\times$), to compute the provenance of all assertions \wrt $\Omc^{\mathbb{F}}$, we would not need to store in the saturation all elements $n\in[0,1]$ generated by the completion rules for a given axiom, but only the \emph{maximal} one (as we did in the case of \lin-annotated ontologies in Section~\ref{sec:linprov}). We would thus avoid the exponential blow-up due to the monomials, hence retain polynomial complexity in the case of $\ELHIbotrestr$.}

\subsection{Future Work}\label{sec:future-work}
For semirings that are not \timesidem, Theorem~\ref{th:sem-entailment} does not hold \new{for GCIs} and some concept subsumptions entailed by the ontology may have an unexpected provenance~\new{$\zero$}. It will nevertheless be interesting to investigate \blue{further} the case where the TBox is not annotated, in particular in the case of \plusidem, absorptive semirings (\sorp provenance) since it covers useful semirings such as the Tropical semiring (costs), the Viterbi semiring (confidence) or the \L ukasiewicz semiring (truth values). 
Possible directions for future work also include considering other query answering methods for \why-annotated $\ELHIbot$ ontologies, for example based on Datalog rewritings \cite{DBLP:conf/rweb/BienvenuO15}. 

\new{We could also try to extend the framework in several directions. First, we focussed on $\ELHIbot$ but the semantics can be easily defined for other constructors that have been considered in the \EL family, such as complex role inclusions ($R_1\circ\dots\circ R_n\sqsubseteq R$), nominals ($\{a\}$), concrete domains, or concept products ($C\times D\sqsubseteq R$) \cite{BBL-EL08,DBLP:conf/dlog/RudolphKH08}, or for Horn versions of expressive DLs, such as Horn $\mathcal{SHIQ}$ for which a completion algorithm has been proposed \cite{DBLP:conf/ijcai/Kazakov09}. However, adapting our provenance-aware completion algorithm for such an extension of $\ELHIbot$ (with suitable restrictions to ensure decidability, cf. \cite{BBL-EL08}) will not be trivial. Indeed, we have seen that existing algorithms cannot be straightforwardly adapted to handle provenance annotations (\cf Example~\ref{ex:pb-with-standard-rules}), and algorithms that allow to deal properly with such constructs are often already complex without annotation (see, e.g., \cite{DBLP:conf/kr/KazakovKS12} for nominals). Note that even if there exist consequence-based procedures for non Horn DLs \shortcite{DBLP:conf/ijcai/SimancikKH11,DBLP:journals/jair/BateMGCSH18}, extending our semantics to such DLs remains a challenge, since it is not clear how to handle full negation in our framework (\cf Section~\ref{sec:syntacticRestrictions}).  
A recent line of research developed semiring provenance for very general logical languages, such as first-order logic and fixed-point logic \cite{DBLP:journals/siglog/Tannen17,DBLP:conf/csl/DannertGNT21,DBLP:journals/corr/abs-2412-07986}. In this context, interpretations associate a semiring element to literals (ground facts and their negation) and are extended to more complex formula as expected \new{(in the spirit of the work on $\mathcal{ALC}$ by \citet{provenance-DL-dannert-gradel} discussed in Section~\ref{sec:relatedwork})}. This could be a track to explore, even if this notion--defined for model checking over finite interpretations rather than entailment from possibly infinite models--seems difficult to adapt to our purposes.}

\new{Another way to extend the framework would be to consider more expressive query languages.   
There have been several attempts at extending semiring provenance for relational databases to more expressive query languages, for example featuring aggregation \cite{DBLP:conf/pods/AmsterdamerDT11}. In particular, non-monotone queries with difference  attracted interest: \citet{DBLP:journals/japll/GeertsP10} proposed to use a monus operator on the underlying semiring, and call m\mbox{-}semirings the class of semirings with such an operator, but \citet{DBLP:conf/tapp/AmsterdamerDT11}  showed limitations of the approach. The case of Boolean provenance (\posbool) of non-monotone queries is however well-understood \cite{DBLP:conf/icalp/AmarilliBS15}. 
Similar questions will certainly arise if we try to extend our framework to DLs or queries featuring negation. The recent work by \citet{DBLP:conf/birthday/Suciu24} that explores three ways of
adding a difference operator to a semiring may provide some leads on these questions.}

\begin{acks}
\new{We are grateful to the anonymous reviewers who pointed out   issues in a previous version of this work and whose detailed feedback greatly helped us to improve this paper.} This work was supported by the ANR projects CQFD (ANR-18-CE23-0003) and EXPAND (ANR-25-CE23-1215), and by the MUR for the Department of Excellence DISCo at the University of Milano-Bicocca and under the PRIN project PINPOINT Prot.\ 2020FNEB27, CUP H45E21000210001. Ana Ozaki is supported by   NFR   through the project
``Learning Description Logic Ontologies,'' project number 316022 and by NFR through its Centre of Excellence Integreat - The Norwegian Centre for Knowledge-driven Machine Learning, project number 332645.
\end{acks}

\printbibliography

\appendix

\section{Proofs for Section \ref{sec:prov-definition}}\label{app:proofs-prov-def}

\LemComplexConceptQueryInter*
\begin{proof}
The proof is made by structural induction. In the base case, $C\in\NC\cup\{\top\}$ and $\ext{q_C}(x)=\exists t_0\, C(x,t_0)$: 
\begin{align*}
\{\bigotimes_{P(\vec{t},t)\in \ext{q_C}(x)} \pi(t) \mid \pi\in\nu_\Imc(\ext{q_C}(x)), \pi(x)=d\}
=\{ \pi(t_0) \mid \pi\in\nu_\Imc(C(d,t_0))\}
=\{\elem \mid (d,\elem)\in C^\Imc\}. 
\end{align*}
Induction step. 
If $C=C_1\sqcap C_2$, then $\ext{q_C}(x)= \ext{q_1}(x)\wedge \ext{q_2}(x)$ where $\ext{q_1}(x)$ and $\ext{q_2}(x)$ are the extended versions of the queries that retrieve the instances of $C_1$ and $C_2$ respectively. 
\begin{align*}
\{\bigotimes_{P(\vec{t},t)\in \ext{q_C}(x)} \pi(t) \mid \pi\in\nu_\Imc(\ext{q_C}(x)), \pi(x)=d\}=&\{\bigotimes_{P(\vec{t},t)\in \ext{q_1}(x)} \pi_1(t)\otimes\bigotimes_{P(\vec{t},t)\in \ext{q_2}(x)} \pi_2(t)  \mid \\ 
\pi_1\in\nu_\Imc(\ext{q_1}&(x)),\pi_2\in\nu_\Imc(\ext{q_2}(x)),\pi_1(x)=\pi_2(x)=d\} \\
=&\{\elem_1\otimes\elem_2 \mid (d,\elem_1)\in C_1^\Imc, (d,\elem_2)\in C_2^\Imc\}\\
=&\{\elem \mid (d,\elem)\in C^\Imc\}
\end{align*}
%
If $C=\exists R. C_1$, then $\ext{q_C}(x)=\exists yt_0\,  R(x,y,t_0) \wedge \ext{q_1}(y)$ where $\ext{q_1}(y)$ is the extended version of the query that retrieves the instances of $C_1$.  
\begin{align*}
\{\bigotimes_{P(\vec{t},t)\in \ext{q_C}(x)} \pi(t) \mid \pi\in\nu_\Imc(\ext{q_C}(x)), \pi(x)=d\}
=&\{ \pi_{R}(t_0)\otimes\bigotimes_{P(\vec{t},t)\in \ext{q_1}(y)} \pi_1(t)  \mid\\ 
\pi_R\in\nu_\Imc(R(x,y,t_0)),&\pi_1\in\nu_\Imc(\ext{q_1}(y)), \pi_R(x)=d, \pi_1(y)=\pi_R(y)\}\\
=&\{\elem_R\otimes\elem_1 \mid (d,e,\elem_R)\in R^\Imc, (e,\elem_1)\in C_1^\Imc\}\\
=&\{\elem \mid (d,\elem)\in C^\Imc\}
\end{align*}
%
The case $C=\exists R^-. C_1$, where $\ext{q_C}(x)=\exists yt_0\, R(y,x,t_0) \wedge \ext{q_1}(y)$, is similar.\qedhere
\end{proof}

\Theorembasic*
\begin{proof}
By Lemma \ref{lem:complex-concept-query-inter}, for every model $\Imc$ of $\Omc^\semiringshort$, it holds that
$$\p{\Imc}{\ext{q_C}(a)}=\{\bigotimes_{P(\vec{t},t)\in \ext{q_C}(a)} \pi(t) \mid \pi\in\nu_\Imc(\ext{q_C}(a))\}=\{\elem\mid (a^\Imc, \elem)\in C^\Imc\}.$$
Moreover,  
(i)~for every model $\Imc$ of $\Omc^\semiringshort\cup\{(C\sqsubseteq A_C,\one)\}$, $C^\Imc\subseteq A_C^\Imc$ and the annotated interpretation $\Jmc$ obtained from $\Imc$ by setting $A_C^\Jmc=C^\Imc$ is also a model of $\Omc^\semiringshort\cup\{(C\sqsubseteq A_C,\one)\}$, so $\bigcap_{\Imc\models \Omc^\semiringshort\cup\{(C\sqsubseteq A_C,\one)\}}\{\elemc\mid (a^\Imc,\elemc)\in A_C^\Imc\}=\bigcap_{\Imc\models \Omc^\semiringshort\cup\{(C\sqsubseteq A_C,\one)\}}\{\elemc\mid (a^\Imc,\elemc)\in C^\Imc\}$, and
(ii)~models of $\Omc^\semiringshort\cup\{(C\sqsubseteq A_C,\one)\}$ are models of $\Omc^\semiringshort$ and for every model $\Imc$ of $\Omc^\semiringshort$, the interpretation $\Jmc$ that extends $\Imc$ by setting $A_C^\Jmc=C^\Imc$ is a model of $\Omc^\semiringshort\cup\{(C\sqsubseteq A_C,\one)\}$ so $\bigcap_{\Imc\models \Omc^\semiringshort}\{\elemc\mid (a^\Imc,\elemc)\in C^\Imc\}=\bigcap_{\Imc\models \Omc^\semiringshort\cup\{(C\sqsubseteq A_C,\one)\}}\{\elemc\mid (a^\Imc,\elemc)\in C^\Imc\}$. 
%
It follows that
\begin{align*}\Pmc(A_C(a), \Omc^\semiringshort\cup\{(C\sqsubseteq A_C,\one)\})=&\bigoplus_{\Omc^\semiringshort\cup\{(C\sqsubseteq A_C,\one)\}\models (A_C(a),\elem)} \elem\\
= & \bigoplus_{\elem\in\bigcap_{\Imc\models \Omc^\semiringshort\cup\{(C\sqsubseteq A_C,\one)\}}\{\elemc\mid (a^\Imc,\elemc)\in A_C^\Imc\}} \elem\\
= & \bigoplus_{\elem\in\bigcap_{\Imc\models \Omc^\semiringshort\cup\{(C\sqsubseteq A_C,\one)\}}\{\elemc\mid (a^\Imc,\elemc)\in C^\Imc\}} \elem\\
= & \bigoplus_{\elem\in\bigcap_{\Imc\models \Omc^\semiringshort}\{\elemc\mid (a^\Imc,\elemc)\in C^\Imc\}} \elem\\
=& \bigoplus_{\elem\in\bigcap_{\Imc\models \Omc^{\semiringshort}}\p{\Imc}{\ext{q_C}(a)}} \elem \\
=&\ \Pmc(q_C(a), \Omc^\semiringshort).\qedhere
\end{align*} 
\end{proof}

\subsection{Relationship with Semantics for Specific Annotations}
 
\new{The proofs of Propositions~\ref{prop:poss:prov} and~\ref{prop:access:prov} are based on Proposition~\ref{prop:ap:pos} and Theorem~\ref{th:sem-commut-hom}. 
This is not an issue since Propositions \ref{prop:poss:prov} and \ref{prop:access:prov} are not used to prove any other results in this paper.}

\proppossprov*
\begin{proof}
\new{Let $\Omc^{\mathbb{F}}=\tup{\Omc,\lambda}$ be a satisfiable $\mathbb{F}$-annotated $\ELHIbot$ ontology and let $\Omc^{\posbool}=\tup{\Omc,\lambda_\semiringVars}$ be such that $\lambda_\semiringVars$ is an injective function from $\Omc$ to the set of variables $\semiringVars$ (assuming that $|\Omc|\leq|\semiringVars|$). Recall that $\posbool$ specializes correctly to every commutative semiring that is \timesidem and absorptive, hence in particular to $\mathbb{F}$. It follows that there exists a unique semiring homomorphism $h$ from $\posbool$ to $\mathbb{F}$ such that $h(\lambda_\semiringVars(\alpha))=\lambda(\alpha)$ for every $\alpha\in\Omc$ and $h(x)=0$ for every $x\in\semiringVars\setminus\{\lambda_\semiringVars(\alpha)\mid\alpha\in \Omc\}$, \ie $\lambda=h\circ\lambda_\semiringVars$. 
By Theorem~\ref{th:sem-commut-hom}, since $\posbool$ and $\mathbb{F}$ are commutative $\omega$-complete semirings (recall that $\semiringVars$ is finite) that are \plusidem and \timesidem, and $h$ is an $\omega$-complete semiring homomorphism from $\posbool$ to $\mathbb{F}$: $h(\Pmc(\alpha,\Omc^{\posbool})) = \Pmc(\alpha, \Omc^{\mathbb{F}})$ holds~if: 
\begin{enumerate}
\item $\alpha$ is a BCQ, an assertion, or an RI whose left-hand side is satisfiable \wrt $\Omc$; or
\item $\Omc$ does not contain any GCI with $\top$ as left-hand side, and $\alpha$ is a GCI whose left-hand side is satisfiable \wrt~$\Omc$.	
\end{enumerate}
In both cases, by Proposition~\ref{prop:ap:pos} (and absorptivity of \posbool which implies that considering also non-minimal $\Mmc$ in the disjunction does not change the provenance value), 
$\Pmc(\alpha,\Omc^{\posbool})=\bigvee_{\Mmc\subseteq\Omc, \Mmc\models\alpha}\bigwedge_{\beta\in\Mmc}\lambda_\semiringVars(\beta)$. Therefore, 
$$\Pmc(\alpha, \Omc^{\mathbb{F}})=h(\Pmc(\alpha,\Omc^{\posbool})) =\max\{\min_{\beta\in\Mmc}(\lambda(\beta))\mid \Mmc\subseteq\Omc, \Mmc\models\alpha\}.$$ 
It is then easy to check that $\Pmc(\alpha, \Omc^{\mathbb{F}})\geq n$ iff there exists $\Mmc\subseteq\Omc$, such that $\Mmc\models\alpha$ and all axioms in $\Mmc$ are annotated with degrees greater or equal to $n$, \ie iff $\Omc_{\geq n}\models\alpha$.}
\end{proof}

\propaccessprov*
\begin{proof}
\new{The proof is analogous to that of Proposition~\ref{prop:poss:prov}, using $\mathbb{L}_{\mathbb{A}}$ instead of $\mathbb{F}$. Indeed, $\mathbb{L}_{\mathbb{A}}$ is also \timesidem and absorptive. Hence, we get 
$\Pmc(\alpha, \Omc^{\mathbb{L}_{\mathbb{A}}})=h(\Pmc(\alpha,\Omc^{\posbool})) =\sup\{\inf_{\beta\in\Mmc}(\lambda(\beta))\mid \Mmc\subseteq\Omc, \Mmc\models\alpha\}$.  }
\end{proof}

\section{Proofs for Section \ref{sec:classical-res-counterparts}}
\subsection{\new{Normal Form}}
\NormalisationELHI*
\begin{proof}
For the first point, we show by induction that for every annotated $\ELHIbot$ ontology $\Omc^\semiringshort$, 
if $\NF^i(\Omc^\semiringshort)$ is obtained from $\Omc^\semiringshort$ by applying $i$ normalization steps, then every model $\Imc$ of $\NF^i(\Omc^\semiringshort)$ is also a model of $ \Omc^\semiringshort$. 

\noindent\emph{Base case: $i=1$.} Let $\Omc^\semiringshort$ be an annotated $\ELHIbot$ ontology and assume that $\NF^1(\Omc^\semiringshort)$ is obtained from $\Omc^\semiringshort$ by applying a single normalization rule $\NF$. We have four cases.
\begin{description}
\item[$\NF=\NF_1$:] Let $\Imc$ be a model of $\NF^1(\Omc^\semiringshort)=\Omc^\semiringshort\setminus\{(C\sqcap\widehat{D}\sqsubseteq E,\, \elem_0)\}
\cup\{(\widehat{D}\sqsubseteq A,\, \one), \, (C\sqcap A\sqsubseteq E, \, \elem_0)\}$ and let $(d,\elem)\in(C\sqcap\widehat{D})^\Imc$. 
There exist $\elem_1,\elem_2\in\semiringset$ such that $(d,\elem_1)\in C^\Imc$, $(d,\elem_2)\in \widehat{D}^\Imc$, and $\elem_1\otimes \elem_2=\elem$. 
Since $\Imc \models(\widehat D\sqsubseteq A,\one)$, then $(d,\elem_2)\in A^\Imc$.
Hence $(d,\elem_1\otimes \elem_2)\in (C\sqcap A)^\Imc$. 
Since $\Imc\models (C\sqcap A\sqsubseteq E, \, \elem_0)$, it follows that 
$(d,\elem_1\otimes \elem_2\otimes \elem_0)\in E^\Imc$, i.e., $(d,\elem\otimes \elem_0)\in E^\Imc$. 
Thus $\Imc\models (C\sqcap\widehat{D}\sqsubseteq E,\, \elem_0)$, and $\Imc$ is a model of $\Omc^\semiringshort$. 

\item[$\NF=\NF_2$:] This case is analogous to the case $\NF=\NF_1$.

\item[$\NF=\NF_3$:] Let $\Imc$ be a model of 
$\NF^1(\Omc^\semiringshort)=\Omc^\semiringshort\setminus\{(\exists P.\widehat{C}\sqsubseteq D,\, \elem_0)\}\cup\{(\widehat{C}\sqsubseteq A,\,
 \one), \, (\exists P. A\sqsubseteq D, \, \elem_0)\}$ and let $(d,\elem)\in(\exists P.\widehat{C})^\Imc$. 
  There exist $e \in \Delta^\Imc$ and $\elem_1,\elem_2\in\semiringset$ such that $(d, e,\elem_1)\in P^\Imc$, $(e,\elem_2)\in \widehat{C}^\Imc$, and $\elem_1\otimes \elem_2=\elem$. Since $\Imc\models (\widehat{C}\sqsubseteq A,\,
 \one)$, then $(e,\elem_2)\in A^\Imc$. Hence $(d,\elem_1\otimes \elem_2)\in (\exists P. A)^\Imc$. 
 Since $\Imc\models (\exists P. A\sqsubseteq D, \, \elem_0)$, it follows that $(d,\elem_1\otimes \elem_2\otimes \elem_0)\in D^\Imc$, i.e. $(d,\elem\otimes \elem_0)\in D^\Imc$. Thus $\Imc\models (\exists P.\widehat{C}\sqsubseteq D,\, \elem_0)$, and $\Imc$ is a model of~$\Omc^\semiringshort$. 

\item[$\NF=\NF_4$:] Let $\Imc$ be a model of 
$\NF^1(\Omc^\semiringshort)=\Omc^\semiringshort\setminus\{(\widehat{C}\sqsubseteq\exists P,\, \elem_0)\}\cup\{(\widehat{C}\sqsubseteq A,\, \one), 
\, (A\sqsubseteq \exists P, \, \elem_0)\}$ and let $(d,\elem)\in\widehat{C}^\Imc$. Since $\Imc\models (\widehat{C}\sqsubseteq A,\, \one)$, then $(d,\elem)\in A^\Imc$. Since $\Imc\models (A\sqsubseteq \exists P, \, \elem_0)$, it follows that $(d,\elem\otimes\elem_0)\in (\exists P)^\Imc$. Thus 
$\Imc\models(\widehat{C}\sqsubseteq\exists P,\, \elem_0)$, and $\Imc$ is a model of $\Omc^\semiringshort$. 
\end{description}

\noindent\emph{Induction step.} Assume that the property is true for some $i\geq 1$ and let $\Omc^\semiringshort$ be an annotated $\ELHIbot$ ontology, $\NF^{i+1}(\Omc^\semiringshort)$ be the result of applying $i+1$ normalization steps to $\Omc^\semiringshort$ and $\NF^i(\Omc^\semiringshort)$ that of applying the first $i$ steps to $\Omc^\semiringshort$.  
By applying the induction hypothesis on $\NF^i(\Omc^\semiringshort)$, we obtain that every model of $\NF^i(\Omc^\semiringshort)$ is a model of $\Omc^\semiringshort$. 
Then, since $\NF^{i+1}(\Omc^\semiringshort)$ results from applying one normalization step to $\NF^i(\Omc^\semiringshort)$, we obtain that every model of $\NF^{i+1}(\Omc^\semiringshort)$ is a model of $\NF^i(\Omc^\semiringshort)$, and thus also a model of $\Omc^\semiringshort$. 
\smallskip

Conversely, we show by induction that for every $i$ and for every annotated $\ELHIbot$  ontology $\Omc^\semiringshort$, 
 if $\NF^i(\Omc^\semiringshort)$ can be obtained from $\Omc^\semiringshort$ by applying $i$ normalization rules, then if $\NF^i(\Omc^\semiringshort)\models (\alpha,\elem)$ and every concept name occurring in $\alpha$ occurs in $\Omc^\semiringshort$, it holds that $\Omc^\semiringshort\models (\alpha,\elem)$. 

\noindent\emph{Base case: $i=1$.} Let $\Omc^\semiringshort$ be an annotated $\ELHIbot$ ontology, $\alpha$ an axiom such that every concept name occurring in $\alpha$ occurs in $\Omc^\semiringshort$ and $\chi\in\semiringset$. Assume that $\NF^1(\Omc^\semiringshort)$ is obtained from $\Omc^\semiringshort$ by applying a normalization rule $\NF$ and that $\NF^1(\Omc^\semiringshort)\models (\alpha,\chi)$. We have four cases. 
\begin{description}
\item[$\NF=\NF_1$:] $\NF^1(\Omc^\semiringshort)=\Omc^\semiringshort\setminus\{(C\sqcap\widehat{D}\sqsubseteq E,\, \elem_0)\}\cup\{(\widehat{D}\sqsubseteq A,\, \one), \, (C\sqcap A\sqsubseteq E, \, \elem_0)\}$. Let $\Imc$ be a model of~$\Omc^\semiringshort$ and $\Jmc$ be the interpretation that extends $\Imc$ with $A^\Jmc=\widehat{D}^\Imc$. 
Clearly, $\Jmc\models (\widehat{D}\sqsubseteq A,\, \one)$. Let $(d,\elem)\in (C\sqcap A)^\Jmc$. There exist $\elem_1,\elem_2  \in \semiringset$ such that  
$(d,\elem_1)\in C^\Jmc=C^\Imc$, $(d,\elem_2)\in A^\Jmc=\widehat{D}^\Imc$, and $\elem_1\otimes \elem_2=\elem$. 
It follows that $(d,\elem_1\otimes \elem_2)\in (C\sqcap\widehat{D})^\Imc$. 
Since $\Imc\models (C\sqcap\widehat{D}\sqsubseteq E,\, \elem_0)$, it follows that $(d,\elem_1\otimes \elem_2\otimes \elem_0)\in E^\Imc=E^\Jmc$, i.e. $(d,\elem\otimes \elem_0)\in E^\Jmc$. 
Hence $\Jmc\models (C\sqcap A\sqsubseteq E, \, \elem_0)$ and $\Jmc$ is a model of $\NF^1(\Omc^\semiringshort)$. 
It follows that $\Jmc\models (\alpha,\chi)$. Since $\alpha$ does not contain $A$, then $\Imc\models (\alpha,\chi)$. Hence $\Omc^\semiringshort\models (\alpha,\chi)$.

\item[$\NF=\NF_2$:] This case is analogous to the case $\NF=\NF_1$.

\item[$\NF=\NF_3$:] $\NF^1(\Omc^\semiringshort)=\Omc^\semiringshort\setminus\{(\exists P.\widehat{C}\sqsubseteq D,\, \elem_0)\}\cup\{(\widehat{C}\sqsubseteq A,\, \one), \, (\exists P. A\sqsubseteq D, \, \elem_0)\}$. Let $\Imc$ be a model of $\Omc^\semiringshort$ and $\Jmc$ be the interpretation that extends $\Imc$ with $A^\Jmc=\widehat{C}^\Imc$. 
Clearly, $\Jmc\models (\widehat{C}\sqsubseteq A,\, \one)$. Let $(d,\elem)\in (\exists P. A)^\Jmc$. There exist $e  \in \Delta^\Jmc$ and $\elem_1,\elem_2 \in \semiringset$ such that  $(d,e,\elem_1)\in P^\Jmc=P^\Imc$, $(e,\elem_2)\in A^\Jmc=\widehat{C}^\Imc$, and $\elem_1\otimes \elem_2=\elem$. It follows that $(d,\elem_1\otimes \elem_2)\in (\exists P. \widehat{C})^\Imc$. Since $\Imc\models (\exists P.\widehat{C}\sqsubseteq D,\, \elem_0)$, it follows that $(d,\elem_1\otimes \elem_2\otimes \elem_0)\in D^\Imc=D^\Jmc$, i.e. $(d,\elem\otimes \elem_0)\in D^\Jmc$. 
Hence $\Jmc\models (\exists P. A\sqsubseteq D, \, \elem_0)$ and $\Jmc$ is a model of~$\NF^1(\Omc^\semiringshort)$. 
It follows that $\Jmc\models (\alpha,\chi)$. Since $\alpha$ does not contain $A$, then $\Imc\models (\alpha,\chi)$. Hence $\Omc^\semiringshort\models (\alpha,\chi)$.

\item[$\NF=\NF_4$:] $\NF^1(\Omc^\semiringshort)=\Omc^\semiringshort\setminus\{(\widehat{C}\sqsubseteq\exists P,\, \elem_0)\}\cup\{(\widehat{C}\sqsubseteq A,\, \one), \, (A\sqsubseteq \exists P, \, \elem_0)\}$. Let $\Imc$ be a model of $\Omc^\semiringshort$ and $\Jmc$ be the interpretation that extends $\Imc$ with $A^\Jmc=\widehat{C}^\Imc$. 
Clearly, $\Jmc\models (\widehat{C}\sqsubseteq A,\, \one)$. Let $(d,\elem)\in A^\Jmc=\widehat{C}^\Imc$. Since $\Imc\models (\widehat{C}\sqsubseteq\exists P,\, \elem_0)$, it follows that $(d,\elem\otimes \elem_0)\in (\exists P)^\Imc=(\exists P)^\Jmc$. 
Hence $\Jmc\models (A\sqsubseteq \exists P, \, \elem_0)$ and $\Jmc$ is a model of $\NF^1(\Omc^\semiringshort)$. 
It follows that $\Jmc\models (\alpha,\chi)$. Since $\alpha$ does not contain $A$, then $\Imc\models (\alpha,\chi)$. Hence $\Omc^\semiringshort\models (\alpha,\chi)$.
\end{description}

\noindent\emph{Induction step.} Assume that the property is true for some $i$ and let $\Omc^\semiringshort$ be an annotated $\ELHIbot$ ontology, $\alpha$ an axiom such that every concept name occurring in $\alpha$ occurs in $\Omc^\semiringshort$ and $\elem\in\semiringset$. 
Let $\NF^{i+1}(\Omc^\semiringshort)$ 
be obtained by applying $i+1$ normalization rules to $\Omc^\semiringshort$ and $\NF^{i}(\Omc^\semiringshort)$ 
be obtained by applying the first $i$ normalization rules to $\Omc^\semiringshort$  (so that $\NF^{i+1}(\Omc^\semiringshort)$ is obtained by applying a normalization rule $\NF$ to $\NF^{i}(\Omc^\semiringshort)$). 
Assume that $\NF^{i+1}(\Omc^\semiringshort)\models (\alpha,\elem)$.  
Since the normalization rules can only introduce new concept names, the concept names occurring in $\Omc^\semiringshort$ are a subset of those occurring in $\NF^{i}(\Omc^\semiringshort)$, so every concept name occurring in $\alpha$ occurs in $\NF^{i}(\Omc^\semiringshort)$. 
Since we have that $\NF^{i+1}(\Omc^\semiringshort)\models (\alpha,\elem)$, 
that $\NF^{i+1}(\Omc^\semiringshort)$ results from the application of a single normalization rule to $\NF^{i}(\Omc^\semiringshort)$, and that all concept names in $\alpha$ occur in $\NF^{i}(\Omc^\semiringshort)$,  
the base case applies and we obtain that $\NF^{i}(\Omc^\semiringshort)\models (\alpha,\elem)$. Hence by the induction hypothesis, $\Omc^\semiringshort\models (\alpha,\elem)$. 
\end{proof}

\subsection{Canonical Model}
We start by proving two lemmas.

\begin{restatable}{lemma}{lemhom}\label{lem:hom}
	Let $\Imc=(\Delta^\Imc,K,\cdot^\Imc)$ and $\Jmc=(\Delta^\Jmc,K,\cdot^\Jmc)$ be $\semiringshort$-annotated interpretations such that there is a homomorphism
	$\homo: \Imc \rightarrow \Jmc$.
For every $\elem\in K$ and every $\ELHIbot$ concept expression $C$,
if $(d,\elem)\in C^\Imc$ then $(\homo(d),\elem)\in C^\Jmc$. 	
\end{restatable}
\begin{proof}
	The proof is by structural induction. In the base case,
	$C$ is a concept name $A$ and the lemma holds by definition of $\homo$. For $C=\top$ the lemma holds trivially. Now, suppose that the lemma holds
	for $D,D_1,D_2$. We make a case distinction:
	\begin{itemize}
		\item $C=D_1\sqcap D_2$: if $(d,\elem)\in (D_1\sqcap D_2)^\Imc$ then, by the semantics of $\semiringshort$-annotated $\ELHIbot$, 
		there exist $(d,\elem_1)\in D^\Imc_1$
		and $(d,\elem_2)\in D^\Imc_2$ such that $\elem=\elem_1\otimes\elem_2$. 
		By the inductive hypothesis, $(\homo(d),\elem_1)\in D^\Jmc_1$
		and $(\homo(d),\elem_2)\in D^\Jmc_2$.
		Then, by the semantics of $\semiringshort$-annotated
		$\ELHIbot$, we have that $(\homo(d),\elem)\in (D_1\sqcap D_2)^\Jmc$. 
		\item $C=\exists P.D$ (where $P$ can be a role name or an inverse role): if $(d,\elem)\in (\exists P.D)^\Imc$ then, by the semantics of $\semiringshort$-annotated $\ELHIbot$, 
		there exist $e\in\Delta^\Imc$ with $(d,e,\elem_1)\in P^\Imc$
		and $(e,\elem_2)\in D^\Imc$ 
		such that $\elem=\elem_1\otimes\elem_2$. 
		By definition of $\homo$, 
		we have that $(d,e,\elem_1)\in P^\Imc$
		implies $(\homo(d),\homo(e),\elem_1)\in P^\Jmc$
		and, by the inductive hypothesis, $(\homo(e),\elem_2)\in D^\Jmc$. 
		Then, by the semantics of $\semiringshort$-annotated
		$\ELHIbot$, we have that $(\homo(d),\elem)\in (\exists P.D)^\Jmc$. 
		\qedhere
	\end{itemize} 
\end{proof}

\begin{restatable}{lemma}{lemhomcq}\label{lem:homcq}
	Let $\Imc=(\Delta^\Imc,K,\cdot^\Imc)$ and $\Jmc=(\Delta^\Jmc,K,\cdot^\Jmc)$ be $\semiringshort$-annotated interpretations such that there is a homomorphism
	$\homo: \Imc \rightarrow \Jmc$.
	For every $\elem\in K$ and every BCQ $q$,
	if $\Imc\models (q,\elem)$ then $\Jmc\models (q,\elem)$. 	
\end{restatable}
\begin{proof}
	Assume $\Imc\models (q,\elem)$ and let $\ext{q}$ be the extended version of $q$. Denote by $\nu_\Imc(\ext{q})$ the set of all matches  of $\ext{q}$ in \Imc. By semantics of $(q,\elem)$, we have that  $\Imc\models (q,\elem)$ iff 
	there is a match $\pi: \mn{terms}(\ext{q}) \rightarrow \Delta^\Imc\cup \semiringset$ in $\nu_\Imc(\ext{q})$ such that  
	$\elem=\bigotimes_{P(\vec{t},t)\in \ext{q}} \pi(t)$.
	To prove this lemma, we need to construct a match $\pi': \mn{terms}(\ext{q}) \rightarrow \Delta^\Jmc\cup \semiringset$ such that  
	$\elem=\bigotimes_{P(\vec{t},t)\in \ext{q}} \pi'(t)$.
	
	For every $t\in \mn{terms}(\ext{q})$, we define $\pi'(t):=g(\pi(t))$ if $\pi(t)\notin K$ and $\pi'(t):=\pi(t)$ otherwise. 
	It is clear that $\elem=\bigotimes_{P(\vec{t},t)\in \ext{q}} \pi'(t)$ from the definition of $\pi'$. We argue that $\pi'$ is a match of $\ext{q}$ in \Jmc. 
	For every $a \in \NI$, $\pi(a)=a^\Imc$ because $\pi$ is a match, and $g(a^\Imc)=a^\Jmc$ because $g$ is a homomorphism.  	Hence $\pi'(a)=g(\pi(a))=a^\Jmc$ for every $a \in \NI$. 
	It remains to show that
	$\pi'(\vec{t},t)\in P^\Jmc$ for every $P(\vec{t},t)\in \mn{atoms}(\ext{q})$.
	That is,
	(i) $(\pi'(t_1),\pi'(t))\in A^\Jmc$ for every $A(t_1,t)\in \mn{atoms}(\ext{q})$
	and (ii) $(\pi'(t_1),\pi'(t_2),\pi'(t))\in R^\Jmc$ for every $R(t_1,t_2,t)\in \mn{atoms}(\ext{q})$. 
	(i) If $A(t_1,t)\in \mn{atoms}(\ext{q})$ then, since $\pi$ is a match, 
	$(\pi(t_1),\pi(t))\in A^\Imc$, where $\pi(t)$ is necessarily equal to some $\elem\in K$ by definition of \Imc. Since 
	$(\pi(t_1),\elem)\in A^\Imc$ then, by definition of $g$, we have that $(g(\pi(t_1)),\elem)\in A^\Jmc$,  
	\ie $(\pi'(t_1),\pi'(t))\in A^\Jmc$, by definition of $\pi'$. 
	The proof of (ii) is similar.
\end{proof}

\lemhomomo*
\begin{proof}
 Let $\Imc$ be a model of $\Omc^\semiringshort$. 
	The proof is by induction on the sequence of interpretations of the canonical model 
	$\Imc_{\Omc^\semiringshort}$. 
	We define $\homo := \bigcup_{n \geq 0} \homo_n$ and show that for every $n$, $\homo_n : \Imc_n \rightarrow \Imc$ is a homomorphism.  	
	Set $\homo_0: \Delta^{\Imc_{0}} \rightarrow \Delta^{\Imc}$ 
	with  $\homo_0(a) = a^\Imc$ for all $a \in \Delta^{\Imc_{0}}$ 
	(recall that $\Delta^{\Imc_{0}}=\NI$ and $a^{\Imc_{0}}=a$ for all $a \in\NI$). 
	By definition of $\Imc_0$, $(a,\elem) \in A^{\Imc_{0}}$ iff $(A(a),\elem) \in \Omc^\semiringshort$.  
	Since $\Imc$ is a model of $\Omc^\semiringshort$, if $(A(a),\elem) \in \Omc^\semiringshort$ then 
	$(a^{\Imc},\elem) \in A^{\Imc}$. 
	So $(a,\elem) \in A^{\Imc_0}$ 
	implies $(\homo_0(a),\elem) \in A^{\Imc}$. 
	Similarly, if $(a,b,\elem) \in R^{\Imc_0}$ then 
	$(R(a,b),\elem)\in\Omc^\semiringshort$, so 
	$(a^{\Imc},b^{\Imc},\elem) \in R^{\Imc}$.   
	Then, $(a,b,\elem) \in R^{\Imc_0}$ implies $(\homo_0(a),\homo_0(b),\elem) \in R^{\Imc}$. 
	Thus, $\homo_0: \Imc_0 \rightarrow \Imc$ is a homomorphism. 
	
	Suppose it was proven that $\homo_n : \Imc_n \rightarrow \Imc$ is a 
	homomorphism.  
	We want to show that there is a homomorphism
	$\homo_{n+1} : \Imc_{n+1} \rightarrow \Imc$ that extends $g_n$.
	By definition of the canonical model,
	given $\Imc_n$,
	the interpretation $\Imc_{n+1}$ is \new{obtained from $\Imc_n$ by applying the chase rule to some $(\alpha,\elem)\in\Omc^\semiringshort$ and $(\vec{d},\elem')\in E^{\Imc_n}$. 
		We are in one of the following cases:
\begin{enumerate}[(1)]
\item \blue{$\alpha=P\sqsubseteq Q$, $(\vec{d})=(d,d')$, $E=P$, $Q^{\Imc_{n+1}}=Q^{\Imc_n}\cup\{(d,d',\elem\otimes \elemb)\}$;}
\item 
\blue{$\alpha=C\sqsubseteq A$, $(\vec{d})=(d)$, $E=C$, $A^{\Imc_{n+1}}=A^{\Imc_{n}}\cup\{(d,\elem\otimes \elemb)\}$;}
\item \blue{$\alpha=C\sqsubseteq \exists P$,  $(\vec{d})=(d)$, $E=C$, $\Delta^{\Imc_{n+1}}=\Delta^{\Imc_{n}}\cup\{d_f\}$ with $d_f\notin\Delta^{\Imc_{n}}$,  
$P^{\Imc_{n+1}}=P^{\Imc_{n}}\cup\{(d,d_f,\elem\otimes \elemb)\}$.}
 \end{enumerate}}
	In the first two cases, we define $\homo_{n+1}:=\homo_{n}$. In the third case, we need to map $d_f$ to an element of \blue{$\Delta^\Imc$}. 
	Since $\homo_{n}: \Imc_n \rightarrow \Imc$ is a homomorphism, by Lemma~\ref{lem:hom}, 
	$(d,\elemb)\in C^{\Imc_n}$ implies that $(\homo_{n}(d),\elemb)\in C^{\Imc}$.
	As  $\Imc$ is a model of $\Omc^\semiringshort$, 
	if $\pair{C\sqsubseteq \exists P}{\elem}\in \Omc^\semiringshort$ and $(\homo_{n}(d),\elemb)\in C^{\Imc}$ then there is
	$e\in\Delta^\Imc$ such that
	$(\homo_{n}(d),e,\elem\otimes\elemb)\in P^\Imc$.
	We take a fixed but arbitrary such $e$
	and define $\homo_{n+1}$ in the same way 
	as $\homo_{n}$ except that $\homo_{n+1}(d_f):=e$.
	
	We now show that in each case $\homo_{n+1} : \Imc_{n+1} \rightarrow \Imc$ is a homomorphism.
	\begin{itemize}
		\item[(1)] The difference between $\Imc_n$
		and $\Imc_{n+1}$ is that
		now we have $(d,d',\elem\otimes \elemb)\in Q^{\Imc_{n+1}}$.
		By assumption $\homo_{n}$ is a homomorphism,
		so $(d,d',\elemb)\in P^{\Imc_n}$ implies that
		$(\homo_{n}(d),\homo_{n}(d'),\elemb)\in P^{\Imc}$.
		Since in this case $\homo_{n+1}=\homo_{n} $,
		we also have that
		$(\homo_{n+1}(d),\homo_{n+1}(d'),\elemb)\in P^{\Imc}$. As \Imc is a model of  $\Omc^{\semiringshort}$,
		if
		$\pair{P\sqsubseteq Q}{\elem}\in 	\Omc^{\semiringshort}$ then
		$(\homo_{n+1}(d),\homo_{n+1}(d'),\elem\otimes \elemb)\in Q^{\Imc}$,
		which means that $\homo_{n+1} : \Imc_{n+1} \rightarrow \Imc$ is a homomorphism, as required. 
		\item[(2)] In this case,	the difference between $\Imc_n$
		and $\Imc_{n+1}$ is that
		now we have $(d,\elem\otimes \elemb)\in A^{\Imc_{n+1}}$.
		By assumption $\homo_{n}$ is a homomorphism,
		so (by Lemma~\ref{lem:hom}) $(d,\elemb)\in C^{\Imc_n}$ implies that
		$(\homo_{n}(d),\elemb)\in C^{\Imc}$.
		Since in this case $\homo_{n+1}=\homo_{n} $,
		we also have that
		$(\homo_{n+1}(d),\elemb)\in C^{\Imc}$. As \Imc is a model of  $\Omc^{\semiringshort}$,
		if
		$\pair{C\sqsubseteq A}{\elem}\in 	\Omc^{\semiringshort}$ then
		$(\homo_{n+1}(d),\elem\otimes\elemb)\in A^{\Imc}$,
		which means that $\homo_{n+1} : \Imc_{n+1} \rightarrow \Imc$ is a homomorphism, as required.
		\item[(3)] Finally, here the difference between $\Imc_n$
		and $\Imc_{n+1}$ 
		is that
		$\Delta^{\Imc_{n+1}}=\Delta^{\Imc_{n}}\cup\{d_f\}$
		and
		$	(d,d_f,\elem\otimes \elemb)\in P^{\Imc_{n+1}}$,
		where $d_f$ is a fresh element. 
		In this case $\homo_{n+1}$ is the same as $\homo_{n}$ except for $\homo_{n+1}(d_f)$,
		which is mapped to an element in $e\in\Delta^\Imc$
		such that $(\homo_{n+1}(d),e,\elem\otimes\elemb)\in P^\Imc$ (recall that $\homo_{n+1}(d)=\homo_{n}(d)$).
		So $(d,d_f,\elem\otimes\elemb)\in P^{\Imc_{n+1}}$
		implies $(\homo_{n+1}(d),\homo_{n+1}(d_f),\elem\otimes\elemb)\in P^\Imc$, as required.
	\end{itemize}
	
	Since $\Imc_{\Omc^{\semiringshort}} = \bigcup_{n \geq 0} \Imc_n$, 
	there exists a homomorphism $\homo: \Imc_{\Omc^{\semiringshort}} \rightarrow \Imc$.
\end{proof}

\thmcanmodelmain*
\begin{proof}
Since annotated BCQs subsume annotated assertions, the two first points are consequences of the last one. 
Assume $\Omc^\semiringshort\models (q,\elem)$. By Proposition~\ref{lem:can},
	$\Imc_{\Omc^\semiringshort}\models \Omc^\semiringshort$. So
	$\Omc^\semiringshort\models (q,\elem)$ implies
	$\Imc_{\Omc^\semiringshort}\models (q,\elem)$. 
Conversely,  assume that $\Imc_{\Omc^\semiringshort}\models (q,\elem)$ and let $\Imc$ be a model of $\Omc^\semiringshort$. 
By Lemma~\ref{lem:homomo}, 
	there exists a homomorphism 
	$\homo: \Imc_{\Omc^\semiringshort} \rightarrow \Imc$.  
	By Lemma~\ref{lem:homcq}, we thus have that  $\Imc\models (q,\elem)$. 
	It follows that $\Omc^\semiringshort\models (q,\elem)$.	
\end{proof}

\new{Given an annotation $\elem\in K$, we define $\Imc_{C,\elem,\Omc^\semiringshort}$ in the same way as $\Imc_{C,\Omc^\semiringshort}$ except that we use $\elem$ instead of $\one$ in the definition of $\Imc_C$ (\cf paragraph before Theorem~\ref{thm:can-model-main-gci}).}

\begin{restatable}{lemma}{lemhomobasic}\label{lem:homomobasic} \new{Let $\Imc$ be a model of $\Omc^\semiringshort$ with $(d,\elem)\in C^\Imc$, where $C$ is a \emph{basic} concept. 
	Then there exists a homomorphism 
	$\homo: \Imc_{C,\elem,\Omc^\semiringshort} \rightarrow \Imc$  such that $\homo(d_C)=d$.}
\end{restatable}
\begin{proof}
\new{	Similar to the case of Lemma~\ref{lem:homomo}, the proof is by induction on the sequence of interpretations of the canonical model, which now is 
	$\Imc_{C,\elem,\Omc^\semiringshort}$ with $C$ a basic concept. 
	We define $\homo := \bigcup_{n \geq 0} \homo_n$ and set  $\homo_0: \Delta^{\Imc_{0}} \rightarrow \Delta^{\Imc}$ 
	with  $g_0(d_C)=d$ and $\homo_0(a) = a^\Imc$ for all $a \in \NI$. 
	By definition of $\Imc_0$, we have that $(d_C,\elem)\in C^{\Imc_0}$ and
	by assumption $(d,\elem)\in C^\Imc$.
	Since $C$ is a basic concept, it is either a concept name $A$ or of the form $\exists P$, with $P$ a (possibly inverse) role. For $C=A$ setting $g_0(d_A)=d$ clearly
	satisfies the homomorphism property required by this lemma.
	For $C=\exists P$, we know that $(d,\elem)\in (\exists P)^\Imc$ 
	holds iff there is $d'\in\Delta^\Imc$ such that $(d,d',\elem)\in P^\Imc$.
	We map $d_f\in \Delta^{\Imc_0}$ (see $d_f$ in the definition of the canonical model $\Imc_C$ for a basic concept $C$)
	to such $d'$, that is, $g_0(d_f)=d'$.   Then we have that
	$(d_C,d_f,\elem)\in P^{\Imc_0}$ implies $(g_0(d_C),g_0(d_f),\elem)\in P^\Imc$, which satisfies the homomorphism property required by this lemma. The argument for the assertions is as in Lemma~\ref{lem:homomo}: 
	if $(a,\elem) \in A^{\Imc_{0}}$, then $(A(a),\elem) \in \Omc^\semiringshort$ so $(a^{\Imc},\elem) \in A^{\Imc}$, \ie $(a,\elem) \in A^{\Imc_0}$ 
	implies $(\homo_0(a),\elem) \in A^{\Imc}$, and similarly for role assertions.  
	We have  shown that  $\homo_0:\Imc_{0} \rightarrow \Imc$ is a homomorphism. The argument for the inductive step is as in Lemma~\ref{lem:homomo}.}
\end{proof}

	\begin{lemma}\label{cl:helper}
	Assume $\semiringshort$ is a commutative \timesidem\ semiring. 
		If $(d_C,\elem)\in D^{\Imc_{C,\Omc^{\semiringshort}}}$, where $D$ is a basic concept, 
		then $(d_C,\elem'\otimes \elem)\in D^{\Imc_{C,\elem',\Omc^{\semiringshort}}}$.
	\end{lemma}	 
	\begin{proof}
	\new{
By construction of $\Imc_{C,\Omc^\semiringshort}=\bigcup_{n\geq 0}\Imc_n$, one can show by induction on $n$ such that $(d_C,\elem)\in D^{\Imc_{n}}$ that there exists a sequence of GCIs and role inclusions from $\Omc^\semiringshort$ with annotation $\elem_1,\dots,\elem_k$ which forms a subset of the axioms used in the $n$ chase rule applications that go from $\Imc_0$ to $\Imc_n$ and is such that $\elem=\one\otimes\elem_1\otimes\dots\otimes\elem_k$. This can be shown thanks to $\otimes$-idempotency (for example, if we apply the chase rule with $(A_1\sqcap A_2\sqsubseteq B,\elem_4)$ and some $(e,\one\otimes\elem_1\otimes\elem_2\otimes\elem_1\otimes\elem_3)\in (A_1\sqcap A_2)^{\Imc_i}$ that comes from $(e,\one\otimes\elem_1\otimes\elem_2)\in A_1^{\Imc_i}$ and $(e,\one\otimes\elem_1\otimes\elem_3)\in A_2^{\Imc_i}$, we obtain $(e,\one\otimes\elem_1\otimes\elem_2\otimes\elem_3\otimes\elem_4)\in B^{\Imc_{i+1}}$ because $\elem_1\otimes\elem_1=\elem_1$). }

\new{Recall that $\Imc_{C,\elem',\Omc^{\semiringshort}}=\bigcup_{n\geq 0}\Imc'_n$ is defined in the same way as $\Imc_{C,\Omc^\semiringshort}$ except that we use $\elem'$ instead of $\one$ in the definition of $\Imc_C$. 
	Since $\Omc$ does not contain any GCI with $\top$ as left-hand side, and $d_C$ only occurs in $\Imc'_0$ in $(d_C,\elem')$ in $C^{\Imc'_0}$ and in $(d_C,\one)$ in $\top^{\Imc'_0}$,  when applying the same rules starting with $(d_C,\elem')$ in $C^{\Imc'_0}$ instead of $(d_C,\one)$, one can show that we obtain $\elem'\otimes\elem_1\otimes\dots\otimes\elem_k=\elem'\otimes\elem$ instead of $\elem$ (again, thanks to $\otimes$-idempotency).} 
	\end{proof}

\thmcanmodelmaingci*
\begin{proof}
	Assume ${\Omc^\semiringshort}\models (C\sqsubseteq D,\elem)$, where $C,D$ are basic concepts (that is, either a concept name $A$ or of the form $\exists P$) \new{and $C$ is satisfiable \wrt $\Omc^\semiringshort$}. This means that, for every interpretation 
	\Imc \new{such that $\Imc\models\Omc^\semiringshort$, }
	if $(d,\elem')\in C^\Imc$ then $(d,\elem\otimes \elem')\in D^\Imc$.
	In particular, this holds for the canonical model $\Imc_{C,\Omc^{\semiringshort}}$ 
	of  $\Omc^\semiringshort$ and $C$.
	Since $(d_C,\one)\in C^{\Imc_{C,\Omc^{\semiringshort}}}$ and $\elem\otimes \one=\elem$ we have that
	$(d_C,\elem)\in D^{\Imc_{C,\Omc^{\semiringshort}}}$. 
	
	Conversely, assume $(d_C,\elem)\in D^{\Imc_{C,\Omc^{\semiringshort}}}$. 
	Let $\Imc$ be a model of $\Omc^\semiringshort$. If $(d,\elem')\in C^\Imc$, 
	then by Lemma~\ref{lem:homomobasic} 
	there exists a homomorphism 
	$\homo: \Imc_{C,\elem',\Omc^\semiringshort} \rightarrow \Imc$  such that 
	$g(d_C)=d$. 
	By Lemma~\ref{cl:helper}, $(d_C,\elem'\otimes \elem)\in D^{\Imc_{C,\elem',\Omc^{\semiringshort}}}$, 
	so by Lemma~\ref{lem:hom}, $(d,\elem'\otimes\elem)\in D^\Imc$. 	
	Hence $\Imc\models (C\sqsubseteq D,\elem)$.
	Since \Imc was an arbitrary model of $\Omc^\semiringshort$, 
	we obtain ${\Omc^\semiringshort}\models (C\sqsubseteq D,\elem)$. 	
\end{proof}

\thmcanmodelmainri*
\begin{proof}
\new{
Assume ${\Omc^\semiringshort}\models (P\sqsubseteq Q,\elem)$. This means that, for every interpretation 
	\Imc such that $\Imc\models\Omc^\semiringshort$, 
	if $(d,e,\elem')\in P^\Imc$ then $(d,e,\elem\otimes \elem')\in Q^\Imc$.
	In particular, this holds for the canonical model $\Imc_{P,\Omc^{\semiringshort}}$ 
	of  $\Omc^\semiringshort$ and $P$.
	Since $(d_1,d_2,\one)\in P^{\Imc_{P,\Omc^{\semiringshort}}}$ and $\elem\otimes \one=\elem$ we have that
	$(d_1,d_2,\elem)\in Q^{\Imc_{P,\Omc^{\semiringshort}}}$. 
	}
	
	\new{
	Conversely, assume $(d_1,d_2,\elem)\in Q^{\Imc_{P,\Omc^{\semiringshort}}}$. By construction of $\Imc_{P,\Omc^\semiringshort}=\bigcup_{n\geq 0}\Imc_n$, since $(d_1,d_2,\one)\in P^{\Imc_0}$ and $d_1,d_2$ do not occur anywhere else in $\Imc_0$, one can show by induction on $n$ such that $(d_1,d_2,\elem)\in Q^{\Imc_{n}}$ that there exists a sequence of role inclusions $(P_1\sqsubseteq P_2,\elem_1),\dots, (P_{k-1}\sqsubseteq P_k,\elem_k)$ from $\Omc^\semiringshort$ which forms a subset of the axioms used in the $n$ chase rule applications that go from $\Imc_0$ to $\Imc_n$ and is such that $\elem=\one\otimes\elem_1\otimes\dots\otimes\elem_k$. 
	Let $\Imc$ be a model of $\Omc^\semiringshort$. Since $\Imc$ is a model of $(P_1\sqsubseteq P_2,\elem_1),\dots, (P_{k-1}\sqsubseteq P_k,\elem_k)$, one can show that $(d,e,\elem')\in P^\Imc$ implies that $(d,e,\elem'\otimes\elem)\in Q^\Imc$. Hence $\Omc^\semiringshort\models (P\sqsubseteq Q,\elem)$.
}
\end{proof}

\subsection{\new{Reduction Between Assertion and GCI or RI Entailment}}

\new{We start with two lemmas that will be used to prove the second point of Theorem~\ref{th:red-concept}.}

\begin{lemma}\label{lem:basic-reduction-lemma-base-case}
\new{Let $\semiringshort$ be a commutative semiring, $\Omc^\semiringshort$ a $\semiringshort$-annotated $\ELHIbot$ ontology such that $\Omc$ does not contain any GCI with $\top$ as left-hand side,  $B(a_0)$ a concept assertion and $\elem_0\in\semiringset$. Let $\Tmc^\semiringshort$ be defined as in the second point of Theorem~\ref{th:red-concept}. Assume that $C_{a_0}$ is satisfiable \wrt $\Tmc^\semiringshort$ and let $\Jmc:=\Imc_{C_{a_0},\Tmc^\semiringshort}$ be the canonical model of $C_{a_0}$ and $\Tmc^\semiringshort$ and $\Jmc_0$ be first interpretation built in the construction of $\Jmc$, \ie 
\begin{itemize}
	\item $\Delta^{\Jmc_0}:=\NI\cup\{d_{C_{a_0}}\}$;
	\item $a^{\Jmc_0}:=a$,  
	for all $a\in\NI$;
	\item $C_{a_0}^{\Jmc_0}=\{(d_{C_{a_0}},\one)\}$;
	\item $(a,\elem)\in A^{\Jmc_0}$ iff 
$(A(a),\elem)\in\Omc^{\semiringshort}$; 
	\item $(a,b,\elem)\in R^{\Jmc_0}$ iff 
$(R(a,b),\elem)\in\Omc^{\semiringshort}$. 
\end{itemize}
The following properties hold.
\begin{enumerate}
\item For every $R\in\NR$, for every $b,c\in\NI$ and $\elem\in\semiringset$, $(b,c,\elem)\in R^{\Jmc_0}$ implies that for every $(e,\one)\in C_b^\Jmc$, there exists $(d,\one)\in C_c^\Jmc$ such that $(e,d,\elem)\in R^\Jmc$.
\item For every $\ELHIbot$ concept $D$, for every $b\in\NI$ and $\elem\in\semiringset$, $(b,\elem)\in D^{\Jmc_0}$ implies that for every $(e,\one)\in C_b^\Jmc$, $(e,\elem)\in D^\Jmc$.
\end{enumerate}}
\end{lemma}

\begin{proof}
\new{
For point (1), assume that $(b,c,\elem)\in R^{\Jmc_0}$. By definition of $\Jmc_0$, $(R(b,c),\elem)\in\Omc^{\semiringshort}$. Hence the following axioms belong to $\Tmc^\semiringshort$: $(R_{bc}\sqsubseteq R,\elem)$, $(C_b\equiv \exists R_{bc},\one)$, $(C_c\equiv \exists R_{bc}^-,\one)$. Let $(e,\one)\in C_b^\Jmc$. Since $\Jmc$ is a model of $\Tmc^\semiringshort$, there exists $d\in\Delta^\Jmc$ such that $(e,d,\one)\in R_{bc}^\Jmc$,  $(e,d,\elem)\in R^\Jmc$ and $(d,\one)\in C_c^\Jmc$. 
}

\new{
We now show point (2) by structural induction. 
\begin{itemize}
\item Base case: $D\in\NC\cup\{\top\}$. Let $b\in\NI$, and $\elem\in\semiringset$ be such that $(b,\elem)\in D^{\Jmc_0}$.
\begin{itemize}
\item If $D=\top$, $(b,\elem)\in \top^{\Jmc_0}$ implies that $\elem=\one$ and for every $(e,\one)\in C_b^\Jmc$, $(e,\one)\in \top^\Jmc$. 
\item If $D\in\NC$, $(b,\elem)\in D^{\Jmc_0}$ implies that $(D(b),\elem)\in\Omc^{\semiringshort}$, so that $(C_b\sqsubseteq D,\elem)\in\Tmc^\semiringshort$. Hence, since $\Jmc\models\Tmc^\semiringshort$, for every $(e,\one)\in C_b^\Jmc$, it holds that $(e,\elem)\in D^\Jmc$.
\end{itemize}
\item Induction step: 
\begin{itemize}
\item Let $D=C_1\sqcap C_2$ with $C_1$ and $C_2$ $\ELHIbot$ concepts such that the property holds. Let $b\in\NI$ and $\elem\in\semiringset$ be such that $(b,\elem)\in D^{\Jmc_0}$. There exist $\elem_1$ and $\elem_2$ such that $(b,\elem_1)\in C_1^{\Jmc_0}$, $(b,\elem_2)\in C_2^{\Jmc_0}$ and $\elem_1\otimes\elem_2=\elem$. By induction hypothesis, for $i\in\{1,2\}$, $(b,\elem_i)\in C_i^{\Jmc_0}$ implies that for every $(e,\one)\in C_b^\Jmc$, $(e,\elem_i)\in C_i^\Jmc$, so  $(e,\elem_1\otimes\elem_2)\in (C_1\sqcap C_2)^\Jmc$, \ie $(e,\elem)\in D^\Jmc$.
\item Let $D=\exists P.C$ with $P$ a role name or an inverse role and $C$ an $\ELHIbot$ concept such that the property holds. Let $b\in\NI$ and $\elem\in\semiringset$ be such that $(b,\elem)\in D^{\Jmc_0}$. There exist $c\in\Delta^{\Jmc_0}$, $\elem_1$ and $\elem_2$ such that $(b,c,\elem_1)\in P^{\Jmc_0}$, $(c,\elem_2)\in C^{\Jmc_0}$ and $\elem_1\otimes\elem_2=\elem$. 
Let $(e,\one)\in C_b^\Jmc$. 
First note that $c\in\NI$. Indeed, it cannot be the case that $c=d_{C_{a_0}}$ and $(b,c,\elem_1)\in P^{\Jmc_0}$ by definition of $\Jmc_0$. 
Hence, by point~(1), $(b,c,\elem_1)\in P^{\Jmc_0}$ and $(e,\one)\in C_b^\Jmc$ implies that there exists $(d,\one)\in C_c^\Jmc$ such that $(e,d,\elem_1)\in P^\Jmc$. Moreover, by induction hypothesis, $(c,\elem_2)\in C^{\Jmc_0}$ and $(d,\one)\in C_c^\Jmc$ implies $(d,\elem_2)\in C^\Jmc$. It follows that $(e,\elem_1\otimes\elem_2)\in (\exists P.C)^\Jmc$, \ie $(e,\elem)\in D^\Jmc$. \qedhere
\end{itemize}
\end{itemize}
}
\end{proof}

\begin{lemma}\label{lem:basic-reduction-lemma}
\new{Let $\semiringshort$ be a commutative semiring, $\Omc^\semiringshort$ a $\semiringshort$-annotated $\ELHIbot$ ontology such that $\Omc$ does not contain any GCI with $\top$ as left-hand side,  $B(a_0)$ a concept assertion and $\elem_0\in\semiringset$. Let $\Tmc^\semiringshort$ be defined as in the second point of Theorem~\ref{th:red-concept}. Assume that $C_{a_0}$ is satisfiable \wrt $\Tmc^\semiringshort$ and let $\Jmc:=\Imc_{C_{a_0},\Tmc^\semiringshort}$ be the canonical model of $C_{a_0}$ and $\Tmc^\semiringshort$. 
The following properties hold.
\begin{enumerate}
\item For every $R\in\NR$, for every $b,c\in\NI$ and $\elem\in\semiringset$, $(b,c,\elem)\in R^{\Jmc}$ implies that for every $(e,\one)\in C_b^\Jmc$, there exists $(d,\one)\in C_c^\Jmc$ such that $(e,d,\elem)\in R^\Jmc$.
\item For every $\ELHIbot$ concept $D$, for every $b\in\NI$ and $\elem\in\semiringset$, $(b,\elem)\in D^{\Jmc}$ implies that for every $(e,\one)\in C_b^\Jmc$, $(e,\elem)\in D^\Jmc$.
\end{enumerate}}
\end{lemma}
\begin{proof}
\new{We denote by $\Jmc_0,\Jmc_1,\dots$ the annotated interpretations built in the construction of $\Jmc=\bigcup_{n\geq 0}\Jmc_n$ (\cf Section~\ref{sec:canonical-model}). }

\new{
We show point (1) by proving by induction on $n$ that for every $R\in\NR$, for every $b,c\in\NI$ and $\elem\in\semiringset$, $(b,c,\elem)\in R^{\Jmc_n}$ implies that for every $(e,\one)\in C_b^\Jmc$, there exists $(d,\one)\in C_c^\Jmc$ such that $(e,d,\elem)\in R^\Jmc$. 
\begin{itemize}
\item The base case ($n=0$) follows from point (1) of Lemma~\ref{lem:basic-reduction-lemma-base-case}. 
\item Assume now that the property is true for some $n\geq 0$.  
Let $(b,c,\elem)\in R^{\Jmc_{n+1}}$. 
If $(b,c,\elem)\in R^{\Jmc_{n}}$, we obtain the result by induction hypothesis. Otherwise, the rule applied to obtain $\Jmc_{n+1}$ from $\Jmc_n$ added  $(b,c,\elem)$ to $R^{\Jmc_{n}}$ using some $(P\sqsubseteq R,\elem_1)\in\Tmc^\semiringshort$ (with $P$ a role name or an inverse role) such that $(b,c,\elem_2)\in P^{\Jmc_n}$ with $\elem_1\otimes\elem_2=\elem$. 
Let $(e,\one)\in C_b^\Jmc$. By induction hypothesis, there exists $(d,\one)\in C_c^\Jmc$ s.t.\ $(e,d,\elem_2)\in P^\Jmc$. 
Since $\Jmc\models \Tmc^\semiringshort$, it follows that $(e,d,\elem_1\otimes\elem_2)\in R^\Jmc$, \ie $(e,d,\elem)\in R^\Jmc$. 
\end{itemize}
}

\new{
We now show point (2) by proving by induction on $n$ that for every $\ELHIbot$ concept $D$, for every $b\in\NI$ and $\elem\in\semiringset$, $(b,\elem)\in D^{\Jmc_n}$ implies that for every $(e,\one)\in C_b^\Jmc$, $(e,\elem)\in D^\Jmc$.  
\begin{itemize}
\item The base case ($n=0$) follows from point (2) of Lemma~\ref{lem:basic-reduction-lemma-base-case}. 
\item Assume that for every $\ELHIbot$ concept $D$, for every $b\in\NI$ and $\elem\in\semiringset$, $(b,\elem)\in D^{\Jmc_n}$ implies that for every $(e,\one)\in C_b^\Jmc$, $(e,\elem)\in D^\Jmc$. 
We prove by structural induction that for every $\ELHIbot$ concept $D$, for every $b\in\NI$ and $\elem\in\semiringset$, $(b,\elem)\in D^{\Jmc_{n+1}}$ implies that for every $(e,\one)\in C_b^\Jmc$, $(e,\elem)\in D^\Jmc$. 
\begin{itemize}
\item Base case: $D\in\NC\cup\{\top\}$. If $D=\top$, $\elem=\one$ and for every $(e,\one)\in C_b^\Jmc$, $(e,\one)\in \top^\Jmc$. If $D\in\NC$, let $(e,\one)\in C_b^\Jmc$.  If $(b,\elem)\in D^{\Jmc_{n}}$, we obtain $(e,\elem)\in D^\Jmc$ by induction. Otherwise, the rule applied to obtain $\Jmc_{n+1}$ from $\Jmc_n$ added $(b,\elem)$ to $D^{\Jmc_{n}}$ using some $(E\sqsubseteq D,\elem_1)\in\Tmc^\semiringshort$ (with $E$ an $\ELHIbot$ concept) such that $(b,\elem_2)\in E^{\Jmc_n}$ with $\elem_1\otimes\elem_2=\elem$. Since $(b,\elem_2)\in E^{\Jmc_n}$, by induction hypothesis, $(e,\elem_2)\in E^{\Jmc}$. Hence, since $\Jmc\models \Tmc^\semiringshort$, $(e,\elem_1\otimes\elem_2)\in D^{\Jmc}$, \ie $(e,\elem)\in D^{\Jmc}$. 
\item Induction step:
\begin{itemize}
\item Assume that $D=C_1\sqcap C_2$ with $C_1$ and $C_2$ such that for every $c\in\NI$ and $\chi\in\semiringset$, $(c,\chi)\in C_i^{\Jmc_{n+1}}$ implies that for every $(d,\one)\in C_c^\Jmc$, $(d,\chi)\in C_i^{\Jmc}$. 
Let $(e,\one)\in C_b^\Jmc$. 
Since $(b,\elem)\in D^{\Jmc_{n+1}}$, there exist $(b,\elem_1)\in C_1^{\Jmc_{n+1}}$ and $(b,\elem_2)\in C_2^{\Jmc_{n+1}}$ such that $\elem_1\otimes\elem_2=\elem$. 
Hence, $(e,\elem_1)\in C_1^{\Jmc}$ and $(e,\elem_2)\in C_2^{\Jmc}$, so $(e,\elem)\in D^{\Jmc}$. 
\item Assume that $D=\exists P.C$ with $C$ such that for every $c\in\NI$ and $\chi\in\semiringset$, $(c,\chi)\in C^{\Jmc_{n+1}}$ implies that for every $(d,\one)\in C_c^\Jmc$, $(d,\chi)\in C^{\Jmc}$, and 
let $(e,\one)\in C_b^\Jmc$. 
Since $(b,\elem)\in D^{\Jmc_{n+1}}$, there must exist $(b,c,\elem_1)\in P^{\Jmc_{n+1}}$ and $(c,\elem_2)\in C^{\Jmc_{n+1}}$ such that $\elem_1\otimes\elem_2=\elem$. 
By point (1), there exists $(d,\one)\in C_c^\Jmc$ such that $(e,d,\elem_1)\in P^\Jmc$. 
Since $(d,\one)\in C_c^\Jmc$ and  $(c,\elem_2)\in C^{\Jmc_{n+1}}$, we get that $(d,\elem_2)\in C^{\Jmc}$. 
Hence $(e,\elem)\in D^{\Jmc}$. \qedhere
\end{itemize}
\end{itemize}
\end{itemize}
}
\end{proof}

\thReductionConcept*
\begin{proof}
\new{For the first point we show that ${\Omc^\semiringshort}\models (C\sqsubseteq D,\elem_0)$ iff $\Omc^\semiringshort\cup\Tmc^\semiringshort_D\cup\Amc^\semiringshort_C\models (E(a_0), \elem_0)$. Recall that $C$ and $D$ are \emph{basic concepts} (concept names or of the form $\exists P$). 
First note that if ${\Omc^\semiringshort}$ is unsatisfiable, so is $\Omc^\semiringshort\cup\Tmc^\semiringshort_D\cup\Amc^\semiringshort_C$, and both annotated ontologies entail every annotated axiom. 
In the same way, if $C$ is unsatisfiable \wrt $\Omc$, $\Omc^\semiringshort\models(C\sqsubseteq D,\elem)$ for every $\elem\in\semiringset$ and $\Omc^\semiringshort\cup\Tmc^\semiringshort_D\cup\Amc^\semiringshort_C$ is unsatisfiable so entails $(E(a_0),\elem)$ for every $\elem\in\semiringset$. 
Hence, we next focus on the case where ${\Omc^\semiringshort}$ is satisfiable and $C$ is satisfiable \wrt $\Omc$, so that $\Omc^\semiringshort\cup\Tmc^\semiringshort_D\cup\Amc^\semiringshort_C$ is satisfiable. }

\noindent\new{($\Rightarrow$) 
Let $\Imc$ be a model of $\Omc^\semiringshort\cup\Tmc^\semiringshort_D\cup\Amc^\semiringshort_C$. By construction of $\Amc^\semiringshort_C$, $(a_0^\Imc,\one)\in C^\Imc$. Since ${\Omc^\semiringshort}\models (C\sqsubseteq D,\elem_0)$ and $\Imc\models\Omc^\semiringshort$, it follows that $(a_0^\Imc,\one\otimes\elem_0)\in D^\Imc$, \ie $(a_0^\Imc,\elem_0)\in D^\Imc$. Hence, since $\Imc\models\Tmc^\semiringshort_D$, $(a_0^\Imc,\one\otimes\elem_0)\in E^\Imc$, \ie $(a_0^\Imc,\elem_0)\in E^\Imc$. }

\noindent\new{($\Leftarrow$) We show the other direction by contrapositive: We assume that $\Omc^\semiringshort\not\models (C\sqsubseteq D,\elem_0)$ and show that $\Omc^\semiringshort\cup\Tmc^\semiringshort_D\cup\Amc^\semiringshort_C\not\models (E(a_0),\elem_0)$. 
Let $\Imc_{C,\Omc^\semiringshort}$ be the canonical model of $C$ and $\Omc^\semiringshort$. Recall that $\Imc_{C,\Omc^\semiringshort}$ is a model of $\Omc^\semiringshort$ such that its domain element $d_C$ is such that $(d_C,\one)\in C^{\Imc_{C,\Omc^\semiringshort}}$ (\cf Section~\ref{sec:canonical-model}). 
Since $\semiringshort$ is \timesidem and $\Omc$ does not contain any GCI with $\top$ as left-hand side, by Theorem~\ref{thm:can-model-main-gci}, $\Omc^\semiringshort\not\models (C\sqsubseteq D,\elem_0)$ implies that $(d_C,\elem_0)\notin D^{\Imc_{C,\Omc^\semiringshort}}$. We obtain a model $\Imc$ of $\Omc^\semiringshort\cup\Tmc^\semiringshort_D\cup\Amc^\semiringshort_C$ such that $\Imc\not\models (E(a_0),\elem_0)$ as follows: $a_0^\Imc=d_C$ (recall that $a_0$ does not occur in $\Omc$), $E^\Imc=D^{\Imc_{C,\Omc^\semiringshort}}$ and for all other individual names, concept names and role names, $\Imc$ coincides with $\Imc_{C,\Omc^\semiringshort}$.}
\smallskip

\new{
We now show the second point: ${\Omc^\semiringshort}\models (B(a_0),\elem_0)$ iff $\Tmc^\semiringshort\models (C_{a_0}\sqsubseteq B,\elem_0)$. 
Again, note that if ${\Omc^\semiringshort}$ is unsatisfiable, so is $\Tmc^\semiringshort=\Omc^\semiringshort\cup\bigcup_{a\in\individuals{\Omc}}\Tmc^\semiringshort_{C_a}$, and both annotated ontologies entail every annotated axiom. Hence, we next focus on the case where ${\Omc^\semiringshort}$ is satisfiable. }

\noindent\new{($\Rightarrow$)  Assume that $\Tmc^\semiringshort\not\models(C_{a_0}\sqsubseteq B,\elem_0)$. Let $\Imc_{C_{a_0},\Tmc^\semiringshort}$ be the canonical model of $C_{a_0}$ and $\Tmc^\semiringshort$ (note that $C_{a_0}$ is satisfiable \wrt $\Tmc^\semiringshort$, otherwise we would have $\Tmc^\semiringshort\models (C_{a_0}\sqsubseteq B,\elem)$ for any $\elem\in\semiringset$). Since $\semiringshort$ is \timesidem  and $\Omc$ (hence also $\Tmc$) does not contain any GCI with $\top$ as left-hand side, by Theorem~\ref{thm:can-model-main-gci}, $\Tmc^\semiringshort\not\models (C_{a_0}\sqsubseteq B,\elem_0)$ implies that $(d_{C_{a_0}},\elem_0)\notin B^{\Imc_{C_{a_0},\Tmc^\semiringshort}}$. 
By Lemma~\ref{lem:basic-reduction-lemma}, and since $(d_{C_{a_0}},\one)\in C_{a_0}^{\Imc_{C_{a_0},\Tmc^\semiringshort}}$ (by definition of $\Imc_{C_{a_0},\Tmc^\semiringshort}$), $(d_{C_{a_0}},\elem_0)\notin B^{\Imc_{C_{a_0},\Tmc^\semiringshort}}$ implies that  $(a_0,\elem_0)\notin B^{\Imc_{C_{a_0},\Tmc^\semiringshort}}$. Since $\Omc^\semiringshort\subseteq\Tmc^\semiringshort$, $\Imc_{C_{a_0},\Tmc^\semiringshort}$ is a model of $\Omc^\semiringshort$.  Hence, $\Omc^\semiringshort\not\models(B(a_0),\elem_0)$.}

\noindent\new{($\Leftarrow$) 
For the converse, assume that $\Omc^\semiringshort\not\models(B(a_0),\elem_0)$. 
We show that $\Tmc^\semiringshort\not\models (C_{a_0}\sqsubseteq B,\elem_0)$. 
Let $\Imc$ be a model of $\Omc^\semiringshort$ such that $\Imc\not\models (B(a_0),\elem_0)$, i.e. $(a_0^\Imc, \elem_0)\notin B^\Imc$. 
Let $\Jmc$ be the interpretation that extends $\Imc$ with $C_{a}^\Jmc=\{(a^\Imc, \one)\}$ for every $a\in\individuals{\Omc}$, and $R_{ab}^\Jmc=\{(a^\Imc, b^\Imc, \one)\}$ for all $a,b\in\individuals{\Omc}$ and every $R\in\NR$. 
Since $(a_0^\Imc, \one)\in C_{a_0}^\Jmc$ and $(a_0^\Imc, \elem_0)\notin B^\Jmc$, then $\Jmc\not\models (C_{a_0}\sqsubseteq B,\elem_0)$. 
We show that $\Jmc$ is a model of $\Tmc^\semiringshort$, so that $\Tmc^\semiringshort\not\models (C_{a_0}\sqsubseteq B,\elem_0)$. 
It is clear that $\Jmc$ is a model of $\Omc^\semiringshort$ since interpretations of individuals, concepts and roles that occur in $\Omc^\semiringshort$ are not modified.  
We now consider the different kinds of RIs and GCIs in $\Tmc^\semiringshort\setminus\Omc^\semiringshort$. }
\new{
\begin{itemize}
\item Let $(R_{ab}\sqsubseteq R, \elem)\in\Tmc^\semiringshort$. By construction of $\Jmc$, $R_{ab}^\Jmc=\{(a^\Imc, b^\Imc, \one)\}$  and since $(R(a,b),\elem)\in\Omc^\semiringshort$ and $\Imc$ is a model of $\Omc^\semiringshort$, then $(a^\Imc, b^\Imc, \elem)\in R^\Jmc$. Thus we have $\Jmc\models (R_{ab}\sqsubseteq R, \elem)$. 
\item Let $(C_a\sqsubseteq A,\elem)\in\Tmc^\semiringshort\setminus\Omc^\semiringshort$ with $A\in\NC$. 
Since $(A(a),\elem)\in\Omc^\semiringshort$ and $\Imc\models \Omc^\semiringshort$, then $(a^\Imc, \elem)\in A^\Jmc$. 
Thus, since $C_{a}^\Jmc=\{(a^\Imc, \one)\}$, it follows that $\Jmc\models (C_a\sqsubseteq A,\elem)$.
\item Let $(C_a\equiv \exists R_{ab},\one)\in\Tmc^\semiringshort\setminus\Omc^\semiringshort$. 
By construction $C_a^\Jmc=\{(a^\Imc,\one)\}$ and $R_{ab}^\Jmc=\{(a^\Imc, b^\Imc, \one)\}$ so $(\exists R_{ab})^\Jmc=\{(a^\Imc, \one)\}$. Hence $\Jmc\models (C_a\equiv \exists R_{ab},\one)$. 
\item Let $(C_b\equiv \exists R_{ab}^-,\one)\in\Tmc^\semiringshort\setminus\Omc^\semiringshort$. 
By construction $C_b^\Jmc=\{(b^\Imc,\one)\}$ and $R_{ab}^\Jmc=\{(a^\Imc, b^\Imc, \one)\}$ so $(\exists R_{ab}^-)^\Jmc=\{(b^\Imc, \one)\}$. Hence $\Jmc\models (C_b\sqsubseteq \exists R_{ab}^-,\one)$. 
\end{itemize}}
\new{We conclude that $\Jmc\models\Tmc^\semiringshort$, so $\Tmc^\semiringshort\not\models (C_{a_0}\sqsubseteq B, \elem_0)$.
} 
\end{proof}
\thReductionRole*
\begin{proof}
\new{We start with the first point and show that $\Omc^\semiringshort{\models} (P_1\sqsubseteq P_2,\elem_0)$ iff $\Omc^\semiringshort\cup\{(P_1(a_0,b_0), \one)\}{\models}(P_2(a_0,b_0),\elem_0)$. 
First note that if the ontology ${\Omc^\semiringshort}$ is unsatisfiable, so is $\Omc^\semiringshort\cup\{(P_1(a_0,b_0), \one)\}$, and both annotated ontologies entail every annotated axiom. Moreover, if $P_1$ is unsatisfiable \wrt ${\Omc^\semiringshort}$, $\Omc^\semiringshort\models (P_1\sqsubseteq P_2,\elem)$ for every $\elem\in\semiringset$ and $\Omc^\semiringshort\cup\{(P_1(a_0,b_0), \one)\}$ is unsatisfiable so $\Omc^\semiringshort\cup\{(P_1(a_0,b_0), \one)\}\models (P_2(a_0,b_0), \elem)$ for every $\elem\in\semiringset$. 
We next focus on the case where ${\Omc^\semiringshort}$ is satisfiable and $P_1$ is satisfiable \wrt ${\Omc^\semiringshort}$.}

\noindent\new{($\Rightarrow$) Let $\Imc$ be a model of $\Omc^\semiringshort\cup\{(P_1(a_0,b_0), \one)\}$. By construction, $(a_0^\Imc,b_0^\Imc,\one)\in P_1^\Imc$. Since ${\Omc^\semiringshort}\models (P_1\sqsubseteq P_2,\elem_0)$ and $\Imc\models\Omc^\semiringshort$, it follows that $(a_0^\Imc,b_0^\Imc,\one\otimes\elem_0)\in P_2^\Imc$, \ie $(a_0^\Imc,b_0^\Imc,\elem_0)\in P_2^\Imc$.}

\noindent\new{($\Leftarrow$) 
We show the other direction by contrapositive: We assume that $\Omc^\semiringshort\not\models (P_1\sqsubseteq P_2,\elem_0)$ and show that $\Omc^\semiringshort\cup\{(P_1(a_0,b_0), \one)\}\not\models (P_2(a_0,b_0),\elem_0)$. 
Let $\Imc_{P_1,\Omc^\semiringshort}$ be the canonical model of $P_1$ and $\Omc^\semiringshort$. Recall that $\Imc_{P_1,\Omc^\semiringshort}$ is a model of $\Omc^\semiringshort$ such that its domain elements $d_1$ and $d_2$ are such that $(d_1,d_2,\one)\in P_1^{\Imc_{P_1,\Omc^\semiringshort}}$ (\cf Section~\ref{sec:canonical-model}). 
By Theorem~\ref{thm:can-model-main-ri}, $\Omc^\semiringshort\not\models (P_1\sqsubseteq P_2,\elem_0)$ implies that $(d_1,d_2,\elem_0)\notin P_2^{\Imc_{P_1,\Omc^\semiringshort}}$. We obtain a model $\Imc$ of $\Omc^\semiringshort\cup\{(P_1(a_0,b_0), \one)\}$ such that $\Imc\not\models (P_2(a_0,b_0),\elem_0)$ as follows: $a_0^\Imc=d_1$, $b_0^\Imc=d_2$ (recall that $a_0$ and $b_0$ do not occur in $\Omc$), and for all other individual names, concept names and role names, $\Imc$ coincides with $\Imc_{P_1,\Omc^\semiringshort}$.
}
\smallskip

\new{We now show the second point: ${\Omc^\semiringshort}\models(R(a_0,b_0),\elem_0)$ iff $\Tmc^\semiringshort_{S_{a_0,b_0}}\models (S\sqsubseteq R,\elem_0)$. 
Again, note that if ${\Omc^\semiringshort}$ is unsatisfiable, so is $\Tmc^\semiringshort_{S_{a_0,b_0}}$, and both annotated ontologies entail every annotated axiom. Hence, we next focus on the case where ${\Omc^\semiringshort}$ is satisfiable.}

\noindent\new{($\Rightarrow$)  Assume that $\Tmc^\semiringshort_{S_{a_0,b_0}}\not\models(S\sqsubseteq R,\elem_0)$. Let $\Jmc:=\Imc_{S,\Tmc^\semiringshort_{S_{a_0,b_0}}}$ be the canonical model of $S$ and $\Tmc^\semiringshort_{S_{a_0,b_0}}$ (note that $S$ is satisfiable \wrt $\Tmc^\semiringshort_{S_{a_0,b_0}}$, otherwise we would have $\Tmc^\semiringshort_{S_{a_0,b_0}}\models (S\sqsubseteq R,\elem)$ for any $\elem\in\semiringset$). By Theorem~\ref{thm:can-model-main-ri}, $\Tmc^\semiringshort_{S_{a_0,b_0}}\not\models (S\sqsubseteq R,\elem_0)$ implies that $(d_1,d_2,\elem_0)\notin R^\Jmc$. Since $\Omc^\semiringshort\subseteq\Tmc^\semiringshort_{S_{a_0,b_0}}$, $\Jmc$ is a model of $\Omc^\semiringshort$. To obtain that $\Omc^\semiringshort\not\models (R(a_0,b_0),\elem_0)$, we show that for every role $P$ and $\elem\in\semiringset$, $(a_0^\Jmc, b_0^\Jmc,\elem)\in P^\Jmc$ implies that $(d_1,d_2,\elem)\in P^\Jmc$. 
Indeed, by construction of $\Jmc$, $(a_0^\Jmc, b_0^\Jmc,\elem)\in P^\Jmc$ means that there exist a role assertion $(R'(a_0,b_0),\elem_1)$ or $(R'(b_0,a_0),\elem_1)$ in $\Omc^\semiringshort$ (since $\Tmc^\semiringshort_{S_{a_0,b_0}}\setminus\Omc^\semiringshort$ does not contain any assertion) and a sequence of role inclusions $(P_1\sqsubseteq P_2,\elem_2),\dots, (P_{n-1}\sqsubseteq P_n,\elem_n)$ in $\Tmc^\semiringshort_{S_{a_0,b_0}}$ such that applying the chase rules corresponding to $(P_1\sqsubseteq P_2,\elem_2),\dots, (P_{n-1}\sqsubseteq P_n,\elem_n)$ to $(a_0,b_0,\elem_1)\in R'^{\Jmc_0}$ or $(b_0,a_0,\elem_1)\in R'^{\Jmc_0}$ respectively leads to $(a_0^\Jmc, b_0^\Jmc,\elem)\in P^\Jmc$. 
Since $(d_1,d_2,\one)\in S^{\Jmc_0}$, and $(R'(a_0,b_0),\elem_1)\in\Omc^\semiringshort$ implies that $(S\sqsubseteq R',\elem_1)\in\Tmc^\semiringshort_{S_{a_0,b_0}}$ (resp.\ $(R'(b_0,a_0),\elem_1)\in\Omc^\semiringshort$ implies that $(S\sqsubseteq R'^-,\elem_1)\in\Tmc^\semiringshort_{S_{a_0,b_0}}$), it holds that $(d_1,d_2,\elem_1)\in R'^\Jmc$ (resp.\ $(d_2,d_1,\elem_1)\in R'^\Jmc$) and applying the chase rules corresponding to $(P_1\sqsubseteq P_2,\elem_2),\dots, (P_{n-1}\sqsubseteq P_n,\elem_n)$ yields $(d_1,d_2,\elem)\in P^\Jmc$. 
We conclude that $(a_0^\Jmc, b_0^\Jmc,\elem_0)\notin R^\Jmc$, so that $\Omc^\semiringshort\not\models(R(a_0,b_0),\elem_0)$.}

\noindent\new{($\Leftarrow$) Assume that $\Omc^\semiringshort\not\models(R(a_0,b_0),\elem_0)$ and let $\Imc$ be a model of $\Omc^\semiringshort$ s.t.\ $\Imc\not\models(R(a_0,b_0),\elem_0)$. Let $\Jmc$ the interpretation that extends $\Imc$ with $S^\Jmc=\{(a_0^\Imc,b_0^\Imc,\one)\}$. Since $\Imc$ is a model of $\Omc^\semiringshort$, it is easy to see that $\Jmc$ is a model of $\Tmc^\semiringshort_{S_{a_0,b_0}}$ (in particular, for every $(S\sqsubseteq R',\elem)\in\Tmc^\semiringshort_{S_{a_0,b_0}}$, $(R'(a_0,b_0),\elem)\in\Omc^\semiringshort$ so $(a_0^\Imc,b_0^\Imc,\elem)\in R'^\Jmc$, and similarly for $(S\sqsubseteq R'^-,\elem)\in\Tmc^\semiringshort_{S_{a_0,b_0}}$). 
Moreover, $\Jmc\not\models (S\sqsubseteq R,\elem_0)$ so $\Tmc^\semiringshort_{S_{a_0,b_0}}\not\models (S\sqsubseteq R,\elem_0)$. 
}
\end{proof}

\section{Proofs for Section \ref{sec:expected-properties}}
\new{The following lemma will be useful to prove Theorem~\ref{th:sem-entailment}.}
\begin{lemma}\label{cl:auxcancisubset}
\new{Let $\semiringshort$ be a commutative semiring and $\Omc^\semiringshort$ be a $\semiringshort$-annotated $\ELHIbot$ ontology. If either 
\begin{enumerate}[(i)]
\item $\alpha$ is an assertion, a BCQ, or an RI, or
\item $\alpha$ is a GCI between basic concepts, $\semiringshort$ is \timesidem and $\Omc$ does not contain any GCI with $\top$ as left-hand side, 
\end{enumerate}then
	if $S^\semiringshort=\{(\alpha_1, \elem_1), \ldots, (\alpha_n, \elem_n)\}$ is a minimal subset of $\Omc^\semiringshort$ such that
	$S\models \alpha$, it holds that  
	$\Omc^\semiringshort\models (\alpha,\bigotimes^n_{i=1}\elem_i^{\new{p_i}})$ where $p_i\geq 1$ for $1\leq i\leq n$}.
\end{lemma}
\begin{proof} 
\new{If $\Omc^\semiringshort$ is unsatisfiable, $\Omc^\semiringshort\models (\alpha,\elemc)$ for every $\elemc\in K$ so the result holds trivially. 
In what follows, we assume that $\Omc^\semiringshort$ is satisfiable.}
	
\new{We start with case (ii), assuming that $\alpha$ is of the form $C\sqsubseteq D$, with $C,D$ basic concepts, $\semiringshort$ is \timesidem and $\Omc$ does not contain any GCI with $\top$ as left-hand side.}  
Assume that $C$ is satisfiable \wrt $\Omc$ (otherwise  
$\Omc^\semiringshort\models (C\sqsubseteq D,\elemc)$ for every $\elemc\in K$). 
	Let $\Imc_{C,S}$ be the canonical model of $C$ and $S$  (we construct it in the same way as $\Imc_{C,S^\semiringshort}$ except that we do not have the annotations). 
	Since $S\models \alpha$ and $S$ is a minimal set of axioms that entails $\alpha$, 
	by construction of $\Imc_{C,S}$, 
	there is a sequence $(\alpha_1,\Imc^1_{C,S}),\ldots,(\alpha_m,\Imc^m_{C,S})$
	of axioms and interpretations such that $d_C\in D^{\Imc^m_{C,S}}$, \new{and $\Imc^j_{C,S}$ is obtained from $\Imc^{j-1}_{C,S}$ by applying the chase rule using $\alpha_j$.} 
	By minimality of $S$ all axioms in $S$  occur in this sequence.  
	By construction of $\Imc_{C,S^{\semiringshort}}$, we have 
	an analogous sequence	
	$((\alpha_1,\elem_1),\Imc^1_{C,S^{\semiringshort}}),\ldots,((\alpha_m,\elem_m),\Imc^m_{C,S^{\semiringshort}})$,
	with the same  axioms in $S$  except that now they are annotated, 
	and we have $(d_C,\elem)\in D^{\Imc^m_{C,S^{\semiringshort}}}$ with
	$\elem=\bigotimes^m_{i=1}\elem_i$. 
	Since $\semiringshort$ is $\otimes$-idempotent, 
	we have that
	$\elem=\bigotimes^n_{i=1}\elem_i$.
	By Theorem~\ref{thm:can-model-main-gci}, \new{since $\semiringshort$ is $\otimes$-idempotent, and $\Omc$ (hence $S$) does not contain any GCI with $\top$ as left-hand side} $S^{\semiringshort}\models (C\sqsubseteq D,\elem)$.
	Since $S^{\semiringshort}$ is a subset of $\Omc^{\semiringshort}$,
	by the semantics of $\semiringshort$-annotated $\ELHIbot$,
	we have that $\Omc^{\semiringshort}\models S^{\semiringshort}\models (C\sqsubseteq D,\elem)$, which means that $\Omc^{\semiringshort}\models   (C\sqsubseteq D,\elem)$. 
	
\new{We now consider case (i): $\alpha$ is an assertion, a BCQ, or an RI and $\semiringshort$ may not be \mbox{\timesidem}.} 
\new{If $\alpha$ is a BCQ (or an assertion),} let $\Imc_{S^\semiringshort}$ be the canonical model of $S^\semiringshort$. As in the GCI case, we can obtain a  sequence	
$((\alpha_1,\elem_1),\Imc^1_{S^{\semiringshort}}),\ldots,((\alpha_m,\elem_m),\Imc^m_{S^{\semiringshort}})$ where \new{$\Imc^j_{S^{\semiringshort}}$ is obtained from $\Imc^{j-1}_{S^{\semiringshort}}$ by applying the chase rule using $(\alpha_j,\elem_j)$} and $\Imc^m_{S^{\semiringshort}}\models (\alpha,\elem)$ with $\elem=\bigotimes^m_{i=1}\elem_i=\bigotimes^n_{i=1}\elem_i^{\new{p_i}}$ \new{where each $p_i\geq 1$ is the number of times $(\alpha_i,\elem_i)$ occurs in the sequence $((\alpha_1,\elem_1),\Imc^1_{S^{\semiringshort}}),\ldots,((\alpha_m,\elem_m),\Imc^m_{S^{\semiringshort}})$}. 
Then by Theorem~\ref{thm:can-model-main}, we obtain $S^{\semiringshort}\models (\alpha,\elem)$ and $\Omc^{\semiringshort}\models (\alpha,\elem)$. 
\new{Finally, if $\alpha$ is an RI of the form $P\sqsubseteq Q$, we proceed in the same way, using the canonical model $\Imc_{P,S}$ of $P$ and $S$ and Theorem~\ref{thm:can-model-main-ri}.}	
\end{proof}

\thsementailment*

\begin{proof}
Assume that $\Omc\not\models\alpha$, \ie there is a model $\Imc$ of \Omc such that $\Imc\not\models\alpha$. By Claim (\ref{claim-add-annot}) of Lemma \ref{lem:relationship-annotated-standard-models}, there exists a model  $\Imc^\semiringshort$ of $\Omc^\semiringshort$ that coincides with $\Imc$ on its non-annotated part. There is no $\elem\in\semiringset$ such that $\Imc^\semiringshort\models (\alpha,\elem)$ (otherwise we would contradict the fact that $\Imc\not\models\alpha$). It follows that there is no $\elem\in\semiringset$ such that $\Omc^\semiringshort\models (\alpha,\elem)$, thus $\Pmc(\alpha,\Omc^\semiringshort) = \zero$. 

Assume that $\semiringshort$ is positive,  $\Pmc(\alpha,\Omc^\semiringshort) = \zero$ \new{and either (i) $\alpha$ is an assertion, a BCQ, or an RI or (ii) $\semiringshort$ is \timesidem and $\Omc$ does not contain any GCI with $\top$ as left-hand side}. 
 Since $\semiringshort$ is positive \new{and complete (or $\omega$-complete with $\semiringset$ countable)}, $\Pmc(\alpha,\Omc^\semiringshort) = \zero$ means that ($\dagger$) there is no $\elem\in\semiringset$ such that $\Omc^\semiringshort\models (\alpha,\elem)$ and $\elem\neq \zero$. \new{Indeed, by ($\omega$-)completeness, for every $\elem_0$ such that $\Omc^\semiringshort\models (\alpha,\elem_0)$, it holds that $\Pmc(\alpha,\Omc^\semiringshort) = \bigoplus_{\Omc^\semiringshort\models (\alpha,\elem)}\elem=\elem_0\oplus \bigoplus_{\Omc^\semiringshort\models (\alpha,\elem), \elem\neq\elem_0}\elem$, so by positivity, $\elem_0=\zero$.}  
Assume for a contradiction that $\Omc\models\alpha$.
Then, there is a minimal subset $S^\semiringshort=\{(\alpha_1, \elem_1), \ldots, (\alpha_n, \elem_n)\}$  of $\Omc^\semiringshort$ such that
$S\models \alpha$.
By Lemma~\ref{cl:auxcancisubset},  
$\Omc^\semiringshort\models (\alpha,\bigotimes^n_{i=1}\elem_i^{\new{p_i}})$ \new{with all $p_i$'s greater or equal to $1$}. 
Since $\semiringshort$ is positive and every $\elem_i$ is different from $\zero$ by definition of annotated ontologies, then $\bigotimes^n_{i=1}\elem_i^{\new{p_i}}\neq\zero$, 
which contradicts ($\dagger$). 
\end{proof}

\thsemalgebraconsquery* 
\begin{proof}
Let $\Imc_{\Omc^\semiringshort}$ be the canonical model of $\Omc^\semiringshort$ (note that $\Omc^\semiringshort$ is satisfiable since it contains only assertions). 
By Theorem \ref{thm:can-model-main}, for every $\elem\in K$, we have
$\Omc^\semiringshort\models (q,\elem)$ iff $\Imc_{\Omc^\semiringshort}\models (q,\elem)$. 
Hence $\Pmc(q, \Omc^\semiringshort)= \bigoplus_{ \Omc^{\semiringshort} \models(q,\elem)} \elem= \bigoplus_{\Imc_{\Omc^\semiringshort} \models(q,\elem)}\elem$. 
By definition, $\provdb(q,\Omc^\semiringshort)=\bigoplus_{\pi\in\Pi(q,\Omc) }\bigotimes_{P(\vec{t})\in q}\lambda(\pi(P(\vec{t})))$, where $\Pi(q,\Omc)$ is the set of all matches of $q$ in $\Omc$. 
We prove that
\[
\{\elem\mid  \Imc_{\Omc^\semiringshort} \models(q,\elem)\}= 
	\{\bigotimes_{P(\vec{t})\in q} \lambda(\pi(P(\vec t))) \mid \pi\in\Pi(q,\Omc)\}. 
\]
The result follows immediately from this and the idempotency of $\oplus$. 

Since $\Omc^\semiringshort$ contains only assertions, by construction of $\Imc_{\Omc^\semiringshort}$, it holds that $A^{\Imc_{\Omc^\semiringshort}}=\{(a,\lambda(A(a)))\mid A(a)\in\Omc\}$ for every $A\in\NC$ and $R^{\Imc_{\Omc^\semiringshort}}=\{(a,b,\lambda(R(a,b)))\mid R(a,b)\in\Omc\}$ for every $R\in\NR$. Let $\ext{q}$ be the extended version of $q$.

\noindent\textbf{[``$\subseteq$'']} Let $\elem$ be s.t.\ $\Imc_{\Omc^\semiringshort} \models(q,\elem)$. There is a match 
$\pi$ of $\ext{q}$ in $\Imc_{\Omc^\semiringshort}$ s.t.\ $\elem=\bigotimes_{P(\vec{t},t)\in \ext{q}} \pi(t)$. It is easy to see that $\pi\in \Pi(q,\Omc)$, and that for every $P(\vec{t})\in q$ with $P(\vec{t},t)\in\ext{q}$, $\lambda(\pi(P(\vec{t})))=\pi(t)$. 
Hence  $\elem\in\{\bigotimes_{P(\vec{t})\in q} \lambda(\pi(P(\vec t))) \mid \pi\in\Pi(q,\Omc)\}$. 

\noindent\textbf{[``$\supseteq$'']} Given a match $\pi\in\Pi(q,\Omc)$, the function $\pi'$ which for every $P(\vec{t},t)\in\ext{q}$ 
maps $\vec{t}$ to $\pi(\vec{t})$ and $t$ to $\lambda(\pi(P(\vec{t})))$ is a match of $\ext{q}$ in $\Imc_{\Omc^\semiringshort}$. Thus, for every match $\pi\in\Pi(q,\Omc)$, it holds that $\bigotimes_{P(\vec{t})\in q} \lambda(\pi(P(\vec t)))\in \{\elem\mid  \Imc_{\Omc^\semiringshort} \models(q,\elem)\}$. 
\end{proof}

\thsemalgebracons*
\begin{proof}
Let $\Dmc^\semiringshort=\tup{\Dmc,\lambda}$ be the annotated database corresponding to the assertions in $\Omc^\semiringshort$ and let $\{C_i\sqsubseteq A\mid 1\le i \le n\}$ be the set of all GCIs in $\Omc$ and $\Pi(q_{C_i},\Dmc)$ be the set of all matches of the query $q_{C_i}(a)$ in $\Dmc$; $1\le i\le n$. 
Let $\Imc_{\Omc^\semiringshort}$ be the canonical model of $\Omc^\semiringshort$  (note that $\Omc^\semiringshort$ is satisfiable since it does not feature $\bot$). 
Since $\Omc^\semiringshort$ contains only assertions and GCIs of the form $C_i\sqsubseteq A$ labelled with $\one$ where $C_i$ does not contain $A$, by construction of $\Imc_{\Omc^\semiringshort}$, it holds that
\begin{itemize}
\item $B^{\Imc_{\Omc^\semiringshort}}=\{(b,\lambda(B(b)))\mid B(b)\in\Dmc\}$, for every $B\in\NC\setminus\{A\}$;
\item  $R^{\Imc_{\Omc^\semiringshort}}=\{(b,c,\lambda(R(b,c)))\mid R(b,c)\in\Dmc\}$, for every $R\in\NR$; and
\item $A^{\Imc_{\Omc^\semiringshort}}=\{(b,\lambda(
A(b)))\mid A(b)\in\Dmc\}\cup\bigcup_{i=1}^n C_i^{\Imc_{\Omc^\semiringshort}}$. 
\end{itemize}
By Lemma~\ref{lem:complex-concept-query-inter},  
$\{\elem\mid (d, \elem)\in {C_i}^\Imc\}=\{\bigotimes_{P(\vec{t},t)\in \ext{q_{C_i}}(x)} \pi(t) \mid \pi\in\nu_\Imc(\ext{q_{C_i}}(x)), \pi(x)=d\}$, where $\ext{q_{C_i}}(x)$ is the extended version of ${q_{C_i}}(x)$ ($1\leq i\leq n$). 
Hence, 
 $$A^{\Imc_{\Omc^\semiringshort}}=\{(b,\lambda(
A(b)))\mid A(b)\in\Dmc\}\cup\bigcup_{i=1}^n \{(b,\bigotimes_{P(\vec{t},t)\in \ext{q_{C_i}}(b)} \pi(t)) \mid \pi\in\nu_{\Imc_{\Omc^\semiringshort}}(\ext{q_{C_i}}(b)) \}.$$

It follows that  
\begin{align*}
\{\elem\mid  \Imc_{\Omc^\semiringshort} \models(A(a),\elem)\}= &
	\{\lambda(
A(a))\mid A(a)\in\Dmc\}\cup\bigcup_{i=1}^n \left\{\bigotimes_{P(\vec{t},t)\in \ext{q_{C_i}}(a)} \pi(t)) \mid \pi\in\nu_{\Imc_{\Omc^\semiringshort}}(\ext{q_{C_i}}(a))\right\}\\
=&  \left\{\lambda(\pi(A(a))) \mid \pi\in\Pi(A(a),\Dmc)\right\} \cup
	\bigcup_{i=1}^n\left\{\bigotimes_{P(\vec{t})\in q_{C_i}(a)} \lambda(\pi(P(\vec t))) \mid \pi\in\Pi(q_{C_i}(a),\Dmc)\right\}.
\end{align*}
As in the proof of Theorem \ref{th:sem-algebra-cons-query}, $\Pmc(A(a), \Omc^\semiringshort)=  \bigoplus_{\Imc_{\Omc^\semiringshort} \models(A(a),\elem)}\elem$ (by Theorem \ref{thm:can-model-main}) and the result follows immediately from the equality above and idempotency of $\oplus$. 
\end{proof}

\thsemconsdatalogSAM*

\begin{proof}
We denote by $\ext{q}$ the extended version of $q$. We show that $$\bigcap_{\Imc\models \Omc^{\semiringshort}} \p{\Imc}{\ext{q}}=\bigcap_{(I,\mu^I)\models (\Sigma,\Dmc^\semiringshort)} \mu^I(\mn{goal}),$$ so that  
$$\Pmc(q, \Omc^\semiringshort) = \bigoplus_{\elem\in\bigcap_{\Imc\models \Omc^{\semiringshort}}\p{\Imc}{\ext{q}}} \elem= \bigoplus_{\elem\in \bigcap_{(I,\mu^I)\models (\Sigma,\Dmc^\semiringshort)} \mu^I(\mn{goal})} \elem=\provdat^{\texttt{SAM}}(\Sigma,\Dmc^\semiringshort,\mn{goal}).$$

\noindent\textbf{[``$\supseteq$'']}  For every $\Imc$ such that $\Imc\models \Omc^{\semiringshort}$, we define 
$f(\Imc)=(I,\mu^I)$ where 
\begin{align*}
I= {} & \{A(a)\mid \exists \elem,\Imc\models (A(a),\elem)\}\cup\{R(a,b)\mid \exists \elem,\Imc\models (R(a,b),\elem)\}\cup {} \\
	& \{F_\alpha\mid \alpha\text{ GCI or RI of }\Omc\}\cup\{\mn{goal}\mid \Imc\models \ext{q}\},
\end{align*}
and 
$\mu^I(\alpha)=\p{\Imc}{\alpha}$ if $\alpha$ is an assertion, 
$\mu^I(F_\alpha)=\{\lambda(\alpha)\}$ if $\alpha$ is a GCI or an RI, and 
$\mu^I(\mn{goal})=\p{\Imc}{\ext{q}}=\{\bigotimes_{P(\vec{t},t)\in \ext{q}} \pi(t) \mid \pi\in\nu_\Imc(\ext{q})\}$. 
We show that $f(\Imc)\models (\Sigma,\Dmc^\semiringshort)$. 
\begin{enumerate}
\item Since $\Imc$ satisfies every annotated assertion of $\Omc^\semiringshort$, it is easy to check that $\Dmc\subseteq I$, that for every $\alpha\in\Dmc\cap\Omc$, $\lambda'(\alpha)=\lambda(\alpha)$ belongs to $\mu^I(\alpha)=\p{\Imc}{\alpha}$ and that for every $F_\alpha\in\Dmc\setminus\Omc$, $\lambda'(F_\alpha)=\lambda(\alpha)$ belongs to $\mu^I(F_\alpha)=\{\lambda(\alpha)\}$.
\item Let $\phi(\vec{x},\vec{y} ) \rightarrow H(\vec{x})\in\Sigma$ be such that there is a homomorphism $h:\phi(\vec{x},\vec{y})\mapsto I$. 
We consider only the case where $\phi(\vec{x},\vec{y} ) \rightarrow H(\vec{x})$ corresponds to an annotated GCI $(C\sqsubseteq A,\elemb)\in\Omc^\semiringshort$; the case of role inclusions is analogous: $\phi(\vec{x},\vec{y} ) \rightarrow H(\vec{x})$ is of the form $q_C(x)\wedge F_{C\sqsubseteq A} \rightarrow A(x)$. 
Assume that $h(q_C(x))=\beta_1\wedge\dots\wedge\beta_n$. 
By construction of $I$ and $\mu^I$, for every $1\leq i\leq n$, since $\beta_i\in I$, $\mu^I(\beta_i)\neq\emptyset$ and for every $\elem_i\in\mu^I(\beta_i)$, $\Imc\models (\beta_i,\elem_i)$. 
By Lemma \ref {lem:complex-concept-query-inter}, we have 
$$\{\bigotimes_{P(\vec{t},t)\in \ext{q_C}(x)} \pi(t) \mid \pi\in\nu_\Imc(\ext{q_C}(x)), \pi(x)=h(x)^\Imc\}=\{\elem\mid (h(x)^\Imc, \elem)\in C^\Imc\},$$ 
where $\ext{q_C}$ is the extended version of $q_C$. 
For every $(\elem_1,\dots,\elem_n)\in \mu^I(\beta_1)\times\dots\times\mu^I(\beta_n)$, there exists 
$\pi\in\nu_\Imc(\ext{q_C}(x))$ such that $\pi$ maps $\ext{q_C}(x)$ to $\{(\beta_1,\elem_1),\dots,(\beta_n,\elem_n)\}$, 
so 
 $(h(x)^\Imc,\elem_1\otimes\dots\otimes\elem_n)\in C^\Imc$.  
As $\Imc\models (C\sqsubseteq A,\elemb)$, we get $(h(x)^\Imc,\elem_1\otimes\dots\otimes\elem_n\otimes\elem')\in A^\Imc$. It then follows, since $\mu^I(F_{C\sqsubseteq A})=\{\elemb\}$ by definition of $\mu^I$, that 
$h(A(x))\in I$ and 
$\{\bigotimes_{i=1}^{n+1} \elem_i \mid (\elem_1,\dots,\elem_n,\elem_{n+1})\in\mu^I(\beta_1)\times\dots\times\mu^I(\beta_n)\times\mu^I(F_{C\sqsubseteq A})\}\subseteq \mu^I(h(A(x)))$. 
\end{enumerate}
Hence,
$$ 
\bigcap_{\Imc\models \Omc^{\semiringshort}} \p{\Imc}{\ext{q}}\ =\bigcap_{\Imc\models \Omc^{\semiringshort}, f(\Imc)=(I,\mu^I)} \mu^I(\mn{goal})\ \supseteq\  \bigcap_{(I,\mu^I)\models (\Sigma,\Dmc^\semiringshort)} \mu^I(\mn{goal}).$$

\noindent\textbf{[``$\subseteq$'']} For every $(I,\mu^I)$ s.t.\ $(I,\mu^I)\models (\Sigma,\Dmc^\semiringshort)$, we define 
$g(I,\mu^I)=\Imc=(\Delta^\Imc,K, \cdot^\Imc)$ by setting $\Delta^\Imc=\NI$, $a^\Imc=a$ for every $a\in\NI$, 
$A^\Imc=\{(a,\elem)\mid A(a)\in I, \elem\in\mu^I(A(a))\}$ for $A\in\NC$, and 
$R^\Imc=\{(a,b,\elem)\mid R(a,b)\in I, \elem\in\mu^I(R(a,b))\}$ for $R\in\NR$. 
We show that $g(I,\mu^I)\models \Omc^{\semiringshort}$ and that $\p{\Imc}{\ext{q}}\subseteq \mu^I(\mn{goal})$. 

\begin{itemize}
\item Since $\Dmc\subseteq I$ and for every $\alpha\in\Dmc$, $\lambda'(\alpha)\in\mu^I(\alpha)$, since $\lambda'(\alpha)=\lambda(\alpha)$ for every assertion of $\Omc$, it is easy to check that $\Imc$ satisfies every annotated assertion of $\Omc^\semiringshort$. 
\item Let $(C\sqsubseteq A,\elemb)$ be an annotated GCI of $\Omc^\semiringshort$ (the case of RIs is similar). 
Let $(e,\elem)\in C^\Imc$ and $\ext{q_C}$ be the extended version of $q_C$. 
By Lemma \ref{lem:complex-concept-query-inter}, we get that 
$\{\bigotimes_{P(\vec{t},t)\in \ext{q_C}(x)} \pi(t) \mid \pi\in\nu_\Imc(\ext{q_C}(x)), \pi(x)=e\}=\{\elem\mid (e, \elem)\in C^\Imc\}$, 
so there exists some $\pi\in\nu_\Imc(\ext{q_C}(x))$ such that $\pi(x)=e$ and $\bigotimes_{P(\vec{t},t)\in \ext{q_C}(x)} \pi(t) =\elem$. 
Let $(\beta_1,\elem_1),\dots,(\beta_n,\elem_n)$ be the image of $\ext{q_C}(x)$ by $\pi$ and $h$ be the restriction of $\pi$ to the non-annotated part of $\Imc$. 
By construction, $\elem=\elem_1\otimes\dots\otimes\elem_n$ and 
for $1\leq i\leq n$, $\beta_i\in I$, $\elem_i\in\mu^I(\beta_i)$,  and $h$ is a homomorphism from $q_C$ to $I$ such that $h(q_C(x))=\beta_1\wedge\dots\wedge\beta_n$ and $h(x)=e$.

Since $q_C(x)\wedge F_{C\sqsubseteq A}\rightarrow A(x)$ is in $\Sigma$ and $h$ is a homomorphism from $q_C(x)\wedge F_{C\sqsubseteq A}$ to $I$, since $(I,\mu^I)\models (\Sigma,\Dmc^\semiringshort)$, it follows that $A(e)\in I$ and $\elem_1\otimes\dots\otimes\elem_n\otimes\elemb$ is in $\mu^I(A(e))$, \ie $(e,\elem\otimes\elemb)\in A^\Imc$. Hence $\Imc\models (C\sqsubseteq A,\elemb)$. 
\item Let $\elem\in\p{\Imc}{\ext{q}}=\{\bigotimes_{P(\vec{t},t)\in \ext{q}} \pi(t) \mid \pi\in\nu_\Imc(\ext{q})\}$. There is a mapping $\pi$ from $\ext{q}$ to $\Imc$ such that $\elem=\bigotimes_{P(\vec{t},t)\in \ext{q}} \pi(t)$. By construction of $\Imc$, this means that for each $P(\vec{t},t)\in \ext{q}$, $\pi(P(\vec{t}))\in I$ and $\pi(t)\in\mu^I(\pi(P(\vec{t})))$. 
Since $q \rightarrow \mn{goal}$ is in $\Sigma$ and $\pi$ is a homomorphism from $q$ to $I$, since $(I,\mu^I)\models (\Sigma,\Dmc^\semiringshort)$, it follows that $\mn{goal}\in I$ and $\bigotimes_{P(\vec{t},t)\in \ext{q}} \pi(t)\in\mu^I(\mn{goal})$. 
Hence 
$\elem\in\mu^I(\mn{goal})$ and $\p{\Imc}{\ext{q}}\subseteq\mu^I(\mn{goal})$.
\end{itemize}
Hence
$$ \bigcap_{\Imc\models \Omc^{\semiringshort}} \p{\Imc}{\ext{q}}\ \subseteq \bigcap_{(I,\mu^I)\models (\Sigma,\Dmc^\semiringshort), g(I,\mu^I)=\Imc} \p{\Imc}{\ext{q}}\ \subseteq \bigcap_{(I,\mu^I)\models (\Sigma,\Dmc^\semiringshort)} \mu^I(\mn{goal}).$$
\end{proof}

\thsemcommuthom*
\begin{proof}
	Let $\Omc$ be a satisfiable
	$\ELHIbot$ ontology with annotated versions $\Omc^{\semiringshort_1}=\tup{\Omc,\lambda}$ and $\Omc^{\semiringshort_2}=\tup{\Omc,h\circ\lambda}$, 
	where 
	$\semiringshort_1 = (K_1, \oplus ,\otimes, \zero, \one)$ and $\semiringshort_2 = (K_2, + ,\cdot, 0, 1)$
	are commutative additively idempotent \new{$\omega$-complete} semirings.
		\[\Omc^{\semiringshort_1}=\{(\alpha,\lambda(\alpha))\mid \alpha\in\Omc\}\text{ and }
	\Omc^{\semiringshort_2}=\{(\alpha,h(\lambda(\alpha)))\mid \alpha\in\Omc\}.\]
	 Recall that since $\Omc$ is satisfiable, by Lemma \ref{lem:relationship-annotated-standard-models}, $\Omc^{\semiringshort_1}$ and $\Omc^{\semiringshort_2}$ are satisfiable as well. 
	Moreover, since $h$ is 
	a  semiring homomorphism, 
	$h(\zero) = 0$, $h(\one) = 1$, and for all $a,b\in K_1$, $h(a \oplus b) = h(a) + h(b)$ and $h(a \otimes b) = h(a) \cdot h(b)$,
	and \new{since this homomorphism is $\omega$-complete, $h(\bigoplus_{i\in I}a_i)=\bigplus_{i\in I}h(a_i)$ for every countable index set $I$. We prove the theorem using two claims.}
	\begin{claim}\label{cl:hom-aux}
		Let	$\Omc^{\semiringshort_1}=\tup{\Omc,\lambda}$ and $\Omc^{\semiringshort_2}=\tup{\Omc,h\circ\lambda}$
		be as in this proof.
\begin{itemize}
\item For every assertion or BCQ $\alpha$, and every $\elem\in K_1$, if 
		$\Omc^{\semiringshort_1}\models (\alpha,\elem)$ then
		$\Omc^{\semiringshort_2}\models (\alpha,h(\elem))$.
\item \new{For every RI $\alpha$ whose left-hand side is satisfiable \wrt $\Omc$, and every $\elem\in K_1$, if 
		$\Omc^{\semiringshort_1}\models (\alpha,\elem)$ then
		$\Omc^{\semiringshort_2}\models (\alpha,h(\elem))$.}		
\item If $\semiringshort_1$ and $\semiringshort_2$ are also multiplicatively idempotent \new{and $\Omc$ does not contain GCI with $\top$ as left-hand side,} then, for every GCI $\alpha$ between basic concepts \new{whose left-hand side is satisfiable \wrt $\Omc$}, for every $\elem\in K_1$, if 
		$\Omc^{\semiringshort_1}\models (\alpha,\elem)$ then
		$\Omc^{\semiringshort_2}\models (\alpha,h(\elem))$.
\end{itemize}		
	\end{claim}	
	\begin{proof}[Proof of the claim]\renewcommand{\qedsymbol}{}
	First assume  $\alpha$ is an assertion or BCQ. 
	Let $\Imc_{\Omc^{\semiringshort_1}}$
	and $\Imc_{\Omc^{\semiringshort_2}}$ be the canonical models of 
	$\Omc^{\semiringshort_1}$ and $\Omc^{\semiringshort_2}$, respectively (see Section~\ref{sec:canonical-model} for the definition). 	
	By Theorem~\ref{thm:can-model-main},  if 
	$\Omc^{\semiringshort_1}\models (\alpha,\elem)$ then
	$\Imc_{\Omc^{\semiringshort_1}}\models (\alpha,\elem)$.
	By definition of $\Omc^{\semiringshort_2}$ \new{and the fact that $h$ is a semiring homomorphism, so that $h(a)\cdot h(b)=h(a\otimes b)$}, the canonical model $\Imc_{\Omc^{\semiringshort_2}}$
	has the same construction as $\Imc_{\Omc^{\semiringshort_1}}$, except
	for the annotations: \new{one can show by induction that if $(\vec{e},\chi)\in E^{\Imc_{\Omc^{\semiringshort_1}}}$ for some concept or role name $E$, then $(\vec{e},h(\chi))\in E^{\Imc_{\Omc^{\semiringshort_2}}}$. 
	It follows that} $\Imc_{\Omc^{\semiringshort_2}}\models (\alpha,h(\elem))$
	and, by Theorem~\ref{thm:can-model-main},   
	$\Omc^{\semiringshort_2}\models (\alpha,h(\elem))$.
	
	\new{The proofs for the cases where $\alpha$ is an RI of the form $P\sqsubseteq Q$ or a GCI of the form $C\sqsubseteq D$ and fulfills the conditions stated in the claim are similar to the case where $\alpha$ is a BCQ, except that we use the canonical models $\Imc_{P,\Omc^{\semiringshort_i}}$ of $P$ and $\Omc^{\semiringshort_i}$ or $\Imc_{C,\Omc^{\semiringshort_i}}$ of $C$ and $\Omc^{\semiringshort_i}$ respectively (see Section~\ref{sec:canonical-model} for the definition) and Theorems~\ref{thm:can-model-main-ri} and \ref{thm:can-model-main-gci} respectively. Note that the conditions of the claim are the same as the conditions to apply these theorems.}	
	\end{proof}
	\begin{claim}\label{cl:return}
		Let	$\Omc^{\semiringshort_1}=\tup{\Omc,\lambda}$ and $\Omc^{\semiringshort_2}=\tup{\Omc,h\circ\lambda}$
		be as in this proof.		
\begin{itemize}
\item For every assertion or BCQ $\alpha$, for every $\elem'\in K_2$, if $\Omc^{\semiringshort_2}\models (\alpha, \elem')$ then there exists $\elem\in K_1$ such that $h(\elem)=\elem'$ and $\Omc^{\semiringshort_1}\models (\alpha,\elem)$.
\item \new{For every RI $\alpha$ whose left-hand side is satisfiable \wrt $\Omc$, for every $\elem'\in K_2$, if $\Omc^{\semiringshort_2}\models (\alpha, \elem')$ then there exists $\elem\in K_1$ such that $h(\elem)=\elem'$ and $\Omc^{\semiringshort_1}\models (\alpha,\elem)$.}	
\item If $\semiringshort_1$ and $\semiringshort_2$ are also multiplicatively idempotent \new{and $\Omc$ does not contain GCI with $\top$ as left-hand side,} then, for every GCI $\alpha$ between basic concepts \new{whose left-hand side is satisfiable \wrt $\Omc$}, for every $\elem'\in K_2$, if $\Omc^{\semiringshort_2}\models (\alpha, \elem')$ then there exists $\elem\in K_1$ such that $h(\elem)=\elem'$ and $\Omc^{\semiringshort_1}\models (\alpha,\elem)$.
\end{itemize}		
	\end{claim}
	\begin{proof}[Proof of the claim]	\renewcommand{\qedsymbol}{}
	Let $\Imc_{\Omc^{\semiringshort_1}}$
	and $\Imc_{\Omc^{\semiringshort_2}}$ be the canonical models of 
	$\Omc^{\semiringshort_1}$ and $\Omc^{\semiringshort_2}$, respectively.
	Assume  $\alpha$ is an assertion or a BCQ.
	By Theorem~\ref{thm:can-model-main}, if 
	$\Omc^{\semiringshort_2}\models (\alpha,\elem')$ then
	$\Imc_{\Omc^{\semiringshort_2}}\models (\alpha,\elem')$.
	By definition of   $\Omc^{\semiringshort_2}$ \new{and the fact that $h$ is a semiring homomorphism}, the canonical model $\Imc_{\Omc^{\semiringshort_2}}$
	has the same construction as $\Imc_{\Omc^{\semiringshort_1}}$, except
	for the annotations: \new{one can show by induction that if $(\vec{e},\chi')\in E^{\Imc_{\Omc^{\semiringshort_2}}}$ for some concept or role name $E$, then there exists some $(\vec{e},\chi)\in E^{\Imc_{\Omc^{\semiringshort_1}}}$ such that $\chi'=h(\chi)$. 
	It follows that there exists $\elem\in\semiringset_1$ such that $\elem'=h(\elem)$ and} $\Imc_{\Omc^{\semiringshort_1}}\models (\alpha,\elem)$
	and, by Theorem~\ref{thm:can-model-main},   
	$\Omc^{\semiringshort_1}\models (\alpha,\elem)$.
	
		\new{The proofs for the cases where $\alpha$ is an RI of the form $P\sqsubseteq Q$ or a GCI of the form $C\sqsubseteq D$ and fulfills the conditions stated in the claim are similar to the case where $\alpha$ is a BCQ, but use Theorems~\ref{thm:can-model-main-ri} and \ref{thm:can-model-main-gci} respectively (as we did in the proof of Claim~\ref{cl:hom-aux}).} 
\end{proof}

	By Claims~\ref{cl:hom-aux} and~\ref{cl:return}, the definitions of $\Pmc(\alpha, \Omc^{\semiringshort_1})$, $\Pmc(\alpha, \Omc^{\semiringshort_2})$, \new{$\omega$-completeness of $h$} and  additive idempotency,
	\begin{align*}
		\Pmc(\alpha, \Omc^{\semiringshort_2}) = & \bigplus_{\Omc^{\semiringshort_2}\models (\alpha,\elem')} \elem' = \bigplus_{\Omc^{\semiringshort_1}\models (\alpha,\elem)} h(\elem)= h(\bigoplus_{\Omc^{\semiringshort_1}\models (\alpha,\elem)} \elem) = 	h(\Pmc(\alpha, \Omc^{\semiringshort_1})).\qedhere
	\end{align*} 
\end{proof}


\section{Proofs for Section \ref{sec:why}}

\new{Recall that in this section, we consider the case where a \why-annotated ontology $\Omc$ is annotated by a function $\lambda_\semiringVars:\Omc\mapsto\semiringVars\cup\{1\}$.} 
We will often use the canonical model of the ontology. For convenience, we recall in Figure~\ref{fig:rules-canonical-why-normal} its construction in more details for the case where $\Omc$ is a satisfiable \why-annotated ontology in \emph{normal form}.

\begin{figure}[tbh]
\framebox{\begin{minipage}{\linewidth}
$\Imc_\Omc=\bigcup_{i\geq 0}\Imc_i$ where  
$\Delta^{\Imc_0}=\NI$, $a^{\Imc_0}=a$ for every $a\in\NI$, $A^{\Imc_0}=\{(a,\nonomial)\mid (A(a),\nonomial)\in\Omc\}$ for every $A\in\NC$ and $R^{\Imc_0}=\{(a,b,\nonomial)\mid (R(a,b),\nonomial)\in\Omc\}$ for every $R\in\NR$, and $\Imc_{i+1}$ results from applying \new{the chase rule to $\Imc_i$ so that we are in one of the following cases.
\begin{enumerate}
\item[$\R_1$] $(R\sqsubseteq S, \nonomial) \in \Omc$ or $(R^-\sqsubseteq S^-, \nonomial) \in \Omc$, $(d,e, \onomial)\in R^{\Imc_i}$, and $S^{\Imc_{i+1}}=S^{\Imc_i}\cup\{(d,e, {\onomial\times \nonomial})\}$.
\item[$\R_2$] $(R\sqsubseteq S^-, \nonomial) \in \Omc$ or $(R^-\sqsubseteq S, \nonomial) \in \Omc$, $(d,e, \onomial)\in R^{\Imc_i}$, and $S^{\Imc_{i+1}}=S^{\Imc_i}\cup\{(e,d, {\onomial\times \nonomial})\}$.
\item[$\R_3$] $(A\sqsubseteq B, \nonomial) \in\Omc$, $(d, \onomial)\in A^{\Imc_i}$, and $B^{\Imc_{i+1}}=B^{\Imc_i}\cup\{(d,{\onomial\times\nonomial})\}$.
\item[$\R_4$] $(A_1\sqcap A_2\sqsubseteq B, \nonomial) \in\Omc$, $(d, \onomial_1)\in A_1^{\Imc_i}$,  $(d, \onomial_2)\in A_2^{\Imc_i}$, and $B^{\Imc_{i+1}}=B^{\Imc_i}\cup\{(d,{\onomial_1\times \onomial_2\times\nonomial})\}$.
\item[$\R_5$] $(\exists R.A\sqsubseteq B, \nonomial)\in\Omc$, $(d,e, \onomial)\in R^{\Imc_i}$, $(e, \onomial')\in A^{\Imc_i}$, and $B^{\Imc_{i+1}}=B^{\Imc_i}\cup\{(d, {\onomial\times \onomial'\times \nonomial})\}$.
\item[$\R_6$] $(\exists R^-.A\sqsubseteq B, \nonomial)\in\Omc$, $(d,e, \onomial)\in R^{\Imc_i}$, $(d, \onomial')\in A^{\Imc_i}$, and $B^{\Imc_{i+1}}=B^{\Imc_i}\cup\{(e, {\onomial\times \onomial'\times \nonomial})\}$.
\item[$\R_7$] $(A\sqsubseteq \exists R, \nonomial)\in\Omc$, $(d, \onomial)\in A^{\Imc_i}$,  $\Delta^{\Imc_{i+1}}=\Delta^{\Imc_i}\cup\{x\}$ and $R^{\Imc_{i+1}}=R^{\Imc_i}\cup\{(d, x, {\onomial\times \nonomial})\}$ where $x\notin\Delta^{\Imc_i}$ is a fresh domain element.
\item[$\R_8$] $(A\sqsubseteq \exists R^-, \nonomial)\in\Omc$,  $(d, \onomial)\in A^{\Imc_i}$, $\Delta^{\Imc_{i+1}}=\Delta^{\Imc_i}\cup\{x\}$ and $R^{\Imc_{i+1}}=R^{\Imc_i}\cup\{(x,d, {\onomial\times \nonomial})\}$ where $x\notin\Delta^{\Imc_i}$ is a fresh domain element.
\end{enumerate}}
\end{minipage}}
\caption{Construction of the canonical model of a satisfiable \why-annotated $\ELHIbot$ ontology in normal form. 
 $A,A_1,A_2\in\NC\cup\{\top\}$, $B\in\NC$, and $R,S\in \NR$ (recall that RIs and GCIs of $\Omc$ are of these forms, except negative RIs and GCIs with $\bot$ as right-hand side, which are not needed since $\Omc$ is satisfiable). \new{Since we assume that axioms of $\Omc$ are annotated with variables or $1$, all annotations are~monomials.}
}
\label{fig:rules-canonical-why-normal}
\end{figure}

\propFocusMonomialForWhy*
\begin{proof}
Let $\Imc_\Omc$ be the canonical model of $\Omc$. 
Since every axiom of $\Omc$ is annotated with some variable from $\semiringVars$ \new{or $1$}, it follows from the construction of $\Imc_\Omc$ that for every $A\in\NC$, $(a,\monomial)\in A^{\Imc_\Omc}$ implies that $\monomial$ is a product of variables from $\semiringVars$, \ie a monomial over $\semiringVars$, and similarly for roles. 
If $\alpha$ is an assertion $P(\vec{a})$, then by Theorem \ref{thm:can-model-main}, $\Omc\models(\alpha,\monomial)$ implies $\Imc_\Omc\models (P(\vec{a}),\monomial)$, \ie $(\vec{a},\monomial)\in P^{\Imc_\Omc}$, so $\monomial$ is a monomial. 
If $\alpha$ is a BCQ $q$ whose extended version is $\ext{q}$, then by Theorem \ref{thm:can-model-main}, $\Omc\models(\alpha,\monomial)$ implies $\Imc_\Omc\models(q,\monomial)$, \ie $\monomial\in\p{\Imc_\Omc}{\ext{q}}=\{\bigotimes_{P(\vec{t},t)\in \ext{q}} \pi(t) \mid \pi\in\nu_{\Imc_\Omc}(\ext{q})\}$ where $\pi(t)$ is the last element of the tuple $\pi(\vec{t},t)\in P^{\Imc_\Omc}$. Hence, since a product of monomials is a monomial, $\monomial$ is a monomial.
\end{proof}

\subsection{Annotated \new{Assertion} Entailment from $\ELHIbot$ Ontologies}\label{app:annotated}

\paragraph{Proof of Theorem~\ref{prop:completionalgorithmELHI}} 
Before proving the theorem, we first establish the following technical lemmas.
Recall that $\mn{saturate}(\Omc)$ is obtained by starting with 
\begin{align*}
\Smc:= \Omc&\cup\{(\top(a),1)\mid a\in\individuals{\Omc}\new{\cup\{a_\top\}} \}
\\&
\cup\{(A\sqsubseteq A,1) \mid A\in(\NC\cap\signature{\Omc})\cup\{\top,\bot\}\}
\\&
\cup\{(R\sqsubseteq R,1),(R^-\sqsubseteq R^-,1),(\exists R.\bot \sqsubseteq \bot,1), (\exists R^-.\bot \sqsubseteq \bot,1) \mid R\in\NR\cap\signature{\Omc}\}
\\&
\cup\{(\mn{inv}(P_1)\sqsubseteq \mn{inv}(P_2),v)\mid (P_1\sqsubseteq P_2,v)\in\Omc\}
\\&
\cup\{(\mn{inv}(P_1)\sqcap\mn{inv}(P_2)\sqsubseteq \bot,v)\mid (P_1\sqcap P_2\sqsubseteq \bot,v)\in\Omc\},
\end{align*}
and extending $\Smc$ through an iterative application of the rules from Table~\ref{tab:completionRules} until no more rules are applicable. 
For point (1) of Theorem \ref{prop:completionalgorithmELHI}, \ie soundness of the completion algorithm, we will use the following lemma.

\begin{lemma}
\label{lem:saturation-correct-for-axioms}
\new{Let $\mn{saturate}(\Omc)$ be the result of saturating \Omc and $\Imc$ be a model of $\Omc$. For every assertion, GCI, or RI $\alpha$ and every monomial $\monomial$, if $(\alpha, {\monomial})\in\mn{saturate}(\Omc)$ it holds that: 
 \begin{enumerate}
\item if $\alpha$ is an assertion or an RI (positive or negative), then $\Imc\models(\alpha, {\monomial})$,
\item if $\alpha$ is a GCI of the form $C\sqsubseteq D$, then for every \emph{monomial} $\nonomial$ and domain element $e\in\Delta^\Imc$, $(e,\nonomial)\in C^\Imc$ implies $(e, \nonomial\times\monomial)\in D^\Imc$.
\end{enumerate}
}
\end{lemma}
\begin{proof}
\new{Let $\Smc_0=\Smc$ and $\Smc_{i+1}$ be obtained by applying a completion rule to $\Smc_i$, so that $\mn{saturate}(\Omc)=\Smc_n$ for some $n\geq 0$. 
We show by induction on $i$ that for every assertion, GCI, or RI $\alpha$ and every monomial $\monomial$, if $(\alpha,\monomial)\in\Smc_i$ it holds that: 
 \begin{enumerate}
\item if $\alpha$ is an assertion or an RI (positive or negative), then $\Imc\models(\alpha, \monomial)$,
\item if $\alpha$ is a GCI of the form $C\sqsubseteq D$, then for every \emph{monomial} $\nonomial$ and domain element $e\in\Delta^\Imc$, $(e,\nonomial)\in C^\Imc$ implies $(e, \nonomial\times\monomial)\in D^\Imc$.
\end{enumerate}}
 
For the base case $i=0$, we \new{can actually show that $(\alpha,\monomial)\in\Smc_0=\Smc$ implies $\Imc\models (\alpha,\monomial)$, which directly implies points (1-2). Indeed, we} have the following cases. 
\begin{itemize}
\item If $(\alpha,\monomial)\in\Omc$, since $\Imc$ is a model of $\Omc$, $\Imc\models (\alpha,\monomial)$.

\item If $(\alpha, {\monomial})$ is of one of the forms: $(\top(a),1)$ for some $\new{a\in\individuals{\Omc}\cup\{a_\top\}}$, $(A\sqsubseteq A,1)$ for some $ A\in(\NC\cap\signature{\Omc})\cup\{\top,\bot\}$, $(R\sqsubseteq R,1)$, $(R^-\sqsubseteq R^-,1)$, $(\exists R.\bot \sqsubseteq \bot,1)$ or $(\exists R^-.\bot \sqsubseteq \bot,1)$ for some $R\in\NR\cap\signature{\Omc}$, $\Imc\models (\alpha, \monomial)$ trivially (recall that $\bot^\Imc=\emptyset$ and $(\exists R^{(-)}.\bot)^\Imc=\emptyset$).

\item If $(\alpha, {\monomial})$ is of the form  $(\mn{inv}(P_1)\sqsubseteq \mn{inv}(P_2),\monomial)$ for some $(P_1\sqsubseteq P_2,\monomial)\in\Omc$, 
since $\Imc$ is a model of $\Omc$, $\Imc\models (P_1\sqsubseteq P_2,\monomial)$. It follows directly from the definition of the annotated interpretation of $\mn{inv}(P_1)$ and $\mn{inv}(P_2)$ that $\Imc\models (\mn{inv}(P_1)\sqsubseteq \mn{inv}(P_2),\monomial)$.

\item If $(\alpha, {\monomial})$ is of the form  $(\mn{inv}(P_1)\sqcap\mn{inv}(P_2)\sqsubseteq \bot,\monomial)$ for some $(P_1\sqcap P_2\sqsubseteq \bot,\monomial)\in\Omc$, since $\Imc$ is a model of $\Omc$, $(P_1\sqcap P_2)^\Imc=\emptyset$ and so $(\mn{inv}(P_1)\sqcap\mn{inv}(P_2))^\Imc=\emptyset$. Hence $\Imc\models (\mn{inv}(P_1)\sqcap\mn{inv}(P_2)\sqsubseteq \bot,\monomial)$
\end{itemize}

\new{Assume that the property is true for some $i$ and consider $(\alpha, {\monomial})\in\Smc_{i+1}$. If $(\alpha, {\monomial})\in\Smc_{i}$, we obtain the result by induction hypothesis. Otherwise, $(\alpha, {\monomial})$ has been added by the rule applied to obtain $\Smc_{i+1}$ from $\Smc_i$ and we have 11 possible cases depending on which completion rule has been applied.}
\begin{description}
\item[$\CR^T_{0}$] $(\alpha, {\monomial})=(A\sqsubseteq \bot,{\monomial_0\times\monomial_1\times \monomial_2\times\monomial_3})$ and the induction hypothesis applies to the axioms $(A\sqsubseteq \exists P,\monomial_0)$, $(P\sqsubseteq P_1,\monomial_1)$, $(P\sqsubseteq P_2,\monomial_2)$ and $(P_1\sqcap P_2\sqsubseteq \bot,\monomial_3)$. 
\new{Moreover, by the form of the axioms that can be added by the saturation rules, $(A\sqsubseteq \exists P,\monomial_0)$ is actually in $\Omc$ so satisfied by $\Imc$, and either $(P_1\sqcap P_2\sqsubseteq \bot,\monomial_3)$ is in $\Omc$, or $(\mn{inv}(P_1)\sqcap \mn{inv}(P_2)\sqsubseteq \bot,\monomial_3)$ is in $\Omc$, and in both cases, $\Imc\models (P_1\sqcap P_2\sqsubseteq \bot,\monomial_3)$. Thus} 
$\Imc\models (A\sqsubseteq \exists P,\monomial_0)$, $\Imc\models (P\sqsubseteq P_1,\monomial_1)$, $\Imc\models (P\sqsubseteq P_2,\monomial_2)$ and $\Imc\models (P_1\sqcap P_2\sqsubseteq \bot,\monomial_3)$. 
It follows that $A^\Imc=\emptyset$: otherwise, if \new{there was $c\in\Delta^\Imc$ and $\nonomial\in\why$ such that} $(c,\nonomial)\in A^\Imc$, there would be $(c,d,\nonomial\times\monomial_0)\in P^\Imc$, so $(c,d,\nonomial\times\monomial_0\times\monomial_1)\in P_1^\Imc$ and $(c,d,\nonomial\times\monomial_0\times\monomial_2)\in P_2^\Imc$, hence $(c,d,\new{\nonomial\times}\nonomial\times\new{\monomial_0\times}\monomial_0\times\monomial_1\times\monomial_2)\in (P_1\sqcap P_2)^\Imc$, contradicting $(P_1\sqcap P_2)^\Imc=\emptyset$. 
Hence $\Imc\models (A\sqsubseteq \bot,\monomial)$ for every $\monomial$, in particular for $\monomial={\monomial_0\times\monomial_1\times \monomial_2\times\monomial_3}$, \new{and the condition required for the GCI $\alpha$ follows.}

\item[$\CR^T_{1}$] $(\alpha, {\monomial})=(P_1\sqsubseteq P_3,{\monomial_1\times \monomial_2})$ and the induction hypothesis applies to $(P_1\sqsubseteq P_2,\monomial_1)$ and $(P_2\sqsubseteq P_3,\monomial_2)$, so that $\Imc\models(P_1\sqsubseteq P_2,\monomial_1)$ and $\Imc\models(P_2\sqsubseteq P_3,\monomial_2)$. 
\new{For all $c,d\in\Delta^\Imc$ and $\nonomial\in\why$,} if $(c,d,\nonomial)\in P_1^\Imc$, then $(c,d,\nonomial\times \monomial_1)\in P_2^\Imc$ which means that $(c,d,\nonomial\times \monomial_1\times\monomial_2)\in P_3^\Imc$. 
Hence $\Imc\models (P_1\sqsubseteq P_3,{\monomial_1\times \monomial_2})$, \new{as required for the RI~$\alpha$}.

\item[$\CR^T_{2}$] $(\alpha, {\monomial})=(M\sqcap N\sqsubseteq C,{\monomial_1\times \monomial_2})$ and the induction hypothesis applies to $(M\sqsubseteq A,\monomial_1)$ and $(A\sqcap N\sqsubseteq C,\monomial_2)$. 
\new{Let $c\in\Delta^\Imc$ and $\nonomial$ be a monomial.} If $(c,\nonomial)\in (M\sqcap N)^\Imc$, there exists $(c,\nonomial_1)\in M^\Imc$ and $(c,\nonomial_2)\in N^\Imc$ such that $\nonomial_1\times\nonomial_2=\nonomial$. \new{Moreover, $\nonomial_1$ and $\nonomial_2$ must be monomials.} 
By induction hypothesis, it follows that $(c,\nonomial_1\times\monomial_1)\in A^\Imc$, so $(c,\nonomial_1\times\monomial_1\times\nonomial_2)\in (A\sqcap N)^\Imc$. By induction hypothesis, \new{since $\nonomial_1\times\monomial_1\times\nonomial_2$ is a monomial}, we obtain $(c,\nonomial_1\times\monomial_1\times\nonomial_2\times\monomial_2)\in C^\Imc$, \ie $(c,\nonomial\times\monomial_1\times\monomial_2)\in C^\Imc$, \new{as required for the GCI $\alpha$}.

\item[$\CR^T_{3}$] $(\alpha, {\monomial})=(A\sqcap A_1\sqcap\dots\sqcap A_k \sqsubseteq D,\monomial\times\nonomial\times\onomial\times\monomial_0\times\Pi_{i=1}^k(\monomial_i\times\nonomial_i)\times\Pi_{i=1}^{k'}\onomial_i)$ and the induction hypothesis applies to 
\begin{itemize}
\item $(A\sqsubseteq \exists Q,\monomial_0)$, \new{which actually belongs to $\Omc$ by the form of the axioms that can be added by the saturation rules, so is satisfied by $\Imc$}, 
\item $(Q\sqsubseteq P, \monomial)$,  
\item $(Q\sqsubseteq P_i,\monomial_i )$, $1\leq i\leq k$, $k\geq 0$,
\item $(\exists \mn{inv}(P_i).A_i\sqsubseteq B_i,\nonomial_i)$, $1\leq i\leq k$, $k\geq 0$, \new{which actually belong to $\Omc$ by the form of the axioms that can be added by the saturation rules, so are satisfied by $\Imc$},
\item $(\top\sqsubseteq B'_i,\onomial_i)$, $1\leq i\leq k'$, $k'\geq 0$,
\item $(B_1\sqcap\dots\sqcap B_k\sqcap B'_1\sqcap\dots\sqcap B'_{k'}\sqsubseteq C,\nonomial),$ and
\item $(\exists P.C\sqsubseteq D, \onomial)$, \new{which actually belongs to $\Omc$ by the form of the axioms that can be added by the saturation rules, so is satisfied by $\Imc$}. 
\end{itemize} 
\new{Let $c\in\Delta^\Imc$ and $r$ be a monomial. 
Assume that} $(c,r)\in (A\sqcap A_1\sqcap\dots\sqcap A_k)^\Imc$. 
There exists $(c,r_0)\in A^\Imc$, and $(c,r_i)\in A_i^\Imc$ for $1\leq i\leq k$ such that $r=r_0\times\Pi_{i=1}^k r_i$. \new{Note that all $r_j$ in this product must be monomials since $r$ is a monomial.} 

Since $\Imc\models (A\sqsubseteq \exists Q,\monomial_0)$, there exists $(c,d,r_0\times\monomial_0)\in Q^\Imc$. 

For every $1\leq i\leq k$, since $\Imc\models (Q\sqsubseteq P_i,\monomial_i )$, $(c,d,r_0\times\monomial_0\times\monomial_i)\in P_i^\Imc$.  
Hence, since  $(c,r_i)\in A_i^\Imc$, 
it follows that $(d,r_0\times\monomial_0\times\monomial_i\times r_i)\in (\exists\mn{inv}(P_i).A_i)^\Imc$. 
Since \new{$\Imc\models (\exists\mn{inv}(P_i).A_i\sqsubseteq B_i,\nonomial_i)$, } 
then $(d, r_0\times\monomial_0\times\monomial_i\times r_i\times\nonomial_i)\in B_i^\Imc$. 

For every $1\leq i\leq k'$, since \new{$\top^\Imc=\Delta^\Imc\times\{1\}$ and $1$ is a monomial, by induction hypothesis on $(\top\sqsubseteq B'_i,\onomial_i)$,} 
then $(d, \onomial_i)\in {B'_i}^\Imc$. 

\new{Since the product of a monomial by itself is the same monomial in the \why semiring,} it follows that $(d, r_0\times\monomial_0\times\Pi_{i=1}^k(\monomial_i\times r_i\times\nonomial_i)\times\Pi_{i=1}^{k'}\onomial_i)\in (B_1\sqcap\dots\sqcap B_k\sqcap B'_1\sqcap\dots\sqcap B'_{k'})^\Imc$, \ie $(d, r\times\monomial_0\times\Pi_{i=1}^k(\monomial_i\times
\nonomial_i)\times\Pi_{i=1}^{k'}\onomial_i)\in (B_1\sqcap\dots\sqcap B_k\sqcap B'_1\sqcap\dots\sqcap B'_{k'})^\Imc$. 
Hence, \new{by induction hypothesis on $(B_1\sqcap\dots\sqcap B_k\sqcap B'_1\sqcap\dots\sqcap B'_{k'}\sqsubseteq C,\nonomial)$, it follows that} 
$(d, r\times\monomial_0\times\nonomial\times\Pi_{i=1}^k(\monomial_i\times
\nonomial_i)\times\Pi_{i=1}^{k'}\onomial_i)\in C^\Imc$. 

Moreover, since $\Imc\models (Q\sqsubseteq P, \monomial)$, $(c,d,r_0\times\monomial_0\times\monomial)\in P^\Imc$, so that 
$(c, r\times\monomial_0\times\monomial\times\nonomial\times\Pi_{i=1}^k(\monomial_i\times
\nonomial_i)\times\Pi_{i=1}^{k'}\onomial_i)\in (\exists P.C)^\Imc$ \new{(again, this holds because we are multiplying monomials)}. 

Finally, \new{since $\Imc\models(\exists P.C\sqsubseteq D, \onomial)$,} 
$(c, r\times\monomial_0\times\monomial\times\nonomial\times\onomial\times\Pi_{i=1}^k(\monomial_i\times
\nonomial_i)\times\Pi_{i=1}^{k'}\onomial_i)\in D^\Imc$, \new{as required for the GCI $\alpha$}. 

\item[$\CR^A_{1}$] $(\alpha, {\monomial})=(B(a),{\monomial_1\times \dots\times\monomial_k\times \monomial})$ and the induction hypothesis applies to the axioms $(A_i(a),\monomial_i)$, $1\leq i\leq k$, and $(A_1\sqcap \dots\sqcap A_k\sqsubseteq B,\monomial)$. 
\new{Since} $\Imc\models(A_i(a),\monomial_i)$, $1\leq i\leq k$, 
it follows that $(a^\Imc, \monomial_1\times\dots\times\monomial_k)\in (A_1\sqcap \dots\sqcap A_n)^\Imc$.  
\new{By induction hypothesis on $(A_1\sqcap \dots\sqcap A_k\sqsubseteq B,\monomial)$, since $\monomial_1\times\dots\times\monomial_k$ is a monomial, it follows that $(a^\Imc, \monomial_1\times\dots\times\monomial_k\times\monomial)\in B^\Imc$. }
Hence $\Imc\models (B(a),{\monomial_1\times \dots\times\monomial_k\times \monomial})$, \new{as required for the assertion $\alpha$}. 

\item[$\CR^A_{2}$] to $\CR^A_{5}$ are easy to show in a similar way.

\item[$\CR^A_{6}$] and $\CR^A_{7}$ are never applied because $\Omc$ is satisfiable. We show this for $\CR^A_{6}$ case. 
Assume for a contradiction that $(\alpha, {\monomial})=(\bot(a),{\monomial_1\times \monomial_2\times\monomial_3})$ has been added \new{to $\Smc_i$ to obtain $\Smc_{i+1}$} by $\CR^A_{6}$. In this case, 
the induction hypothesis applies to $(R(a,b),\monomial_1)$ and $(S(a,b),\monomial_2)$, \new{while by the form of axioms that can be added by the saturation rules, either $(R\sqcap S\sqsubseteq\bot,\monomial_3)$ is in $\Omc$, or $(\mn{inv}(R)\sqcap \mn{inv}(S)\sqsubseteq\bot,\monomial_3)$ is in $\Omc$}. Thus, $\Imc\models (R(a,b),\monomial_1)$, $\Imc\models(S(a,b),\monomial_2)$ and $\Imc\models(R\sqcap S\sqsubseteq\bot,\monomial_3)$. 
It must be the case that $(R\sqcap S)^\Imc=\emptyset$ while $(a^\Imc,b^\Imc,\monomial_1\times \monomial_2)\in (R\sqcap S)^\Imc$, which is a contradiction. Hence $\CR^A_{6}$ is never applied when $\Omc$ is satisfiable.
\end{description}
This concludes the proof of the lemma.
\end{proof}

The following lemma will be useful to handle RIs in the proofs of this section.

\begin{lemma}\label{lem:inverse-roles-completion}
If $(P_1\sqsubseteq P_2,\monomial)\in\mn{saturate}(\Omc)$, then 
\begin{enumerate}
\item $(\mn{inv}(P_1)\sqsubseteq \mn{inv}(P_2),\monomial)\in\mn{saturate}(\Omc)$;
\item \new{$P_1\sqsubseteq_\Omc P_2$, where $\sqsubseteq_\Omc$ is the transitive closure of the relation  $\{(S,P), (\mn{inv}(S),\mn{inv}(P))\mid S\sqsubseteq P\in\Omc\}\cup\{(R,R), (R^-,R^-)\mid R\in\signature{\Omc}\cap\NR\}$.} 
\end{enumerate}
\end{lemma}
\begin{proof}
We show by induction on the number of completion rules applied before adding $(P_1\sqsubseteq P_2,\monomial)$ to $\mn{saturate}(\Omc)$ that $(\mn{inv}(P_1)\sqsubseteq \mn{inv}(P_2),\monomial)\in\mn{saturate}(\Omc)$ \new{and $P_1\sqsubseteq_\Omc P_2$}.

\begin{itemize}
\item Base case: $(P_1\sqsubseteq P_2,\monomial)\in\Smc$. If $(P_1\sqsubseteq P_2,\monomial)\in\Omc$ then $(\mn{inv}(P_1)\sqsubseteq \mn{inv}(P_2),\monomial)\in\Smc$. Otherwise, $(P_1\sqsubseteq P_2,\monomial)\in\Smc\setminus\Omc$ and there is two possibilities by definition of $\Smc$: either $(P_1\sqsubseteq P_2,\monomial)=(P\sqsubseteq P, 1)$, or $(P_1\sqsubseteq P_2,\monomial)=(\mn{inv}(P_3)\sqsubseteq \mn{inv}(P_4),\monomial)$ for some $(P_3\sqsubseteq P_4,\monomial)\in\Omc$. Since for every $R$, both $(R\sqsubseteq R,1)$ and $(R^-\sqsubseteq R^-,1)$ are in $\Smc$ and $\mn{inv}(\mn{inv}(P))=P$, it follows that in every case, $(\mn{inv}(P_1)\sqsubseteq \mn{inv}(P_2),\monomial)\in\mn{saturate}(\Omc)$. \new{Moreover, in every case, $P_1\sqsubseteq_\Omc P_2$.}

\item Assume that for every $(P\sqsubseteq Q,\nonomial)\in\mn{saturate}(\Omc)$ obtained after at most $k$ applications of the completion rules, $(\mn{inv}(P)\sqsubseteq \mn{inv}(Q),\nonomial)\in\mn{saturate}(\Omc)$ \new{and $P\sqsubseteq_\Omc Q$}, and let $(P_1\sqsubseteq P_2,\monomial)\in\mn{saturate}(\Omc)$ be obtained after applying $k+1$ completion rules. 
Since the only completion rule that adds role inclusions is $\CR^T_{1}$, $(P_1\sqsubseteq P_2,\monomial)$ has been added by applying $\CR^T_{1}$ to some $(P_1\sqsubseteq P',\monomial_1)$ and $(P'\sqsubseteq P_2,\monomial_2)$ in $\mn{saturate}(\Omc)$ such that $\monomial_1\times\monomial_2=\monomial$. 
By induction hypothesis, $(\mn{inv}(P_1)\sqsubseteq \mn{inv}(P'),\monomial_1)$ and $(\mn{inv}(P')\sqsubseteq \mn{inv}(P_2),\monomial_2)$ are in $\mn{saturate}(\Omc)$, \new{and $P_1\sqsubseteq_\Omc P'$ and $P'\sqsubseteq_\Omc P_2$}. 
Hence, by $\CR^T_{1}$, $(\mn{inv}(P_1)\sqsubseteq \mn{inv}(P_2),\monomial)\in\mn{saturate}(\Omc)$, \new{and by definition of $\sqsubseteq_\Omc$, $P_1\sqsubseteq_\Omc P_2$}. \qedhere
\end{itemize}
\end{proof}

For point (2) of Theorem \ref{prop:completionalgorithmELHI}, \ie completeness of the completion algorithm, we will use the canonical model $\Imc_\Omc=\bigcup_{i\geq 0}\Imc_i$ of $\Omc$ (\cf Figure~\ref{fig:rules-canonical-why-normal}). 
The next lemmas will be crucial to handle the anonymous part of the canonical model in several proofs.

\begin{lemma}\label{lem:canonical-saturation-RIs}
For all $i\geq 0$, $x,y\in\Delta^{\Imc_\Omc}$ and $R$ role name or inverse role, if $(x,y,\onomial)\in R^{\Imc_i}$ and $y\notin\NI$ has been introduced between $\Imc_{j-1}$ and $\Imc_j$ ($j\leq i$) to satisfy an inclusion of the form $(C\sqsubseteq \exists S, \monomial_0)\in\Omc$ by applying \new{the chase rule in case} $\R_7$ or $\R_8$ to some $(x,s)\in C^{\Imc_{j-1}}$, 
then there exists $(S\sqsubseteq R, \monomial_{S\sqsubseteq R}  )\in\mn{saturate}(\Omc)$ such that $s\times \monomial_0\times\monomial_{S\sqsubseteq R}=\onomial$.
\end{lemma}
\begin{proof}
The proof is by induction on $l=i-j$. 
\smallskip

\noindent\emph{Base case: $l=0$.} The only possibility to obtain  $(x,y,\onomial)\in R^{\Imc_i}$ by \new{a single application of the chase rule} that introduces $y$ to satisfy an inclusion of the form $(C\sqsubseteq \exists S, \monomial_0)$ with some $(x,s)\in C^{\Imc_{i-1}}$ is that $S=R$ and $s\times \monomial_0=\onomial$. Since $(R\sqsubseteq R,1)\in\mn{saturate}(\Omc)$, we obtain the property. 
\smallskip

\noindent\emph{Inductive step.} 
Assume that the property is true for every integer up to $l$ and consider the case where $i-j=l+1$. 
Let $x,y\in\Delta^{\Imc_\Omc}$, and assume that $(x,y,\onomial)\in R^{\Imc_i}$, and $y\notin\NI$ has been introduced between $\Imc_{j-1}$ and $\Imc_j$, to satisfy $(C\sqsubseteq \exists S, \monomial_0)\in\Omc$ by applying \new{the chase rule in case} $\R_7$ or $\R_8$ to $(x,s)\in C^{\Imc_{j-1}}$ (thus adding $(x,y,s\times\monomial_0)\in S^{\Imc_{j}}$). 
We make a case analysis on whether $(x,y,\onomial)\in R^{\Imc_i}$ has been added by applying \new{the chase rule in case} $\R_1$ or $\R_2$ with a RI of the form: 
\begin{itemize}
\item  $(P_{p-1}\sqsubseteq R, r_p)\in\Omc$: 
It holds that $(x,y,\onomial'')\in P_{p-1}^{\Imc_{i-1}}$ (with $\onomial=\onomial''\times r_p$) so by induction hypothesis (since $i-1-j=l$ and the property is assumed to be true for $l$), there exists $(S\sqsubseteq P_{p-1}, \monomial_{S\sqsubseteq P_{p-1}}  )\in\mn{saturate}(\Omc)$ such that $s\times \monomial_0\times\monomial_{S\sqsubseteq P_{p-1}}=\onomial''$, which implies by $\CR^T_{1}$ that $(S\sqsubseteq R, \monomial_{S\sqsubseteq R}  )\in\mn{saturate}(\Omc)$ with $\monomial_{S\sqsubseteq R} =\monomial_{S\sqsubseteq P_{p-1}}\times r_p$ so that $s\times \monomial_0\times\monomial_{S\sqsubseteq R}=\onomial''\times r_p=\onomial$.

\item $(P_{p-1}\sqsubseteq \mn{inv}(R), r_p)\in\Omc$: It holds that $(x,y,\onomial'')\in \mn{inv}(P_{p-1})^{\Imc_{i-1}}$ (with $\onomial=\onomial''\times r_p$) so by induction hypothesis there exists $(S\sqsubseteq \mn{inv}(P_{p-1}), \monomial_{S\sqsubseteq \mn{inv}(P_{p-1})}  )\in\mn{saturate}(\Omc)$ with $s\times \monomial_0\times\monomial_{S\sqsubseteq \mn{inv}(P_{p-1})}=\onomial''$. 
Since $(P_{p-1}\sqsubseteq \mn{inv}(R), r_p)\in\mn{saturate}(\Omc)$, by Lemma \ref{lem:inverse-roles-completion},  
$(\mn{inv}(P_{p-1})\sqsubseteq R, r_p)\in\mn{saturate}(\Omc)$. 
Hence, by $\CR^T_{1}$, $(S\sqsubseteq R, \monomial_{S\sqsubseteq R}  )\in\mn{saturate}(\Omc)$ with $\monomial_{S\sqsubseteq R} =\monomial_{S\sqsubseteq \mn{inv}(P_{p-1})}\times r_p$ so that $s\times \monomial_0\times\monomial_{S\sqsubseteq R}=\onomial''\times r_p=\onomial$.
\qedhere
\end{itemize}

\end{proof}

\new{For a multiset $\{A_1,\dots, A_k\}$ of concept names, we use $\mn{Conj}(A_1,\dots,A_k)$ and $\mn{Conj}(A_1\sqcap\dots\sqcap A_k)$ to denote the conjunction obtained from $A_1\sqcap\dots\sqcap A_k$ by limiting the number of times each concept name can occur to $\mn{Card}(\Omc)$, \ie replacing $\underbrace{A\sqcap\dots\sqcap A}_{\mn{Card}(\Omc)+n\text{ times}}$ by $\underbrace{A\sqcap\dots\sqcap A}_{\mn{Card}(\Omc)\text{ times}}$.}

\begin{lemma}\label{lem:choice}
\new{Let $k\in\mathbb{N}$, and for every $j$ such that $1\leq j\leq k$, assume that we are given $k_j\geq 0$ and $k_j'\geq 0$, monomials $s^j_i$ for $1\leq i\leq k_j$, and a set of GCIs of the form:
\begin{itemize}		
\item $(S\sqsubseteq P^j_i, \monomial_{S\sqsubseteq P^j_i})$,  $(\exists \mn{inv}(P_i^j).A^j_i\sqsubseteq B^j_i, \monomial_{\exists \mn{inv}(P_i^j).A^j_i\sqsubseteq B^j_i})$, \mbox{$1\leq i\leq k_j$,}  
\item $(\top\sqsubseteq B'^j_i, \monomial_{\top\sqsubseteq B'^j_i})$, $1\leq i\leq k_j'$.
\end{itemize} 
Assume that all monomials are built from variables that annotate $\Omc$.}

\new{If $B\in\NC$ occurs $\mn{Card}(\Omc)+n$ times in the multiset $\bigcup_{j=1}^k\{B^j_i\mid 1\leq i\leq k_j\}\cup\{B'^j_i\mid 1\leq i\leq k_j'\}$, then there are $n$ occurrences of $B$ that we can assume to correspond to $\bigcup_{j=1}^k\{B^j_{h_j+1},\dots, B^j_{k_j}, B'^j_{h'_j+1},\dots, B'^j_{k'_j}\}$ for some $h_j\leq k_j$ and $h'_j\leq k'_j$ such that 
\begin{align*}
\Pi_{j=1}^k\Pi_{i=1}^{k_j}(\monomial_{S\sqsubseteq P^j_i}\times \monomial_{\exists \mn{inv}(P_i^j).A^j_i\sqsubseteq B^j_i} \times s^j_i)\times \Pi_{i=1}^{k'_j}\monomial_{\top\sqsubseteq B'^j_i}\\
=\ \Pi_{j=1}^k\Pi_{i=1}^{h_j}(\monomial_{S\sqsubseteq P^j_i}\times \monomial_{\exists \mn{inv}(P_i^j).A^j_i\sqsubseteq B^j_i} \times s^j_i)\times \Pi_{i=1}^{h'_j}\monomial_{\top\sqsubseteq B'^j_i}.
\end{align*}
}
\end{lemma}
\begin{proof}
\new{Assume w.l.o.g.\ that for every $1\leq j\leq k$, the occurences of the concept name $B$ come last in each multiset $\{B^j_i\mid 1\leq i\leq k_j\}$ and $\{B'^j_i\mid 1\leq i\leq k_j'\}$, that is, there exist $r_j\leq k_j$ and $r'_j\leq k'_j$ such that $\{B^j_i\mid r_j+1\leq i\leq k_j\}$ and $\{B'^j_i\mid r'_j+1\leq i\leq k_j'\}$ contains only $B$ and there is no occurence of $B$ in $\{B^j_i\mid 1\leq i\leq r_j\}$ and $\{B'^j_i\mid 1\leq i\leq r_j'\}$. }

\new{Since there are $\mn{Card}(\Omc)+n$ occurences of the concept name $B$ in total, the multiset of monomials $\bigcup_{j=1}^k\{\monomial_{S\sqsubseteq P^j_i}\times \monomial_{\exists \mn{inv}(P_i^j).A^j_i\sqsubseteq B^j_i} \times s^j_i\mid r_j+1\leq i\leq k_j\}\cup\{\monomial_{\top\sqsubseteq B'^j_i}\mid r'_j+1\leq i\leq k'_j\}$ contains $\mn{Card}(\Omc)+n$ monomials. The total number of variables that may occur in some monomial is $\mn{Card}(\Omc)$, so among $\mn{Card}(\Omc)+n$ monomials over $\mn{Card}(\Omc)$ variables, there must be $n$ monomials such that all their variables are present in the product of the remaining $\mn{Card}(\Omc)$ monomials. We can assume w.l.o.g.\ (by re-ordering the GCIs and monomials) that they come last in each multiset $\{\monomial_{S\sqsubseteq P^j_i}\times \monomial_{\exists \mn{inv}(P_i^j).A^j_i\sqsubseteq B^j_i} \times s^j_i\mid r_j+1\leq i\leq k_j\}$ and $\{\monomial_{\top\sqsubseteq B'^j_i}\mid r'_j+1\leq i\leq k'_j\}$, that is, there exist $r_j+1\leq h_j\leq k_j$ and $r'_j+1\leq h'_j\leq k'_j$ such that 
\begin{align*}
\Pi_{j=1}^k\Pi_{i=r_j+1}^{k_j}(\monomial_{S\sqsubseteq P^j_i}\times \monomial_{\exists \mn{inv}(P_i^j).A^j_i\sqsubseteq B^j_i} \times s^j_i)\times \Pi_{i=r'_j+1}^{k'_j}\monomial_{\top\sqsubseteq B'^j_i}\\
=\ \Pi_{j=1}^k\Pi_{i=r_j+1}^{h_j}(\monomial_{S\sqsubseteq P^j_i}\times \monomial_{\exists \mn{inv}(P_i^j).A^j_i\sqsubseteq B^j_i} \times s^j_i)\times \Pi_{i=r'_j+1}^{h'_j}\monomial_{\top\sqsubseteq B'^j_i}.
\end{align*}
The lemma follows.}
\end{proof}

\begin{lemma}\label{lem:choice2}
\new{Let $k\in\mathbb{N}$, and for every $j$ such that $1\leq j\leq k$, assume that we are given $k_j\geq 0$, monomials $s^j_i$ for $1\leq i\leq k_j$, and a set of GCIs of the form:
\begin{itemize}		
\item $(S\sqsubseteq P^j_i, \monomial_{S\sqsubseteq P^j_i})$,  $(\exists \mn{inv}(P_i^j).A^j_i\sqsubseteq B^j_i, \monomial_{\exists \mn{inv}(P_i^j).A^j_i\sqsubseteq B^j_i})$, \mbox{$1\leq i\leq k_j$}.  
\end{itemize} 
Assume that all monomials are built from variables that annotate $\Omc$.
}

\new{If $A\in\NC$ occurs $\mn{Card}(\Omc)+n$ times in the multiset $\bigcup_{j=1}^k\{A^j_i\mid 1\leq i\leq k_j\}$, then there are $n$ occurrences of $A$ that we can assume to correspond to $\bigcup_{j=1}^k\{A^j_{h_j+1},\dots, A^j_{k_j}\}$ for some $h_j\leq k_j$ such that 
\begin{align*}
\Pi_{j=1}^k\Pi_{i=1}^{k_j}(\monomial_{S\sqsubseteq P^j_i}\times \monomial_{\exists \mn{inv}(P_i^j).A^j_i\sqsubseteq B^j_i} \times s^j_i)
=\ \Pi_{j=1}^k\Pi_{i=1}^{h_j}(\monomial_{S\sqsubseteq P^j_i}\times \monomial_{\exists \mn{inv}(P_i^j).A^j_i\sqsubseteq B^j_i} \times s^j_i).
\end{align*}
}
\end{lemma}
\begin{proof}
\new{The proof is analogous to that of Lemma~\ref{lem:choice}. If $A\in\NC$ occurs $\mn{Card}(\Omc)+n$ times in the multiset $\bigcup_{j=1}^k\{A^j_i\mid 1\leq i\leq k_j\}$, then there are $n$ occurrences of $A$ that we can assume (by re-ordering the GCIs)  to correspond to $\bigcup_{j=1}^k\{A^j_{h_j+1},\dots, A^j_{k_j}\}$ for some $h_j\leq k_j$ and are such that (by the same argument as in the proof of Lemma~\ref{lem:choice}) every variable in the corresponding monomials $\monomial_{S\sqsubseteq P^j_i}\times \monomial_{\exists \mn{inv}(P_i^j).A^j_i\sqsubseteq B^j_i} \times s^j_i$ occurs in the product of the remaining $\mn{Card}(\Omc)$ monomials of this form with $A^j_i=A$.}
\end{proof}

\allowdisplaybreaks

\begin{lemma}\label{lem:canonical-saturation-GCIs}
For all $i\geq 0$, $x,y\in\Delta^{\Imc_\Omc}$ and $A\in\NC\cup\{\top\}$  if $(y,\onomial')\in A^{\Imc_i}$ and $y\notin\NI$ has been introduced between $\Imc_{j-1}$ and $\Imc_j$ ($j\leq i$) to satisfy an inclusion of the form $(C\sqsubseteq \exists S, \monomial_0)\in\Omc$ by applying \new{the chase rule in case} $\R_7$ or $\R_8$ to some $(x,s)\in C^{\Imc_{j-1}}$, 
then the following holds.
\begin{enumerate}
\item There exist $k\geq 0$ and $k'\geq 0$ such that:
\begin{itemize} 								
\item $(S\sqsubseteq P_i, \monomial_{S\sqsubseteq P_i})$ and  $(\exists \mn{inv}(P_i).A_i\sqsubseteq B_i, \monomial_{\exists \mn{inv}(P_i).A_i\sqsubseteq B_i})\in\mn{saturate}(\Omc)$, \mbox{$1\leq i\leq k$,}  

\item $(\top\sqsubseteq B'_i, \monomial_{\top\sqsubseteq B'_i})\in\mn{saturate}(\Omc)$, $1\leq i\leq k'$, 

\item $(B_1\sqcap\dots\sqcap B_k\sqcap B'_1\sqcap\dots\sqcap B'_{k'}\sqsubseteq A, p)\in\mn{saturate}(\Omc)$,
\end{itemize} 
\new{and $\mn{Conj}(C,A_1,\dots,A_k)=C\sqcap A_1\sqcap\dots\sqcap A_k$ (\ie there is no concept name that occurs more than $\mn{Card}(\Omc)$ times in $\{C, A_1,\dots,A_k\}$).}

Moreover, if $A=\top$, then $k=k'=0$ and $p=1$ (recall that the empty conjunction is $\top$ and $(\top\sqsubseteq\top,1)\in\mn{saturate}(\Omc)$).
 
\item For every $1\leq \ell\leq k$, there exists $(x,s_\ell)\in A_\ell^{\Imc_{i-1}}$.  
 
\item The monomials are related as follows:
\begin{itemize}
\item If $k\neq 0$,  $s\times\monomial_0\times\Pi_{\ell=1}^k s_\ell\times \Pi_{i=1}^k (\monomial_{S\sqsubseteq P_i}\times \monomial_{\exists \mn{inv}(P_i).A_i\sqsubseteq B_i})\times\Pi_{i=1}^{k'}\monomial_{\top\sqsubseteq B'_i}\times p =\onomial'$.
\item If $k=0$, 
$\Pi_{i=1}^{k'}\monomial_{\top\sqsubseteq B'_i}\times p =\onomial'$.
\end{itemize}
\end{enumerate} 
\end{lemma}
\begin{proof}
The proof is by induction on $l=i-j$. In the case where $A=\top$, the property trivially holds with $k=k'=0$ for every $l\geq 0$ (since in this case $\onomial'=1=p$). 
\smallskip

\noindent\emph{Base case.} 
In the case where $A\neq \top$, the base case is $l=1$: $y\notin\NI$ has been introduced between $\Imc_{i-2}$ and $\Imc_{i-1}$ to satisfy $(C\sqsubseteq \exists S, \monomial_0)$ by applying \new{the chase rule in case} $\R_7$ or $\R_8$ to $(x,s)\in C^{\Imc_{i-2}}$, adding $(x,y,s\times\monomial_0)\in S^{\Imc_{i-1}}$. By construction of $\Imc_\Omc$, it must be the case that either (i) the \new{chase rule has been applied between $\Imc_{i-1}$ and $\Imc_{i}$ with} $(\exists \mn{inv}(S).\top\sqsubseteq A, \monomial_1)$ (\new{case} $\R_6$ or $\R_5$), so that $(\exists \mn{inv}(S).\top\sqsubseteq A, \monomial_1 )\in\Omc$, and $s\times \monomial_0\times\monomial_1=\onomial'$, 
or (ii) the \new{chase rule has been applied between $\Imc_{i-1}$ and $\Imc_{i}$ with} $(\top\sqsubseteq A, \monomial_1)$ (\new{case} $\R_3$), so that $(\top\sqsubseteq A, \monomial_1 )\in\Omc$, and $\monomial_1=\onomial'$. 
This shows the property in this case:
\begin{enumerate}
\item In case (i), for $k=1$ and $k'=0$, it holds that
\begin{itemize} 
								
\item $(S\sqsubseteq S, 1)$ and $(\exists \mn{inv}(
S).\top\sqsubseteq A, \monomial_1)\in\mn{saturate}(\Omc)$ (take $P_1=S$, $A_1=\top$, $B_1=A$, $\monomial_{S\sqsubseteq P_1}=1$, and $\monomial_{\mn{inv}(P_1).A_1\sqsubseteq B_1}=\monomial_1$),

\item $(A\sqsubseteq A, 1)\in\mn{saturate}(\Omc)$ (take $B_1=A$ and $p=1$), 
\end{itemize}
\new{and $\mn{Conj}(C,\top)=C$}. 

In case (ii), for $k=0$ and $k'=1$, it holds that 
 \begin{itemize} 								
\item $(\top\sqsubseteq A, \monomial_1)\in\mn{saturate}(\Omc)$ (take $B'_1=A$ and $\monomial_{\top\sqsubseteq B'_i}=\monomial_1$)

\item $(A\sqsubseteq A, 1)\in\mn{saturate}(\Omc)$ (take $B'_1=A$ and $p=1$),
\end{itemize} 
\new{and $\mn{Conj}(C)=C$}. 
 
\item In the case $k=1$, $(x,1)\in \top^{\Imc_{i-1}}$ holds trivially (recall that $A_1=\top$ and take $s_1=1$).  
 
\item The monomials are related as follows:
\begin{itemize}
\item If $k\neq 0$, we are in case (i) and $s\times \monomial_0\times 1\times 1\times\monomial_1\times 1=\onomial'$.
\item If $k=0$, we are in case (ii) and 
$\monomial_1\times 1 =\onomial'$.
\end{itemize}
\end{enumerate} 

\noindent\emph{Inductive step.} 
Assume that the property is true for every integer up to $l$ and consider the case where $i-j=l+1$. 
Let $x,y\in\Delta^{\Imc_\Omc}$, and assume that $(y,\onomial')\in A^{\Imc_i}$, and $y\notin\NI$ has been introduced between $\Imc_{j-1}$ and $\Imc_j$, to satisfy $(C\sqsubseteq \exists S, \monomial_0)\in\Omc$ by applying \new{the chase rule in case} $\R_7$ or $\R_8$ to $(x,s)\in C^{\Imc_{j-1}}$ (thus adding $(x,y,s\times\monomial_0)\in S^{\Imc_{j}}$). 
We make a case analysis on the last \new{chase} rule applied to add $(y,\onomial')\in A^{\Imc_i}$ (among \new{cases} $\R_3$, $\R_4$, $\R_5$ and $\R_6$).

\smallskip
\new{\textbf{$\boldsymbol{\R_3}$: $\boldsymbol{(y,\onomial')\in A^{\Imc_i}}$ added by applying the chase rule with $\boldsymbol{(D\sqsubseteq A, \monomial_{D\sqsubseteq A})\in\Omc}.$} 
\\
There are $(y,\onomial)\in D^{\Imc_{i-1}}$ and $\onomial'=\monomial_{D\sqsubseteq A}\times \onomial$. 
Since $(y,\onomial)\in D^{\Imc_{i-1}}$, by induction hypothesis: 
\begin{enumerate}
\item There exist $k\geq 0$ and $k'\geq 0$ such that: 
\begin{itemize} 								
\item $(S\sqsubseteq P_i, \monomial_{S\sqsubseteq P_i})$ and $(\exists \mn{inv}(P_i).A_i\sqsubseteq B_i, \monomial_{\exists \mn{inv}(P_i).A_i\sqsubseteq B_i})\in\mn{saturate}(\Omc)$, \mbox{$1\leq i\leq k$,}  
\item $(\top\sqsubseteq B'_i, \monomial_{\top\sqsubseteq B'_i})\in\mn{saturate}(\Omc)$, $1\leq i\leq k'$, 
\item $(B_1\sqcap\dots\sqcap B_{k}\sqcap B'_1\sqcap\dots\sqcap B'_{k'}\sqsubseteq D, p')\in\mn{saturate}(\Omc)$,
\end{itemize} 
and $\mn{Conj}(C,A_1,\dots,A_k)=C\sqcap A_1\sqcap\dots\sqcap A_k$.
\item For every $1\leq \ell\leq k$, there exists $(x,s_\ell)\in {A_\ell}^{\Imc_{i-2}}$. 
\item The monomials are related as follows:
\begin{itemize}
\item If $k\neq 0$,  $s\times\monomial_0\times\Pi_{\ell=1}^{k} s_\ell\times \Pi_{i=1}^{k} (\monomial_{S\sqsubseteq P_i}\times \monomial_{\exists \mn{inv}(P_i).A_i\sqsubseteq B_i})\times\Pi_{i=1}^{k'}\monomial_{\top\sqsubseteq B'_i}\times p' =\onomial$.
\item If $k=0$, 
$\Pi_{i=1}^{k'}\monomial_{\top\sqsubseteq B'_i}\times p' =\onomial$.
\end{itemize}
\end{enumerate} 
Since $(D\sqsubseteq A, \monomial_{D\sqsubseteq A})\in\mn{saturate}(\Omc)$ and $(B_1\sqcap\dots\sqcap B_{k}\sqcap B'_1\sqcap\dots\sqcap B'_{k'}\sqsubseteq D, p')\in\mn{saturate}(\Omc)$, by $\CR^T_{2}$, 
$(B_1\sqcap\dots\sqcap B_{k}\sqcap B'_1\sqcap\dots\sqcap B'_{k'}\sqsubseteq A, \monomial_{D\sqsubseteq A}\times p')\in\mn{saturate}(\Omc)$.}

\new{We thus obtain items 1 and 2 of the property by taking $p=\monomial_{D\sqsubseteq A}\times p'$. 
To obtain item~3 of the property, we observe that the monomials are related as follows:}
\new{\begin{align*}
\text{If $k\neq 0$: }\onomial'={} &\monomial_{D\sqsubseteq A}\times \onomial
\\
{} = {}&\monomial_{D\sqsubseteq A}\times s\times\monomial_0\times(\Pi_{\ell=1}^{k} s_\ell\times \Pi_{i=1}^{k} (\monomial_{S\sqsubseteq P_i}\times \monomial_{\exists \mn{inv}(P_i).A_i\sqsubseteq B_i})\times\Pi_{i=1}^{k'}\monomial_{\top\sqsubseteq B'_i}\times p') \\
{}= {}&
s\times\monomial_0 \times\Pi_{\ell=1}^k s_\ell\times \Pi_{i=1}^k (\monomial_{S\sqsubseteq P_i}\times \monomial_{\exists \mn{inv}(P_i).A_i\sqsubseteq B_i})\times\Pi_{i=1}^{k'}\monomial_{\top\sqsubseteq B'_i}\times p.
\\
\text{If $k=0$: }
\onomial'={}&\monomial_{D\sqsubseteq A}\times \onomial
= \monomial_{D\sqsubseteq A}\times (\Pi_{i=1}^{k'}\monomial_{\top\sqsubseteq B'_i}\times p') 
= \Pi_{i=1}^{k'}\monomial_{\top\sqsubseteq B'_i}\times p.
\end{align*} }
This shows the property in this case.

\smallskip
\new{\textbf{$\boldsymbol{\R_4}$: $\boldsymbol{(y,\onomial')\in A^{\Imc_i}}$ added by applying the chase rule with}}   
		\new{\textbf{$\boldsymbol{(D_1\sqcap D_2\sqsubseteq A, \monomial_{D_1\sqcap D_2\sqsubseteq A})\in\Omc}.$}} 
\\		
There are $(y,\onomial_1)\in D_1^{\Imc_{i-1}}$ and $(y,\onomial_2)\in D_2^{\Imc_{i-1}}$ and $\onomial'=\monomial_{D_1\sqcap D_2\sqsubseteq A}\times \onomial_1\times\onomial_2$. 
For $j\in\{1,2\}$, since $(y,\onomial_j)\in D_j^{\Imc_{i-1}}$, by induction hypothesis: 
\begin{enumerate}
\item There exist $k_j\geq 0$ and $k_j'\geq 0$ such that: 
\begin{itemize} 								
\item $(S\sqsubseteq P^j_i, \monomial_{S\sqsubseteq P^j_i})$ and $(\exists \mn{inv}(P_i^j).A^j_i\sqsubseteq B^j_i, \monomial_{\exists \mn{inv}(P_i^j).A^j_i\sqsubseteq B^j_i})\in\mn{saturate}(\Omc)$, \mbox{$1\leq i\leq k_j$,}  

\item $(\top\sqsubseteq B'^j_i, \monomial_{\top\sqsubseteq B'^j_i})\in\mn{saturate}(\Omc)$, $1\leq i\leq k_j'$, 

\item $(B^j_1\sqcap\dots\sqcap B^j_{k_j}\sqcap B'^j_1\sqcap\dots\sqcap B'^j_{k_j'}\sqsubseteq D_j, p_j)\in\mn{saturate}(\Omc)$,
\end{itemize} 
 \new{and $\mn{Conj}(C,A^j_1,\dots,A^j_{k_j})=C\sqcap A^j_1\sqcap\dots\sqcap A^j_{k_j}$}.
\item For every $1\leq \ell\leq k_j$, there exists $(x,s^j_\ell)\in {A^j_\ell}^{\Imc_{i-2}}$. 
 
\item The monomials are related as follows:
\begin{itemize}
\item If $k_j\neq 0$,  $s\times\monomial_0\times\Pi_{\ell=1}^{k_j} s^j_\ell\times \Pi_{i=1}^{k_j} (\monomial_{S\sqsubseteq P^j_i}\times \monomial_{\exists \mn{inv}(P_i^j).A^j_i\sqsubseteq B^j_i})\times\Pi_{i=1}^{k_j'}\monomial_{\top\sqsubseteq B'^j_i}\times p_j =\onomial_j$.
\item If $k_j=0$, 
$\Pi_{i=1}^{k_j'}\monomial_{\top\sqsubseteq B'^j_i}\times p_j =\onomial_j$.
\end{itemize}
\end{enumerate} 
\noindent\new{We thus have the following relationships between monomials ($\dagger$). 
\begin{align*}
\text{If $k_1\neq 0$ or $k_2\neq 0$: }\onomial'={}&\monomial_{D_1\sqcap D_2\sqsubseteq A}\times \onomial_1\times\onomial_2
\\
{} = {} &\monomial_{D_1\sqcap D_2\sqsubseteq A}\times s\times\monomial_0
\\
&\times\Pi_{j=1}^2 (\Pi_{\ell=1}^{k_j} s^j_\ell\times \Pi_{i=1}^{k_j} (\monomial_{S\sqsubseteq P^j_i}\times \monomial_{\exists \mn{inv}(P_i^j).A^j_i\sqsubseteq B^j_i})\times\Pi_{i=1}^{k_j'}\monomial_{\top\sqsubseteq B'^j_i}\times p_j) \\
{}= {}&
s\times\monomial_0\times\monomial_{D_1\sqcap D_2\sqsubseteq A}\times p_1\times p_2
\\&\times \Pi_{j=1}^2 \Pi_{\ell=1}^{k_j} s^j_\ell
\times \Pi_{j=1}^2 \Pi_{i=1}^{k_j} (\monomial_{S\sqsubseteq P^j_i}\times \monomial_{\exists \mn{inv}(P_i^j).A^j_i\sqsubseteq B^j_i})\times \Pi_{j=1}^2\Pi_{i=1}^{k_j'}\monomial_{\top\sqsubseteq B'^j_i}.
\\
\text{If $k_1=k_2=0$: }\onomial'={}&\monomial_{D_1\sqcap D_2\sqsubseteq A}\times \onomial_1\times\onomial_2
= \monomial_{D_1\sqcap D_2\sqsubseteq A}\times \Pi_{j=1}^2 p_j\times\Pi_{j=1}^2\Pi_{i=1}^{k_j'}\monomial_{\top\sqsubseteq B'^j_i}.
\end{align*} 
}

\blue{By Lemmas \ref{lem:choice} and \ref{lem:choice2}, we can re-order the RIs and GCIs from point~(1) above in a way that there exist $h_1\leq k_1$, $h'_1\leq k'_1$, $h_2\leq k_2$, and $h'_2\leq k'_2$ and such that $$\mn{Conj}(\bigsqcap_{j=1}^2 (B^j_1\sqcap\dots\sqcap B^j_{k_j}\sqcap B'^j_1\sqcap\dots\sqcap B'^j_{k_j'}))=\mn{Conj}(\bigsqcap_{j=1}^2 (B^j_1\sqcap\dots\sqcap B^j_{h_j}\sqcap B'^j_1\sqcap\dots\sqcap B'^j_{h_j'}))=\bigsqcap_{j=1}^2 (B^j_1\sqcap\dots\sqcap B^j_{h_j}\sqcap B'^j_1\sqcap\dots\sqcap B'^j_{h_j'}),$$  $$\mn{Conj}(C,A^1_1,\dots,A^1_{k_1},A^2_1,\dots,A^2_{k_2})=\mn{Conj}(C,A^1_1,\dots,A^1_{h_1},A^2_1,\dots,A^2_{h_2})=C\sqcap \bigsqcap_{j=1}^2(A^j_1\sqcap\dots\sqcap A^j_{h_j})$$ and the relationships ($\dagger$) are still true if $k_j$ and $k'_j$ are replaced by $h_j$ and $h'_j$, respectively.}

Since $(D_1\sqcap D_2\sqsubseteq A, \monomial_{D_1\sqcap D_2\sqsubseteq A})\in\mn{saturate}(\Omc)$ and $(B^1_1\sqcap\dots\sqcap B^1_{k_1}\sqcap B'^1_1\sqcap\dots\sqcap B'^1_{k_1'}\sqsubseteq D_1, p_1)\in\mn{saturate}(\Omc)$, by $\CR^T_{2}$, 
\new{$(\mn{Conj}(B^1_1\sqcap\dots\sqcap B^1_{k_1}\sqcap B'^1_1\sqcap\dots\sqcap B'^1_{k_1'}\sqcap D_2)\sqsubseteq A, \monomial_{D_1\sqcap D_2\sqsubseteq A}\times p_1)\in\mn{saturate}(\Omc)$}. 
By applying again $\CR^T_{2}$, we obtain $(\mn{Conj}(\bigsqcap_{j=1}^2 (B^j_1\sqcap\dots\sqcap B^j_{k_j}\sqcap B'^j_1\sqcap\dots\sqcap B'^j_{k_j'}))\sqsubseteq A, p)\in\mn{saturate}(\Omc)$ with $p=\monomial_{D_1\sqcap D_2\sqsubseteq A}\times p_1 \times p_2$, 
\blue{\ie $(\bigsqcap_{j=1}^2 (B^j_1\sqcap\dots\sqcap B^j_{h_j}\sqcap B'^j_1\sqcap\dots\sqcap B'^j_{h_j'})\sqsubseteq A, p)\in\mn{saturate}(\Omc)$ by the equality above.}

If we let \blue{$k= h_1+h_2$, $k'=h'_1+h'_2$,  $f(i,j)=\Sigma_{\ell=1}^{j-1} h_\ell+i$ and $f'(i,j)=\Sigma_{\ell=1}^{j-1} h'_\ell+i$}, we can rename the $P_i^j$, $A_i^j$, $B_i^j$, ${B'_i}^j$ and $s_\ell^j$ by $P_{f(i,j)}$, $A_{f(i,j)}$, $B_{f(i,j)}$, $B'_{f'(i,j)}$ and $s_{f(\ell,j)}$ respectively to obtain items 1 and 2 of the property. 
To obtain item 3 of the property, we observe that the monomials are related as follows, \blue{using the fact that ($\dagger$) is still true if $k_j$ and $k'_j$ are replaced by $h_j$ and $h'_j$, respectively}:
\begin{align*}
\text{If $k\neq 0$, \ie $\blue{h}_1\neq 0$ or $\blue{h}_2\neq 0$: }\onomial'=&
s\times\monomial_0\times\monomial_{D_1\sqcap D_2\sqsubseteq A}\times  p_1\times p_2
\\&\times \Pi_{j=1}^2 \Pi_{\ell=1}^{\blue{h}_j} s^j_\ell
\times \Pi_{j=1}^2 \Pi_{i=1}^{\blue{h}_j} (\monomial_{S\sqsubseteq P^j_i}\times \monomial_{\exists \mn{inv}(P_i^j).A^j_i\sqsubseteq B^j_i})\times \Pi_{j=1}^2\Pi_{i=1}^{\blue{h}_j'}\monomial_{\top\sqsubseteq B'^j_i}\\
=&
s\times\monomial_0\times p \times\Pi_{\ell=1}^k s_\ell\times \Pi_{i=1}^k (\monomial_{S\sqsubseteq P_i}\times \monomial_{\exists \mn{inv}(P_i).A_i\sqsubseteq B_i})\times\Pi_{i=1}^{k'}\monomial_{\top\sqsubseteq B'_i}.
\\
\text{If $k=0$, \ie $\blue{h}_1=\blue{h}_2=0$:  }
\onomial'
=& \monomial_{D_1\sqcap D_2\sqsubseteq A}\times p_1\times p_2\times\Pi_{j=1}^2\Pi_{i=1}^{\blue{h}_j'}\monomial_{\top\sqsubseteq B'^j_i}
=p\times \Pi_{i=1}^{k'}\monomial_{\top\sqsubseteq B'_i}.
\end{align*} 
This shows the property in this case.

\smallskip
\new{\textbf{$\boldsymbol{\R_5}$ or $\boldsymbol{\R_6}$: $\boldsymbol{(y,\onomial')\in A^{\Imc_i}}$ added by applying the chase rule with}}  \new{\textbf{$\boldsymbol{(\exists P.D\sqsubseteq A, \monomial_{\exists P.D\sqsubseteq A})\in\Omc}.$}} 
\\
There exist $(y,z,\onomial_{yz})\in P^{\Imc_{i-1}}$ and $(z,\onomial_z)\in D^{\Imc_{i-1}}$ such that $\onomial'=\onomial_{yz}\times \onomial_z\times \monomial_{\exists P.D\sqsubseteq A}$. We distinguish two subcases:
\begin{enumerate}[(i)]
\item either $z=x$,
\item or $z\neq x$, which implies that $z$ has been introduced between $j'-1$ and $j'$ ($j\leq j'\leq i$) to satisfy an inclusion of the form $(E_0\sqsubseteq \exists S_z, \monomial_z)\in\Omc$ by applying \new{the chase rule in case} $\R_7$ or $\R_8$.
\end{enumerate}
In case (i), it holds that 
$(y,x,\onomial_{yx})\in P^{\Imc_{i-1}}$, $(x,\onomial_x)\in D^{\Imc_{i-1}}$ and $\onomial'=\onomial_{yx}\times \onomial_x\times \monomial_{\exists P.D\sqsubseteq A}$. 
 Since $(y,x,\onomial_{yx})\in P^{\Imc_{i-1}}$, \ie  $(x,y,\onomial_{yx})\in \mn{inv}(P)^{\Imc_{i-1}}$, and $y$ has been introduced to satisfy $(C\sqsubseteq \exists S,\monomial_0)\in\Omc$ by applying \new{the chase rule in case} $\R_7$ to some $(x,s)\in C^{\Imc_{j-1}}$, 
then by Lemma \ref{lem:canonical-saturation-RIs}, there is $(S\sqsubseteq \mn{inv}(P), \monomial_{S\sqsubseteq \mn{inv}(P)})\in\mn{saturate}(\Omc)$ s.t.\ $s\times\monomial_0\times \monomial_{S\sqsubseteq \mn{inv}(P)}=\onomial_{yx}$. 
This shows the property in case (i), indeed:
\begin{enumerate}
\item If we take $k=1$ and $k'=0$, it holds that:
\begin{itemize} 								
\item $(S\sqsubseteq \mn{inv}(P), \monomial_{S\sqsubseteq \mn{inv}(P)})$ and $(\exists P.D\sqsubseteq A, \monomial_{\exists P.D\sqsubseteq A})\in\mn{saturate}(\Omc)$ (take $P_1=\mn{inv}(P)$, $A_1=D$, $B_1=A$),

\item $(A\sqsubseteq A, 1)\in\mn{saturate}(\Omc)$ (take $B_1=A$ and $p=1$), 
\end{itemize} 
 \new{and $\mn{Conj}(C,D)=C\sqcap D$}.
\item $(x,\onomial_x)\in D^{\Imc_{i-1}}$ (take $A_1=D$ and $s_1=\onomial_x$).  
 
\item The monomials are related as follows (and $k\neq 0$):
\begin{align*}
\onomial'= \onomial_{yx}\times \onomial_x\times \monomial_{\exists P.D\sqsubseteq A}
=s\times\monomial_0\times \onomial_x\times \monomial_{S\sqsubseteq \mn{inv}(P)}\times \monomial_{\exists P.D\sqsubseteq A}\times 1.
\end{align*}
\end{enumerate} 
We now consider case (ii): $(y,z,\onomial_{yz})\in P^{\Imc_{i-1}}$, $(z,\onomial_z)\in D^{\Imc_{i-1}}$, $\onomial'=\onomial_{yz}\times \onomial_z\times \monomial_{\exists P.D\sqsubseteq A}$ and $z$ has been introduced between $j'-1$ and $j'$ ($j\leq j'\leq i$) to satisfy an inclusion of the form $(E_0\sqsubseteq \exists S_z, \monomial_z)\in\Omc$ by applying \new{the chase rule in case} $\R_7$ or $\R_8$. 
There must exist $(y,s_0)\in E_0^{\Imc_{j'-1}}$ to which \new{the chase rule in case} $\R_7$ or $\R_8$ has been applied to satisfy $(E_0\sqsubseteq \exists S_z, \monomial_z)$. 
Hence
 by Lemma \ref{lem:canonical-saturation-RIs} there exists $(S_z\sqsubseteq P, \monomial_{S_z\sqsubseteq P}  )\in\mn{saturate}(\Omc)$ such that $s_0\times  \monomial_z\times \monomial_{S_z\sqsubseteq P}=\onomial_{yz}$, 
and by induction hypothesis, the following statements hold.
\begin{itemize}
\item There exist $k\geq 0$ and $k'\geq 0$ such that:
\begin{itemize} 								
\item $(S_z\sqsubseteq P_i, \monomial_{S_z\sqsubseteq P_i})$ and  $(\exists \mn{inv}(P_i).E_i\sqsubseteq B_i, \monomial_{\exists \mn{inv}(P_i).E_i\sqsubseteq B_i})\in\mn{saturate}(\Omc)$, \mbox{$1\leq i\leq k$,}  

\item $(\top\sqsubseteq B'_i, \monomial_{\top\sqsubseteq B'_i})\in\mn{saturate}(\Omc)$, $1\leq i\leq k'$, 

\item $(B_1\sqcap\dots\sqcap B_k\sqcap B'_1\sqcap\dots\sqcap B'_{k'}\sqsubseteq D, p_z)\in\mn{saturate}(\Omc)$,
\end{itemize} 
\new{and $\mn{Conj}(E_0,E_1,\dots,E_k)=E_0\sqcap E_1\sqcap\dots\sqcap E_k$}. 

Since $(\exists P.D\sqsubseteq A, \monomial_{\exists P.D\sqsubseteq A})\in\Omc$ and $(E_0\sqsubseteq \exists S_z, \monomial_z)\in\Omc$, it follows by $\CR^T_{3}$ that $(E_0\sqcap E_1\sqcap\dots\sqcap E_k\sqsubseteq A,r_z )\in\mn{saturate}(\Omc)$ \new{(as $\mn{Conj}(E_0,E_1,\dots,E_k)=E_0\sqcap E_1\sqcap\dots\sqcap E_k$)}
with $r_z= \monomial_z\times \monomial_{\exists P.D\sqsubseteq A}\times \monomial_{S_z\sqsubseteq P}\times\Pi_{i=1}^k (\monomial_{S_z\sqsubseteq P_i}\times \monomial_{\exists \mn{inv}(P_i).E_i\sqsubseteq B_i})\times\Pi_{i=1}^{k'}\monomial_{\top\sqsubseteq B'_i}\times p_z$.
\item For every $1\leq \ell\leq k$, there exists $(y,s_\ell)\in E_\ell^{\Imc_{i-1}}$. 
\item The monomials are related as follows:
\begin{itemize}
\item If $k\neq 0$,  $s_0\times \monomial_z\times\Pi_{\ell=1}^k s_\ell\times \Pi_{i=1}^k (\monomial_{S_z\sqsubseteq P_i}\times \monomial_{\exists \mn{inv}(P_i).E_i\sqsubseteq B_i})\times\Pi_{i=1}^{k'}\monomial_{\top\sqsubseteq B'_i}\times p_z =\onomial_{z}$.
\item If $k=0$, 
$\Pi_{i=1}^{k'}\monomial_{\top\sqsubseteq B'_i}\times p_z =\onomial_{z}$.
\end{itemize}

\end{itemize} 
For every $0\leq j\leq k$, since $(y,s_j)\in E_j^{\Imc_{i-1}}$, by induction hypothesis: 
\begin{itemize} 
					
\item There exist $k_j\geq 0$ and $k_j'\geq 0$ such that:
\begin{itemize}		
\item $(S\sqsubseteq P^j_i, \monomial_{S\sqsubseteq P^j_i})$,  $(\exists \mn{inv}(P_i^j).A^j_i\sqsubseteq B^j_i, \monomial_{\exists \mn{inv}(P_i^j).A^j_i\sqsubseteq B^j_i})\in\mn{saturate}(\Omc)$, \mbox{$1\leq i\leq k_j$,}  

\item $(\top\sqsubseteq B'^j_i, \monomial_{\top\sqsubseteq B'^j_i})\in\mn{saturate}(\Omc)$, $1\leq i\leq k_j'$, 

\item $(B^j_1\sqcap\dots\sqcap B^j_{k_j}\sqcap B'^j_1\sqcap\dots\sqcap B'^j_{k_j'}\sqsubseteq E_j, p_j)\in\mn{saturate}(\Omc)$,
\end{itemize}
\new{and $\mn{Conj}(C,A^j_1,\dots,A^j_{k_j})=C\sqcap A^j_1\sqcap\dots\sqcap A^j_{k_j}$}.
\item For every $1\leq \ell\leq k_j$, there exists $(x,s^j_\ell)\in {A^j_\ell}^{\Imc_{i-1}}$. 

\item The monomials are related as follows:
\begin{itemize}
\item If $k_j\neq 0$, $s\times\monomial_0\times\Pi_{\ell=1}^{k_j} s^j_\ell\times \Pi_{i=1}^{k_j} (\monomial_{S\sqsubseteq P^j_i}\times \monomial_{\exists \mn{inv}(P_i^j).A^j_i\sqsubseteq B^j_i})\times\Pi_{i=1}^{k_j'}\monomial_{\top\sqsubseteq {B'}^j_i}\times p_j =s_j$.
\item If $k_j=0$, 
$\Pi_{i=1}^{k_j'}\monomial_{\top\sqsubseteq {B'}^j_i}\times p_j =s_j$.
\end{itemize}
\end{itemize} 

\noindent\new{Combining the equalities we have, we obtain the following: 
\begin{align*}
\text{If $k\neq 0$: }\onomial'={} & \onomial_{yz}\times \onomial_z\times \monomial_{\exists P.D\sqsubseteq A}\\
={}& s_0\times  \monomial_z\times \monomial_{S_z\sqsubseteq P}
\times
\Pi_{j=1}^k s_j\times \Pi_{i=1}^k (\monomial_{S_z\sqsubseteq P_i}\times \monomial_{\exists \mn{inv}(P_i).E_i\sqsubseteq B_i})\times\Pi_{i=1}^{k'}\monomial_{\top\sqsubseteq B'_i}\times p_z
\\&\times \monomial_{\exists P.D\sqsubseteq A}\\
{} ={}& 
\Pi_{j=0}^k s_j\! \times\!  \Pi_{i=1}^k (\monomial_{S_z\sqsubseteq P_i}\! \times\!  \monomial_{\exists \mn{inv}(P_i).E_i\sqsubseteq B_i})\! \times\! \Pi_{i=1}^{k'}\monomial_{\top\sqsubseteq B'_i}\! \times\!  p_z
\! \times\! \monomial_z\! \times\!  \monomial_{S_z\sqsubseteq P}\! \times\! \monomial_{\exists P.D\sqsubseteq A}\\
{}={}&  \Pi_{j=0}^k s_j \times r_z.
\\
\text{If $k=0$: }
\onomial'= &\onomial_{yz}\times \onomial_z\times \monomial_{\exists P.D\sqsubseteq A}
= s_0\times  \monomial_z\times \monomial_{S_z\sqsubseteq P}\times\Pi_{i=1}^{k'}\monomial_{\top\sqsubseteq B'_i}\times p_z
\times \monomial_{\exists P.D\sqsubseteq A}
= s_0\times r_z.
\end{align*}
So in both cases, $\onomial'=\Pi_{j=0}^k s_j\times r_z$.}

\noindent\new{We thus have the following relationships between monomials ($\dagger$). 
 \begin{align*}
\text{When $\Sigma_{j=0}^k k_j\neq 0$: }\onomial'={}& \Pi_{j=0}^k s_j
\times r_z\\
{}={}&  s\times\monomial_0\times\Pi_{j=0}^k (\Pi_{\ell=1}^{k_j} s^j_\ell\times \Pi_{i=1}^{k_j} (\monomial_{S\sqsubseteq P^j_i}\times \monomial_{\exists \mn{inv}(P_i^j).A^j_i\sqsubseteq B^j_i})\times\Pi_{i=1}^{k_j'}\monomial_{\top\sqsubseteq {B'}^j_i}\times p_j)
\times r_z\\
{}={}& s\times\monomial_0 \\&
\times\Pi_{j=0}^k\Pi_{\ell=1}^{k_j} s^j_\ell \times \Pi_{j=0}^k \Pi_{i=1}^{k_j} (\monomial_{S\sqsubseteq P^j_i}\times \monomial_{\exists \mn{inv}(P_i^j).A^j_i\sqsubseteq B^j_i})\times \Pi_{j=0}^k\Pi_{i=1}^{k_j'}\monomial_{\top\sqsubseteq {B'}^j_i}\\&
\times\Pi_{j=0}^k p_j \times r_z.
\\
\text{When $\Sigma_{j=0}^k k_j= 0$: }\onomial'={}& \Pi_{j=0}^k s_j
\times r_z
= \Pi_{j=0}^k \Pi_{i=1}^{k_j'}\monomial_{\top\sqsubseteq {B'}^j_i}\times p_j 
\times r_z.
\end{align*} 
}

\blue{Again, by Lemmas \ref{lem:choice} and \ref{lem:choice2}, we can re-order the RIs and GCIs above in a way that for $0\leq j\leq k$, there exist $h_j\leq k_j$, $h'_j\leq k'_j$ such that 
$$\mn{Conj}(\bigsqcap_{j=0}^k (B^j_1\sqcap\dots\sqcap B^j_{k_j}\sqcap B'^j_1\sqcap\dots\sqcap B'^j_{k_j'}))=\mn{Conj}(\bigsqcap_{j=0}^k (B^j_1\sqcap\dots\sqcap B^j_{h_j}\sqcap B'^j_1\sqcap\dots\sqcap B'^j_{h_j'}))=\bigsqcap_{j=0}^k (B^j_1\sqcap\dots\sqcap B^j_{h_j}\sqcap B'^j_1\sqcap\dots\sqcap B'^j_{h_j'}),$$  
$$\mn{Conj}(C,A^0_1,\dots,A^0_{k_0},\dots,A^k_1,\dots,A^k_{k_k})=\mn{Conj}(C,A^0_1,\dots,A^0_{h_0},\dots,A^k_1,\dots,A^k_{h_k})=C\sqcap \bigsqcap_{j=0}^k(A^j_1\sqcap\dots\sqcap A^j_{h_j})$$
and the relationships ($\dagger$) are still true if $k_j$ and $k'_j$ are replaced by $h_j$ and $h'_j$, respectively.}

Since $(E_0\sqcap E_1\sqcap\dots\sqcap E_{k}\sqsubseteq A, r_z)\in\mn{saturate}(\Omc)$ and $(B^0_1\sqcap\dots\sqcap B^0_{k_0}\sqcap B'^0_1\sqcap\dots\sqcap B'^0_{k_0'}\sqsubseteq E_0, p_0)\in\mn{saturate}(\Omc)$, by $\CR^T_{2}$, 
\new{$(\mn{Conj}(B^0_1\sqcap\dots\sqcap B^0_{k_0}\sqcap B'^0_1\sqcap\dots\sqcap B'^0_{k_0'}\sqcap E_1\sqcap\dots\sqcap E_{k})\sqsubseteq A, r_z\times p_0)\in\mn{saturate}(\Omc)$}. 
By applying successively $\CR^T_{2}$, we obtain \new{$(\mn{Conj}(\bigsqcap_{j=0}^k (B^j_1\sqcap\dots\sqcap B^j_{k_j}\sqcap B'^j_1\sqcap\dots\sqcap B'^j_{k_j'}))\sqsubseteq A, p)\in\mn{saturate}(\Omc)$}, \ie $(\bigsqcap_{j=0}^k (B^j_1\sqcap\dots\sqcap B^j_{\blue{h}_j}\sqcap B'^j_1\sqcap\dots\sqcap B'^j_{\blue{h}_j'})\sqsubseteq A, p)\in\mn{saturate}(\Omc)$, with $p=r_z\times \Pi_{j=0}^k p_j$. 

If we let $K=\Sigma_{j=0}^k \blue{h}_j$,  $K'=\Sigma_{j=1}^k \blue{h}'_j$, and rename the $P_i^j$, $A_i^j$, $B_i^j$, ${B'_i}^j$ and $s^j_\ell$, we obtain items 1 and 2 of the property. We now show item 3 of the property from ($\dagger$). 

\noindent\new{ 
 \begin{align*}
\text{In the case where $K\neq 0$: }\onomial'=& s\times\monomial_0 \\&
\times\Pi_{j=0}^k\Pi_{\ell=1}^{\blue{h}_j} s^j_\ell \times \Pi_{j=0}^k \Pi_{i=1}^{\blue{h}_j} (\monomial_{S\sqsubseteq P^j_i}\times \monomial_{\exists \mn{inv}(P_i^j).A^j_i\sqsubseteq B^j_i})\times \Pi_{j=0}^k\Pi_{i=1}^{\blue{h}_j'}\monomial_{\top\sqsubseteq {B'}^j_i}\\&
\times\Pi_{j=0}^k p_j \times r_z\\
=&  s\times\monomial_0 
\times\Pi_{j=0}^k\Pi_{\ell=1}^{\blue{h}_j} s^j_\ell \times \Pi_{j=0}^k \Pi_{i=1}^{\blue{h}_j} (\monomial_{S\sqsubseteq P^j_i}\times \monomial_{\exists \mn{inv}(P_i^j).A^j_i\sqsubseteq B^j_i})\times \Pi_{j=0}^k\Pi_{i=1}^{\blue{h}_j'}\monomial_{\top\sqsubseteq {B'}^j_i}\times p
\\
=&s\times\monomial_0\times\Pi_{\ell=1}^K s_\ell\times \Pi_{i=1}^K (\monomial_{S\sqsubseteq P_i}\times \monomial_{\exists \mn{inv}(P_i).A_i\sqsubseteq B_i})\times\Pi_{i=1}^{K'}\monomial_{\top\sqsubseteq B'_i}\times p.
\\
\text{ In the case $K=0$: }
\onomial'={}&\Pi_{j=0}^k \Pi_{i=1}^{\blue{h}_j'}\monomial_{\top\sqsubseteq {B'}^j_i}\times p_j 
\times r_z\\
{}={}& \Pi_{j=0}^k \Pi_{i=1}^{\blue{h}_j'}\monomial_{\top\sqsubseteq {B'}^j_i}\times p\\
{}={}&\Pi_{i=1}^{K'}\monomial_{\top\sqsubseteq B'_i}\times p.
\end{align*}}
This shows the property in case (ii) and finishes the proof of the lemma.
\end{proof}

\begin{lemma}\label{claimcomplELHI} 
For all $i\geq 0$, $x,y\in\Delta^{\Imc_\Omc}$, $R$ role name or inverse role, and $A\in\NC\cup\{\top\}$, if $(x,y,\onomial)\in R^{\Imc_i}$, $(y,\onomial')\in A^{\Imc_i}$, $(\exists R.A\sqsubseteq B,\nonomial)\in\Omc$, and $y\notin\NI$ has been introduced between $\Imc_{j-1}$ and $\Imc_j$ ($j\leq i$), to satisfy an inclusion of the form $(C\sqsubseteq \exists S, \monomial_0)\in\Omc$ by applying \new{the chase rule in case} $\R_7$ or $\R_8$ to some $(x,s)\in C^{\Imc_{j-1}}$, 
then  
\begin{itemize} 
\item there exists $(C\sqcap A_1\sqcap\dots\sqcap A_k\sqsubseteq B,r )\in\mn{saturate}(\Omc)$ for some $k\geq 0$, such that for every $1\leq \ell\leq k$, $(x,s_\ell)\in A_\ell^{\Imc_i}$, 
\item and $s\times\Pi_{\ell=1}^k s_\ell\times r= \onomial\times\onomial'\times\nonomial$. 
\end{itemize} 
\end{lemma}
\begin{proof}
By Lemma \ref{lem:canonical-saturation-RIs}, there exists $(S\sqsubseteq R, \monomial_{S\sqsubseteq R}  )\in\mn{saturate}(\Omc)$ such that $$s\times \monomial_0\times\monomial_{S\sqsubseteq R}=\onomial,$$ and 
by Lemma \ref{lem:canonical-saturation-GCIs}, 
the following statements hold.
\begin{enumerate}
\item There exist $k\geq 0$ and $k'\geq 0$ such that:
\begin{itemize} 								
\item $(S\sqsubseteq P_i, \monomial_{S\sqsubseteq P_i})$ and  $(\exists \mn{inv}(P_i).A_i\sqsubseteq B_i, \monomial_{\exists \mn{inv}(P_i).A_i\sqsubseteq B_i})\in\mn{saturate}(\Omc)$, $1\leq i\leq k$,  

\item $(\top\sqsubseteq B'_i, \monomial_{\top\sqsubseteq B'_i})\in\mn{saturate}(\Omc)$, $1\leq i\leq k'$, 

\item $(B_1\sqcap\dots\sqcap B_k\sqcap B'_1\sqcap\dots\sqcap B'_{k'}\sqsubseteq A, p)\in\mn{saturate}(\Omc)$.
\end{itemize} 
\new{and $\mn{Conj}(C,A_1,\dots,A_k)=C\sqcap A_1\sqcap\dots\sqcap A_k$ (\ie there is no concept name that occurs more than $\mn{Card}(\Omc)$ times in $\{C, A_1,\dots,A_k\}$)}.

Moreover, if $A=\top$, then $k=k'=0$ and $p=1$.
 
\item For every $1\leq \ell\leq k$, there exists $(x,s_\ell)\in A_\ell^{\Imc_{i-1}}$.  
 
\item The monomials are related as follows:
\begin{itemize}
\item If $k\neq 0$,  $s\times\monomial_0\times\Pi_{\ell=1}^k s_\ell\times \Pi_{i=1}^k (\monomial_{S\sqsubseteq P_i}\times \monomial_{\exists \mn{inv}(P_i).A_i\sqsubseteq B_i})\times\Pi_{i=1}^{k'}\monomial_{\top\sqsubseteq B'_i}\times p =\onomial'$.
\item If $k=0$, 
$\Pi_{i=1}^{k'}\monomial_{\top\sqsubseteq B'_i}\times p =\onomial'$.
\end{itemize}
\end{enumerate} 

\noindent Hence, 
if $(\exists R.A\sqsubseteq B,\nonomial)\in\Omc$, since $(C\sqsubseteq \exists S, \monomial_0)\in\Omc$, it follows by $\CR^T_{3}$ that \new{$(\mn{Conj}(C\sqcap A_1\sqcap\dots\sqcap A_k)\sqsubseteq B,r )\in\mn{saturate}(\Omc)$, \ie $(C\sqcap A_1\sqcap\dots\sqcap A_k\sqsubseteq B,r )\in\mn{saturate}(\Omc)$} with $$r= \monomial_0\times \nonomial\times \monomial_{S\sqsubseteq R}\times\Pi_{i=1}^k (\monomial_{S\sqsubseteq P_i}\times \monomial_{\exists \mn{inv}(P_i).A_i\sqsubseteq B_i})\times\Pi_{i=1}^{k'}\monomial_{\top\sqsubseteq B'_i}\times p,$$ so that  
\begin{align*}
s\times\Pi_{\ell=1}^k s_\ell\times r 
{}={}&
s\times\Pi_{\ell=1}^k s_\ell\times \monomial_0\times
\nonomial\times\monomial_{S\sqsubseteq R}\times \Pi_{i=1}^k (\monomial_{S\sqsubseteq P_i}\times \monomial_{\exists \mn{inv}(P_i).A_i\sqsubseteq B_i})\times\Pi_{i=1}^{k'}\monomial_{\top\sqsubseteq B'_i}\times p
\\
{}={}&
\monomial_{S\sqsubseteq R}\times s\times \monomial_0\times\Pi_{\ell=1}^k s_\ell\times \Pi_{i=1}^k (\monomial_{S\sqsubseteq P_i}\times \monomial_{\exists \mn{inv}(P_i).A_i\sqsubseteq B_i})\times\Pi_{i=1}^{k'}\monomial_{\top\sqsubseteq B'_i}\times p \times
\nonomial
\\
{}={}&\onomial\times\onomial'\times\nonomial.\qedhere
\end{align*}
\end{proof}

We are now ready to prove our theorem.

\CorrectnessELHI*
\begin{proof}
(1) \new{If $\Omc$ is unsatisfiable, $\Omc\models(\alpha,\monomial)$ for every $(\alpha,\monomial)$. If $\Omc$ is satisfiable, by point (1) of Lemma~\ref{lem:saturation-correct-for-axioms}, $\Imc\models (\alpha,\monomial)$ for every model $\Imc$ of $\Omc$, so $\Omc\models(\alpha,\monomial)$. Moreover, if $\alpha=A(a_\top)$, since $a_\top\notin\individuals{\Omc}$, it is easy to check that $\Omc\models (A(a_\top),\monomial)$ implies $\Omc\models (A(c),\monomial)$ for every $c\in\NI$ (e.g., by considering the canonical model $\Imc_\Omc$ of $\Omc$).} 

\medskip
\noindent(2.a) Assume that $\Omc$ is satisfiable and let $\alpha$ be an assertion \new{of the form $A(a)$ or $R(a,b)$ with $a,b\in\individuals{\Omc}$} and $\monomial$ a monomial. We show that if $\Omc\models(\alpha,\monomial)$ then $(\alpha,\monomial)\in\mn{saturate}(\Omc)$ 
by contrapositive: 
assuming that $(\alpha, {\monomial})\notin \mn{saturate}(\Omc)$, we show that the canonical model $\Imc_\Omc=\bigcup_{i\geq 0}\Imc_i$ of $\Omc$ (\cf Figure~\ref{fig:rules-canonical-why-normal}) is such that $\Imc_\Omc\not\models(\alpha,\monomial)$. 
We show by induction that for every $i$, for every assertion $\beta$ \new{of the form $S(a,b)$ or $B(a)$ with $a,b\in\individuals{\Omc}$} and every monomial $\monomial$, if $(\beta, {\monomial})\notin\mn{saturate}(\Omc)$, then $\Imc_i\not\models (\beta,\monomial)$. 

For $i=0$, for every assertion $\beta$, $(\beta, {\monomial})\notin \mn{saturate}(\Omc)$ implies $(\beta, {\monomial})\notin \Omc$, so $\Imc_0\not\models (\beta,\monomial)$ by construction of $\Imc_0$. 

Assume that the property holds for $i\geq 0$ and let $\beta$ \new{be an assertion of the form $S(a,b)$ or $B(a)$ with $a,b\in\individuals{\Omc}$} such that $(\beta, {\monomial})\notin \mn{saturate}(\Omc)$. Assume for a contradiction that $\Imc_{i+1}\models (\beta,\monomial)$. 
Since $\Imc_{i}\not\models (\beta,\monomial)$ by the induction hypothesis, it follows  that $\Imc_{i+1}$ has been obtained from $\Imc_i$ by applying \new{the chase rule in one of the cases} $\R_1$ to $\R_6$ (since the tuples added by \new{applying the chase rule in cases} $\R_7$ and $\R_8$ involve at least one domain element $x\in\Delta^\Imc\setminus\NI$). 
We next show that in every case, $(\beta,{\monomial})\in\mn{saturate}(\Omc)$. 
\begin{itemize}
\item $\R_1$: $(\beta,\monomial)=(S(a,b),\onomial\times\nonomial)$ and it holds that $\Imc_i\models(R(a,b),\onomial)$ and $(R\sqsubseteq S,\nonomial)\in\Omc$. 
By induction hypothesis, since $\Imc_i\models(R(a,b),\onomial)$, then $(R(a,b),\onomial)\in\mn{saturate}(\Omc)$. 
Hence, it follows from the construction of $\mn{saturate}(\Omc)$ (by $\CR^A_{4}$) that $(\beta,\monomial)\in\mn{saturate}(\Omc)$.

\item $\R_2$: $(\beta,\monomial)=(S(a,b),\onomial\times\nonomial)$ and it holds that $\Imc_i\models(R(b,a),\onomial)$ and $(R\sqsubseteq S^-,\nonomial)\in\Omc$. 
By induction hypothesis, since $\Imc_i\models(R(b,a),\onomial)$, then $(R(b,a),\onomial)\in\mn{saturate}(\Omc)$. 
Thus, it follows from the construction of $\mn{saturate}(\Omc)$ (by $\CR^A_{5}$) that $(\beta,\monomial)\in\mn{saturate}(\Omc)$.

\item \new{$\R_3$: $(\beta,\monomial)=(B(a),\onomial\times\nonomial)$ and it holds that $\Imc_i\models(A(a),\onomial)$ and $(A\sqsubseteq B, \nonomial)\in\Omc$. 
By induction hypothesis, since $\Imc_i\models(A(a),\onomial)$, then $(A(a),\onomial)\in\mn{saturate}(\Omc)$. 
Hence, it follows from the construction of $\mn{saturate}(\Omc)$ (by $\CR^A_{1}$) that $(\beta,\monomial)\in\mn{saturate}(\Omc)$.}

\item $\R_4$: $(\beta,\monomial)=(B(a),\onomial_1\times\onomial_2\times\nonomial)$ and it holds that $\Imc_i\models(A_1(a),\onomial_1)$, $\Imc_i\models(A_2(a),\onomial_2)$ and $(A_1\sqcap A_2\sqsubseteq B, \nonomial)\in\Omc$. 
By induction hypothesis, since $\Imc_i\models(A_j(a),\onomial_j)$ for $1\leq j\leq 2$, then $(A_j(a),\onomial_j)\in\mn{saturate}(\Omc)$. 
Hence, it follows from the construction of $\mn{saturate}(\Omc)$ (by $\CR^A_{1}$) that $(\beta,\monomial)\in\mn{saturate}(\Omc)$.

\item $\R_5$: $(\beta,\monomial)=(B(a),\onomial\times\onomial'\times\nonomial)$ and it holds that $\Imc_i\models(R(a,x),\onomial)$, $\Imc_i\models(A(x),\onomial')$, and $(\exists R.A\sqsubseteq B, \nonomial)\in\Omc$. 

If \new{$x\in\NI$, since $a\in\individuals{\Omc}$, it must be the case that $x\in\individuals{\Omc}$ by construction of $\Imc_i$ (since the chase rule never adds two individuals in a role interpretation if they are not already related together in some role interpretation). Hence}, by induction hypothesis, since $\Imc_i\models(R(a,x),\onomial)$ and $\Imc_i\models(A(x),\onomial')$, then $(R(a,x),\onomial)\in\mn{saturate}(\Omc)$ and $(A(x),\onomial')\in\mn{saturate}(\Omc)$. 
Hence, it follows from the construction of $\mn{saturate}(\Omc)$ (by $\CR^A_{2}$) that $(\beta,\monomial)\in\mn{saturate}(\Omc)$.

Otherwise, if $x\notin\NI$, $x$ has been introduced during the construction of $\Imc_i$, let us say between $\Imc_{j-1}$ and $\Imc_j$ (with $j\leq i$), to satisfy an inclusion of the form $(C\sqsubseteq \exists S, \monomial_0)$ that belongs to $\Omc$, when applying \new{the chase rule in case} $\R_7$ or $\R_8$ ($S$ can be a role name or an inverse role). 
In this case, there exists a monomial $s$ such that $(a,s)\in C^{\Imc_{j-1}}$ to which \new{the chase rule} has been applied to add $(a,x,s\times\monomial_0)$ to $S^{\Imc_{j-1}}$, so $(C(a),s)\in\mn{saturate}(\Omc)$ by induction. Moreover, by Lemma~\ref{claimcomplELHI}:
\begin{itemize} 
\item there exists $(C\sqcap A_1\sqcap\dots\sqcap A_k\sqsubseteq B,r )\in\mn{saturate}(\Omc)$ for some $k\geq 0$, such that for every $1\leq \ell\leq k$, $(a,s_\ell)\in A_\ell^{\Imc_i}$, so by induction $(A_\ell(a),s_\ell)\in\mn{saturate}(\Omc)$, 
\item and $s\times\Pi_{\ell=1}^k s_\ell\times r= \onomial\times\onomial'\times\nonomial =\monomial$. 
\end{itemize}
 
Hence, $(\beta, {\monomial})\in\mn{saturate}(\Omc)$ by $\CR_{1}^A$.

\item The case $\R_6$ is similar to the case $\R_5$, using $\CR^A_{3}$ instead of $\CR^A_{2}$ in the case $x\in\NI$.

\end{itemize}
We have thus shown that $(\beta, {\monomial})\in\mn{saturate}(\Omc)$ regardless the form of the rule applied between $\Imc_i$ and $\Imc_{i+1}$, which contradicts our original assumption.  
Hence $\Imc_{i+1}\not\models (\beta,\monomial)$. 
We conclude by induction that for every $(\beta,\monomial)$ \new{with $\beta$ of the form $S(a,b)$ or $B(a)$ with $a,b\in\individuals{\Omc}$} such that $(\beta, {\monomial})\notin\mn{saturate}(\Omc)$, $\Imc_i\not\models (\beta,\monomial)$ for every $i\geq 0$, so that $\Imc_\Omc\not\models (\beta,\monomial)$. 
In particular, $\Imc_\Omc\not\models (\alpha,\monomial)$ so $\Omc\not\models (\alpha,\monomial)$.

\medskip
\noindent\new{(2.b) Assume that $\Omc$ is satisfiable and that $\Omc\models (A(c),\monomial)$ with $c\in\NI\setminus\individuals{\Omc}$. It is easy to check that this implies that $\Omc\models (A(d),\monomial)$ for every $d\in\NI$, in particular, $\Omc\models (A(a_\top),\monomial)$. Consider $\Omc'=\Omc\cup\{(\mi{Top}(a_\top),1)\}$ where $\mi{Top}$ is a fresh concept name. Clearly, $\Omc$ and $\Omc'$ entail the same annotated assertions, except for $(\mi{Top}(a_\top),1)$ which is entailed by $\Omc'$ but not by $\Omc$. In particular, $\Omc'\models (A(a_\top),\monomial)$. Hence, by point (2.a), $(A(a_\top),\monomial)\in\mn{saturate}(\Omc')$. Moreover, it is easy to see that $\mn{saturate}(\Omc')=\mn{saturate}(\Omc)\blue{\cup\{(\mi{Top}(a_\top),1)\}}\cup\Bmc_\top$ where $\Bmc_\top$ contains only assertions on the fresh element $b_\top\notin\individuals{\Omc'}$ introduced by the initialization of the completion algorithm for $\Omc'$ (which is different from $a_\top\in\individuals{\Omc'}\setminus\individuals{\Omc}$). It follows that $(A(a_\top),\monomial)\in\mn{saturate}(\Omc)$.}

\medskip
\noindent(3) Assume that $\Omc$ is unsatisfiable. We show that there exists $(\bot(a),\nonomial)\in\mn{saturate}(\Omc)$ where $a\in\individuals{\Omc}\new{\cup\{a_\top\}}$. 

\new{Let $\Omc'$ be the set of annotated axioms obtained from $\mn{saturate}(\Omc)$ as follows: remove all assertions on $a_\top$; replace $\bot$ by a fresh concept name $\mi{Bot}$ in the assertions and GCIs; remove all negative RIs. 
Since $\bot$ does not occur in $\Omc'$, $\Omc'$ is satisfiable, and we have $\individuals{\Omc'}=\individuals{\Omc}$. 
We can show that $\mn{saturate}(\Omc')$ is as follows, assuming w.l.o.g.~that the same fresh individual $a_\top$ is used when initializing the saturation algorithm for $\Omc$ and $\Omc'$ (note that $\signature{\Omc'}=\signature{\Omc'}\cup\{\mi{Bot}\}$ and that $(\mi{Bot}\sqsubseteq\mi{Bot},1)\in\Omc'$ because $(\bot\sqsubseteq\bot,1)\in\mn{saturate}(\Omc)$).
\begin{align*}
(\dagger)\ \mn{saturate}(\Omc')=&\Omc'\cup\{(\exists R.\bot\sqsubseteq \bot,1),(\exists R^-.\bot\sqsubseteq\bot,1)\mid R\in\NR\cap\signature{\Omc}\}\cup \{(\top(a_\top),1)\}
\\&\cup\{(A(a_\top),\monomial)\mid (A(a_\top),\monomial)\in\mn{saturate}(\Omc),  A\in\NC\}
\\&\cup\{(\mi{Bot}(a_\top),\monomial)\mid (\bot(a_\top),\monomial)\in\mn{saturate}(\Omc)\}.
\end{align*}
Indeed, the saturation rules from Table \ref{tab:completionRules} treat $\bot$ in the same way as a concept name when it occurs in a GCI that appears as a premise of the rule, and $(\bot(a_\top),\monomial)$ can be added to $\mn{saturate}(\Omc)$ only by $\CR^A_1$-$\CR^A_3$ (and not by $\CR^A_6$ and $\CR^A_7$) since the completion rules will never add $a_\top$ in a role assertion.}
 
\new{Let $\Imc'$ be the canonical model of $\Omc'$ (we extend straightforwardly the definition of the canonical model of an annotated ontology to the canonical model of a set of annotated axioms since the unicity of axiom annotation does not play any role for the canonical model).  
Since $\Omc$ is unsatisfiable, $\Omc$ does not have any model so $\Imc'$ is not a model of $\Omc$. Since $\Omc'$ contains all axioms of $\Omc$ except those with $\bot$ as right-hand side}, it follows that $\Imc'$ does not satisfy (i) a GCI of $\Omc$ with $\bot$ as right-hand side or (ii) a negative RI of $\Omc$. 

\begin{itemize}
\item Case (i): $\Imc'$ does not satisfy a GCI $(D\sqsubseteq \bot,v)\in\Omc$, \ie there exists $(d,\monomial)\in D^{\Imc'}$. By construction of $\Omc'$, $(D\sqsubseteq \mi{Bot},v)\in\Omc'$ so by construction of $\Imc'$, $(d,\monomial\times v)\in \mi{Bot}^{\Imc'}$. 

 First note that 
if $d\notin\NI$, then 
$d$ has been introduced during the construction of $\Imc'$, let us say between $\Imc'_{j-1}$ and $\Imc'_j$, to satisfy an inclusion of the form $(C\sqsubseteq \exists P, \monomial_0)$ that belongs to $\Omc'$, hence to $\Omc$ \new{(since completion rules do not introduce axioms of this form)}, when applying \new{the chase rule in case} $\R_7$ or $\R_8$  with some $(c,s)\in C^{\Imc'}$.

It is easy to check by induction that there exist $c_0=a\in\new{\NI}$ and $c_1,\dots, c_k\notin\NI$ such that $k\geq 0$, $d=c_k$ and each $c_i$ has been introduced during the construction of $\Imc'$ to satisfy an inclusion of the form $(C_i\sqsubseteq \exists P_i, \monomial_i)$ that belongs to $\Omc$, when applying \new{the chase rule in case} $\R_7$ or $\R_8$ with $(c_{i-1},s_{i-1})\in C_{i-1}^{\Imc'}$, so that $(c_{i-1},c_i, s_{i-1}\times\monomial_i)\in P_i^{\Imc'}$. 
\new{We show by induction on $k$ that $(d,\monomial\times v)\in \mi{Bot}^{\Imc'}$ implies that there exists $(\mi{Bot}(a), \nonomial)\in\mn{saturate}(\Omc')$ if $a\in\individuals{\Omc}$ and $(\mi{Bot}(a_\top), \nonomial)\in\mn{saturate}(\Omc')$ if $a\in\NI\setminus\individuals{\Omc}$. By $(\dagger)$ and by definition of $\Omc'$, it will follow that $(\bot(a),\nonomial)\in\mn{saturate}(\Omc)$ if $a\in\individuals{\Omc}$ and $(\bot(a_\top), \nonomial)\in\mn{saturate}(\Omc)$ otherwise.}

\begin{itemize}
\item Base case: $k=0$, \ie $d=a\in\new{\NI}$. \new{Since $a\in\NI$ and $\Imc'\models (\mi{Bot}(a),\monomial\times v)$, it follows by Theorem \ref{thm:can-model-main} (straightforwardly extended to use a set of annotated axioms instead of an annotated ontology) that $\Omc'\models(\mi{Bot}(a),\monomial\times v)$.}
\begin{itemize}
\item \new{If $a\in\individuals{\Omc}=\individuals{\Omc'}$, since $\Omc'$ is satisfiable, by point (2.a), we obtain that $(\mi{Bot}(a),\monomial\times v)\in\mn{saturate}(\Omc')$.}
\item \new{Otherwise, if $a\in\NI\setminus\individuals{\Omc}=\NI\setminus\individuals{\Omc'}$, since $\Omc'$ is satisfiable, by point (2.b), we obtain that $(\mi{Bot}(a_\top),\monomial\times v)\in\mn{saturate}(\Omc')$.}
\end{itemize}
\item Induction step: Assume that the property is true for $k$ and consider $(d,\monomial\times v)\in \mi{Bot}^{\Imc'}$ such that there exists $c_0=a\in\new{\NI}$ and $c_1,\dots, c_{k+1}\notin\NI$ with $d=c_{k+1}$ and each $c_i$ has been introduced during the construction of $\Imc'$ to satisfy some $(C_i\sqsubseteq \exists P_i, \monomial_i)\in\Omc$, when applying \new{the chase rule in case} $\R_7$ or $\R_8$ with $(c_{i-1},s_{i-1})\in C_{i-1}^{\Imc'}$, so that $(c_{i-1},c_i, s_{i-1}\times\monomial_i)\in P_i^{\Imc'}$. 
Since $(c_{k},d, s_{k}\times\monomial_{k+1})\in P_{k+1}^{\Imc'}$, $(d,\monomial\times v)\in \mi{Bot}^{\Imc'}$, and $(\exists P_{k+1}.\bot\sqsubseteq \bot,1)\in\mn{saturate}(\Omc)$ (by initialization of $\Smc$), so that $(\exists P_{k+1}.\mi{Bot}\sqsubseteq \mi{Bot},1)\in\Omc'$, it follows that $(c_k, s_{k}\times\monomial_{k+1}\times\monomial\times v)\in \mi{Bot}^{\Imc'}$. 
Hence, by induction hypothesis, there exists $(\mi{Bot}(a),\nonomial)\in\mn{saturate}(\Omc')$ \new{if $a\in\individuals{\Omc}$ and $(\mi{Bot}(a_\top), \nonomial)\in\mn{saturate}(\Omc')$ otherwise}.
\end{itemize}

\item Case (ii):  $\Imc'$ does not satisfy a RI $(P\sqcap Q\sqsubseteq \bot,v)\in\Omc$, \ie there exists $(c,d,\monomial_1)\in P^{\Imc'}$ and $(c,d,\monomial_2)\in Q^{\Imc'}$. 
\begin{itemize}
\item If \new{$c,d\in\NI$, then by the construction of $\Imc'$, it must be the case that  $c,d\in\individuals{\Omc}$ (since the chase rule never adds two individuals in a role interpretation if they are not already related together in some role interpretation).} Since $\Imc'\models (P(c,d),\monomial_1)$ and $\Imc'\models (Q(c,d),\monomial_2)$, then \new{by Theorem~\ref{thm:can-model-main} $\Omc'$ satisfies $(P(c,d),\monomial_1)$ and $(Q(c,d),\monomial_2)$, and by point (2.a) of the theorem, since $\Omc'$ is satisfiable, $(P(c,d),\monomial_1)$ and $(Q(c,d),\monomial_2)$ are in $\mn{saturate}(\Omc')$ hence in $\mn{saturate}(\Omc)$ by $(\dagger)$ and construction of $\Omc'$.} 
Hence, since $(P\sqcap Q\sqsubseteq\bot,\variable)\in\Omc$, by $\CR^A_{6}$ or $\CR^A_{7}$, $(\bot(c),\variable\times\monomial_1\times\monomial_2)\in\mn{saturate}(\Omc)$.

\item If $c$ or $d$ is not in $\NI$,  assume that $d\notin\NI$ has been introduced after $c$ in the construction of $\Imc'$ (the case where $c$ has been introduced after $d$ is similar). 
Hence $d$ has been introduced during the construction of $\Imc'$, let us say between $\Imc'_{j-1}$ and $\Imc'_j$, to satisfy an inclusion of the form $(C\sqsubseteq \exists S, \monomial_0)$ or $(C\sqsubseteq \exists S^-, \monomial_0)$ that belongs to $\Omc$, when applying \new{the chase rule in case} $\R_7$ or $\R_8$ respectively. 
Let us consider the case $(C\sqsubseteq \exists S, \monomial_0)$ (the case $(C\sqsubseteq \exists S^-, \monomial_0)$ is similar). 
We can show that there exist
\begin{itemize}
\item $(c,s)\in C^{\Imc'}$,
\item $(S\sqsubseteq P, \onomial_1)\in\Omc'$ hence in $\mn{saturate}(\Omc)$ such that $s\times\monomial_0\times\onomial_1=\monomial_1$, and 
\item $(S\sqsubseteq Q, \onomial_2)\in\Omc'$ hence in $\mn{saturate}(\Omc)$ such that $s\times\monomial_0\times\onomial_2=\monomial_2$.
\end{itemize}
By $\CR^T_{0}$, since $(P\sqcap Q\sqsubseteq\bot,\variable)\in\mn{saturate}(\Omc)$, it follows that $(C\sqsubseteq \bot, \monomial_0\times\onomial_1\times\onomial_2\times\variable)\in \mn{saturate}(\Omc)$. 
Hence $(C\sqsubseteq \mi{Bot}, \monomial_0\times\onomial_1\times\onomial_2\times\variable)\in \Omc'$ and by the construction of $\Imc'$, $(c, s\times\monomial_0\times\onomial_1\times\onomial_2\times\variable) \in \mi{Bot}^{\Imc'}$, \ie $(c,\monomial_1\times\monomial_2\times\variable)\in\mi{Bot}^{\Imc'}$. We can now use the argument of case (i) to conclude that there exists $a\in\individuals{\Omc}\new{\cup\{a_\top\}}$ such that $(\bot(a),\nonomial)\in\mn{saturate}(\Omc)$ for some $\nonomial$.
 \end{itemize}
 \end{itemize}

It follows that in all cases, there exists $(\bot(a),\nonomial)\in\mn{saturate}(\Omc)$ where $\new{a\in\individuals{\Omc}\cup\{a_\top\}}$.

\medskip
\noindent(4) $\mn{saturate}(\Omc)$ can be computed in exponential time \wrt the size of $\Omc$. Indeed, every annotated axiom $(\alpha,\monomial)\in\mn{saturate}(\Omc)\setminus\Omc$ is such that $\monomial$ has at most $\mn{Card}(\Omc)$ variables 
\new{(since they represent sets of variables by definition of $\why$)},
so the number of possible $\monomial$ is exponential in the size of $\Omc$ \new{(more precisely, there are $2 ^{\mn{Card}(\Omc)}$ monomials since each monomial is obtained by choosing for each variable whether it is in the product or not)}, and $\alpha$ is of one of the following form:
\begin{itemize}
\item a concept or role assertion built from individual names and predicates that occurs in $\Omc$ \new{(plus $a_\top$, $\top$, $\bot$)}, and there are polynomially many such assertions,
\item a RI $P_1\sqsubseteq P_2$ where $P_1, P_2$ or their inverses occur in $\Omc$, and there are polynomially many such RIs,
\item a negative RI $P_1\sqcap P_2\sqsubseteq \bot$ or a GCI of the form $\exists P.\bot\sqsubseteq \bot$ where $P_1, P_2$ or their inverses occur in $\Omc$, and there are polynomially many such axioms,
\item a GCI $A_1\sqcap\dots\sqcap A_k\sqsubseteq B$, where $A_1,\dots,A_k$ are concept names that occur in $\Omc$ (or $k=1$ and $A_1=\top$) and $B$ is a concept name that occurs in $\Omc$ or $\bot$. Moreover, \new{each concept name can occur at most $\mn{Card}(\Omc)$ times in $A_1\sqcap\dots\sqcap A_k$ (since conjunctions are treated as multisets with maximal multiplicity $\mn{Card}(\Omc)$). Hence the number of possible left-hand sides $A_1\sqcap\dots\sqcap A_k$ is bounded by $K^N$ where $K$ is the biggest possible $k$ and $N$ is the number of concept names that occur in $\Omc$ plus 1, since each such left-hand side can be obtained by choosing for each $1\leq i\leq K$ one concept name or $\top$. 
The number of such GCIs is thus bounded by $N*K^N$ (since there are $N$ possibilities for the right-hand side $B$). 
Since $N\leq |\signature{\Omc}|+1\leq |\Omc|$ and $K\leq |\signature{\Omc}|*\mn{Card}(\Omc)\leq |\Omc|^2$, then $N*K^N\leq |\Omc|^{2|\Omc|+1}$.}
\end{itemize}
\new{Hence, in the worst case, $\mn{saturate}(\Omc)$ adds to $\Omc$ $O(|\Omc|^{2|\Omc|+1})$ axioms where each axiom may be annotated with $2^{\mn{Card}(\Omc)}\leq 2^{|\Omc|}$ monomials. The total number of annotated axioms added is thus at most exponential in the size of $\Omc$ (there are \blue{$O(e^{p_1(|\Omc|)})$ possible annotated axioms where $p_1$ is a polynomial function}).} 

\blue{Moreover, the number of premises of the rules is bounded by $p_2(|\Omc|)$ for some polynomial function $p_2$ (since we have seen that the size of the conjunctions in the left-hand side of GCIs is bounded by $|\Omc|^2$), so for each rule, the number of rule instantiations \wrt the axioms already derived is bounded by $(e^{p_1(|\Omc|)})^{p_2(|\Omc|)}=e^{p_3(|\Omc|)}$ where $p_3(x)=p_1(x)*p_2(x)$ is still a polynomial function.}

\blue{Hence, for each annotated axiom in $\mn{saturate}(\Omc)$, there have been at most $O(e^{p_3(|\Omc|)})$ rule instantiation evaluations during the step that added this axiom, so we can bound the total run time of the algorithm by $O(e^{p_1(|\Omc|)+p_3(|\Omc|)})=O(e^{p(|\Omc|)})$, where $p(x)=p_1(x)+p_3(x)$ is still a polynomial function.} 
It follows that the algorithm runs in exponential time.
\end{proof}

\paragraph{Proof of Theorem~\ref{prop:completionalgorithmELHIrestr}} 
Recall that $\Omc$ belongs to $\ELHIbotrestr$ if $\Omc$ is an $\ELHIbot$ ontology in normal form such that  \new{if $C\sqsubseteq \exists P_1\in\Omc$, $P_1\sqsubseteq_\Omc P_2$, and $\exists\mn{inv}(P_2).A\sqsubseteq B \in \Omc$, then $A=\top$.}  
In this case, we modify the completion algorithm by restricting the rules from Table~\ref{tab:completionRules} that may introduce exponentially many axioms, and add two additional rules $\CR^T_{4}$ and $\CR^T_{5}$ (see Figure~\ref{fig:completion-rules-restr}). 
The following lemmas are the counterparts  of Lemmas~\ref{lem:canonical-saturation-GCIs} and \ref{claimcomplELHI} for the $\ELHIbotrestr$ case.  
We consider a \why-annotated $\ELHIbotrestr$ ontology $\Omc$,  its canonical model $\Imc_\Omc=\bigcup_{i\geq 0}\Imc_i$ (\cf Figure~\ref{fig:rules-canonical-why-normal}) and the set $\mn{saturate}(\Omc)$ obtained from $\Omc$ using the modified completion rules of Figure~\ref{fig:completion-rules-restr}.

\begin{lemma}\label{lem:canonical-saturation-GCIs-bis}
For all $i\geq 0$, $x,y\in\Delta^{\Imc_\Omc}$ and $A\in\NC\cup\{\top\}$  if $(y,\onomial')\in A^{\Imc_i}$ and $y\notin\NI$ has been introduced between $\Imc_{j-1}$ and $\Imc_j$ ($j\leq i$) to satisfy an inclusion of the form $(C\sqsubseteq \exists S, \monomial_0)\in\Omc$ by applying \new{the chase rule in case} $\R_7$ or $\R_8$ to some $(x,s)\in C^{\Imc_{j-1}}$, 
then the following holds.
\begin{enumerate}
\item There exist $k\geq 0$ and $k'\geq 0$ such that $k+k'\leq 2$ and:
\begin{itemize}								
\item $(S\sqsubseteq P_i, \monomial_{S\sqsubseteq P_i})$ and  $(\exists \mn{inv}(P_i).\top\sqsubseteq B_i, \monomial_{\exists \mn{inv}(P_i).\top\sqsubseteq B_i})\in\mn{saturate}(\Omc)$, $1\leq i\leq k$,  

\item $(\top\sqsubseteq B'_i, \monomial_{\top\sqsubseteq B'_i})\in\mn{saturate}(\Omc)$, $1\leq i\leq k'$, 

\item $(B_1\sqcap\dots\sqcap B_k\sqcap B'_1\sqcap\dots\sqcap B'_{k'}\sqsubseteq A, p)\in\mn{saturate}(\Omc)$ (note that the conjunction contains at most two concept names since $k+k'\leq 2$).
\end{itemize}

Moreover, if $A=\top$, then $k=k'=0$ and $p=1$ (recall that the empty conjunction is $\top$ and $(\top\sqsubseteq\top,1)\in\mn{saturate}(\Omc)$).
 
\item The monomials are related as follows:
\begin{itemize}
\item If $k\neq 0$,  $s\times\monomial_0\times \Pi_{i=1}^k (\monomial_{S\sqsubseteq P_i}\times \monomial_{\exists \mn{inv}(P_i).\top\sqsubseteq B_i})\times\Pi_{i=1}^{k'}\monomial_{\top\sqsubseteq B'_i}\times p =\onomial'$.
\item If $k=0$, 
$\Pi_{i=1}^{k'}\monomial_{\top\sqsubseteq B'_i}\times p =\onomial'$.
\end{itemize}
\end{enumerate} 
\end{lemma}
\begin{proof}
The proof is by induction on $l=i-j$. In the case where $A=\top$, the property trivially holds with $k=k'=0$ for every $l\geq 0$. 
\smallskip

\noindent\emph{Base case.} 
In the case where $A\neq \top$, the base case is $l=1$ and is exactly as in the proof of Lemma \ref{lem:canonical-saturation-GCIs}. 

\noindent\emph{Induction step.} 
Assume that the property is true for every integer up to $l$ and consider the case where $i-j=l+1$. 
Let $x,y\in\Delta^{\Imc_\Omc}$, and assume that $(y,\onomial')\in A^{\Imc_i}$, and $y\notin\NI$ has been introduced between $\Imc_{j-1}$ and $\Imc_j$, to satisfy $(C\sqsubseteq \exists S, \monomial_0)\in\Omc$ by applying \new{the chase rule in case} $\R_7$ or $\R_8$ to $(x,s)\in C^{\Imc_{j-1}}$ (thus adding $(x,y,s\times\monomial_0)\in S^{\Imc_{j}}$). 
We make a case analysis on the last \new{chase} rule applied to add $(y,\onomial')\in A^{\Imc_i}$. 

\smallskip
\new{\textbf{$\boldsymbol{\R_3}$: $\boldsymbol{(y,\onomial')\in A^{\Imc_i}}$ added by applying the chase rule with $\boldsymbol{(D\sqsubseteq A, \monomial_{D\sqsubseteq A})\in\Omc}.$} 
\\
There are $(y,\onomial)\in D^{\Imc_{i-1}}$ and $\onomial'=\monomial_{D\sqsubseteq A}\times \onomial$. 
Since $(y,\onomial)\in D^{\Imc_{i-1}}$, by induction hypothesis: 
\begin{itemize}
\item There exist $k\geq 0$ and $k'\geq 0$ such that $k+k'\leq 2$ and: 
\begin{itemize}  								
\item $(S\sqsubseteq P_i, \monomial_{S\sqsubseteq P_i})$,  $(\exists \mn{inv}(P_i).\top\sqsubseteq B_i, \monomial_{\exists \mn{inv}(P_i).\top\sqsubseteq B_i})\in\mn{saturate}(\Omc)$, \mbox{$1\leq i\leq k$,}  
\item $(\top\sqsubseteq B'_i, \monomial_{\top\sqsubseteq B'_i})\in\mn{saturate}(\Omc)$, $1\leq i\leq k'$, 
\item $(B_1\sqcap\dots\sqcap B_{k_j}\sqcap B'_1\sqcap\dots\sqcap B'_{k'}\sqsubseteq D, p')\in\mn{saturate}(\Omc)$.
\end{itemize} 
\item The monomials are related as follows
\begin{itemize}
\item If $k\neq 0$, $s\times\monomial_0\times \Pi_{i=1}^{k} (\monomial_{S\sqsubseteq P_i}\times \monomial_{\exists \mn{inv}(P_i).\top\sqsubseteq B_i})\times\Pi_{i=1}^{k'}\monomial_{\top\sqsubseteq B'_i}\times p' =\onomial$.
\item If $k=0$, $\Pi_{i=1}^{k'}\monomial_{\top\sqsubseteq B'_i}\times p' =\onomial$.
\end{itemize}
\end{itemize} 
It follows that: 
\begin{itemize}
\item If $k\geq 1$:
since $(S\sqsubseteq P_i, \monomial_{S\sqsubseteq P_i})\in\mn{saturate}(\Omc)$, then by Lemma~\ref{lem:inverse-roles-completion} (whose proof only uses the completion rule $\CR^{T}_1$ which is not modified for $\ELHIbotrestr$, \cf Figure~\ref{fig:completion-rules-restr}), we have that $(\mn{inv}(S)\sqsubseteq \mn{inv}(P_i), \monomial_{S\sqsubseteq P_i})\in\mn{saturate}(\Omc)$. 
Thus, since $(\exists \mn{inv}(P_i).\top\sqsubseteq B_i, \monomial_{\exists \mn{inv}(P_i).\top\sqsubseteq B_i})$,  $(\top\sqsubseteq B'_i, \monomial_{\top\sqsubseteq B'_i})$, and $(B_1\sqcap\dots\sqcap B_{k}\sqcap B'_1\sqcap\dots\sqcap B'_{k'}\sqsubseteq D, p')$ are in $\mn{saturate}(\Omc)$ (recall that $k+k'\leq 2$), then by $\CR^T_{4}$, $(\exists \mn{inv}(S).\top\sqsubseteq D, \monomial_{\exists \mn{inv}(S).\top\sqsubseteq D} )\in\mn{saturate}(\Omc)$ where $\monomial_{\exists \mn{inv}(S).\top\sqsubseteq D}=\Pi_{i=1}^{k}(\monomial_{S\sqsubseteq P_i}\times\monomial_{\exists \mn{inv}(P_i).\top\sqsubseteq B_i})\times\Pi_{i=1}^{k'}\monomial_{\top\sqsubseteq B'_i}\times p'$.
\item If $k=0$: since $(\top\sqsubseteq B'_i, \monomial_{\top\sqsubseteq B'_i})$, $1\leq i\leq k'$, and $(B'_1\sqcap\dots\sqcap B'_{k_j'}\sqsubseteq D, p')$ are in $\mn{saturate}(\Omc)$ (recall that $k'\leq 2$), 
then by $\CR^T_{5}$, $(\top\sqsubseteq D, \monomial_{\top\sqsubseteq D})\in\mn{saturate}(\Omc)$ where $ \monomial_{\top\sqsubseteq D} = \Pi_{i=1}^{k'}\monomial_{\top\sqsubseteq B'_i}\times p$. 
\end{itemize}
}
 
\new{This shows item 1 of the property: by taking $K=1$ and $K'=0$ if $k\geq 1$, and $K=0$, $K'=1$ if $k=0$, it holds that: 
\begin{itemize}
\item $(S\sqsubseteq S, 1)$ and  $(\exists \mn{inv}(S).\top\sqsubseteq D, \monomial_{\exists \mn{inv}(S).\top\sqsubseteq D} )\in\mn{saturate}(\Omc)$, $1\leq i\leq K$, (take $P_i=S$, $\monomial_{S\sqsubseteq P_i}=1$, and $B_i=D$)
\item $(\top\sqsubseteq D, \monomial_{\top\sqsubseteq D})\in\mn{saturate}(\Omc)$, $1\leq i\leq K'$, (take $B'_i=D$)
\item $(D\sqsubseteq A, p)\in\mn{saturate}(\Omc)$ with $p= \monomial_{D\sqsubseteq A}$ (note that $K+K'=1$).
\end{itemize}}

\new{To show item 2 of the property, we observe that if $K\neq 0$ (\ie $k\neq 0$), then 
\begin{align*}
\onomial'=&\monomial_{D\sqsubseteq A}\times \onomial =p\times \onomial
\\
=&p\times s\times\monomial_0\times( \Pi_{i=1}^{k} (\monomial_{S\sqsubseteq P_i}\times \monomial_{\exists \mn{inv}(P_i).\top\sqsubseteq B_i})\times\Pi_{i=1}^{k'}\monomial_{\top\sqsubseteq B'_i}\times p' )
\\
=&s\times\monomial_0\times p\times\Pi_{i=1}^K (1\times \monomial_{\exists \mn{inv}(S).\top\sqsubseteq D})\times\Pi_{i=1}^{K'}\monomial_{\top\sqsubseteq D},
\end{align*} 
and if $K=0$ (\ie $k= 0$ and $K'=1$)
\begin{align*}
\onomial'=\monomial_{D\sqsubseteq A}\times \onomial
=p\times \onomial
=p\times (\Pi_{i=1}^{k'}\monomial_{\top\sqsubseteq B'_i}\times p' )
= p\times\Pi_{i=1}^{K'}\monomial_{\top\sqsubseteq D}.
\end{align*} 
This shows the property in this case.}

\smallskip
\new{\textbf{$\boldsymbol{\R_4}$: $\boldsymbol{(y,\onomial')\in A^{\Imc_i}}$ added by applying the chase rule with}}   
		 \new{\textbf{$\boldsymbol{(D_1\sqcap D_2\sqsubseteq A, \monomial_{D_1\sqcap D_2\sqsubseteq A})\in\Omc}.$}}
		 \\ 
There are $(y,\onomial_1)\in D_1^{\Imc_{i-1}}$ and $(y,\onomial_2)\in D_2^{\Imc_{i-1}}$ and $\onomial'=\monomial_{D_1\sqcap D_2\sqsubseteq A}\times \onomial_1\times\onomial_2$. 
For $j\in\{1,2\}$, since $(y,\onomial_j)\in D_j^{\Imc_{i-1}}$, by induction hypothesis: 

\begin{itemize}
\item There exist $k_j\geq 0$ and $k_j'\geq 0$ such that $k_j+k'_j\leq 2$ and: 
\begin{itemize}  
								
\item $(S\sqsubseteq P^j_i, \monomial_{S\sqsubseteq P^j_i})$,  $(\exists \mn{inv}(P_i^j).\top\sqsubseteq B^j_i, \monomial_{\exists \mn{inv}(P_i^j).\top\sqsubseteq B^j_i})\in\mn{saturate}(\Omc)$, \mbox{$1\leq i\leq k_j$,}  

\item $(\top\sqsubseteq B'^j_i, \monomial_{\top\sqsubseteq B'^j_i})\in\mn{saturate}(\Omc)$, $1\leq i\leq k_j'$, 

\item $(B^j_1\sqcap\dots\sqcap B^j_{k_j}\sqcap B'^j_1\sqcap\dots\sqcap B'^j_{k_j'}\sqsubseteq D_j, p_j)\in\mn{saturate}(\Omc)$.
\end{itemize} 
\item The monomials are related as follows
\begin{itemize}
\item If $k_j\neq 0$, $s\times\monomial_0\times \Pi_{i=1}^{k_j} (\monomial_{S\sqsubseteq P^j_i}\times \monomial_{\exists \mn{inv}(P_i^j).\top\sqsubseteq B^j_i})\times\Pi_{i=1}^{k_j'}\monomial_{\top\sqsubseteq B'^j_i}\times p_j =\onomial_j$.
\item If $k_j=0$, $\Pi_{i=1}^{k_j'}\monomial_{\top\sqsubseteq B'^j_i}\times p_j =\onomial_j$.
\end{itemize}
\end{itemize} 

It follows that for $j\in\{1,2\}$: 
\begin{itemize}
\item If $k_j\geq 1$:
since $(S\sqsubseteq P^j_i, \monomial_{S\sqsubseteq P^j_i})\in\mn{saturate}(\Omc)$, then by Lemma~\ref{lem:inverse-roles-completion}, we have that $(\mn{inv}(S)\sqsubseteq \mn{inv}(P^j_i), \monomial_{S\sqsubseteq P^j_i})\in\mn{saturate}(\Omc)$. 
Since $(\exists \mn{inv}(P_i^j).\top\sqsubseteq B^j_i, \monomial_{\exists \mn{inv}(P_i^j).\top\sqsubseteq B^j_i})$,  $(\top\sqsubseteq B'^j_i, \monomial_{\top\sqsubseteq B'^j_i})$, and $(B^j_1\sqcap\dots\sqcap B^j_{k_j}\sqcap B'^j_1\sqcap\dots\sqcap B'^j_{k_j'}\sqsubseteq D_j, p_j)$ are in $\mn{saturate}(\Omc)$ (recall that $k_j+k'_j\leq 2$), then by $\CR^T_{4}$, $(\exists \mn{inv}(S).\top\sqsubseteq D_j, \monomial_{\exists \mn{inv}(S).\top\sqsubseteq D_j} )\in\mn{saturate}(\Omc)$ where $\monomial_{\exists \mn{inv}(S).\top\sqsubseteq D_j}=\Pi_{i=1}^{k_j}(\monomial_{S\sqsubseteq P^j_i}\times\monomial_{\exists \mn{inv}(P_i^j).\top\sqsubseteq B^j_i})\times\Pi_{i=1}^{k'_j}\monomial_{\top\sqsubseteq B'^j_i}\times p_j$.

\item If $k_j=0$: since $(\top\sqsubseteq B'^j_i, \monomial_{\top\sqsubseteq B'^j_i})$, $1\leq i\leq k_j'$, and $(B'^j_1\sqcap\dots\sqcap B'^j_{k_j'}\sqsubseteq D_j, p_j)$ are in $\mn{saturate}(\Omc)$ (recall that $k'_j\leq 2$), 
then by $\CR^T_{5}$, $(\top\sqsubseteq D_j, \monomial_{\top\sqsubseteq D_j})\in\mn{saturate}(\Omc)$ where $ \monomial_{\top\sqsubseteq D_j} = \Pi_{i=1}^{k'_j}\monomial_{\top\sqsubseteq B'^j_i}\times p_j$. 
\end{itemize}
Let $K= 0$ if $k_1=k_2=0$, $K=2$ if $k_1\neq 0$ and $k_2\neq 0$, and $K=1$ otherwise, 
and let $K'=2-K$. Note that $K,K'\geq 0$ and $K+K'\leq 2$. 
Rename $D_1$ by $E_1$ if $k_1\neq 0$ and $E'_1$ if $k_1=0$, and rename $D_2$ by $E_2$ if $k_1\neq 0$ and $k_2\neq 0$, $E_1$ if $k_1=0$ and $k_2\neq 0$, $E'_1$ if $k_1\neq 0$ and $k_2=0$, and $E'_2$ if $k_1=0$ and $k_2=0$. This shows item 1 of the property: 
\begin{itemize}
\item $(S\sqsubseteq S, 1)$ and  $(\exists \mn{inv}(S).\top\sqsubseteq E_i, \monomial_{\exists \mn{inv}(S).\top\sqsubseteq E_i} )\in\mn{saturate}(\Omc)$, $1\leq i\leq K$, (take $P_i=S$, $\monomial_{S\sqsubseteq P_i}=1$, and $B_i=E_i$)
\item $(\top\sqsubseteq E'_i, \monomial_{\top\sqsubseteq E'_i})\in\mn{saturate}(\Omc)$, $1\leq i\leq K'$, (take $B'_i=E'_i$)
\item $(E_1\sqcap\dots\sqcap E_K\sqcap E'_1\sqcap\dots\sqcap E'_{K'}\sqsubseteq A, p)\in\mn{saturate}(\Omc)$ where $E_1\sqcap\dots\sqcap E_K\sqcap E'_1\sqcap\dots\sqcap E'_{K'}=D_1\sqcap D_2$ and $p= \monomial_{D_1\sqcap D_2\sqsubseteq A}$.
\end{itemize}

To show item 2 of the property, we observe that if $K\neq 0$ (\ie $k_1\neq 0$ or $k_2\neq\emptyset$), then 
\begin{align*}
\onomial'={}&\monomial_{D_1\sqcap D_2\sqsubseteq A}\times \onomial_1\times\onomial_2
\\
{}={}&p\times \onomial_1\times\onomial_2
\\
{}={}&p\times s\times\monomial_0\times\Pi_{j=1}^2 ( \Pi_{i=1}^{k_j} (\monomial_{S\sqsubseteq P^j_i}\times \monomial_{\exists \mn{inv}(P_i^j).\top\sqsubseteq B^j_i})\times\Pi_{i=1}^{k_j'}\monomial_{\top\sqsubseteq B'^j_i}\times p_j )
\\
{}={}&s\times\monomial_0\times p\times\Pi_{i=1}^K (1\times \monomial_{\exists \mn{inv}(S).\top\sqsubseteq E_i})\times\Pi_{i=1}^{K'}\monomial_{\top\sqsubseteq E'_i},
\end{align*} 
and if $K=0$ (\ie $k_1=k_2= 0$ and $K'=2$)
\begin{align*}
\onomial'=\monomial_{D_1\sqcap D_2\sqsubseteq A}\times \onomial_1\times\onomial_2
=p\times \onomial_1\times\onomial_2
=p\times \Pi_{j=1}^2 (\Pi_{i=1}^{k_j'}\monomial_{\top\sqsubseteq B'^j_i}\times p_j )
= p\times\Pi_{i=1}^{K'}\monomial_{\top\sqsubseteq E'_i}.
\end{align*} 
This shows the property in this case.
\smallskip

\new{\textbf{$\boldsymbol{\R_5}$ or $\boldsymbol{\R_6}$:  $\boldsymbol{(y,\onomial')\in A^{\Imc_i}}$ added by applying the chase rule  with}} \new{\textbf{$\boldsymbol{(\exists P.D\sqsubseteq A, \monomial_{\exists P.D\sqsubseteq A})\in\Omc}.$}} 
\\
There exist $(y,z,\onomial_{yz})\in P^{\Imc_{i-1}}$ and $(z,\onomial_z)\in D^{\Imc_{i-1}}$ such that $\onomial'=\onomial_{yz}\times \onomial_z\times \monomial_{\exists P.D\sqsubseteq A}$. We distinguish two subcases:
\begin{enumerate}[(i)]
\item either $z=x$, 
\item or $z\neq x$, which implies that $z$ has been introduced between $j'-1$ and $j'$ ($j\leq j'\leq i$) to satisfy an inclusion of the form $(E_0\sqsubseteq \exists S_z, \monomial_z)\in\Omc$ by applying \new{the chase rule in case} $\R_7$ or $\R_8$.
\end{enumerate}
We start with case (i). It holds that 
$(y,x,\onomial_{yx})\in P^{\Imc_{i-1}}$, $(x,\onomial_x)\in D^{\Imc_{i-1}}$ and $\onomial'=\onomial_{yx}\times \onomial_x\times \monomial_{\exists P.D\sqsubseteq A}$. 
Since $(y,x,\onomial_{yx})\in P^{\Imc_{i-1}}$, \ie  $(x,y,\onomial_{yx})\in \mn{inv}(P)^{\Imc_{i-1}}$, $(y,1)\in\top^{\Imc_{i-1}}$ and $y$ has been introduced to satisfy an inclusion of the form $(C\sqsubseteq \exists S,\monomial_0)\in\Omc$ by applying \new{the chase rule in case} $\R_7$ or $\R_8$ to $(x,s)\in C^{\Imc_{j-1}}$, by Lemma \ref{lem:canonical-saturation-RIs} (whose proof only uses the completion rule $\CR^{T}_1$ which is not modified for $\ELHIbotrestr$, \cf Figure~\ref{fig:completion-rules-restr}), there exists $(S\sqsubseteq \mn{inv}(P), \monomial_{S\sqsubseteq \mn{inv}(P)})\in\mn{saturate}(\Omc)$ such that $s\times\monomial_0\times \monomial_{S\sqsubseteq \mn{inv}(P)}=\onomial_{yx}$. 

\new{By Lemma~\ref{lem:inverse-roles-completion}, $(S\sqsubseteq \mn{inv}(P), \monomial_{S\sqsubseteq \mn{inv}(P)})\in\mn{saturate}(\Omc)$ implies $S\sqsubseteq_\Omc \mn{inv}(P)$}. 
\new{Since $(C\sqsubseteq \exists S,\monomial_0)\in\Omc$, $S\sqsubseteq_\Omc \mn{inv}(P)$ and} $(\exists P.D\sqsubseteq A, \monomial_{\exists P.D\sqsubseteq A})\in\mn{saturate}(\Omc)$, it follows that $D=\top$: Indeed, if $(\exists P.D\sqsubseteq A, \monomial_{\exists P.D\sqsubseteq A})$ has been added by the completion algorithm, \ie by $\CR^T_{4}$, $D=\top$ by the definition of $\CR^T_{4}$, and otherwise, $(\exists P.D\sqsubseteq A, \monomial_{\exists P.D\sqsubseteq A})\in\Omc$ so since $\Omc$ belongs to $\ELHIbotrestr$, it must be the case that $D=\top$. 
Since $D=\top$ and $(x,\onomial_x)\in D^{\Imc_{i-1}}$, it must be the case that $\onomial_x=1$. 
This shows the property in case (i), indeed:
\begin{enumerate}
\item If we take $k=1$ and $k'=0$, it holds that:
\begin{itemize} 
								
\item $(S\sqsubseteq \mn{inv}(P), \monomial_{S\sqsubseteq \mn{inv}(P)})$ and $(\exists P.\top\sqsubseteq A, \monomial_{\exists P.\top\sqsubseteq A})\in\mn{saturate}(\Omc)$ (take $P_1=\mn{inv}(P)$, $B_1=A$),

\item $(A\sqsubseteq A, 1)\in\mn{saturate}(\Omc)$ (take $B_1=A$ and $p=1$). 
\end{itemize} 
\item $
\onomial'= \onomial_{yx}\times \onomial_x\times \monomial_{\exists P.\top\sqsubseteq A}
=\onomial_{yx}\times 1\times \monomial_{\exists P.\top\sqsubseteq A}
=s\times\monomial_0\times \monomial_{S\sqsubseteq \mn{inv}(P)}\times \monomial_{\exists P.\top\sqsubseteq A}\times p.$
\end{enumerate} 
We now consider case (ii): $(y,z,\onomial_{yz})\in P^{\Imc_{i-1}}$, $(z,\onomial_z)\in D^{\Imc_{i-1}}$, $\onomial'=\onomial_{yz}\times \onomial_z\times \monomial_{\exists P.D\sqsubseteq A}$ and $z$ has been introduced between $j'-1$ and $j'$ ($j\leq j'\leq i$) by applying \new{the chase rule in case} $\R_7$ or $\R_8$ to satisfy an inclusion of the form $(E_0\sqsubseteq \exists S_z, \monomial_z)\in\Omc$.
There must exist $(y,s_0)\in E_0^{\Imc_{j'-1}}$ to which \new{the chase rule in case} $\R_7$ or $\R_8$ has been applied. Hence by Lemma~\ref{lem:canonical-saturation-RIs} there exists $(S_z\sqsubseteq P, \monomial_{S_z\sqsubseteq P}  )\in\mn{saturate}(\Omc)$ such that   $s_0\times  \monomial_z\times \monomial_{S_z\sqsubseteq P}=\onomial_{yz}$, and 
by induction hypothesis, the following statements hold.
\begin{itemize}
\item There exist $k\geq 0$ and $k'\geq 0$ such that $k+k'\leq 2$ and:
\begin{itemize} 
\item $(S_z\sqsubseteq P_i, \monomial_{S_z\sqsubseteq P_i})$ and  $(\exists \mn{inv}(P_i).\top\sqsubseteq B_i, \monomial_{\exists \mn{inv}(P_i).\top\sqsubseteq B_i})\in\mn{saturate}(\Omc)$, \mbox{$1\leq i\leq k$,}  

\item $(\top\sqsubseteq B'_i, \monomial_{\top\sqsubseteq B'_i})\in\mn{saturate}(\Omc)$, $1\leq i\leq k'$, 

\item $(B_1\sqcap\dots\sqcap B_k\sqcap B'_1\sqcap\dots\sqcap B'_{k'}\sqsubseteq D, p)\in\mn{saturate}(\Omc)$.
\end{itemize} 
Since $(\exists P.D\sqsubseteq A, \monomial_{\exists P.D\sqsubseteq A})\in\Omc$ and $(E_0\sqsubseteq \exists S_z, \monomial_z)\in\Omc$, it follows by the modified version of $\CR^T_{3}$ (\cf Figure~\ref{fig:completion-rules-restr}) that $(E_0\sqsubseteq A,r_z )\in\mn{saturate}(\Omc)$ 
with $$r_z= \monomial_z\times \monomial_{\exists P.D\sqsubseteq A}\times \monomial_{S_z\sqsubseteq P}\times\Pi_{i=1}^k (\monomial_{S_z\sqsubseteq P_i}\times \monomial_{\exists \mn{inv}(P_i).\top\sqsubseteq B_i})\times\Pi_{i=1}^{k'}\monomial_{\top\sqsubseteq B'_i}\times p.$$ 
\item  The monomials are related as follows
\begin{itemize}
\item If $k\neq 0$, $s_0\times \monomial_z\times \Pi_{i=1}^k (\monomial_{S_z\sqsubseteq P_i}\times \monomial_{\exists \mn{inv}(P_i).\top\sqsubseteq B_i})\times\Pi_{i=1}^{k'}\monomial_{\top\sqsubseteq B'_i}\times p =\onomial_{z}$.
\item If $k=0$, $\Pi_{i=1}^{k'}\monomial_{\top\sqsubseteq B'_i}\times p =\onomial_{z}$.
\end{itemize}
\end{itemize} 

Moreover since $(y,s_0)\in E_0^{\Imc_{i-1}}$, by induction hypothesis:
\begin{itemize} 
\item There exist $k_0\geq 0$ and $k_0'\geq 0$ such that $k_0+k_0'\leq 2$ and: 
\begin{itemize}								
\item $(S\sqsubseteq P^0_i, \monomial_{S\sqsubseteq P^0_i})$,  $(\exists \mn{inv}(P_i^0).\top\sqsubseteq B^0_i, \monomial_{\exists \mn{inv}(P_i^0).\top\sqsubseteq B^0_i})\in\mn{saturate}(\Omc)$, \mbox{$1\leq i\leq k_0$,}  

\item $(\top\sqsubseteq B'^0_i, \monomial_{\top\sqsubseteq B'^0_i})\in\mn{saturate}(\Omc)$, $1\leq i\leq k_0'$, 

\item $(B^0_1\sqcap\dots\sqcap B^0_{k_0}\sqcap B'^0_1\sqcap\dots\sqcap B'^0_{k_0'}\sqsubseteq E_0, p_0)\in\mn{saturate}(\Omc)$.
\end{itemize}
\item The monomials are related as follows
\begin{itemize}
\item If $k_0\neq 0$, $s\times\monomial_0\times \Pi_{i=1}^{k_0} (\monomial_{S\sqsubseteq P^0_i}\times \monomial_{\exists \mn{inv}(P_i^0).\top\sqsubseteq B^0_i})\times\Pi_{i=1}^{k_0'}\monomial_{\top\sqsubseteq {B'}^0_i}\times p_0 =s_0$.
\item If $k_0=0$,  $
\Pi_{i=1}^{k_0'}\monomial_{\top\sqsubseteq {B'}^0_i}\times p_0 =s_0$.
\end{itemize}
\end{itemize}

Since $(B^0_1\sqcap\dots\sqcap B^0_{k_0}\sqcap B'^0_1\sqcap\dots\sqcap B'^0_{k_0'}\sqsubseteq E_0, p_0)$ (with $k_0+k_0'\leq 2$) and $(E_0\sqsubseteq A,r_z )$ are in $\mn{saturate}(\Omc)$, by $\CR^T_{2}$ (\cf Figure~\ref{fig:completion-rules-restr}), 
$(B^0_1\sqcap\dots\sqcap B^0_{k_0}\sqcap B'^0_1\sqcap\dots\sqcap B'^0_{k_0'}\sqsubseteq A, p_0\times r_z)\in \mn{saturate}(\Omc)$ and this shows item 1 of the property. 
To get item 2 of the property, observe that 
\begin{align*}
\onomial'={}& \onomial_{yz}\times \onomial_z\times \monomial_{\exists P.D\sqsubseteq A}\\
{}={} & s_0\times  \monomial_z\times \monomial_{S_z\sqsubseteq P}\times  \Pi_{i=1}^k (\monomial_{S_z\sqsubseteq P_i}\times \monomial_{\exists \mn{inv}(P_i).\top\sqsubseteq B_i})\times\Pi_{i=1}^{k'}\monomial_{\top\sqsubseteq B'_i}\times p\times\monomial_{\exists P.D\sqsubseteq A}
\end{align*}
so that if $k_0\neq 0$, then  
\begin{align*}
\onomial'={}& s\times\monomial_0\times \Pi_{i=1}^{k_0} (\monomial_{S\sqsubseteq P^0_i}\times \monomial_{\exists \mn{inv}(P_i^0).\top\sqsubseteq B^0_i})\times\Pi_{i=1}^{k_0'}\monomial_{\top\sqsubseteq {B'}^0_i}\times p_0\\
&\times  \monomial_z\times \monomial_{S_z\sqsubseteq P}\times  \Pi_{i=1}^k (\monomial_{S_z\sqsubseteq P_i}\times \monomial_{\exists \mn{inv}(P_i).\top\sqsubseteq B_i})\times\Pi_{i=1}^{k'}\monomial_{\top\sqsubseteq B'_i}\times p\times\monomial_{\exists P.D\sqsubseteq A}\\
{}={}& s\times\monomial_0\times\Pi_{i=1}^{k_0} (\monomial_{S\sqsubseteq P^0_i}\times \monomial_{\exists \mn{inv}(P_i^0).\top\sqsubseteq B^0_i})\times\Pi_{i=1}^{k_0'}\monomial_{\top\sqsubseteq {B'}^0_i}\times p_0\times r_z
\end{align*}
and if $k_0=0$, then 
\begin{align*}
\onomial'={}& \Pi_{i=1}^{k_0'}\monomial_{\top\sqsubseteq {B'}^0_i}\times p_0\\&
\times  \monomial_z\times \monomial_{S_z\sqsubseteq P}\times  \Pi_{i=1}^k (\monomial_{S_z\sqsubseteq P_i}\times \monomial_{\exists \mn{inv}(P_i).\top\sqsubseteq B_i})\times\Pi_{i=1}^{k'}\monomial_{\top\sqsubseteq B'_i}\times p\times\monomial_{\exists P.D\sqsubseteq A}\\
{}={}& \Pi_{i=1}^{k_0'}\monomial_{\top\sqsubseteq {B'}^0_i}\times p_0\times r_z
\end{align*}
This shows the property in case (ii) and finishes the proof of the lemma.
\end{proof}

\begin{lemma}\label{claimcomplELHI-bis}
For all $i\geq 0$, $x,y\in\Delta^{\Imc_\Omc}$, $R$ role name or inverse role, and $A\in\NC\cup\{\top\}$, if $(x,y,\onomial)\in R^{\Imc_i}$, $(y,\onomial')\in A^{\Imc_i}$, $(\exists R.A\sqsubseteq B,\nonomial)\in\Omc$, and $y\notin\NI$ has been introduced between $\Imc_{j-1}$ and $\Imc_j$ ($j\leq i$), to satisfy an inclusion of the form $(C\sqsubseteq \exists S, \monomial_0)\in\Omc$ by applying \new{the chase rule in case} $\R_7$ or $\R_8$ to some $(x,s)\in C^{\Imc_{j-1}}$, 
then  
\begin{itemize} 
\item there exists $(C\sqsubseteq B,r )\in\mn{saturate}(\Omc)$, 
\item and $s\times r= \onomial\times\onomial'\times\nonomial$. 
\end{itemize} 
\end{lemma}
\begin{proof}
By Lemma \ref{lem:canonical-saturation-RIs} (whose proof only uses the completion rule $\CR^{T}_1$ which is not modified for $\ELHIbotrestr$, \cf Figure~\ref{fig:completion-rules-restr}), there exists $(S\sqsubseteq R, \monomial_{S\sqsubseteq R}  )\in\mn{saturate}(\Omc)$ such that $$s\times \monomial_0\times\monomial_{S\sqsubseteq R}=\onomial,$$ and 
by Lemma \ref{lem:canonical-saturation-GCIs-bis}, 
the following statements hold.
\begin{enumerate}
\item There exist $k\geq 0$ and $k'\geq 0$ such that $k+k'\leq 2$ and:
\begin{itemize}								
\item $(S\sqsubseteq P_i, \monomial_{S\sqsubseteq P_i})$ and  $(\exists \mn{inv}(P_i).\top\sqsubseteq B_i, \monomial_{\exists \mn{inv}(P_i).\top\sqsubseteq B_i})\in\mn{saturate}(\Omc)$, $1\leq i\leq k$,  

\item $(\top\sqsubseteq B'_i, \monomial_{\top\sqsubseteq B'_i})\in\mn{saturate}(\Omc)$, $1\leq i\leq k'$, 

\item $(B_1\sqcap\dots\sqcap B_k\sqcap B'_1\sqcap\dots\sqcap B'_{k'}\sqsubseteq A, p)\in\mn{saturate}(\Omc)$ (note that the conjunction contains at most two concept names since $k+k'\leq 2$).
\end{itemize}

Moreover, if $A=\top$, then $k=k'=0$ and $p=1$.
 
\item The monomials are related as follows:
\begin{itemize}
\item If $k\neq 0$,  $s\times\monomial_0\times \Pi_{i=1}^k (\monomial_{S\sqsubseteq P_i}\times \monomial_{\exists \mn{inv}(P_i).\top\sqsubseteq B_i})\times\Pi_{i=1}^{k'}\monomial_{\top\sqsubseteq B'_i}\times p =\onomial'$.
\item If $k=0$, 
$\Pi_{i=1}^{k'}\monomial_{\top\sqsubseteq B'_i}\times p =\onomial'$.
\end{itemize}
\end{enumerate} 

Hence, if $(\exists R.A\sqsubseteq B,\nonomial)\in\Omc$, since $(C\sqsubseteq \exists S, \monomial_0)\in\Omc$, it follows by the modified version of $\CR^T_{3}$ for $\ELHIbotrestr$ (\cf Figure~\ref{fig:completion-rules-restr}) that $(C\sqsubseteq B,r )\in\mn{saturate}(\Omc)$ with 
$$r= \monomial_0\times \nonomial\times \monomial_{S\sqsubseteq R}\times\Pi_{i=1}^k (\monomial_{S\sqsubseteq P_i}\times \monomial_{\exists \mn{inv}(P_i).\top\sqsubseteq B_i})\times\Pi_{i=1}^{k'}\monomial_{\top\sqsubseteq B'_i}\times p,$$ 
so that $s\times r= \onomial\times\onomial'\times\nonomial$. 
\end{proof}

Recall that $\mn{saturate}^k(\Omc)$ is the $k$-saturation of $\Omc$, \ie the set of annotated axioms obtained from $\Omc$ through the completion algorithm restricted to monomials of size at most~$k$. 

\CorrectnessELHIrestr*
\begin{proof}
\new{The proof is similar to the proof of Theorem~\ref{prop:completionalgorithmELHI} so we focus here on what differs.}
\smallskip

\noindent(1) To show that every \new{annotated assertion} $(\alpha,\monomial)\in\mn{saturate}^k(\Omc)$ is entailed by $\Omc$, since the modified version of $\CR^T_{2}$ and $\CR^T_{3}$ are special cases of the original $\CR^T_{2}$ and $\CR^T_{3}$, we only need to show that adding $\CR^T_{4}$ and $\CR^T_{5}$ preserves the soundness of the completion algorithm. This is easy to check (similarly as we do in the proof of Lemma~\ref{lem:saturation-correct-for-axioms}).
\smallskip

\noindent (2.a) To show that if $\Omc$ is satisfiable and $\Omc\models (\alpha,\monomial)$ with \new{$\alpha$ an assertion of the form $A(a)$ or $R(a,b)$ with $a,b\in\individuals{\Omc}$} and $|\monomial|\leq k$, then $(\alpha,\monomial)\in \mn{saturate}^k(\Omc)$, we show (i) that if $\Omc$ is satisfiable and $\Omc\models (\alpha,\monomial)$, then $(\alpha,\monomial)\in \mn{saturate}(\Omc)$ obtained with the modified rules and (ii) that if $|\monomial|\leq k$ and $(\alpha,\monomial)\in \mn{saturate}(\Omc)$, then $(\alpha,\monomial)\in \mn{saturate}^k(\Omc)$.

For point (i), we can use exactly the same proof as for point~(2.a) of Theorem~\ref{prop:completionalgorithmELHI} except that we use Lemma~\ref{claimcomplELHI-bis} instead of Lemma~\ref{claimcomplELHI} (note that $\CR^T_2$ and $\CR^T_3$ are not used in the proof for point~(2) of Theorem~\ref{prop:completionalgorithmELHI}, except for proving Lemma~\ref{claimcomplELHI}, so that our modifications do not affect the rest of the proof).

Point (ii) follows from the fact that every annotated axiom added by a rule application has at least as many variables as the premises of the rule.
\smallskip

\noindent \new{(2.b) The proof of point (2.b) is exactly the same as in the proof of Theorem~\ref{prop:completionalgorithmELHI}.}

\smallskip

\noindent (3) Every annotated axiom $(\alpha,\monomial)\in\mn{saturate}^k(\Omc)$ is such that $\monomial$ has at most $k$ variables, so the number of possible $\monomial$ is exponential in $k$ and polynomial \wrt $|\Omc|$  \new{(more precisely, bounded by $|\Omc|^k$)}, and $\alpha$ is of one of the following form:
\begin{itemize}
\item a concept or role assertion built from individual names and predicates that occurs in $\Omc$ \new{(plus $a_\top$, $\top$, $\bot$)}, and there are polynomially many such assertions,
\item a RI $P_1\sqsubseteq P_2$ where $P_1, P_2$ or their inverses occur in $\Omc$, and there are polynomially many such RIs,
\item a negative RI $P_1\sqcap P_2\sqsubseteq \bot$ or a GCI of the form $\exists P.\bot\sqsubseteq \bot$ where $P_1, P_2$ or their inverses occur in $\Omc$, and there are polynomially many such axioms,
\item a GCI of the form $B\sqsubseteq D$ or $B_1\sqcap B_2\sqsubseteq D$ with $B_{(i)}\in\NC\cup\{\top\}$ and $D\in\NC\cup\{\bot\}$ occurs in $\Omc$ (thanks to the modifications done to $\CR^T_{2}$ and $\CR^T_{3}$), and there are polynomially many such GCIs, or
\item a GCI of the form $\exists P.\top\sqsubseteq B$ (introduced by $\CR^T_{4}$) where $P$ is a role name or inverse role and $B\in\NC\cup\{\bot\}$, and there are polynomially many such GCIs.
\end{itemize}
Hence, in the worst case, $\mn{saturate}^k(\Omc)$ adds to $\Omc$ a polynomial number of axioms where each axiom may be annotated with $|\Omc|^k$ monomials: the total number of annotated axiom is thus \blue{bounded by $p_1(|\Omc|^k)$ for some polynomial function $p_1$}. 
\blue{Moreover, the number of premises of the rules is bounded by 8, so for each rule, the number of rule instantiations \wrt the axioms already derived is bounded by $p_1(|\Omc|^k)^8$.} 
\blue{Hence, for each annotated axiom in $\mn{saturate}^k(\Omc)$, there have been at most $p_1(|\Omc|^k)^8$ rule instantiation evaluations during the step that added this axiom, so we can bound the total run time of the algorithm by $p_1(|\Omc|^k)^9=p(|\Omc|^k)$ with $p$ a polynomial function.}
\end{proof}

\complexityprovmonomial*
\begin{proof}
We show that the algorithm can be modified by ignoring the monomials in the annotated axioms to check in polynomial time whether $\Omc$ is satisfiable. It follows that we can decide whether $\Omc\models(\alpha,\monomial)$ by first checking that $\Omc$ is satisfiable in polynomial time, then, by Theorem~\ref{prop:completionalgorithmELHIrestr},  computing $\mn{saturate}^{|\monomial|}(\Omc)$ in polynomial time \wrt the size of $\Omc$ and exponential time \wrt $|\monomial|$.

Let $\Omc'$ be the result of applying the completion algorithm using the rules modified by Figure~\ref{fig:completion-rules-restr} while ignoring the monomial part of the annotated axioms. 
It is easy to check that $\Omc'=\{\alpha\mid (\alpha,\monomial)\in \mn{saturate}(\Omc)\}$. 
Moreover, it follows from the proof of point (3) of Theorem~\ref{prop:completionalgorithmELHIrestr} that the algorithm adds a polynomial number of axioms, hence it terminates in polynomial time. 
Finally,  we can show as in the proof of point (3) of Theorem~\ref{prop:completionalgorithmELHI} that $\Omc'$ contains some $\bot(a)$ iff $\Omc$ is unsatisfiable.
\end{proof}


\subsection{Conjunctive Queries over $\ELHIbot$ Ontologies}

\paragraph{Proof of Theorem~\ref{th:CQ-answering-algorithm}} 
We start with point (4) of Theorem~\ref{th:CQ-answering-algorithm} by showing that $\Rew(q,\Omc)$ can be computed in exponential time  in $|\Omc|$ and $|q|$ using Algorithm~\ref{alg:rewrite}\footnote{In Line 9 of Algorithm~\ref{alg:rewrite}, we have $R^-(y,x_0)$ which does not belong to the syntax of CQs. This is just used in the intermediate steps of the algorithm, since they are   replaced by concept atoms of the form $D(y_0)$ in the final set of annotated CQs computed by the algorithm.}, \blue{where \mn{RemoveSuperfluousRepeatedAtoms} is a function that given a BCQ $q$ returns the BCQ obtained by limiting the number of repetitions of each concept atom to $\mn{Card}(\Omc)$ (\ie replace $\underbrace{A(x)\wedge\dots\wedge A(x)}_{\mn{Card}(\Omc)+k \text{ times}}$ by 
$\underbrace{A(x)\wedge\dots\wedge A(x)}_{\mn{Card}(\Omc) \text{ times}}$). In this section, we often rely on the extended version $\ext{q}$ of $q$ to distinguish the different occurrences of an atom in $q$, indeed, if $q=A(x)\wedge A(x)$, $\ext{q}=\exists t_1t_2\ A(x,t_1)\wedge A(x,t_2)$.}

\IncMargin{1em}
\begin{algorithm}[t]
\caption{Compute $\Rew(q,\Omc)$}
\label{alg:rewrite}
\BlankLine
\KwInput{a CQ $q$, an $\ELHIbot$ ontology $\Omc$ in normal form and its saturation  \mn{saturate}(\Omc)} 
\KwOutput{a set $\Rew(q,\Omc)$ of pairs of the form $(q^*,\monomial)$ where $q^*$ is a CQ and $\monomial$ a monomial}
\BlankLine
$Result\leftarrow \{(q,1)\}$\; 
$Rew\leftarrow \emptyset$\; 
\While{$Rew\neq Result$}{
	$Rew\leftarrow Result$\; 
	\ForEach(\tcc*[f]{rewrite $(q_i,\monomial_i)$}){$(q_i,\monomial_i)\in Rew$}{\label{foreachquery}
		\ForEach(\tcc*[f]{(S1)}){$x_0$ existentially quantified in $q_i$ such that there is no $R(x_0,x_0)\in\mn{atoms}(q_i)$}{\label{stepS1}
			$q^*_i\leftarrow q_i$\;
			\ForEach{$R(x_0,y)\in\mn{atoms}(q^*_i)$}{\label{stepS2-1}
				\blue{$q^*_i\leftarrow q^*_i[R(x_0,y)\leftarrow R^-(y,x_0)]$};\tcc*[f]{(S2)}\\
			}\label{stepS2-2}
			$V_p\leftarrow\{y\mid Q(y,x_0)\in q^*_i\text{ for some }Q\}$;\tcc*[f]{(S3)}\\
			$RoleAt(x_0)\leftarrow \{Q(y,x_0)\mid Q(y,x_0)\in q^*_i\}$\;
			$ConceptAt(x_0)\leftarrow \{C(x_0)\mid C(x_0)\in q^*_i\}$\;
			\ForEach(\tcc*[f]{$A\in\NC$ since $\Omc$ in normal form}){$(A\sqsubseteq\exists P,v)\in\Omc$}{\label{stepS3}
				$RoleRew\leftarrow\mn{RoleAtRew}(q^*_i,x_0, A\sqsubseteq \exists P)$;\tcc*[f]{Algorithm~\ref{alg:rolerew}}\\
				$ConceptRew\leftarrow\mn{ConceptAtRew}(q^*_i,x_0,A\sqsubseteq \exists P)$;\tcc*[f]{Algorithm~\ref{alg:conceptrew}}\\			
	\ForEach{$mon_r\in RoleRew$}{\label{foreachmonr}
	\ForEach{$(At,mon_c)\in ConceptRew$}{\label{foreachAtmonc}
				$q'_i\leftarrow q^*_i\setminus (RoleAt(x_0)\cup ConceptAt(x_0))$;\label{stepS4}\tcc*[f]{(S4)}\\
				choose $y_0\in V_p$;\label{stepS5-1}\tcc*[f]{(S5)}\\
				$q'_i\leftarrow q'_i[y'\leftarrow y_0\mid y'\in V_p]$\;\label{stepS5-2}
				$\blue{q'_i\leftarrow q'_i\wedge A(y_0) \bigwedge_{j=1}^{|At|}At[j](y_0)}$;\label{stepS6}\tcc*[f]{(S6)}\\
				$\blue{q'_i\leftarrow \mn{RemoveSuperfluousRepeatedAtoms}(q'_i)}$;\\
				$Result\leftarrow Result\cup\{(q'_i,\monomial_i\times v\times mon_r\times mon_c)\}$\;
				}
				}
				}
			}
		}
	}
Output $Result$\;
\end{algorithm}

\begin{algorithm}[t]
\caption{Compute $\mn{RoleAtRew}(q^*_i,x_0, A\sqsubseteq \exists P)$}\label{alg:rolerew}
\BlankLine
\KwInput{a CQ $q^*_i$, a variable $x_0$, a GCI $A\sqsubseteq \exists P$, $\mn{saturate}(\Omc)$} 
\KwOutput{a set $\mn{RoleAtRew}(q^*_i,x_0, A\sqsubseteq \exists P)$ of monomials}
\BlankLine
$ListMonSets\leftarrow\emptyset$\;
\ForEach{\blue{$Q(y,x_0,t)\in \ext{q^*_i}$}}{\label{algo2-foreachatom}
$MonSetQ\leftarrow\emptyset$;\tcc*[f]{store monomials of RIs that allow to rewrite $Q$ in $P$}\\
\ForEach{$(P\sqsubseteq Q,\monomial_{P\sqsubseteq Q})\in\mn{saturate}(\Omc)$}{\label{algo2-foreachaxiom}
$MonSetQ\leftarrow MonSetQ\cup \{\monomial_{P\sqsubseteq Q}\}$\;
}
\If(\tcc*[f]{$Q$ cannot be rewritten into $P$}){$MonSetQ=\emptyset$}{Return $\emptyset$\;}
$ListMonSets\leftarrow ListMonSets\cup \{MonSetQ\}$
}
$Result\leftarrow\{1\}$;\tcc*[f]{compute all possible products of one monomial per \blue{$Q(y,x_0,t)$}}\\
\ForEach{$MonSetQ\in ListMonSets$}{\label{algo2-foreachmonsetq}
$Res\leftarrow \emptyset$\;
\ForEach{$\nonomial_1\in Result$}{\label{algo2-foreachmoninres}
\ForEach{$\nonomial_2\in MonSetQ$}{\label{algo2-foreachmoninmonsetq}
$Res\leftarrow Res\cup\{\nonomial_1\times\nonomial_2\}$\;
}
}
$Result\leftarrow Res$\;
}
Return $Result$\;
\end{algorithm}

\begin{algorithm}[t]
\caption{Compute $\mn{ConceptAtRew}(q^*_i,x_0, A\sqsubseteq \exists P)$}\label{alg:conceptrew}
\BlankLine
\KwInput{a CQ $q^*_i$, a variable $x_0$, a GCI $A\sqsubseteq \exists P$, $\mn{saturate}(\Omc)$} 
\KwOutput{a set $\mn{ConceptAtRew}(q^*_i,x_0, A\sqsubseteq \exists P)$ of pairs of the form $(At,\monomial)$ where \blue{$At$ is a list of concept names} and $\monomial$ is a monomial}
\BlankLine
$ListPairSets\leftarrow\emptyset$\;
\ForEach{\blue{$C(x_0,t)\in \ext{q^*_i}$}}{\label{alg3-foreachatom}
$PairSetC\leftarrow\emptyset$;\tcc*[f]{store pairs $(At,\monomial)$ that correspond to axioms that allow to rewrite $C$ using $P$}\\
\ForEach{$(B_1\sqcap\dots\sqcap B_p\sqsubseteq C, \nonomial_C)\in\mn{saturate}(\Omc)$}{\label{alg3-foreachaxiom1}
$ListPairSBs\leftarrow \emptyset$\;
\ForEach{$1\leq i\leq p$}{\label{alg3-foreachp}
$PairSB_i\leftarrow\emptyset$;\tcc*[f]{store pairs $(At,\monomial)$ that correspond to axioms that allow to rewrite $B_i$ using $P$}\\
\ForEach{$(P\sqsubseteq P_i,\monomial_i)\in\mn{saturate}(\Omc)$}{\label{alg3-foreachaxiom2}
\ForEach{$(\exists \mn{inv}(P_i).A_i\sqsubseteq B_i,v_i)\in\Omc$}{\label{alg3-foreachaxiom3}
$PairSB_i\leftarrow PairSB_i\cup \{(\blue{[A_i]},v_i\times\monomial_i)\}$\;
}}
\ForEach{$(\top\sqsubseteq B_i,\onomial_i)\in\mn{saturate}(\Omc)$}{\label{alg3-foreachaxiom4}
$PairSB_i\leftarrow PairSB_i\cup\{ (\blue{[]}, \onomial_i) \}$\;
}
$ListPairSBs\leftarrow ListPairSBs\cup\{PairSB_i\}$\;
}
$PairSet\leftarrow\{(\blue{[]}, \nonomial_C)\}$;\tcc*[f]{compute the part of $PairSetC$ that uses $(B_1\sqcap\dots\sqcap B_p\sqsubseteq C, \nonomial_C)$ with all possible choices of one pair per $B_i$}\\
\ForEach{$PairSB_i\in ListPairSBs$}{\label{alg3-foreachlistpairb}
$Res\leftarrow\emptyset$\;
\ForEach{$(At_1,\nonomial_1)\in PairSet$}{\label{alg3-foreachpairset}
\ForEach{$(At_2,\nonomial_2)\in PairSB_i$}{\label{alg3-foreachpairsbi}
$Res\leftarrow Res\cup\{(At_1\blue{\cdot} At_2,\nonomial_1\times\nonomial_2))\}$\;
}
}
$PairSet\leftarrow Res$\;\label{alg3pairsetres}
}
}
$PairSetC\leftarrow PairSetC\cup PairSet$\;
\If(\tcc*[f]{$C$ cannot be rewritten using $P$}){$PairSetC=\emptyset$}{Return $\emptyset$\;}
$ListPairSets\leftarrow ListPairSets\cup\{PairSetC\}$
}
$Result\leftarrow\{(\emptyset, 1)\}$;\tcc*[f]{compute all possible products of one pair per \blue{$C(x_0,t)$}}\\
\ForEach{$PairSetC\in ListPairSets$}{\label{alg3-foreachpairsetc}
$Res\leftarrow\emptyset$\;
\ForEach{$(At_1,\nonomial_1)\in Result$}{\label{alg3-foreachmoninres}
\ForEach{$(At_2,\nonomial_2)\in PairSetC$}{\label{alg3-foreachmoninpairsetC}
$Res\leftarrow Res\cup\{(At_1\blue{\cdot}At_2, \nonomial_1\times\nonomial_2)\}$\;
}
}
$Result\leftarrow Res$\;
}
Return $Result$\;
\end{algorithm}
\DecMargin{1em}

\begin{lemma}\label{lem:algorewfcorrect}
Algorithm~\ref{alg:rewrite} computes $\Rew(q,\Omc)$.
\end{lemma}
\begin{proof}
Let $(q^*,\monomial)\in\Rew(q,\Omc)$. There exists a rewriting sequence $(q_0,\monomial_0)\rightarrow_\Omc (q_1,\monomial_1)\rightarrow_\Omc\dots\rightarrow_\Omc(q_k,\monomial_k)$ where $(q,1)=(q_0,\monomial_0)$ and $(q^*,\monomial)=(q_k,\monomial_k)$. We show by induction that $(q_i,\monomial_i)$ is produced by Algorithm~\ref{alg:rewrite} for every $0\leq i\leq k$. 

\noindent\emph{Base case, $i=0$.} Since Algorithm~\ref{alg:rewrite} initializes $Result$ with $\{(q,1)\}$ and only adds elements to $Result$, it holds that Algorithm~\ref{alg:rewrite} produces $(q_0,\monomial_0)$. 

\noindent\emph{Induction step.} Assume that Algorithm~\ref{alg:rewrite} produces $(q_i,\monomial_i)$ and that $(q_i,\monomial_i)\rightarrow_\Omc(q_{i+1},\monomial_{i+1})$, \ie $(q_{i+1},\monomial_{i+1})$ is obtained from $(q_i,\monomial_i)$ by applying steps (S1) to (S6) of Definition~\ref{def:rewriting}. 
After Algorithm~\ref{alg:rewrite} adds $(q_i,\monomial_i)$ to $Result$, $Rew\neq Result$ so the while-loop is entered. Then, since $(q_i,\monomial_i)\in Result$, it is considered in the foreach-loop in line~\ref{foreachquery} of Algorithm~\ref{alg:rewrite}.  
\begin{itemize}
\item Let $x_0$ be the existentially quantified variable of $q_i$ selected in step (S1): $x_0$ fulfills the conditions to be 
 considered in the foreach-loop in line~\ref{stepS1}.
\item Step (S2) corresponds to lines~\ref{stepS2-1} and~\ref{stepS2-2}.
\item Let $(A\sqsubseteq \exists P,v)\in\Omc$ be the GCI selected by step (S3): $(A\sqsubseteq \exists P,v)$ is considered in the foreach-loop in line~\ref{stepS3}. We show that $At(q_i,x_0,A\sqsubseteq \exists P)$ and $mon(q_i,x_0,A\sqsubseteq \exists P)$ constructed in step (S3) are such that there are $mon_r\in \mn{RoleAtRew}(q^*_i,x_0, A\sqsubseteq \exists P)$ and $(At,mon_c)\in \mn{ConceptAtRew}(q^*_i,x_0, A\sqsubseteq \exists P)$ s.t.\ \blue{$At(q_i,x_0,A\sqsubseteq \exists P)=[A]\cdot At$} and $mon(q_i,x_0,A\sqsubseteq \exists P)=v\times mon_r\times mon_c$. 
This follows from the fact that :
\begin{enumerate}
\item Algorithm~\ref{alg:rolerew} returns the set of all monomials that can be obtained as a product of the form $\prod_{\blue{Q(y,x_0,t)\in \ext{q^*_i}}}\monomial_Q$ where for each $\blue{Q(y,x_0,t)\in \ext{q^*_i}}$, $\monomial_Q\in MonSetQ$ and $MonSetQ=\{\monomial_{P\sqsubseteq Q}\mid (P\sqsubseteq Q,\monomial_{P\sqsubseteq Q})\in\mn{saturate}(\Omc)\}$. In particular, it contains the product $mon_r$ of the $\monomial_{P\sqsubseteq Q}$ selected by step (S3)(a).
\item  Algorithm~\ref{alg:conceptrew} returns the set of all pairs $(At,mon_c)$ that can be obtained as the \blue{concatenation} and product of pairs of \blue{lists} of concept names $At_C$ and monomials $\nonomial_C\times\monomial_C$ for all \blue{$C(x_0,t)\in \ext{q^*_i}$}, such that there is 
$(B_1\sqcap\dots\sqcap B_p\sqsubseteq C, \nonomial_C)\in\mn{saturate}(\Omc)$ and for every $1\leq i\leq p$ either there is a pair $(\exists \mn{inv}(P_i).A_i\sqsubseteq B_i,v_i)\in\Omc$ and $(P\sqsubseteq P_i,\monomial_i)\in\mn{saturate}(\Omc)$ or there is $(\top\sqsubseteq B_i,\onomial_i)\in\mn{saturate}(\Omc)$ such that $At_C$ is the \blue{list} of the $A_i$ and $\monomial_C$ is the product of the $v_i\times\monomial_i$ or $\onomial_i$. 
In particular, some $(At,mon_c)$ corresponds to the GCIs and RIs selected by step (S3)(b).
\end{enumerate}

\item Finally, step (S4) corresponds to line~\ref{stepS4}, step (S5) corresponds to lines~\ref{stepS5-1} and~\ref{stepS5-2} and step (S6) corresponds to line~\ref{stepS6}.
\end{itemize}
Hence, $(q_{i+1},\monomial_{i+1})$  is in the set $Result$ returned by Algorithm~\ref{alg:rewrite}.
\smallskip

In the other direction, assume that $(q^*,\monomial)$ is in the set $Result$ returned by Algorithm~\ref{alg:rewrite}. Since Algorithm~\ref{alg:rewrite} only adds elements to $Result$, we can prove that there exists a rewriting sequence $(q_0,\monomial_0)\rightarrow_\Omc (q_1,\monomial_1)\rightarrow_\Omc\dots\rightarrow_\Omc(q_k,\monomial_k)$ where $(q,1)=(q_0,\monomial_0)$ and $(q^*,\monomial)=(q_k,\monomial_k)$ by induction on the number of iterations of the while-loop of Algorithm~\ref{alg:rewrite} before it produces $(q^*,\monomial)$. 

\noindent\emph{Base case.} If $(q^*,\monomial)$ is in the set returned by Algorithm~\ref{alg:rewrite} when the while-loop has not been entered, $(q^*,\monomial)=(q,1)$ and there exists a rewriting sequence as required (with $k=0$). 

\noindent\emph{Induction step.} Assume that every $(q_i,\monomial_i)$ produced by  Algorithm~\ref{alg:rewrite} after $n$ iterations of the while-loop is such that there exists a rewriting sequence as required. 
Let $(q_{i+1},\monomial_{i+1})$ be added to $Result$ by Algorithm~\ref{alg:rewrite} in the $n+1$ iteration of the while-loop and $(q_i,\monomial_i)\in Rew$ be the query that corresponds to the for-loop (line~\ref{foreachquery}) in which $(q_{i+1},\monomial_{i+1})$ is added. It is sufficient to show that $(q_i,\monomial_i)\rightarrow_\Omc(q_{i+1},\monomial_{i+1})$ to conclude by induction. 

\begin{itemize}
\item Let $x_0$, $(A\sqsubseteq \exists P,v)$, $mon_r$, $(At,mon_c)$ correspond to the foreach-loops (lines~\ref{stepS1},~\ref{stepS3},~\ref{foreachmonr}, and~\ref{foreachAtmonc} respectively) in which $(q_{i+1},\monomial_{i+1})$ is added: it is clear that $x_0$ fulfills the conditions to be selected in step (S1) and we show that $(A\sqsubseteq \exists P,v)$ fulfills those to be selected in step (S3) and that step (S3) can choose $At(q_i,x_0,A\sqsubseteq \exists P)$ and $mon(q_i,x_0,A\sqsubseteq \exists P)$ such that $At(q_i,x_0,A\sqsubseteq \exists P)=\blue{[A]\cdot At}$ and $mon(q_i,x_0,A\sqsubseteq \exists P)=v\times mon_r\times mon_c$.
\begin{enumerate}
\item Since $mon_r\in \mn{RoleAtRew}(q^*_i,x_0, A\sqsubseteq \exists P)$, Algorithm~\ref{alg:rolerew} outputs a non-empty set. Hence for every $\blue{Q(y,x_0,t)\in \ext{q^*_i}}$, $MonSetQ\neq\emptyset$, \ie there exists $(P\sqsubseteq Q,\monomial_{P\sqsubseteq Q})\in\mn{saturate}(\Omc)$. Moreover, $mon_r$ is one of the products built by selecting one such $\monomial_{P\sqsubseteq Q}$ per $\blue{Q(y,x_0,t)\in \ext{q^*_i}}$.

\item Since $(At,mon_c)\in \mn{ConceptAtRew}(q^*_i,x_0, A\sqsubseteq \exists P)$, Algorithm~\ref{alg:conceptrew} outputs a non-empty set. Hence for every $\blue{C(x_0,t)\in \ext{q^*_i}}$, $PairSetC\neq\emptyset$, \ie there exists $(B_1\sqcap\dots\sqcap B_p\sqsubseteq C,\nonomial_C)\in\mn{saturate}(\Omc)$ such that the corresponding $PairSet$ is non-empty when added to $PairSetC$. It follows that for every $1\leq i\leq p$, $PairSB_i\neq\emptyset$ (since otherwise $PairSet$ is replaced by $\emptyset$ in line~\ref{alg3pairsetres}). Finally, $PairSB_i\neq\emptyset$ implies the existence of either $(P\sqsubseteq P_i,\monomial_i)\in\mn{saturate}(\Omc)$ and $(\exists \mn{inv}(P_i).A_i\sqsubseteq B_i, v_i)\in \Omc$ or of $(\top\sqsubseteq B_i,\onomial_i)\in\mn{saturate}(\Omc)$. Moreover, each $PairSetC$ contains all possible pairs $(At_C,\nonomial_C\times\monomial_C)$ that correspond to the choice of some $(B_1\sqcap\dots\sqcap B_p\sqsubseteq C,\nonomial_C)$ and for $1\leq i\leq p$ of either $(P\sqsubseteq P_i,\monomial_i)$ and $(\exists \mn{inv}(P_i).A_i\sqsubseteq B_i, v_i)$ or of $(\top\sqsubseteq B_i,\onomial_i)$, so $(At,mon_c)$ is one of the pairs built by selecting one such pair per $\blue{C(x_0,t)\in \ext{q^*_i}}$ and aggregating them.
\end{enumerate}

\item It is then easy to check that step (S2) is performed by lines~\ref{stepS2-1} and~\ref{stepS2-2}, step (S4) by line~\ref{stepS4}, step (S5) by lines~\ref{stepS5-1} and~\ref{stepS5-2} and step (S6) by line~\ref{stepS6}.
\end{itemize}
It follows that we can obtain $(q_{i+1},\monomial_{i+1})$ from $(q_i,\monomial_i)$ by applying steps (S1) to (S6), \ie \mbox{$(q_i,\monomial_i)\rightarrow_\Omc (q_{i+1},\monomial_{i+1})$.}
\end{proof}

\begin{lemma}\label{lem:size-rew}
For every $(q^*,\monomial)\in \Rew(q,\Omc)$,  the size of $(q^*,\monomial)$ is polynomial \wrt $|\Omc|$ and $|q|$.  
Hence the cardinality of $\Rew(q,\Omc)$ is exponential in $|\Omc|$ and $|q|$. 
\end{lemma}
\begin{proof}
For every $(q^*,\monomial)\in \Rew(q,\Omc)$, (i) $\mn{terms}(q^*)\subseteq \mn{terms}(q)$, (ii) the atoms of $q^*$ use concept and role names that occur in $\Omc$, \blue{and (iii) the atoms of $q^*\setminus q$ can be repeated at most $\mn{Card}(\Omc)$ times. It follows that the number of atoms in $q^*$ is bounded by the number of atoms in $q$ plus $|\signature{\Omc}|\times |\mn{terms}(q)|^2\times \mn{Card}(\Omc)$ and the size of $q^*$ is bounded by $|q|+3(|\signature{\Omc}|\times |\mn{terms}(q)|^2\times \mn{Card}(\Omc))$}.  
Hence the size of $q^*$ is polynomial \wrt $|\Omc|$ and $|q|$. Moreover, since $\monomial$ is a product of variables that occur in $\Omc$, the size of $\monomial$ is linear \wrt $|\Omc|$. 

 It follows that the cardinality of $\{q^*\mid (q^*,\monomial)\in\Rew(q,\Omc)\}$ is exponential in $|\Omc|$ and $|q|$ and the cardinality of $\{\monomial \mid (q^*,\monomial)\in\Rew(q,\Omc)\}$ is exponential in $|\Omc|$.
Hence the cardinality of $\Rew(q,\Omc)$ is exponential in $|\Omc|$ and~$|q|$. 
\end{proof}

\begin{lemma}\label{lem:algorewrole}
Algorithm~\ref{alg:rolerew} runs in exponential time \wrt $|q|$ and $|\Omc|$.
\end{lemma}
\begin{proof}
We show that every foreach-loop of Algorithm~\ref{alg:rolerew} iterates over a set whose size is exponentially bounded \wrt $|\Omc|$ and $|q|$.  
%
In line~\ref{algo2-foreachatom}, by Lemma~\ref{lem:size-rew}, the number of atoms in $q^*_i$ is polynomial \wrt $|\Omc|$ and $|q|$.
In line~\ref{algo2-foreachaxiom}, by Theorem~\ref{prop:completionalgorithmELHI}, $\mn{saturate}(\Omc)$ can be computed in exponential time hence is of exponential size \wrt $|\Omc|$. 
In line~\ref{algo2-foreachmonsetq}, $ListMonSets$ contains a polynomial number (bounded by the size of $q^*_i$) of sets $MonSetQ$. 
In line~\ref{algo2-foreachmoninmonsetq}, each set $MonSetQ$ constructed by the algorithm is a set of monomials from $\mn{saturate}(\Omc)$ hence is of exponential size \wrt $|\Omc|$.
Finally, in line~\ref{algo2-foreachmoninres}, $Result$ is bounded by $\prod_{\blue{Q(y,x_0,t)\in \ext{q^*_i}}} |MonSetQ|$ where $|q^*_i|$ is polynomial \wrt $|\Omc|$ and $|q|$ and $|MonSetQ|$ is exponential \wrt $|\Omc|$, hence $|Result|$ is exponential \wrt $|q|$ and $|\Omc|$. 
\end{proof}

\begin{lemma}\label{lem:algorewconcept}
Algorithm~\ref{alg:conceptrew} runs in exponential time \wrt $|q|$ and $|\Omc|$.
\end{lemma}
\begin{proof}
We show that every foreach-loop of Algorithm~\ref{alg:conceptrew} iterates over a set whose size is exponentially bounded \wrt $|\Omc|$ and $|q|$. 
In line~\ref{alg3-foreachatom}, by Lemma~\ref{lem:size-rew}, the number of atoms in $q^*_i$ is polynomial \wrt $|\Omc|$ and $|q|$. 
In lines~\ref{alg3-foreachaxiom1},~\ref{alg3-foreachaxiom2},~\ref{alg3-foreachaxiom3},~\ref{alg3-foreachaxiom4}, by Theorem~\ref{prop:completionalgorithmELHI}, $\mn{saturate}(\Omc)$ can be computed in exponential time hence is of exponential size \wrt $|\Omc|$. 
 In line~\ref{alg3-foreachp}, for every $(B_1\sqcap\dots\sqcap B_p\sqsubseteq C,\nonomial_C)\in\mn{saturate}(\Omc)$, \new{$p$ is bounded by $|\signature{\Omc}|*\mn{Card}(\Omc)$, hence is polynomial in $|\Omc|$}. 
 In line~\ref{alg3-foreachlistpairb}, each $ListPairSBs$ constructed by the algorithm is a set of $p$ sets for some $(B_1\sqcap\dots\sqcap B_p\sqsubseteq C,\nonomial_C)\in\mn{saturate}(\Omc)$, so contains a polynomial number of sets. 
In line~\ref{alg3-foreachpairsbi}, for each $PairSB_i$ constructed by the algorithm, $PairSB_i$ is a set of pairs that consists of a \blue{list} of 0 or 1 concept name that occurs in $\Omc$ and a product of one or two monomials that occur in $\mn{saturate}(\Omc)$, so $|PairSB_i|$ is exponential \wrt $|\Omc|$. 
 In line~\ref{alg3-foreachpairset}, for each $PairSet$ constructed by the algorithm, $|PairSet|$ is bounded by $\prod_{i=1}^p |PairSetB_i|$ where $p$ is polynomial \wrt $|\Omc|$ and $|PairSB_i|$ is exponential \wrt $|\Omc|$, hence $|PairSet|$ is exponential \wrt $|\Omc|$. 
 In line~\ref{alg3-foreachpairsetc}, $ListPairSets$ is bounded by the number of atoms \blue{$C(x_0,t)$ in $\ext{q^*_i}$} so is polynomial \wrt $|\Omc|$ and $|q|$. 
 In line~\ref{alg3-foreachmoninpairsetC}, for each $PairSetC$ constructed by the algorithm, $|PairSetC|$ is bounded by the number of $(B_1\sqcap\dots\sqcap B_p\sqsubseteq C,\nonomial_C)\in\mn{saturate}(\Omc)$ so is exponentially bounded \wrt $|\Omc|$. 
In line~\ref{alg3-foreachmoninres}, $|Result|$ is bounded by $\prod_{\blue{C(x_0,t)\in \ext{q^*_i}}} |PairSetC|$ where $|q^*_i|$ is polynomial \wrt $|\Omc|$ and $|q|$ and $|PairSetC|$ is exponential \wrt $|\Omc|$, hence $|Result|$ is exponential \wrt $|q|$ and~$|\Omc|$.  \qedhere
\end{proof}

\begin{lemma}\label{lem:algorewfull}
Algorithm~\ref{alg:rewrite} runs in exponential time \wrt $|q|$ and $|\Omc|$.
\end{lemma}

\begin{proof}
By Lemma~\ref{lem:size-rew}, $|\Rew(q,\Omc)|$ is exponential in $|\Omc|$ and $|q|$ so Algorithm~\ref{alg:rewrite} runs the while-loop an exponential number of times \wrt $|\Omc|$ and $|q|$. Moreover, at each iteration of the while-loop, every foreach-loop iterates over a set whose size is exponentially bounded \wrt $|\Omc|$ and $|q|$ and each iteration takes at most exponential time \wrt $|\Omc|$ and $|q|$:
\begin{itemize}
\item $|Rew|$ is bounded by $|\Rew(q,\Omc)|$ so exponential \wrt $|\Omc|$ and $|q|$. 

\item For each $(q_i,\monomial_i)$, by Lemma~\ref{lem:size-rew} $(q_i,\monomial_i)$ has polynomial size \wrt $|\Omc|$ and $|q|$ so both the number of variables in $q_i$ and $|\mn{atoms}(q^*_i)|$ are polynomial \wrt $|\Omc|$ and $|q|$.

\item By Lemmas~\ref{lem:algorewrole} and~\ref{lem:algorewconcept}, $\mn{RoleAtRew}(q^*_i,x_0, A\sqsubseteq \exists P)$ and $\mn{ConceptAtRew}(q^*_i,x_0, A\sqsubseteq \exists P)$ are executed in exponential time \wrt $|\Omc|$ and $|q|$, hence their results $RoleRew$ and $ConceptRew$ are of exponential size \wrt $|\Omc|$ and $|q|$.\qedhere
\end{itemize}
\end{proof}

\ThCQansweringAlgo*
\begin{proof}
Claim (4) of the theorem follows from Lemmas~\ref{lem:algorewfcorrect} and~\ref{lem:algorewfull}: $\Rew(q,\Omc)$ can be computed in exponential time in $|\Omc|$ and $|q|$ using Algorithm~\ref{alg:rewrite}. 
Moreover, claim (3) implies claims (1) and (2):
\begin{itemize}
\item Since \why is positive, by Theorem~\ref{th:sem-entailment}, $\Omc'\models q(\vec{a})$ iff $\Pmc(q(\vec{a}),\Omc)\neq 0$, and claim~(3) implies that $\Pmc(q(\vec{a}),\Omc)\neq 0$ iff there exists $(q^*,\monomial^*)\in \Rew(q,\Omc)$ such that there is a match for the extended version of $q^*(\vec{a})$ in~\new{$\Imc_\Dmc$} (note that all monomials that occur in $\Rew(q,\Omc)$ are different from 0 since by construction they are products of variables that annotate $\Omc$), so iff there exists $(q^*,\monomial^*)\in \Rew(q,\Omc)$ such that there is a match for $q^*(\vec{a})$ in~$\Dmc'$.
\item $\Pmc(q(\vec{a}), \Omc)= \sum_{ \Omc\models(q(\vec{a}),\monomial)} \monomial$ so claim (3) implies that $\Omc\models(q(\vec{a}),\monomial)$ iff $\monomial$ occurs in $\new{\sum_{(q^*,\monomial^*)\in \Rew(q,\Omc)}(\monomial^*\times \Sigma_{\onomial\in \p{\Imc_\Dmc}{\ext{q^*(\vec{a})}}}\onomial)}=\sum_{(q^*,\monomial^*)\in\Rew(q,\Omc)} \sum_{ \new{\onomial\in \p{\Imc_\Dmc}{\ext{q^*(\vec{a})}}}} \monomial^*\times\onomial$. 
\end{itemize}
We thus prove claim (3), \ie that $\Pmc(q(\vec{a}),\Omc)=\new{\sum_{(q^*,\monomial^*)\in \Rew(q,\Omc)}(\monomial^*\times \Sigma_{\onomial\in \p{\Imc_\Dmc}{\ext{q^*(\vec{a})}}}\onomial)}$. 
By definition of the provenance of a query and by Theorem~\ref{thm:can-model-main}, it holds that: 
$\Pmc(q(\vec{a}),\Omc)=\sum_{\nonomial\in \p{\Imc_{\Omc}}{\ext{q}(\vec{a})}}\nonomial$  where $\Imc_\Omc$ is the canonical model of $\Omc$ and $\ext{q}(\vec{a})$ is the extended version of $q(\vec{a})$. 
We thus need to show that $$\p{\Imc_{\Omc}}{\ext{q}(\vec{a})}=\{\monomial^*\times\onomial\mid (q^*,\monomial^*)\in \Rew(q,\Omc), \onomial\in \p{\Imc_{\Dmc}}{\ext{q^*}(\vec{a})}\}.$$ 

\noindent\textbf{[``$\supseteq$'']} Let $(q^*,\monomial^*)\in \Rew(q,\Omc)$ and let $\onomial\in \p{\Imc_{\Dmc}}{\ext{q^*}(\vec{a})}$. 
\\
\noindent$\bullet$ There is a match $\pi$ of  $\ext{q^*}$ to $\Imc_{\Dmc}$ such that $\pi(\vec{x})=\vec{a}$ and $\onomial=\prod_{P(\vec{t},t)\in \ext{q^*}} \pi(t)$. 
\new{By definition of $\Imc_\Dmc$}, it holds that for every $P(\vec{t},t)\in \ext{q^*}$:
\begin{itemize}
\item $(P(\pi(\vec{t})),\pi(t))\in \Dmc$ so $(P(\pi(\vec{t})),\pi(t))\in \mn{saturate}(\Omc)$;
\item so $\Omc\models (P(\pi(\vec{t})),\pi(t))$ by Theorem~\ref{prop:completionalgorithmELHI};
\item so $\Imc_\Omc\models (P(\pi(\vec{t})),\pi(t))$ by Theorem~\ref{thm:can-model-main}. 
\end{itemize}
It follows that $\pi$ is a match of $\ext{q^*}$ to $\Imc_{\Omc}$ such that $\pi(\vec{x})=\vec{a}$ and $\onomial=\prod_{P(\vec{t},t)\in \ext{q^*}} \pi(t)$.
\\
\noindent$\bullet$  Since $(q^*,\monomial^*)\in \Rew(q,\Omc)$, there is a rewriting sequence $(q_0,\monomial_0)\rightarrow_\Omc (q_1,\monomial_1)\rightarrow_\Omc\dots\rightarrow_\Omc(q_k,\monomial_k)$ such that $(q_0,\monomial_0)=(q,1)$ and $(q_k,\monomial_k)=(q^*,\monomial^*)$. 
Moreover, for every $0\leq i\leq k-1$, there is $\monomial_{i,i+1}$ such that $\monomial_{i+1}=\monomial_i\times \monomial_{i,i+1}$. 

\noindent$\bullet$ We show that for every $0\leq j\leq k$, there is a match $\pi_j$ of the extended version $\ext{q_j}$ of $q_j$ to $\Imc_{\Omc}$ such that $\pi_j(\vec{x})=\vec{a}$ and $\prod_{P(\vec{t},t)\in \ext{q_j}} \pi_j(t)=\prod_{i=j}^{k-1}\monomial_{i,i+1}\times\onomial$. The proof is by descending induction on $j$.

\noindent\emph{Base case:} $j=k$, $(q_k,\monomial_k)=(q^*,\monomial^*)$ and $\prod_{i=k}^{k-1}\monomial_{i,i+1}\times\onomial=\onomial$. We have shown that there is a match $\pi$ of $\ext{q^*}$ to $\Imc_{\Omc}$ such that $\pi(\vec{x})=\vec{a}$ and $\onomial=\prod_{P(\vec{t},t)\in \ext{q^*}} \pi(t)$ so we just need to take $\pi_k=\pi$. 

\noindent\emph{Induction step.}  Assume that the property is true for some $1\leq j\leq k$: There is a match $\pi_j$ of $\ext{q_j}$ to $\Imc_{\Omc}$ such that $\pi_j(\vec{x})=\vec{a}$ and $\prod_{P(\vec{t},t)\in \ext{q_j}} \pi_j(t)=\prod_{i=j}^{k-1}\monomial_{i,i+1}\times\onomial$. 

\noindent We show that the property still holds for $j-1$. Since $(q_{j-1},\monomial_{j-1})\rightarrow_\Omc (q_j,\monomial_j)$, there is a variable $x_0$ existentially quantified in $q_{j-1}$ such that no atoms of the form $R(x_0,x_0)$ occur in $q_{j-1}$ and there is $(A\sqsubseteq\exists P,v)\in\Omc$ such that $q_j$ has been obtained from $q_{j-1}$ by dropping every atom that contains $x_0$, selecting a variable $y_0\in V_p$ and replacing every occurrence $y'\in V_p$ in $q_{j-1}$ by $y_0$ and \blue{adding atom $D(y_0)$ for each occurrence of concept name $D$ in the list $At(q_{j-1},x_0,A\sqsubseteq \exists P)$ (with a limit of $\mn{Card}(\Omc)$ repetitions) and the following conditions are respected:} 
\begin{itemize}
\item[(a)] for every \blue{$Q(y,x_0,t)\in \ext{q_{j-1}}$}, there exists $(P\sqsubseteq Q,\monomial_{P\sqsubseteq Q}^{\blue{t}})\in \mn{saturate}(\Omc)$, 
\item[(b)] for every \blue{$C(x_0,t)\in \ext{q_{j-1}}$}, there are $p^{C,\blue{t}},p'^{C,\blue{t}}\geq 0$ such that $(B^{C,\blue{t}}_1\sqcap\dots\sqcap B^{C,\blue{t}}_{p^{C,\blue{t}}}\sqcap B'^{C,\blue{t}}_1\sqcap\dots\sqcap B'^{C,\blue{t}}_{p'^{C,\blue{t}}}\sqsubseteq C, \nonomial_C^{\blue{t}})\in\mn{saturate}(\Omc)$, for every $1\leq i\leq p^{C,\blue{t}}$, there exist $(\exists \mn{inv}(P^{C,\blue{t}}_i).A^{C,\blue{t}}_i\sqsubseteq B^{C,\blue{t}}_i,v^{C,\blue{t}}_i)\in\Omc$ and $(P\sqsubseteq P^{C,\blue{t}}_i,\monomial^{C,\blue{t}}_i)\in\mn{saturate}(\Omc)$, and for every $1\leq i\leq p'^{C,\blue{t}}$ there exists $(\top\sqsubseteq B^{C,\blue{t}}_i,\onomial^{C,\blue{t}}_i)\in\mn{saturate}(\Omc)$; 
\item[(c)] \blue{the concept names in $At(q_{j-1},x_0,A\sqsubseteq \exists P)$ (i.e., those that occur in $\ext{q_j}\setminus\ext{q_{j-1}}$) are exactly $A$ and the concept names corresponding to the $A_i^{C,\blue{t}}$'s 
and} if we let $\monomial_C^{\blue{t}}=\prod_{i=1}^{p^{C,\blue{t}}}(v^{C,\blue{t}}_i\times\monomial^{C,\blue{t}}_i)\times\prod_{i=1}^{p'^{C,\blue{t}}}\onomial^{C,\blue{t}}_i$, it holds that:
$$\monomial_{j-1,j}=mon(q_{j-1},x_0,A\sqsubseteq \exists P)=v\times\prod_{\blue{Q(y,x_0,t)\in \ext{q_{j-1}}}}\monomial_{P\sqsubseteq Q}^{\blue{t}}\times\prod_{\blue{C(x_0,t)\in \ext{q_{j-1}}}}\monomial_C^{\blue{t}}\times \nonomial_C^{\blue{t}}.$$
\end{itemize}
\blue{Since $A\in At(q_{j-1},x_0,A\sqsubseteq \exists P)$, there is an atom $A(y_0,t_A)$ in $\ext{\ext{q_j}}$. Since $\pi_j$ is a match of $\ext{q_j}$ to $\Imc_\Omc$  and $(A\sqsubseteq \exists P,v)\in\Omc$,} it follows that there exists $e\in\Delta^{\Imc_\Omc}$ such that $(\pi_j(y_0), e,\pi_j(t_A)\times v)\in P^{\Imc_\Omc}$. 

\noindent Let $P(\vec{t},t)\in \ext{q_{j-1}}$. 
\begin{itemize}
\item If $P(\vec{t},t)\in \ext{q_{j}}$, $\pi_j(\vec{t},t)\in P^{\Imc_\Omc}$ by definition of $\pi_j$. 

\item Otherwise, if $P(\vec{t},t)\notin \ext{q_{j}}$, either $P(\vec{t},t)=Q(y,x_0,t)$ for some $Q\in\NR$ and $y\in V_p$, or $P(\vec{t},t)=C(x_0,t)$ for some $C\in \NC$.
\begin{itemize}
\item In the first case, \blue{$Q(y,x_0,t)\in \ext{q_{j-1}}$ and (by (a)) there is $(P\sqsubseteq Q,\monomial_{P\sqsubseteq Q}^{\blue{t}})\in \mn{saturate}(\Omc)$}. 

By Lemma~\ref{lem:saturation-correct-for-axioms}, \new{since $\Imc_\Omc$ is a model of $\Omc$ and $(P\sqsubseteq Q,\monomial_{P\sqsubseteq Q}^{\blue{t}})$ is an RI, } $\Imc_\Omc\models (P\sqsubseteq Q,\monomial_{P\sqsubseteq Q}^{\blue{t}})$. 

Hence, since $(\pi_j(y_0), e,\pi_j(t_A)\times v)\in P^{\Imc_\Omc}$, it holds that $(\pi_j(y_0), e,\pi_j(t_A)\times v\times \monomial_{P\sqsubseteq Q}^{\blue{t}})\in Q^{\Imc_\Omc}$. 

\item In the second case, \blue{$C(x_0,t)\in \ext{q_{j-1}}$ and (by (b)) there are} $p^{C,\blue{t}},p'^{C,\blue{t}}\geq 0$ such that $(B^{C,\blue{t}}_1\sqcap\dots\sqcap B^{C,\blue{t}}_{p^{C,\blue{t}}}\sqcap B'^{C,\blue{t}}_1\sqcap\dots\sqcap B'^{C,\blue{t}}_{p'^{C,\blue{t}}}\sqsubseteq C, \nonomial_C^{\blue{t}})\in\mn{saturate}(\Omc)$, for every $1\leq i\leq p^{C,\blue{t}}$, there are $(\exists \mn{inv}(P^{C,\blue{t}}_i).A^{C,\blue{t}}_i\sqsubseteq B^{C,\blue{t}}_i,v^{C,\blue{t}}_i)\in\Omc$ and $(P\sqsubseteq P^{C,\blue{t}}_i,\monomial^{C,\blue{t}}_i)\in\mn{saturate}(\Omc)$, and for every $1\leq i\leq p'^{C,\blue{t}}$ there are $(\top\sqsubseteq B^{C,\blue{t}}_i,\onomial^{C,\blue{t}}_i)\in\mn{saturate}(\Omc)$.
 
By Lemma~\ref{lem:saturation-correct-for-axioms}, \new{since $\Imc_\Omc$ is a model of $\Omc$ and \emph{all annotations of $\Imc_\Omc$ are monomials}, it holds that,} 
$\Imc_\Omc\models (B^{C,\blue{t}}_1\sqcap\dots\sqcap B^{C,\blue{t}}_{p^{C,\blue{t}}}\sqcap B'^{C,\blue{t}}_1\sqcap\dots\sqcap B'^{C,\blue{t}}_{p'^{C,\blue{t}}}\sqsubseteq C, \nonomial_C^{\blue{t}})$, for $1\leq i\leq p^{C,\blue{t}}$, $\Imc_\Omc\models (\exists \mn{inv}(P^{C,\blue{t}}_i).A^{C,\blue{t}}_i\sqsubseteq B^{C,\blue{t}}_i,v^{C,\blue{t}}_i)$ and $\Imc_\Omc\models (P\sqsubseteq P^{C,\blue{t}}_i,\monomial^{C,\blue{t}}_i)$, and for every $1\leq i\leq p'^{C,\blue{t}}$, $\Imc_\Omc\models (\top\sqsubseteq B^{C,\blue{t}}_i,\onomial^{C,\blue{t}}_i)$. Hence:

 \begin{itemize}
 \item  For $1\leq i\leq p^{C,\blue{t}}$, $(\pi_j(y_0), e,\pi_j(t_A)\times v\times \monomial^{C,\blue{t}}_i)\in (P^{C,\blue{t}}_i)^{\Imc_\Omc}$ and since $A^{C,\blue{t}}_i\in At(q_{j-1},x_0,A\sqsubseteq \exists P)$, \blue{there is some atom} $A^{C,\blue{t}}_i(y_0,t_{A^{C,\blue{t}}_i})\in\ext{q_j}$ so $(\pi_j(y_0),\pi_j(t_{A^{C,\blue{t}}_i}))\in (A^{C,\blue{t}}_i)^{\Imc_\Omc}$, thus $(e,\pi_j(t_A)\times v\times \monomial^{C,\blue{t}}_i\times v^{C,\blue{t}}_i\times\pi_j(t_{A^{C,\blue{t}}_i}))\in (B^{C,\blue{t}}_i)^{\Imc_\Omc}$.
 \item For $1\leq i\leq p'^{C,\blue{t}}$, $(e,\onomial^{C,\blue{t}}_i)\in (B^{C,\blue{t}}_i)^{\Imc_\Omc}$. 
 \end{itemize}
It follows that $(e,\pi_j(t_A)\times v\times \prod_{i=1}^{p^{C,\blue{t}}}(\monomial^{C,\blue{t}}_i\times v^{C,\blue{t}}_i\times \pi_j(t_{A^{C,\blue{t}}_i}))\times\prod_{i=1}^{p'^{C,\blue{t}}}\onomial^{C,\blue{t}}_i\times\nonomial_C^{\blue{t}})\in C^{\Imc_\Omc}$, \ie $(e,\pi_j(t_A)\times v\times \monomial_C^{\blue{t}}\times\nonomial_C^{\blue{t}}\times\prod_{i=1}^{p^{C,\blue{t}}} \pi_j(t_{A^{C,\blue{t}}_i}))\in C^{\Imc_\Omc}$ (with $\monomial_C^{\blue{t}}=\prod_{i=1}^{p^{C,\blue{t}}}(\monomial^{C,\blue{t}}_i\times v^{C,\blue{t}}_i)\times\prod_{i=1}^{p'^{C,\blue{t}}}\onomial^{C,\blue{t}}_i$).   
\end{itemize}
\end{itemize}
Define $\pi_{j-1}$ as follows:
\begin{itemize}
\item $\pi_{j-1}(x_0)=e$, 
\item $\pi_{j-1}(y)=\pi_{j}(y_0)$ for every $y\in V_p$, 
\item $\pi_{j-1}(x)=\pi_{j}(x)$ for every other variable $x$ of $q_{j-1}$, 
\item $\pi_{j-1}(t)=\pi_{j}(t)$ for every $P(\vec{t},t)\in \ext{q_j}\cap\ext{q_{j-1}}$, 
\item $\pi_{j-1}(t)=\pi_{j}(t_A)\times v\times \monomial_{P\sqsubseteq Q}^{\blue{t}}$ for $t$ such that $Q(y,x_0,t)\in\ext{q_{j-1}}$ and 
\item $\pi_{j-1}(t)=\pi_j(t_A)\times v\times \monomial_C^{\blue{t}}\times\nonomial_C^{\blue{t}}\times\prod_{i=1}^{p^{C,\blue{t}}} \pi_j(t_{A^{C,\blue{t}}_i})$ for $t$ such that $C(x_0,t)\in\ext{q_{j-1}}$.
\end{itemize}
$\pi_{j-1}$ is a match of $\ext{q_{j-1}}$ to $\Imc_{\Omc}$ such that $\pi_{j-1}(\vec{x})=\pi_{j}(\vec{x})=\vec{a}$ (since $x_0\notin\vec{x}$ as it is an existentially quantified variable, and for every $y\in \vec{x}\cap V_p$, $y$ is replaced by $y_0$ in $q_j$). Moreover: 

\begin{align*}
\prod_{P(\vec{t},t)\in \ext{q_{j-1}}} \pi_{j-1}(t)=&\prod_{P(\vec{t},t)\in \ext{q_{j-1}}\cap \ext{q_{j}}} \pi_{j}(t)\times \prod_{\blue{Q(y,x_0,t)\in \ext{q_{j-1}}}}\pi_{j}(t_A)\times v\times \monomial_{P\sqsubseteq Q}^{\blue{t}}\\&\times\prod_{\blue{C(x_0,t)\in \ext{q_{j-1}}}}\pi_j(t_A)\times v\times \monomial_C^{\blue{t}}\times\nonomial_C^{\blue{t}}\times\prod_{i=1}^{p^{C,\blue{t}}} \pi_j(t_{A^{C,\blue{t}}_i})
\\
=&\prod_{P(\vec{t},t)\in \ext{q_{j-1}}\cap \ext{q_{j}}} \pi_{j}(t)\times \pi_{j}(t_A)\times \prod_{\blue{C(x_0,t)\in \ext{q_{j-1}}}}\prod_{i=1}^{p^{C,\blue{t}}} \pi_j(t_{A^{C,\blue{t}}_i}) \\&
\times v\times\prod_{\blue{Q(y,x_0,t)\in \ext{q_{j-1}}}}\monomial_{P\sqsubseteq Q}^{\blue{t}}\times\prod_{\blue{C(x_0,t)\in \ext{q_{j-1}}}}\monomial_C^{\blue{t}}\times\nonomial_C^{\blue{t}}
\\
=&\prod_{P(\vec{t},t)\in \ext{q_{j}}} \pi_{j}(t)\times v\times\prod_{\blue{Q(y,x_0,t)\in \ext{q_{j-1}}}}\monomial_{P\sqsubseteq Q}^{\blue{t}}\times\prod_{\blue{C(x_0,t)\in \ext{q_{j-1}}}}\monomial_C^{\blue{t}}\times\nonomial_C^{\blue{t}}\quad\text{ \blue{(cf. point (c) above)}}
\\
=&\prod_{P(\vec{t},t)\in \ext{q_{j}}} \pi_{j}(t)\times \monomial_{j-1,j} \quad\text{ \blue{(cf. point (c) above)}}
\\
=&\prod_{i=j}^{k-1}\monomial_{i,i+1}\times\onomial\times \monomial_{j-1,j}
\\
=&\prod_{i=j-1}^{k-1}\monomial_{i,i+1}\times\onomial. 
\end{align*}

\noindent We conclude that there exists a match $\pi_0$ of $\ext{q_0}=\ext{q}$ to $\Imc_{\Omc}$ such that $\pi_0(\vec{x})=\vec{a}$ and $\prod_{P(\vec{t},t)\in \ext{q}} \pi_0(t)=\prod_{i=0}^{k-1}\monomial_{i,i+1}\times\onomial$. Since $\monomial_0=1$, $\prod_{i=0}^{k-1}\monomial_{i,i+1}=\monomial_0\times \prod_{i=0}^{k-1}\monomial_{i,i+1}$ so since $\monomial_{i+1}=\monomial_i\times \monomial_{i,i+1}$ for every $0\leq i\leq k-1$, $\prod_{i=0}^{k-1}\monomial_{i,i+1}=\monomial_k=\monomial^*$. 
Hence $\prod_{P(\vec{t},t)\in \ext{q}} \pi_0(t)=\monomial^*\times\onomial$.
\smallskip

We have thus shown that for every $(q^*,\monomial^*)\in \Rew(q,\Omc)$ and $\onomial\in \p{\Imc_{\Dmc}}{\ext{q^*}(\vec{a})}$, $\monomial^*\times\onomial\in\p{\Imc_{\Omc}}{\ext{q}(\vec{a})}$, \ie $\p{\Imc_{\Omc}}{\ext{q}(\vec{a})}\supseteq\{\monomial^*\times\onomial\mid (q^*,\monomial^*)\in \Rew(q,\Omc), \onomial\in \p{\Imc_{\Dmc}}{\ext{q^*}(\vec{a})}\}.$
\medskip

\noindent\textbf{[``$\subseteq$'']} In the other direction, let $\nonomial\in \p{\Imc_{\Omc}}{\ext{q}(\vec{a})}$: There is a match $\pi$ of $\ext{q}$ to $\Imc_{\Omc}$ such that $\pi(\vec{x})=\vec{a}$ and $\nonomial=\prod_{P(\vec{t},t)\in \ext{q}} \pi(t)$. 
\\
\noindent$\bullet$ Let $V_{an}$ be the set of variables of $\ext{q}$ that are mapped by $\pi$ to anonymous individuals from $\Delta^{\Imc_\Omc}\setminus\NI$. Since $\pi(\vec{x})=\vec{a}$ and $c^{\Imc_\Omc}=c$ for every $c\in\NI$, all variables in $V_{an}$ are existentially quantified in $q$. Moreover, the anonymous part of $\Imc_\Omc$ is tree-shaped, so for every $x\in V_{an}$, there are no atoms of the form $R(x,x)$ in $q$.
\\
\noindent$\bullet$ We define a total order $\succ$ over variables from $V_{an}$ as follows. Considering a fixed sequence of applications of \new{the chase rule} that constructs $\Imc_\Omc$, we define $y\succ y'$ iff $\pi(y)$ has been introduced before $\pi(y')$.  
\\
\noindent$\bullet$ We build a rewriting sequence $(q_0,\monomial_0)\rightarrow_\Omc (q_1,\monomial_1)\rightarrow_\Omc\dots\rightarrow_\Omc(q_k,\monomial_k)$  together with a sequence $\pi_0=\pi, \pi_1\dots, \pi_k=\pi^*$ such that:
\begin{enumerate}
\item $(q_0,\monomial_0)=(q,1)$, $(q_k,\monomial_k)=(q^*,\monomial^*)$ and $\mn{terms}(q^*)\cap V_{an}=\emptyset$; 

\item each $\pi_i$ is a match for $\ext{q_i}$ to $\Imc_\Omc$ such that $\pi_i$ coincides with $\pi$ on \blue{$\mn{terms}(q_i)$, \ie on the shared variables of (the non-extended versions of) $q$ and $q_i$ (recall that each rewriting step removes one existentially quantified variable from $q$)};

\item $\prod_{P(\vec{t},t)\in \ext{q}} \pi(t)=\prod_{P(\vec{t},t)\in \ext{q^*}} \pi^*(t)\times \monomial^*$. 
\end{enumerate}

\noindent The construction by induction ensures that (1) $\mn{terms}(q_{i+1})\cap V_{an}\subsetneq \mn{terms}(q_i)\cap V_{an}$, (2) $\pi_{i+1}$ is a match for $\ext{q_{i+1}}$ in $\Imc_\Omc$ such that $\pi_{i+1}$ coincides with $\pi$ on \blue{$\mn{terms}(q_i)$}, and (3) $\prod_{P(\vec{t},t)\in \ext{q_{i+1}}} \pi_{i+1}(t)\times \monomial_{i,i+1}=\prod_{P(\vec{t},t)\in \ext{q_i}} \pi_i(t)$ for some $\monomial_{i,i+1}$ such that $\monomial_{i+1}=\monomial_i\times \monomial_{i,i+1}$. 

\noindent For $i\geq 0$, assuming that $(q_i,\monomial_i)$ and $\pi_i$ are constructed, we obtain $(q_{i+1},\monomial_{i+1})$ and $\pi_{i+1}$ as follows:
\begin{itemize}
\item[(S1)] Choose $x_0$ such that $x_0$ is the least element in $V_{an}\cap \mn{terms}(q_i)$ \wrt $\succ$ (\ie $\pi_i(x_0)=\pi(x_0)$ has been introduced last in the construction of $\Imc_\Omc$). 

\item[(S2)] Replace each role atom of the form $R(x_0,y)$ in $q_i$, where $y$ and $R$ are arbitrary, by the atom $\mn{inv}(R)(y,x_0)$.

\item[(S3)] Let $V_p=\{y\mid Q(y,x_0)\in q\text{ for some }Q\}$. 
By construction of $\Imc_\Omc$, $\pi_i(x_0)$ has been introduced by applying the chase rule with some $(A\sqsubseteq \exists P,v)\in\Omc$ (recall that $\Omc$ is in normal form, so $A\in\NC$) and some $(e_0,\monomial_A)\in A^{\Imc_\Omc}$. 

Let $At(q_i,x_0,A\sqsubseteq \exists P)=\blue{[A]}$ and $mon(q_i,x_0,A\sqsubseteq \exists P)=v$.  
\begin{itemize}
\item For every 
\blue{$Q(y,x_0, t)\in \ext{q_i}$}, since $\pi_i$ is a match for $\ext{q_i}$ in $\Imc_\Omc$, it holds that $(\pi_i(y),\pi_i(x_0),\blue{\pi_i(t)})\in Q^{\Imc_\Omc}$. 

Since $\pi_i(x_0)$ is least \wrt $\succ$ in $V_{an}\cap \mn{terms}(q_i)$, it has been introduced after $\pi_i(y)$. By construction of $\Imc_\Omc$ (since there is no loop in the anonymous part of $\Imc_\Omc$), this implies that $\pi_i(y)=e_0$. 

By Lemma~\ref{lem:canonical-saturation-RIs}, there exists $(P\sqsubseteq Q, \monomial_{P\sqsubseteq Q}^{\blue{t}})\in\mn{saturate}(\Omc)$ such that 
$$\pi_i(\blue{t})=\monomial_A\times v\times\monomial_{P\sqsubseteq Q}^{\blue{t}}.$$

Update $mon(q_i,x_0,A\sqsubseteq \exists P)\leftarrow mon(q_i,x_0,A\sqsubseteq \exists P)\times\monomial_{P\sqsubseteq Q}^{\blue{t}}$. 

\item For every 
\blue{$C(x_0, t)\in \ext{q_i}$}, since $\pi_i$ is a match for $\ext{q_i}$ in $\Imc_\Omc$, $(\pi_i(x_0),\blue{\pi_i(t)})\in C^{\Imc_\Omc}$. 

By Lemma~\ref{lem:canonical-saturation-GCIs}, 
there are $p^{C,\blue{t}}, p'^{C,\blue{t}}\geq 0$ such that
\begin{itemize}
\item for every $1\leq i\leq p^{C,\blue{t}}$, there exist $(\exists \mn{inv}(P^{C,\blue{t}}_i).A^{C,\blue{t}}_i\sqsubseteq B^{C,\blue{t}}_i,v^{C,\blue{t}}_i)\in\Omc$ and $(P\sqsubseteq P^{C,\blue{t}}_i,\monomial^{C,\blue{t}}_i)\in\mn{saturate}(\Omc)$,
\item for every $1\leq i\leq p'^{C,\blue{t}}$, there exists $(\top\sqsubseteq B'^{C,\blue{t}}_i,\onomial^{C,\blue{t}}_i)\in\mn{saturate}(\Omc)$,
\item $(B^{C,\blue{t}}_1\sqcap\dots\sqcap B^{C,\blue{t}}_{p^{C,\blue{t}}}\sqcap B'^{C,\blue{t}}_1\sqcap\dots\sqcap B'^{C,\blue{t}}_{p'^{C,\blue{t}}}\sqsubseteq C, \nonomial_C^{\blue{t}})\in\mn{saturate}(\Omc)$,
\item for every $1\leq i\leq p^{C,\blue{t}}$ there is some $(e_0,\monomial_{A^{C,\blue{t}}_i})\in (A^{C,\blue{t}}_i)^{\Imc_\Omc}$,
\end{itemize}
and if we let $\monomial_C^{\blue{t}}=\prod_{i=1}^{p^{C,\blue{t}}}\monomial^{C,\blue{t}}_i\times v^{C,\blue{t}}_i\times\prod_{i=1}^{p'^{C,\blue{t}}} \onomial^{C,\blue{t}}_i$, then 
\begin{align*}
\text{if $p^{C,\blue{t}}\neq 0$: }\quad\quad
\pi_i(\blue{t})=&\monomial_A\times v\times\prod_{i=1}^{p^{C,\blue{t}}}(\monomial^{C,\blue{t}}_i\times v^{C,\blue{t}}_i\times\monomial_{A^{C,\blue{t}}_i})\times\prod_{i=1}^{p'^{C,\blue{t}}} \onomial^{C,\blue{t}}_i\times \nonomial_C^{\blue{t}}\\
=&\monomial_A\times v\times\monomial_C^{\blue{t}}\times\nonomial_C^{\blue{t}}\times\prod_{i=1}^{p^{C,\blue{t}}}\monomial_{A^{C,\blue{t}}_i}
\\
\text{and if $p^{C,\blue{t}}=0$: }\quad\quad
\pi_i(\blue{t})=&\prod_{i=1}^{p'^{C,\blue{t}}} \onomial^{C,\blue{t}}_i\times \nonomial_C^{\blue{t}}
\\=&\monomial_C^{\blue{t}}\times\nonomial_C^{\blue{t}}.
\end{align*}
Update $At(q_i,x_0,A\sqsubseteq \exists P)\leftarrow At(q_i,x_0,A\sqsubseteq \exists P)\blue{\cdot [}A^{C,\blue{t}}_i\mid 1\leq i\leq p^{C,\blue{t}}\blue{]}$ and $mon(q_i,x_0,A\sqsubseteq \exists P)\leftarrow mon(q_i,x_0,A\sqsubseteq \exists P)\times\monomial_C^{\blue{t}}\times\nonomial_C^{\blue{t}}$. 
\end{itemize}

\item[(S4)] Drop from $q_i$ every atom that contains $x_0$.

\item[(S5)] Select a variable $y_0\in V_p$ and replace every occurrence of $y'\in V_p$ in $q_i$ by $y_0$ (recall that $\pi_i$ maps all these variables to $e_0$).

\item[(S6)] 
\blue{Add atom $D(y_0)$ to $q_i$ for each occurrence of concept name $D$ in the list $At(q_i,x_0,A\sqsubseteq \exists P)$, then limit the total number of occurrences of $D(y_0)$ in the query to $\mn{Card}(\Omc)$. Finally,} multiply $\monomial_i$ by $$\monomial_{i,i+1}=mon(q_i,x_0,A\sqsubseteq \exists P)=v\times\prod_{\blue{Q(y,x_0,t)\in \ext{q_{i}}}}\monomial_{P\sqsubseteq Q}^{\blue{t}}\times\prod_{\blue{C(x_0,t)\in \ext{q_{i}}}}\monomial_C^{\blue{t}}\times\nonomial_C^{\blue{t}}.$$ 
\end{itemize}

\noindent \blue{We define $\pi_{i+1}$ so that it coincides with $\pi_i$ (hence with $\pi$) on $\mn{terms}(q_{i+1})\subseteq \mn{terms}(q_i)$ and on all $t$ such that $P(\vec{t},t)\in \ext{q_{i+1}}$ with $\vec{t}\neq y_0$. Note that 
\begin{enumerate}[(i)]
\item $P(\vec{t},t)\in \ext{q_{i+1}}$ with $\vec{t}\neq y_0$ implies that $P(\vec{t},t)\in \ext{q_i}$ and $x_0\notin \vec{t}$ by the form of atoms added to $q_i$ to obtain $q_{i+1}$ in (S6) and the fact that $x_0$ does not occur in $q_{i+1}$, and 
\item $P(\vec{t},t)\in \ext{q_i}$ with $x_0\notin \vec{t}$ implies that either $P(\vec{t},t)\in \ext{q_{i+1}}$ (or $P(\vec{t},t)[y'\leftarrow y_0 \mid y'\in V_p]\in \ext{q_{i+1}}$ with $P$ a role name) with $\vec{t}\neq y_0$ or $P(\vec{t},t)$ is of the form $D(y',t)$ with $y'\in V_p$. 
\end{enumerate}It remains to define $\pi_{i+1}(z)$ for each $D(y_0,z)\in\ext{q_{i+1}}$, which we do as follows.
\begin{itemize}
\item If no occurrence of $D(y_0)$ has been removed in step (S6):
\begin{itemize}
\item For each $D(y_0,z)\in\ext{q_i}$, let $\pi_{i+1}(z)=\pi_i(z)$.
\item There is a one-to-one correspondence between the atoms $D(y_0,z)\in\ext{q_{i+1}}\setminus\ext{q_i}$ and the $A^{C,t}_i$ that are equal to $D$ in $At(q_i,x_0,A\sqsubseteq \exists P)$. We use this correspondence to define $\pi_{i+1}(z)=\monomial_{A^{C,t}_i}$, for $A^{C,t}_i$ corresponding to $D(y_0,z)$, where $\monomial_{A^{C,t}_i}$ is such that $(e_0,\monomial_{A^{C,\blue{t}}_i})\in (A^{C,\blue{t}}_i)^{\Imc_\Omc}$ (\cf (S3)). 
\end{itemize}
\item If some occurrence of $D(y_0)$ has been removed in step (S6), there are less atoms of the form $D(y_0,z)$ in $\ext{q_{i+1}}\setminus\ext{q_i}[y'\leftarrow y_0\mid y'\in V_p]$ than $A^{C,t}_i$ equal to $D$ in $At(q_i,x_0,A\sqsubseteq \exists P)$. 
In this case, we need to choose a subset $M'$ of at most $\mn{Card}(\Omc)$ monomials from $M=\{\monomial_{A^{C,t}_i}\mid A^{C,t}_i\in At(q_i,x_0,A\sqsubseteq \exists P), A^{C,t}_i=D\}\cup\{\pi_i(z)=\nonomial \mid D(y',z)\in \ext{q_i}, y'\in V_p\}$ whose product is equal to the product of all monomials in $M$ (this is possible since the product cannot contain more than $\mn{Card}(\Omc)$ variables). Since there are $\mn{Card}(\Omc)$ atoms of the form $D(y_0,z)$ in $\ext{q_{i+1}}$, we can define a surjective function $f$ from $\{z\mid D(y_0,z)\in \ext{q_{i+1}}\}$ to $M'$ and define $\pi_{i+1}(z)=f(z)$. 
\end{itemize}
Recall that for every $y\in V_p$, $e_0=\pi_i(y)=\pi_{i+1}(y_0)$ so $\pi_{i+1}$ is indeed a match for $\ext{q_{i+1}}$ in $\Imc_\Omc$. Moreover, the definition of $\pi_{i+1}$ ensures that $$\prod_{D(y_0,t)\in \ext{q_{i+1}}} \pi_{i+1}(t)=\prod_{D(y',t)\in \ext{q_{i}}, y'\in V_p}\pi_{i}(t)\times \prod_{A_i^{C,t}\in At(q_i,x_0,A\sqsubseteq \exists P)} \monomial_{A_i^{C,t}}$$}

\noindent We have that (1) $\mn{terms}(q_{i+1})\cap V_{an}\subsetneq \mn{terms}(q_i)\cap V_{an}$ since $x_0$ has been removed, (2) $\pi_{i+1}$ is a match for $\ext{q_{i+1}}$ in $\Imc_\Omc$ that coincides with $\pi$ on \blue{$\mn{terms}(q_{i+1})$}, and \blue{we obtain (using the fact that for monomials, $\monomial\times\monomial=\monomial$):
\begin{align*}
(3)\ \prod_{P(\vec{t},t)\in \ext{q_{i+1}}}& \pi_{i+1}(t)\times \monomial_{i,i+1}
=\prod_{P(\vec{t},t)\in \ext{q_{i+1}}, \vec{t}\neq y_0} \pi_{i+1}(t)
\times \prod_{D(y_0,t)\in \ext{q_{i+1}}} \pi_{i+1}(t)\times \monomial_{i,i+1}
\\
=&\prod_{P(\vec{t},t)\in \ext{q_{i}}, x_0\notin\vec{t}} \pi_{i}(t)
\times \prod_{D(y',t)\in \ext{q_{i}}, y'\in V_p}\pi_{i}(t)\times \prod_{A_i^{C,t}\in At(q_i,x_0,A\sqsubseteq \exists P)} \monomial_{A_i^{C,t}}\times \monomial_{i,i+1}
\\
=&\prod_{P(\vec{t},t)\in \ext{q_{i}}, x_0\notin\vec{t}} \pi_{i}(t)
\times \prod_{A_i^{C,t}\in At(q_i,x_0,A\sqsubseteq \exists P)} \monomial_{A_i^{C,t}}\times \monomial_{i,i+1}\quad \text{since $y'\in V_p$ implies $x_0\neq y'$}
\\
=&\prod_{P(\vec{t},t)\in \ext{q_{i}}, x_0\notin\vec{t}} \pi_{i}(t)
\times  \monomial_A\times\prod_{A_i^{C,t}\in At(q_i,x_0,A\sqsubseteq \exists P)} \monomial_{A_i^{C,t}}\times \monomial_{i,i+1}
\\
=&\prod_{P(\vec{t},t)\in \ext{q_{i}}, x_0\notin\vec{t}} \pi_{i}(t)
\times  \monomial_A\times\prod_{A_i^{C,t}\in At(q_i,x_0,A\sqsubseteq \exists P)} \monomial_{A_i^{C,t}}
\times v\times\prod_{\blue{Q(y,x_0,t)\in \ext{q_{i}}}}\monomial_{P\sqsubseteq Q}^{\blue{t}}\times\prod_{\blue{C(x_0,t)\in \ext{q_{i}}}}\monomial_C^{\blue{t}}\times\nonomial_C^{\blue{t}}
\\
=&\prod_{P(\vec{t},t)\in \ext{q_{i}}, x_0\notin\vec{t}} \pi_{i}(t)
\times \prod_{\blue{Q(y,x_0,t)\in \ext{q_{i}}}}\monomial_A\times v\times\monomial_{P\sqsubseteq Q}^{\blue{t}}\times \prod_{\blue{C(x_0,t)\in \ext{q_{i}}}} \monomial_A\times v\times\monomial_C^{\blue{t}}\times\nonomial_C ^{\blue{t}}\times\prod_{i=1}^{p^{C,\blue{t}}}\monomial_{A^{C,\blue{t}}_i} 
\\
=&\prod_{P(\vec{t},t)\in \ext{q_{i}}, x_0\notin\vec{t}} \pi_{i}(t)\times\prod_{Q(y,x_0,t)\in \ext{q_{i}}}\ \pi_i(t)\times\prod_{C(x_0,t)\in \ext{q_{i}}}\ \pi_i(t)
\\
=&\prod_{P(\vec{t},t)\in \ext{q_i}} \pi_i(t).
\end{align*} 
}

\noindent The rewriting sequence ends when $q^*$ does not contain any variable from $V_{an}$ and we obtain a match $\pi^*$ for $\ext{q^*}$ in $\Imc_\Omc$ that coincides with $\pi$ on \blue{$\mn{terms}(q^*)$} and such that $\prod_{P(\vec{t},t)\in \ext{q^*}} \pi^*(t)\times \monomial^*=\prod_{P(\vec{t},t)\in \ext{q^*}} \pi^*(t)\times \monomial_0\times\prod_{i=0}^{k-1}\monomial_{i,i+1}=\prod_{P(\vec{t},t)\in \ext{q_k}} \pi_k(t)\times \prod_{i=0}^{k-1}\monomial_{i,i+1}= \prod_{P(\vec{t},t)\in \ext{q_0}} \pi(t)= \prod_{P(\vec{t},t)\in \ext{q}} \pi(t)$. 
\\
\noindent$\bullet$ By definition of $V_{an}$, $\pi^*$ maps every variable in $q^*$ to an individual from $\NI$. Hence, for every $P(\vec{t},t)\in \ext{q^*}$:
\begin{itemize}
\item $\Imc_\Omc\models (P(\pi^*(\vec{t})),\pi^*(t))$; 
\item so $\Omc\models (P(\pi^*(\vec{t})),\pi^*(t))$ by Theorem~\ref{thm:can-model-main};
\item \new{we have two cases to consider:
\begin{itemize}
\item if $\pi^*(\vec{t})\subseteq \individuals{\Omc}$, then $(P(\pi^*(\vec{t})),\pi^*(t))\in\mn{saturate}(\Omc)$ by Theorem~\ref{prop:completionalgorithmELHI}, 
\item otherwise, it must be the case that (i) $(P(\pi^*(\vec{t})),\pi^*(t))$ is of the form $(A(c),\nonomial)$ with $c\in\NI\setminus\individuals{\Omc}$ (since by construction of $\Imc_\Omc$, role interpretations in $\Imc_\Omc$ cannot contain any pair of individual names such that one of them is in $\NI\setminus\individuals{\Omc}$), and (ii) $\vec{t}$ is existentially quantified (since $q(\vec{a})$ does not contain any individual name from $\NI\setminus\individuals{\Omc}$ by assumption on $q(\vec{x})$ and $\vec{a}$), so by Theorem~\ref{prop:completionalgorithmELHI}, $(A(a_\top),\nonomial)\in\mn{saturate}(\Omc)$ and we can redefine $\pi^*$ by $\pi^*(\vec{t})=(a_\top)$, so that $(P(\pi^*(\vec{t})),\pi^*(t))=(A(a_\top),\nonomial)$ is in $\mn{saturate}(\Omc)$;
\end{itemize}}
\item it follows that $(P(\pi^*(\vec{t})),\pi^*(t))\in \Dmc$ by definition of $\Dmc$.
\end{itemize} 
It follows that $\pi^*$ is a match of $\ext{q^*}$ to $\Imc_\Dmc$ such that $\pi^*(\vec{x})=\vec{a}$. 
\\
\noindent$\bullet$ Let $\onomial=\prod_{P(\vec{t},t)\in \ext{q^*}}\pi^*(t)$. 
We obtain $\onomial\in \p{\Imc_\Dmc}{\ext{q^*}(\vec{a})}$ and $\nonomial=\prod_{P(\vec{t},t)\in \ext{q}} \pi(t)=\prod_{P(\vec{t},t)\in \ext{q^*}} \pi^*(t)\times \monomial^*=\onomial\times \monomial^*$. 

We have thus shown that for every $\nonomial\in \p{\Imc_{\Omc}}{\ext{q}(\vec{a})}$, there exists $(q^*,\monomial^*)\in \Rew(q,\Omc)$ and $\onomial\in \p{\Imc_{\Dmc}}{\ext{q^*}(\vec{a})}$ such that $\nonomial=\onomial\times \monomial^*$, \ie $\p{\Imc_{\Omc}}{\ext{q}(\vec{a})}\subseteq\{\monomial^*\times\onomial\mid (q^*,\monomial^*)\in \Rew(q,\Omc), \onomial\in \p{\Imc_{\Dmc}}{\ext{q^*}(\vec{a})}\}.$
\end{proof}

\paragraph{Proof of Corollaries~\ref{th:complexity-cq-elhi-entailment} and~\ref{th:complexity-cq-elhi} } Corollaries~\ref{th:complexity-cq-elhi-entailment} and~\ref{th:complexity-cq-elhi} are direct consequences of Theorem~\ref{th:CQ-answering-algorithm} and Lemma~\ref{lem:complexityevaluationqueryinsat} below. 

\begin{lemma}\label{lem:complexityevaluationqueryinsat}
For every $(q^*,\monomial^*)\in\Rew(q,\Omc)$, if $\Dmc$ is the set of annotated assertions in $\mn{saturate}(\Omc)$:
\begin{itemize}
\item for every $\onomial\in\monomials(\semiringVars)$, \new{one can decide whether $\onomial\in\p{\Imc_\Dmc}{\ext{q^*}}$ } 
in exponential time \wrt $|\Omc|$ and $|q|$, 
\item computing \new{$\Sigma_{\onomial\in \p{\Imc_\Dmc}{\ext{q^*}}}\onomial$} 
can be done in exponential time \wrt $|\Omc|$ and $|q|$.
\end{itemize}
\end{lemma}
\begin{proof}
Let $(q^*,\monomial^*)\in \Rew(q,\Omc)$. By Lemma~\ref{lem:size-rew} the size of $q^*$ (hence of $\ext{q^*}$) is polynomial \wrt $|\Omc|$ and $|q|$, and 
by Theorem~\ref{prop:completionalgorithmELHI}, $\mn{saturate}(\Omc)$ can be computed in exponential time, so $|\Dmc|$ is exponential \wrt $|\Omc|$. 
Deciding whether \new{$\onomial\in\p{\Imc_\Dmc}{\ext{q^*}}$} amounts to deciding whether there is a match $\pi$ for $\ext{q^*}$ in \new{$\Imc_\Dmc$} such that $\onomial=\prod_{P(\vec{t},t)\in \ext{q^*}}\pi(t)$ and 
computing \new{$\Sigma_{\onomial\in \p{\Imc_\Dmc}{\ext{q^*}}}\onomial$} amounts to find the matches for $\ext{q^*}$ in \new{$\Imc_\Dmc$}. 
\begin{itemize}
\item The number of potential matches is exponential \wrt $|\ext{q^*}|$ and polynomial \wrt $|\Dmc|$, hence exponential \wrt $|\Omc|$ and $|q|$. 
\item Each match can be checked in polynomial time \wrt $|\ext{q^*}|$ and $\Dmc$, hence in exponential time \wrt $|\Omc|$ and polynomial time \wrt $|q|$. 
\end{itemize}
Thus deciding whether \new{$\onomial\in\p{\Imc_\Dmc}{\ext{q^*}}$} and computing \new{$\Sigma_{\onomial\in \p{\Imc_\Dmc}{\ext{q^*}}}\onomial$} can both be done in exponential time \wrt $|\Omc|$ and $|q|$.
\end{proof}

\complexityBCQprovmonomial*
\begin{proof}
To decide whether $\Omc\models (q,\monomial)$, we use the following algorithm.
\begin{itemize}
\item As in the proof of Corollary~\ref{th:complexity:provmonomial}, check that $\Omc$ is satisfiable in polynomial time (if $\Omc$ is not satisfiable, then $\Omc\models (q,\monomial)$), then compute $\mn{saturate}^{|\monomial|}(\Omc)$ (with the completion rules modified for $\ELHIbotrestr$, \cf Figure~\ref{fig:completion-rules-restr}) in polynomial time  \wrt  $|\Omc|$ and exponential time \wrt $|\monomial|$. 

\item Guess:
\begin{enumerate}
\item a rewriting sequence $(q_0,\monomial_0)\rightarrow_\Omc (q_1,\monomial_1)\rightarrow_\Omc\dots\rightarrow_\Omc(q_k,\monomial_k)$ such that  $(q,1)=(q_0,\monomial_0)$ and $(q^*,\monomial^*)=(q_k,\monomial_k)$, 
\item for every $0\leq i\leq k-1$, a certificate that $(q_{i+1},\monomial_{i+1})$ is obtained from $(q_i,\monomial_i)$ by applying steps (S1) to (S6) of Definition~\ref{def:rewriting} but using $\mn{saturate}^{|\monomial|}(\Omc)$  instead of $\mn{saturate}(\Omc)$ and 
\item a match $\pi$ of $\ext{q^*}$ in \new{$\Imc_{\Dmc^{|\monomial|}}$ such that $\monomial=\monomial^*\times\prod_{P(\vec{t},t)\in\ext{q^*}}\pi(t)$, where $\Dmc^{|\monomial|}$ is the set of annotated assertions in $\mn{saturate}^{|\monomial|}(\Omc)$}. 
\end{enumerate}
\item Verify (1) to (3) in polynomial time \wrt $|\Omc|$ and $|q|$. 
\begin{enumerate}
\item Since each rewriting step removes a variable from the query, $k\leq |\mn{terms}(q)|$, and by Lemma~\ref{lem:size-rew}, for every $0\leq i\leq k$, the size of $(q_i,\monomial_i)$ is polynomial \wrt $|\Omc|$ and $|q|$.  
\item Since the sizes of $q_i$ and $q_{i+1}$ are polynomial \wrt $|\Omc|$ and $|q|$, a certificate that $(q_i,\monomial_i)\rightarrow_\Omc (q_{i+1},\monomial_{i+1})$ is also of polynomial size and can be checked in polynomial time \wrt $|\Omc|$ and $|q|$. Indeed, it consists of a variable $x_0$ of $q_i$, a GCI $(A\sqsubseteq \exists P,v)$ of $\Omc$, and at most one RI $(P\sqsubseteq Q,\monomial_{P\sqsubseteq Q})$ of $\mn{saturate}^{|\monomial|}(\Omc)$ per role atom of $q_i$ and \new{$5$} GCIs and RIs of $\mn{saturate}^{|\monomial|}(\Omc)$ per concept atom of $q_i$ \new{(since $p+p'\leq 2$ in the GCI of the form $(B_1\sqcap\dots\sqcap B_p\sqcap B'_1\sqcap\dots\sqcap B'_{p'}\sqsubseteq C,n_C)\in\mn{saturate}(\Omc)$)}. 
\item Since the size of $\ext{q^*}$ is polynomial \wrt $|\Omc|$ and $|q|$, so is $\pi$, and since $\mn{saturate}^{|\monomial|}(\Omc)$ is of polynomial size  \wrt  $|\Omc|$ and exponential time \wrt $|\monomial|$, we can check that $\pi$ is indeed a match in polynomial time \wrt $|\Omc|$. 
\end{enumerate}
\end{itemize}

Since $\mn{saturate}^{|\monomial|}(\Omc)\subseteq\mn{saturate}(\Omc)$, it is clear that this algorithm is sound. It remains to show that it is complete. 
Recall that every annotated axiom that belongs to $\mn{saturate}(\Omc)$ and has at most $|\monomial|$ variables belongs to  $\mn{saturate}^{|\monomial|}(\Omc)$ because every annotated axiom added by a rule application has at least as many variables as the premises of the rule.

Assume that $\Omc\models (q,\monomial)$. \new{Since $q$ does not contain any individual name from $\NI\setminus\individuals{\Omc}$,} by Theorem~\ref{th:CQ-answering-algorithm}, \new{there exist $(q^*,\monomial^*)\in\Rew(q,\Omc)$ and $\onomial\in\p{\Imc_\Dmc}{\ext{q^*}}$ such that} $\monomial=\monomial^*\times\onomial$, where $\Dmc$ is the set of annotated assertions in $\mn{saturate}(\Omc)$. Since $\monomial=\monomial^*\times\onomial$, it follows that $|\onomial|\leq |\monomial|$ and $|\monomial^*|\leq |\monomial|$.

Since \new{$\onomial\in\p{\Imc_\Dmc}{\ext{q^*}}$}, there is a match $\pi$ of $\ext{q^*}$ in \new{$\Imc_\Dmc$} such that $\onomial=\prod_{P(\vec{t},t)\in\ext{q^*}}\pi(t)$. Since $|\onomial|\leq |\monomial|$, $\pi$ is actually a match of $\ext{q^*}$ to $\new{\Imc_{\Dmc^{|\monomial|}}}$ such that $\monomial=\monomial^*\times\prod_{P(\vec{t},t)\in\ext{q^*}}\pi(t)$.

Since $(q^*,\monomial^*)\in\Rew(q,\Omc)$, there is a rewriting sequence  $(q_0,\monomial_0)\rightarrow_\Omc (q_1,\monomial_1)\rightarrow_\Omc\dots\rightarrow_\Omc(q_k,\monomial_k)$ such that $(q,1)=(q_0,\monomial_0)$, $(q^*,\monomial^*)=(q_k,\monomial_k)$ and for every $0\leq i\leq k-1$, a certificate that $(q_{i+1},\monomial_{i+1})$ is obtained from $(q_i,\monomial_i)$ by applying steps (S1) to (S6) of Definition~\ref{def:rewriting}. 
Since $|\monomial^*|\leq |\monomial|$, every annotated axiom used in the rewriting sequence from $(q,1)$ to $(q^*,\monomial^*)$ is in $\mn{saturate}^{|\monomial|}(\Omc)$. 
\end{proof}

\section{Proofs for Section~\ref{sec:posbool-lin}}

\propposbooljustif*
\begin{proof}
\new{Since $\posbool$, $\alpha$ and $\Omc$ satisfy the conditions of Lemma~\ref{cl:auxcancisubset}, it follows that for every $\Jmc\in\Just(\alpha)$, it holds that $\Omc\models (\alpha,\bigwedge_{\beta\in\Jmc}\lambda_\semiringVars(\beta))$. Moreover, since $\posbool$ is absorptive, $\bigvee_{\Jmc\in\Just(\alpha)}\ \bigwedge_{\beta\in\Jmc}\lambda_\semiringVars(\beta)=\bigvee_{\Mmc\subseteq\Omc', \Mmc\models\alpha}\bigwedge_{\beta\in\Mmc}\lambda_\semiringVars(\beta)$. Indeed, every $\Mmc\subseteq\Omc'$ such that $\Mmc\models\alpha$ contains a subset $\Jmc\subseteq\Mmc$ which is a justification. To obtain that $\Pmc(\alpha,\Omc)=
	\bigvee_{\Jmc\in\Just(\alpha)}\ \bigwedge_{\beta\in\Jmc}\lambda_\semiringVars(\beta)$, it remains to show that for every subset $V\subseteq\semiringVars$, if $\Omc\models (\alpha,\bigwedge_{v\in V} v)$, then there exists $\Mmc\subseteq\Omc'$ such that $\{\lambda_X(\beta)\mid \beta\in\Mmc\}=V$ and $\Mmc\cup\{\beta\mid (\beta,1)\in\Omc\}\models\alpha$. Let $V\subseteq\semiringVars$ such that $\Omc\models (\alpha,\bigwedge_{v\in V} v)$. }
	
\new{If $\alpha$ is an assertion, then by Theorem~\ref{thm:can-model-main}, $\Imc_\Omc\models (\alpha,\bigwedge_{v\in V} v)$ where $\Imc_\Omc$ is the canonical model of $\Omc$. It follows from the construction of the canonical model of $\Omc$ that if $\Nmc=\{(\beta,x)\mid (\beta,x)\in\Omc, x\in V\cup\{1\}\}$ is the subset of $\Omc$ annotated with variables from $V$ or $1$, then $\Imc_\Nmc\models (\alpha,\bigwedge_{v\in V} v)$. Hence $\Nmc\models(\alpha,\bigwedge_{v\in V} v)$ and by Theorem~\ref{th:sem-entailment}, $\Nmc'\models\alpha$, where $\Nmc'=\{\beta\mid (\beta,x)\in\Nmc\}$. We obtain $\Mmc$ as required by setting $\Mmc=\{\beta\mid \beta\in\Nmc',\lambda_\semiringVars(\beta)\neq 1\}$.}

\new{Consider now the case where $\alpha$ is a GCI of the form $C\sqsubseteq D$ with $C$ and $D$ basic concepts and $C$ satisfiable \wrt $\Omc$, and $\Omc$ does not contain any GCI with $\top$ as left-hand side. Since \posbool is \timesidem, by Theorem~\ref{th:red-concept}, $\Omc\models (C\sqsubseteq D,\bigwedge_{v\in V} v)$ iff $\Omc\cup\Tmc_D\cup\Amc_C\models (E(a_0),\bigwedge_{v\in V} v)$ where $\Tmc_D=\{(D\sqsubseteq E, 1)\}$,
 $\Amc_C=\{(C(a_0),1)\}$ if $C\in\NC$, 
and $\Amc_C=\{(P(a_0,b_0),1)\}$ if $C=\exists P$, 
with $a_0, b_0\in\NI\setminus\individuals{\Omc}$ and $E\in\NC\setminus\signature{\Omc}$. We can then obtain as above $\Mmc\subseteq\Omc'$ such that $\{\lambda_X(\beta)\mid \beta\in\Mmc\}=V$ and $\Mmc\cup\{\beta\mid (\beta,1)\in\Omc\cup\Tmc_D\cup\Amc_C\}\models E(a_0)$ and one can check (using the fact that $\Omc\cup\Tmc_D\cup\Amc_C$ is satisfiable because $C$ is satisfiable \wrt $\Omc$) that $\Mmc\cup\{\beta\mid (\beta,1)\in\Omc\}\models C\sqsubseteq D$.}

\new{The case where $\alpha$ is a positive RI whose left-hand side is satisfiable \wrt $\Omc$ can be handled in the same way, using Theorem~\ref{th:red-role}.}
\end{proof}

Recall that the completion algorithm to compute provenance in the \lin semiring initializes 
\Smc as in Section~\ref{sec:completion} 
and extends it by exhaustively applying the rules in Table~\ref{tab:completionRules}, where rule applications are modified to 
change \Smc into
\[
\Smc\Cup\{(\alpha,\monomial)\} :=
	\begin{cases}
		\Smc\cup\{(\alpha,\monomial)\}  \text{ if there is no $(\alpha,\nonomial)\in\Smc$} \\
		\Smc\setminus\{(\alpha,\nonomial)\}\cup\{(\alpha,\monomial\times\nonomial)\}  \text{ if $(\alpha,\nonomial)\in\Smc$.}
	\end{cases}
\] 

\completionAlgoRelevance*
\begin{proof}
First note that by construction of $\mn{linsat}(\Omc)$, there are no two annotated axioms $(\alpha,\monomial)$ and $(\alpha,\monomial')$ ($\monomial\not=\monomial'$) in $\mn{linsat}(\Omc)$ since $\mn{linsat}(\Omc)$ is initialized in a way that respects this condition and if $(\alpha,\monomial)\in\mn{linsat}(\Omc)$, a rule application modifies the annotation of $\alpha$ in $\mn{linsat}(\Omc)$ instead of adding $(\alpha,\monomial')$. Hence we need to show that 
\begin{enumerate}
\item if $\alpha$ is an assertion of the form $A(a)$ or $R(a,b)$ with $a,b\in\individuals{\Omc}$ such that $\Omc\models \alpha$, then 
$(\alpha,\Pmc(\alpha,\Omc))\in\mn{linsat}(\Omc)$; and 
\item \new{if $\alpha$ is an assertion of the form $A(c)$ with $c\in\NI\setminus\individuals{\Omc}$ such that $\Omc\models \alpha$, then 
$(A(a_\top),\Pmc(A(a_\top),\Omc))\in\mn{linsat}(\Omc)$ (since it is easy to see that $\Pmc(A(c),\Omc)=\Pmc(A(a_\top),\Omc)$, e.g. using the canonical model of $\Omc$).}
\end{enumerate}
By definition of the addition in \lin, for every assertion $\beta$, $\Pmc(\beta,\Omc)=\{v \mid \exists \monomial, \Omc\models  (\beta, \monomial\times v)\}$. 
Hence it is sufficient to check whether for every variable $v$, there exists $\monomial$ such that $\Omc\models  (\alpha, \monomial\times v)$ \new{(resp.\ $\Omc\models  (A(a_\top), \monomial\times v)$)} iff there exists $\nonomial$ such that $(\alpha, \nonomial\times v)\in\mn{linsat}(\Omc)$  \new{(resp.\ $(A(a_\top), \nonomial\times v)\in\mn{linsat}(\Omc)$)}. 

Let $\mn{saturate}(\Omc)$ be the set obtained by the completion algorithm of Section~\ref{sec:subcompletionalgo}. 
By Proposition~\ref{prop:why-lin-pos} (which states that $\Omc^{\lin}$ and $\Omc^{\why}$ entail the same annotated assertions) and Theorem~\ref{prop:completionalgorithmELHI}, $\Omc\models  (\alpha, \monomial\times v)$ iff $(\alpha, \monomial\times v)\in \mn{saturate}(\Omc)$ \new{(resp.\ $\Omc\models  (A(a_\top), \monomial\times v)$ iff $(A(a_\top, \monomial\times v)\in \mn{saturate}(\Omc)$)}. 
We thus only need to show that $\mn{linsat}(\Omc)=\{(\beta,\Pi_{(\beta,\onomial)\in\mn{saturate}(\Omc)} \onomial) \mid (\beta,\onomial)\in \mn{saturate}(\Omc)\}$ to get the theorem's result. 
\begin{itemize}
\item The initial set $\Smc_0$ is the same for the completion algorithm that computes $\mn{saturate}(\Omc)$ and the one that computes $\mn{linsat}(\Omc)$, and $\Smc_0=\{(\beta,\Pi_{(\beta,\onomial)\in\Smc_0} \onomial) \mid (\beta,\onomial)\in \Smc_0\}$ since there is no two annotated axioms $(\beta,\onomial)$ and $(\beta,\onomial')$ in $\Smc_0$. 

\item Moreover, we can apply the completion rules in parallel to obtain $\mn{saturate}(\Omc)$ and $\mn{linsat}(\Omc)$ while preserving $\Smc=\{(\beta,\Pi_{(\beta,\onomial)\in\Smc'} \onomial) \mid (\beta,\onomial)\in \Smc'\}$ where $\Smc$ is an intermediate step in the computation of $\mn{linsat}(\Omc)$ and $\Smc'$ is an intermediate step in the computation of $\mn{saturate}(\Omc)$. Indeed, let $\CR$ be a rule applicable to $\Smc'$ and assume that $\Smc=\{(\beta,\Pi_{(\beta,\onomial)\in\Smc'} \onomial) \mid (\beta,\onomial)\in \Smc'\}$ before the application of $\CR$. 

First apply $\CR$ to $\Smc'$ as follows: 
\begin{itemize}
\item Let $\alpha_1,\dots,\alpha_n$ be the axioms in the premises of $\CR$ and $\beta$ the axiom in its consequence. 

\item For every $1\leq i\leq n$, let $mon(\alpha_i)=\{\onomial \mid (\alpha_i,\onomial)\in \Smc'\}$ and apply $\CR$ to every $\{(\alpha_1,\onomial_1),\dots,(\alpha_n,\onomial_n)\}$ where $\onomial_i\in mon(\alpha_i)$ for each $i$. After these applications, $\{(\beta,\Pi_{i=1}^n \onomial_i) \mid \onomial_i\in mon(\alpha_i), 1\leq i\leq n\}\subseteq \Smc'$. 
\end{itemize}
Then try to apply $\CR$ to $\Smc$. 

\begin{itemize}
\item If $\CR$ is not applicable to $\Smc$, since $\{(\alpha_1,\Pi_{\onomial_1\in mon(\alpha_1)}\onomial_1)\dots,(\alpha_n,\Pi_{\onomial_n\in mon(\alpha_n)}\onomial_n)\}\subseteq\Smc$, this means that $\Smc$ already contains $\beta$ annotated with $\Pi_{i=1}^n \Pi_{\onomial_i\in mon(\alpha_i)}\onomial_i$. 

\item Otherwise, $\CR$ is applicable to $\Smc$ and either adds to $\Smc$ a new annotated axiom $(\beta,\Pi_{i=1}^n \Pi_{\onomial_i\in mon(\alpha_i)}\onomial_i)$, or updates $(\beta,\nonomial)\in\Smc$ by adding to $\nonomial$ the variables from $\Pi_{i=1}^n \Pi_{\onomial_i\in mon(\alpha_i)}\onomial_i$ that do not already belong to it.  
\end{itemize}
In both cases, after applying $\CR$ in this way, $\Smc=\{(\beta,\Pi_{(\beta,\onomial)\in\Smc'} \onomial) \mid (\beta,\onomial)\in \Smc'\}$.
\end{itemize}
Hence, 
$\mn{linsat}(\Omc)=\{(\beta,\Pi_{(\beta,\onomial)\in\mn{saturate}(\Omc)} \onomial) \mid (\beta,\onomial)\in \mn{saturate}(\Omc)\}.$
\end{proof}

\complexityRelevance*
\begin{proof}
We show that the saturated set $\mn{linsat}(\Omc)$ is computed in exponential time (resp.\  polynomial time if the ontology belongs to $\ELHIbotrestr$). \new{The theorem then follows from Theorem~\ref{completionAlgoRelevance} when $\alpha$ is an assertion. When $\alpha$ is a GCI or an RI, one first need to check whether its left-hand side is satisfiable \wrt $\Omc$ in exponential time (resp.\  polynomial time if the ontology belongs to $\ELHIbotrestr$). If it is unsatisfiable, $\Pmc(\alpha,\Omc)=1$. Otherwise, we use the polynomial reduction from annotated GCI or RI entailment to annotated assertion entailment in \timesidem semirings (Theorems~\ref{th:red-concept} and~\ref{th:red-role}) to obtain $\Omc'$ and an assertion $\beta$ such that $\Omc\models (\alpha,\monomial)$ iff $\Omc'\models (\beta,\monomial)$, so that $\Pmc(\alpha,\Omc)=\Pmc(\beta,\Omc')$.} 

Each rule application either adds a new axiom or adds a variable to the axiom annotation in $\Smc$. As the number of variables is linear in $\Omc$, the total number of rule applications is linearly bounded by the number of (non-annotated) axioms that may be added.  
It follows from the proofs of Theorems~\ref{prop:completionalgorithmELHI} and~\ref{prop:completionalgorithmELHIrestr} that this number of axioms is exponential in the size of $\Omc$ for the general $\ELHIbot$ case, and polynomial in the size of $\Omc$ in the case of $\ELHIbotrestr$. \blue{Moreover, it also follows from the proofs of Theorems~\ref{prop:completionalgorithmELHI} and~\ref{prop:completionalgorithmELHIrestr} that for each rule application, there are at most an exponential (polynomial in the case of $\ELHIbotrestr$) number of rule instantiation evaluations}.
\end{proof}

\Propusableaxioms*
\begin{proof}
\new{
We first consider the case where $\alpha$ is an assertion. By Theorem~\ref{thm:can-model-main}, for every $\monomial\in\lin$, $\Omc\models (\alpha,\monomial)$ iff $\Imc_{\Omc}\models (\alpha,\monomial)$, where $\Imc_{\Omc}$ is the canonical model of $\Omc$.  
It follows that $\gamma\in\Omc'$ is relevant to entail $\alpha$ \wrt $\Omc$ iff there exists a monomial $\monomial$ such that $\lambda_\semiringVars(\gamma)$ occurs in $\monomial$ and $\Imc_{\Omc}\models (\alpha,\monomial)$. 
Let $\Jmc_\gamma$ be the (classical) interpretation defined by $\Delta^{\Jmc_\gamma}=\Delta^{\Imc_\Omc}$, $a^{\Jmc_\gamma}=a^{\Imc_\Omc}$ for every $a\in\NI$, and for every $A\in\NC$ and $R\in\NR$:
\begin{itemize}
\item $A^{\Jmc_\gamma}=\{e\mid (e,\nonomial)\in A^{\Imc_\Omc}\}$,
\item ${A^\gamma}^{\Jmc_\gamma}=\{e\mid (e,\nonomial)\in A^{\Imc_\Omc}, \lambda_\semiringVars(\gamma) \text{ occurs in }\nonomial\}$, 
\item $R^{\Jmc_\gamma}=\{(d,e)\mid (d,e,\nonomial)\in R^{\Imc_\Omc}\}$, and 
\item ${R^\gamma}^{\Jmc_\gamma}=\{(d,e)\mid (d,e,\nonomial)\in R^{\Imc_\Omc}, \lambda_\semiringVars(\gamma) \text{ occurs in }\nonomial\}$. 
\end{itemize}
By construction of $\Jmc_\gamma$, $\gamma$ is relevant to entail $\alpha$ \wrt $\Omc$ iff $\Jmc_\gamma\models \alpha^\gamma$. 
We show that $\Jmc_\gamma$ is the canonical model of $\Omc'\cup{\Omc'}^\gamma\cup\{\gamma^\gamma\}$, so that $\gamma$ is relevant to entail $\alpha$ \wrt $\Omc$ 
iff $\Omc'\cup{\Omc'}^\gamma\cup\{\gamma^\gamma\}\models \alpha^\gamma$, \ie $\gamma$ is usable to derive $\alpha$ \wrt~$\Omc'$. 
}

\new{Recall that $\Imc_\Omc=\bigcup_{i\geq 0}\Imc_i$ and let $\Jmc_i$ be defined from $\Imc_i$ as $\Jmc_\gamma$ is defined from $\Imc_\Omc$. 
First, one can easily show (using structural induction in the case of GCI and the fact that GCIs and RIs in ${\Omc'}^\gamma=\bigcup_{\beta\in\Omc'}f^\gamma(\beta)$ are such that their left-hand sides all contain a concept or role name of the form $A^\gamma$ or $R^\gamma$) that:
\begin{itemize}
\item for every GCI $C\sqsubseteq D\in \Omc'$, ${C}^{\Jmc_i}=\{d\mid (d,\nonomial)\in C^{\Imc_i}\}$, 
\item for every GCI $C'\sqsubseteq D'\in {\Omc'}^\gamma$ \blue{such that $C'\sqsubseteq D'\in f^\gamma(C\sqsubseteq D)$}, ${C'}^{\Jmc_i}=\{d\mid (d,\nonomial)\in C^{\Imc_i}, \lambda_\semiringVars(\gamma) \text{ occurs in }\nonomial\}$, 
\item for every RI $P\sqsubseteq Q\in \Omc'$, ${P}^{\Jmc_i}=\{(d,e)\mid (d,e,\nonomial)\in P^{\Imc_i}\}$, 
\item for every RI $P'\sqsubseteq Q'\in {\Omc'}^\gamma$ \blue{such that $P'\sqsubseteq Q'\in f^\gamma(P\sqsubseteq Q)$}, ${P'}^{\Jmc_i}=\{(d,e)\mid (d,e,\nonomial)\in P^{\Imc_i}, \lambda_\semiringVars(\gamma) \text{ occurs in }\nonomial\}$. 
\end{itemize}
It is easy to see that $\Jmc_0$ is the interpretation that satisfies exactly the assertions in $\Omc'\cup{\Omc'}^\gamma\cup\{\gamma^\gamma\}$ (note that they are exactly the assertions in $\Omc'$, plus $\gamma^\gamma$ in the case where $\gamma$ is an assertion, and that in this case, $\Jmc_0\models\gamma^\gamma$ since $\Imc_0\models (\gamma,\lambda_\semiringVars(\gamma))$). Then each application of the chase rule from $\Imc_i$ to $\Imc_{i+1}$ that uses some $(e,\nonomial)\in C^{\Imc_i}$ and $(C\sqsubseteq D,v)\in\Omc$ and yields $(e,\nonomial\times v)\in D^{\Imc_{i+1}}$ corresponds to the following (possibly multiple) applications of the chase rule that build $\Jmc_{i+1}$ from $\Jmc_i$ \blue{using axioms from $\Omc'\cup{\Omc'}^\gamma\cup\{\gamma^\gamma\}$}:
\begin{itemize}
\item using $C\sqsubseteq D\in\Omc'$ and $e\in C^{\Jmc_i}=\{d\mid (d,\onomial)\in C^{\Imc_i}\}$, since it is indeed the case that $e\in D^{\Jmc_{i+1}}=\{d\mid (d,\onomial)\in D^{\Imc_{i+1}}\}$, as $(e,\nonomial\times v)\in D^{\Imc_{i+1}}$, 
\item in the case where $\gamma=C\sqsubseteq D$ (hence $\lambda_\semiringVars=v$), using $\gamma^\gamma=C\sqsubseteq D'$ (where $D'$ is obtained by replacing the unique predicate in $D$ by its adornment by $\gamma$) and $e\in C^{\Jmc_i}=\{d\mid (d,\onomial)\in C^{\Imc_i}\}$, since it is indeed the case that  $e\in {D'}^{\Jmc_{i+1}}=\{d\mid (d,\onomial)\in D^{\Imc_{i+1}}, \lambda_\semiringVars(\gamma) \text{ occurs in }\onomial\}$ since $(e,\nonomial\times v)\in D^{\Imc_{i+1}}$ and $\lambda_\semiringVars=v$, 
\item in the case where $\lambda_\semiringVars$ occurs in $\nonomial$, using every $C'\sqsubseteq D'\in f^\gamma(C\sqsubseteq D)$ and $e\in {C'}^{\Jmc_i}=\{d\mid (d,\onomial)\in C^{\Imc_i}, \lambda_\semiringVars(\gamma) \text{ occurs in }\onomial\}$, since it is indeed the case that $e\in {D'}^{\Jmc_{i+1}}=\{d\mid (d,\onomial)\in D^{\Imc_{i+1}}, \lambda_\semiringVars(\gamma) \text{ occurs in }\onomial\}$ since $(e,\nonomial\times v)\in D^{\Imc_{i+1}}$ and $\lambda_\semiringVars$ occurs in $\nonomial$,
\end{itemize} 
and similarly for application of the chase rule using some RI. Hence $\Jmc_{i+1}$ is obtained from $\Jmc_i$ by applying one or several chase rules using axioms of $\Omc'\cup{\Omc'}^\gamma\cup\{\gamma^\gamma\}$. Moreover, no rule is applied twice with the same axiom and tuple, and the rule application is fair since every rule applicable in $\Jmc_i$ corresponds to a rule applicable in $\Imc_i$ and will thus eventually be applied. We obtain that $\Jmc_\gamma$ is indeed the canonical model of $\Omc'\cup{\Omc'}^\gamma\cup\{\gamma^\gamma\}$.}

\new{We now consider the case where $\alpha$ is a GCI of the form $C\sqsubseteq D$ with $C$ and $D$ basic concepts such that $C$ is satisfiable \wrt $\Omc$ and $\Omc$ does not contain any GCI with $\top$ as left-hand side. Since \lin is \timesidem, by Theorem~\ref{th:red-concept}, for every $\monomial\in\lin$, $\Omc\models (C\sqsubseteq D,\monomial)$ iff $\Omc\cup\Tmc_D\cup\Amc_C\models (E(a_0), \monomial)$ where $\Tmc_D=\{(D\sqsubseteq E, 1)\}$,
 $\Amc_C=\{(C(a_0),1)\}$ if $C\in\NC$, 
and $\Amc_C=\{(P(a_0,b_0),1)\}$ if $C=\exists P$, 
with $a_0, b_0\in\NI\setminus\individuals{\Omc}$ and $E\in\NC\setminus\signature{\Omc}$. Hence, $\gamma\in\Omc'$ is relevant to entail $\alpha$ \wrt $\Omc$ iff there exists a monomial $\monomial$ such that $\lambda_\semiringVars(\gamma)$ occurs in $\monomial$ and $\Omc\cup\Tmc_D\cup\Amc_C\models (E(a_0), \monomial)$, i.e., iff $\lambda_\semiringVars(\gamma)$ is relevant to entail $E(a_0)$ \wrt $\Omc\cup\Tmc_D\cup\Amc_C$. It is easy to check that replacing the annotation $1$ by some fresh variables in $\Tmc_D\cup\Amc_C$, so that $\Omc\cup\Tmc_D\cup\Amc_C$ fulfills the conditions of the proposition (note that since $C$ is satisfiable \wrt $\Omc$, $\Omc\cup\Tmc_D\cup\Amc_C$ is satisfiable), has no impact on whether $\lambda_\semiringVars(\gamma)$ is relevant to entail $E(a_0)$. Hence, using the proposition in the assertion case, $\gamma$ is relevant to entail $E(a_0)$ \wrt $\Omc\cup\Tmc_D\cup\Amc_C$ iff it is usable to derive $E(a_0)$ \wrt the non-annotated version $\Omc'\cup\Tmc'_D\cup\Amc'_C$. Finally, one can check that $\gamma$ is usable to derive $E(a_0)$ \wrt $\Omc'\cup\Tmc'_D\cup\Amc'_C$ iff it is usable to derive $C\sqsubseteq D$ \wrt $\Omc'$.}

\new{The case where $\alpha$ is a positive RI whose left-hand side is satisfiable \wrt $\Omc$ can be handled in the same way, using Theorem~\ref{th:red-role}.}
\end{proof}


\section{Proof of Theorem \ref{theo:ELHIrestrcomplexity} (Complexity of $\ELHIbotrestr$)}\label{app:ELHIrestrcomplexity}
The proof of Theorem \ref{theo:ELHIrestrcomplexity} relies on the algorithms we develop for computing the \mbox{\why-}pro\-venance of assertions and queries \wrt $\ELHIbotrestr$ ontologies in Section \ref{sec:completionELHIrestr}.

\theoELHIrestrcomplexity*
\begin{proof}
Let $\Omc$ be a (non-annotated) ontology that belongs to $\ELHIbotrestr$, 
and let $\Omc'$ be the ontology obtained by applying exhaustively the saturation rules for $\ELHIbotrestr$ defined in Section \ref{sec:completionELHIrestr}, while ignoring the monomials part of the rules. 
\begin{itemize}
\item We have shown in the proof of Corollary~\ref{th:complexity:provmonomial} that $\Omc'$ can be built in polynomial time \wrt $|\Omc|$ and contains some $\bot(a)$ iff $\Omc$ is unsatisfiable, so that satisfiability in $\ELHIbotrestr$ is in \PTime. 

\item Assertion entailment can be reduced to satisfiability in polynomial time as follows: $\Omc\models A(a)$ iff $\Omc\cup\{A\sqcap B\sqsubseteq\bot, B(a)\}$ is unsatisfiable where $B\notin\signature{\Omc}$, and $\Omc\models R(a,b)$ iff $\Omc\cup\{R\sqcap S\sqsubseteq\bot, S(a,b)\}$ is unsatisfiable where $S\notin\signature{\Omc}$. Note that in both cases the modified ontologies belong to $\ELHIbotrestr$. 
Hence assertion entailment is in \PTime. 

\item RI entailment can be reduced to role assertion entailment in polynomial time as follows: $\Omc\models P_1\sqsubseteq P_2$ iff $\Omc\cup\{P_1(a_0,b_0)\}\models P_2(a_0,b_0)$ where $a_0,b_0$ are fresh individual names. Hence RI entailment is in \PTime.

\item GCI entailment can be reduced to concept  assertion entailment in polynomial time. 
\new{Given a GCI $C\sqsubseteq D$, } 
let $\Tmc_D=\emptyset$ if $D=\bot$ and $\Tmc_D$ be the set of GCIs obtained from the normalization of $D\sqsubseteq E$ otherwise, where $E\notin\signature{\Omc}$. Note that $\Tmc_D$ belongs to $\ELHIbotrestr$ since for any $\ELHIbot$ concept $D$, the normalization of $D\sqsubseteq E$ will only produce GCIs with concept names in the right-hand side. 
Let $\Amc_C=f(C,a_0)$ where $a_0$ is an individual name that does not occur in $\Omc$ and $f$ is the function inductively defined as 
follows, where all constants introduced are fresh: 
\begin{itemize}
\item $f(\top,a)=\emptyset$,
\item $f(A,a)=\{A(a)\}$ if $A\in\NC$, 
\item $f(\exists R.B,a)=\{R(a,b)\}\cup f(B,b)$, 
\item $f(\exists R^-.B,a)=\{R(b,a)\}\cup f(B,b)$, 
\item $f(B\sqcap B',a)=f(B,a)\cup f(B',a)$. 
\end{itemize}
We can show that 
$\Omc\models C\sqsubseteq D$ iff $\Omc\cup\Tmc_D\cup\Amc_C\models E(a_0)$. 
Hence GCI entailment is in \PTime.

\item For BCQ entailment, we adapt the rewriting algorithm (Definition \ref{def:rewriting}) so that it does not take into account monomials and uses $\Omc'$ instead of $\mn{saturate}(\Omc)$. The proof is then similar to the one for Theorem \ref{th:complexity:provBCQmonomial}: We obtain a \NP-upper bound by guessing $\Omc'$, a rewriting $q^*$, a rewriting sequence and a match for $q^*$ in the set $\Dmc$ of assertions in $\Omc'$. Since $\Omc'=\{\alpha\mid(\alpha,\monomial)\in \mn{saturate}(\Omc)\}$, it is easy to check that $q^*$ can be obtained by the adapted rewriting algorithm iff there exists $(q^*,\monomial^*)\in\Rew(q^*,\Omc^{\why})$, and that the BCQ entailment algorithm is correct. 
\end{itemize}
\NP-hardness of  BCQ entailment already holds for, e.g., DL-Lite \shortcite[Theorem~44]{DBLP:journals/jar/CalvaneseGLLR07}, \new{and \PTime-hardness of axiom entailment already holds for the language of GCIs of the form $C\sqsubseteq D$ with  $C:= A\mid A_1\sqcap A_2$ and $D:=A \mid \bot$~\cite{Cook2011-COOLFO-2}
	(see also \shortcite[Theorem 4.3, case 1]{DBLP:conf/kr/CalvaneseGLLR06},  for PTime-hardness of instance checking in the language of GCIs of the form $C\sqsubseteq D$, with $C:= A\mid A_1\sqcap A_2 \mid \exists R.A$ and $D:=A$, where $A,A_1,A_2\in \NC$ and $R\in\NR$}). 
\end{proof}

\end{document}